\documentclass[sigplan,10pt]{acmart}

\acmConference[PL'17]{ACM SIGPLAN Conference on Programming Languages}{January 01--03, 2017}{New York, NY, USA}
\acmYear{2017}
\acmISBN{} 
\acmDOI{} 
\startPage{1}

\setcopyright{none}

\usepackage{booktabs}   
\usepackage{subcaption} 

\usepackage{amsmath}
\usepackage[utf8]{inputenc} 
\usepackage{amssymb} 
\usepackage{amsmath}
\usepackage{listings} 
\usepackage{mathpartir} 
\usepackage{microtype} 
\usepackage{stmaryrd} 
\usepackage{graphicx} 
\usepackage{url}
\usepackage{rtyps}
\usepackage[inline]{enumitem}
\usepackage{array} 
\usepackage{hhline}
\usepackage{calrsfs}
\usepackage{prettyref}
\usepackage{amsthm}
\usepackage{multirow}
\usepackage{wrapfig} 
\usepackage{tikz}
\usepackage{pgf}
\usetikzlibrary{arrows,shapes,automata,positioning}
\usepackage{pifont} 

\makeatletter 
\def\arcr{\@arraycr}
\makeatother

\definecolor{notecolor}{rgb}{0.84,0.12,0.5}

\newcommand{\contcost}[0]{\ensuremath{\mathbb{T}}}

\newcommand{\invrwalk}[0]{\ensuremath{\text{inv}(s,s_{\min})}}

\newcommand{\toolname}[0]{\textsf{Absynth}}
\newcommand{\numexamp}[0]{\ensuremath{39}}

\newcommand{\word}[1]{\ensuremath{\mathop{\,\mathsf{#1}\,}}}

\newcommand{\code}[1]{\text{\sl #1}}

\newcommand{\tree}[2][XX]{%
   \ifthenelse{\equal{#1}{XX}}%
      {\ensuremath{T(#2)}}%
            {\ensuremath{T^{#1}(#2)}}}

\newcommand{\tr}[2][XX]{%
   \ifthenelse{\equal{#1}{XX}}%
      {\ensuremath{T(#2)}}%
            {\ensuremath{T^{#1}(#2)}}}

\newlength{\rWidth}

\newcommand{\funold}[3][XX]{%
   \ifthenelse{\equal{#1}{XX}}%
      {\ensuremath{#2 {\rightarrow} #3}}%
      { \settowidth{\rWidth}{\ensuremath{#1}}%
        \ensuremath{ #2 {\xrightarrow{\hspace{\rWidth}}\hspace{-0.84\rWidth}}\!\!\!^%
         {#1}
         \hspace{0.2\rWidth}\,\, #3}}}

\newcommand{\sem}[1]{\ensuremath{\llbracket #1 \rrbracket}}

\newcommand{\Rule}[4][]{\ensuremath{\inferrule*[right={(#2)},#1]{#3}{#4}}}
\newcommand{\RuleToplabel}[4][]{\ensuremath{\inferrule[(#2)]{#3}{#4}}}

\newcommand{\Qplus}[0]{\ensuremath{\mathbb{Q}_{\geq 0}}}
\newcommand{\Q}[0]{\ensuremath{\mathbb{Q}}}

\newcommand{\Rplus}[0]{\ensuremath{\mathbb{R}_{\geq 0}}}

\newcommand{\context}[0]{\ensuremath{\Gamma}}

\newcommand{\expt}[2]{\ensuremath{\mathbb{E}_{#1}(#2)}}

\newcommand{\prob}[1]{\ensuremath{\mathbb{P}(#1)}}
\newcommand{\dist}[1]{\ensuremath{\mu_{#1}}}
\newcommand{\proofcomment}[1]{\ensuremath{\dagger\,\,#1\,\,\dagger}}
\newcommand{\interval}[1]{\ensuremath{{|}{[}{#1}{]}{|}}}
\newcommand{\devcomment}[1]{\ensuremath{\textcolor{ACMDarkBlue}{\{#1\}}}}
\newcommand{\defineas}[0]{\ensuremath{:=}}
\newcommand{\progname}[1]{\ensuremath{\textsl{#1}}}
\newcommand{\State}[0]{\Sigma}                                        
\newcommand{\rul}[1]{\textcolor{ACMRed}{\sc #1}}                      
\newcommand{\transformer}[3]{\ensuremath{\mathop{\,\mathsf{ert}\,}[#1,#2](#3)}}   

\makeatletter
\newcommand{\dotminus}{\mathbin{\text{\@dotminus}}}
\newcommand{\@dotminus}{%
  \ooalign{\hidewidth\raise1ex\hbox{.}\hidewidth\cr$\m@th-$\cr}%
}
\makeatother

\newcommand{\jan}[1]{{\color{blue}Jan: #1}}
\newcommand{\chan}[1]{{\color{orange}Chan: #1}}

\begin{document}

\title[Bounded Expectations: Resource Analysis for Probabilistic Programs]{Bounded Expectations: Resource Analysis for Probabilistic Programs}         


\author{Van Chan Ngo}
\orcid{nnnn-nnnn-nnnn-nnnn}             
\affiliation{
  \institution{Carnegie Mellon University}           
  \streetaddress{5000 Forbes Avenue}
  \city{Pittsburgh}
  \state{PA}
  \country{USA}
}
\email{channgo@cmu.edu}          


\author{Quentin Carbonneaux}

\affiliation{
  \institution{Yale University}           
  \streetaddress{Street3b Address2b}
  \city{New Haven}
  \state{CT}
  \country{USA}
}
\email{quentin.carbonneaux@yale.edu}         

\author{Jan Hoffmann}
\orcid{nnnn-nnnn-nnnn-nnnn}             
\affiliation{
  \institution{Carnegie Mellon University}           
  \streetaddress{5000 Forbes Avenue}
  \city{Pittsburgh}
  \state{PA}
  \country{USA}
}

\email{jhoffmann@cmu.edu}         

\begin{abstract}
  %
  This paper presents a new static analysis for deriving upper bounds
  on the expected resource consumption of probabilistic programs. The
  analysis is fully automatic and derives symbolic bounds that are
  multivariate polynomials of the inputs.
  The new technique combines manual state-of-the-art reasoning
  techniques for probabilistic programs with an effective method
  for automatic resource-bound analysis of deterministic programs.
  It can be seen as both, an extension of automatic amortized resource
  analysis (AARA) to probabilistic programs and an automation of
  manual reasoning for probabilistic programs that is based on weakest
  preconditions.
  An advantage of the technique is that it combines the clarity and
  compositionality of a weakest-precondition calculus with the
  efficient automation of AARA.
  As a result, bound inference can be reduced to off-the-shelf LP
  solving in many cases and automatically-derived bounds can be
  interactively extended with standard program logics if the
  automation fails.
  Building on existing work, the soundness of the analysis is proved
  with respect to an operational semantics that is based on Markov decision
  processes.
  The effectiveness of the technique is demonstrated with a prototype
  implementation that is used to automatically analyze \numexamp{}
  challenging probabilistic programs and randomized algorithms.
  Experimental results indicate that the derived constant factors
  in the bounds are very precise and even optimal for many programs.
\end{abstract}

\ccsdesc[500]{Software and its engineering~General programming languages}
\ccsdesc[300]{Social and professional topics~History of programming languages}

\keywords{probabilistic programming, resource bound analysis, static analysis}  

\maketitle

\vspace{-1.5ex}
\section{Introduction}
\label{sec:intro}

Probabilistic programming~\cite{Kozen81,Pfeffer16} is an increasingly
popular technique for implementing and analyzing Bayesian Networks and
Markov Chains~\cite{GordonHNR14}, randomized
algorithms~\cite{Barthe2016}, cryptographic
constructions~\cite{Barthe2009}, and machine-learning
algorithms~\cite{Ghahramani2015}.
Compared with deterministic programs, reasoning about probabilistic
programs adds additional complexity and challenges. As a result, there
is a renewed interest in developing automatic and manual analysis and
verification techniques that help programmers to reason about their
probabilistic code. Examples of such developments are probabilistic
program logics~\cite{Kozen81,McIver04,CelikuM05,Kamin16,OlmedoKKM16},
automatic probabilistic invariant
generation~\cite{Chakarov13,ChatterjeeFG16}, abstract interpretation
for probabilistic programs~\cite{Monniaux01,Monniaux05,CousotM12},
symbolic inference~\cite{GehrMV16}, and probabilistic model
checking~\cite{Katoen16}.

One important property that is often part of the formal and informal
analysis of programs is \emph{resource bound analysis}: What is the
amount resources such as time, memory or energy that is required to
execute a program? Over the past decade, the programming language
community has developed numerous tools that can be used to
(semi-)automatically derive non-trivial symbolic
resource bounds for
imperative~\cite{BrockschmidtEFFG14,GulwaniMC09,SinnZV14,KincaidBBR2017}
and functional
programs~\cite{Jost03,CicekBGGH16,AvanziniLM15,DannerLR15}. 

Existing techniques for resource bound analysis can be applied
to derive bounds on the \emph{worst-case} resource consumption of
probabilistic programs. However, if the control-flow
is influenced by probabilistic choices then a worst-case analysis is
often not applicable because there is no such upper bound. 
%
Consider the example \code{trader($s_{min}$,s)} in Figure~\ref{fig:rwalk_intro}(a) that
implements a 1-dimensional random walk to model the fluctuations of a
stock price \code{s}. With probability $\frac{1}{4}$
the price increases by $1$
point and with probability $\frac{3}{4}$
the price decreases by $1$
point. After every price change, a trader performs an action
\code{trade(s)}---like buying 10 shares---that can depend on the
current price until the stock price falls to $\code{$s_{min}$}$.
In the worst case, \code{s} will be incremented in every
loop iteration and the loop will not terminate. 
However, the loop terminates with probability $1$, called \emph{almost surely} (a.s.) termination~\cite{Fioriti15}. 

While a.s. termination is a useful property, we might be also
interested in the distribution of the \emph{number of loop iterations}
or, considering a different resource, in \emph{the spending of our
  trader}. The distribution of the spending depends of course on the
implementation of the auxiliary function \code{trade()}. In general,
it is not straightforward to derive such distributions, even for
relatively simple programs. Consider for example the implementation of
\code{trade()} in Figure~\ref{fig:rwalk_intro}(b). It models a trader
that buys between $0$ and $10$ shares according to a uniform
distribution. It is not immediately
clear what the distribution of the cost is. 
%
\begin{figure}[th!]
\vspace{-2ex}
  \centering
\begin{minipage}[b]{0.28\textwidth}
\centering
\begin{lstlisting}[basicstyle=\scriptsize,escapeinside={@}{@}]
void trader(int @$s_{min}$@,int s) {
  assume (@$s_{min}$@>=0);
  while (@$s>s_{min}$@) {
    s=s+1 @$\oplus_{\frac{1}{4}}$@ s=s-1;
    trade(s); 
  }
}
\end{lstlisting}
\vspace{-2ex}
\footnotesize{$\progname{(a)}$}
\vspace{-1ex}
\end{minipage}%
\begin{minipage}[b]{0.20\textwidth}
\centering
\begin{lstlisting}[basicstyle=\scriptsize,escapeinside={@}{@}]
void trade(int s) {
  int nShares;
  nShares=unif(0,10);
  while (nShares>0) {
    nShares=nShares-1;
    cost=cost+s; 
  }
}
\end{lstlisting}
\vspace{-2ex}
  \footnotesize{$\progname{(b)}$}
\vspace{-1ex}
\end{minipage}%
\caption{(a) A 1-dimensional random walk that models the progression
  of a stock price that is incremented with probability $\frac{1}{4}$
  and decremented with probability $\frac{3}{4}$
  while it is greater than $s_{\min}$. (b) A trader that decides
  to buy between $0$
  and $10$
  shares by sampling from a uniform distribution. The program cost is modeled by the global variable
  \code{cost}. 
  }
  \label{fig:rwalk_intro}
\vspace{-3ex}
\end{figure}

In this article, we are introducing a new method for deriving bounds
on the expected resource consumption of probabilistic programs.  Our
technique derives symbolic polynomial bounds, is fully automatic, and
generates certificates that are derivations in a quantitative program
logic~\cite{Kamin16,OlmedoKKM16}. 
For example, given the function \code{trade($s_{min}$,s)}, our technique
automatically derives the bound $2 {\cdot} {\max}(0,  s{-}s_{min})$ on the expected number of loop
iterations. For the total spending of the trader, our technique automatically
derives the bound
$5 {\cdot} {\max}(0, s{-}s_{min}) + 10 {\cdot} {\max}(0, s{-}s_{min})
{\cdot} {\max}(0, s_{min}) + 5 {\cdot} {\max}(0, s{-}s_{min})^2$
on the expected value of the variable \code{cost} in less than $7$
seconds. Both bounds are tight in the sense that they precisely
describe the expected resource consumption. 

To the best of our knowledge, we present the first fully automatic
analysis for deriving symbolic bounds on the expected resource consumption of
probabilistic programs with probabilistic branchings and sampling from discrete
distributions. It is also one of the few techniques that can
automatically derive polynomial properties. 
Different resource metrics can be defined either by
using a resource-counter variable or by using \code{tick} commands. The analysis
is compositional, automatically tracks size changes, and derives whole
program bounds. Note that derived time bounds also imply \emph{positive
  termination}~\cite{Fioriti15}, that is, termination with bounded expected runtime.
Compositionality is particularly tricky for probabilistic programs
since the composition of two positively
terminating programs is not a positively
terminating in general.
While we focus on bounds on the expected cost, the analysis can also
be used to derive worst-case bounds. Moreover, we can adapt the
analysis (following~\cite{NgoDFH16}) to also derive lower bounds on
the expected resource usage.

Resource bound analysis and static analysis of probabilistic programs
have developed largely independently. The key insight of our work is
that there are close connections between (manual) quantitative
reasoning methods for probabilistic programs and automatic resource
analyses for deterministic programs. Our novel analysis combines
probabilistic quantitative reasoning using a weakest precondition (WP)
calculus~\cite{McIver04,CelikuM05,Kamin16,OlmedoKKM16} with an
automatic resource analysis method that is known as automatic
amortized resource analysis
(AARA)~\cite{Jost03,HoffmannW15,CarbonneauxHRZ13,CarbonneauxHZ15},
while preserving the best properties of both worlds. On the one hand,
we have the strength and conceptual simplicity of the WP calculus.  On
the other hand, we get template-based bound inference that can be
efficiently reduced to off-the-shelf LP solving.

We implemented our analysis in the tool \toolname. We currently
support imperative integer programs that features procedures, recursion, and loops. 
We have performed an evaluation with \numexamp{} probabilistic programs
and randomized algorithms that include examples from the literature
and new challenging benchmarks.  To determine the precision of the analysis,
we compared the statically-derived bounds with the experimentally-measured resource
usage for different inputs derived by sampling.
Our experiments show that we often derive a precise bound on
the expected resource usage and the automatic inference
usually takes only seconds.  

In summary, we make the following contributions.
\begin{itemize}
\item We describe the first automatic analysis that derives symbolic
  bounds on the expected resource usage of probabilistic programs.
\item We prove the soundness of the method by showing that a
  successful bound analysis produces derivations in a probabilistic WP
  calculus~\cite{OlmedoKKM16}.
\item We show the effectiveness of the technique with a prototype
  implementation and by successfully analyzing \numexamp{} examples
  from a new benchmark set and from previous work on probabilistic programs and randomized
  algorithms.
\end{itemize}
The advantages of our technique are compositionality, efficient
reduction of bound inference to LP solving, and
compatibility with manual bound derivation in the WP calculus.


\vspace{-1.0ex}
\section{Probabilistic programs}
\label{sec:programmodel}
%
In this section we first recall some essential concepts and notations from probability theory that are used in this paper. We then present the syntax of our imperative probabilistic programming language.
\vspace{-1.0ex}
\subsection{Essential notions and concepts}
The interested readers can find more detailed descriptions in reference textbooks~\cite{Ash00}.
%
\vspace{-.5ex}
\paragraph{Probability space.} 
Consider a random experiment, the set $\Omega$ of all possible
outcomes is called the \emph{sample space}. 
A \emph{probability space} is a triple
$(\Omega, \mathcal{F}, \mathbb{P})$, where $\mathcal{F}$ is a
$\sigma$-algebra of $\Omega$ and $\mathbb{P}$ is a probability measure
for $\mathcal{F}$, that is, a function from $\mathcal{F}$ to the
closed interval $[0,1]$ such that $\prob{\Omega} = 1$ and
$\prob{A \cup B} = \prob{A} + \prob{B}$ for all disjoint sets
$A, B \in \mathcal{F}$. The elements of $\mathcal{F}$ are called
\emph{events}. 
A function $f: \Omega \rightarrow \Omega'$ is \emph{measurable} w.r.t
$\mathcal{F}$ and $\mathcal{F'}$ if $f^{-1}(B) \in \mathcal{F}$ for
all $B \in \mathcal{F'}$.
\vspace{-.5ex}
\paragraph{Random variable.} 
A \emph{random variable} $X$ is a measurable function from a probability space $(\Omega, \mathcal{F}, \mathbb{P})$ to the real numbers, e.g., it is a function $X: \Omega \rightarrow \mathbb{R} \cup \{-\infty,+\infty\}$ such that for every Borel set $B \in \mathcal{B}$, $X^{-1}(B) \defineas \{\omega \in \Omega | X(\omega) \in B\} \in \mathcal{F}$. Then the function $\dist{X}(B) = \prob{X^{-1}(B)}$, called \emph{probability distribution}
. If $\dist{X}$ measures on a countable set of reals, or the range of $X$ is countable, then $X$ is called a \emph{discrete random variable}. If $\dist{X}$ gives zero measure to every singleton set, then $X$ is called a \emph{continuous random variable}. The distribution $\dist{X}$ is often characterized by the \emph{cumulative distribution function} defined by $F_{X}(x) = \prob{X \leq x} = \dist{X}((-\infty,x])$.
\vspace{-.5ex}
\paragraph{Expectation.} 
The \emph{expected value} of a discrete random variable $X$ is the weighted average $\expt{}{X} \defineas \sum_{x_i \in R_{X}}x_i\prob{X=x_i}$, where $R_{X}$ is the range of $X$. To emphasize the distribution $\dist{X}$, we often write $\expt{\dist{X}}{X}$ instead of $\expt{}{X}$. An important property of the expectation is \emph{linearity}. If $X$ and $X'$ are random variables and $\lambda, \mu \in \mathbb{R}$ then $Y = \lambda X + \mu X'$ is a random variable and $\expt{}{Y} = \lambda\expt{}{X} + \mu\expt{}{X'}$.
\vspace{-1.0ex}
\subsection{Syntax of probabilistic programs}
The probabilistic programming language we use is a simple imperative
integer language structured into expressions and
commands. 
The abstract syntax is given by the grammar in Figure \ref{fig:syntax}.
%
%
The command $\word{id} = e \word{bop} R$, where $R$ is a (discrete)
random variable whose probability distribution is $\dist{R}$ (written
as $R \sim \dist{R}$), is a \emph{random sampling} assignment. 
It first samples according to the distribution
$\dist{R}$ to obtain a sample value then evaluates the
expression in which $R$ is replaced by the sample value. Finally, the
evaluated value is assigned to $\word{id}$.
The command $c_1 \oplus_{p} c_2$ is a \emph{probabilistic branching}.
It executes the command $c_1$ with probability $p$, or the command $c_2$
with probability $(1-p)$.

The command $\word{if} \star \; c_1 \word{else} c_2$ is a
\emph{non-deterministic} choice between $c_1$ and $c_2$.  The command
$\word{call} P$ makes a (possibly recursive) call to the procedure
with identifier
$P \in \word{PID}$. In this article, we assume that the procedures only
manipulate the global program state. Thus, we avoid to use local
variables, arguments, and $\word{return}$ commands for passing
information across procedure calls. However, we support local
variables and $\word{return}$ commands in the
implementation. Arguments, can be easily simulated by using global
variables as registers.

%
\begin{wrapfigure}{R}{0.16\textwidth}
\vspace{-2ex}
\centering
\begin{minipage}[b]{0.15\textwidth}
\centering
\begin{lstlisting}[basicstyle=\footnotesize,xleftmargin=0.00cm,stepnumber=1,escapeinside={@}{@}]
while (h<=t) {
  t=t+1;
  h=h+unif(0,10)
  @$\oplus_{\frac{1}{2}}$@ skip;
  tick(1);
}
\end{lstlisting}
\end{minipage}%
\vspace{-3ex}
\caption{The $\progname{race}$ program.}
\vspace{-3ex}
\label{fig:race}
\end{wrapfigure}
We include the built-in primitive \word{assert \; e} that terminates
the program if the expression $e$ evaluates to 0 and does
nothing otherwise. 
The primitive \word{tick(q)}, where $q \in \Qplus$ 
is used to
model resource usage of commands and thus to define the
\emph{cost model}. As we have seen in the introduction, regular
variables can also be used to define a cost model.
%
\begin{figure}[t!]
\vspace{-1ex}
\small {
\begin{displaymath}
\begin{array}{ll}
e & := \word{id} \mid n \mid e_1 \word{bop} e_2 \\
c & := \word{skip} \mid \word{abort} \mid \word{assert} e \mid \word{tick}(q) \mid \word{id} = e \\
  & \mid \word{id} = e \word{bop} R \mid \word{if} e \; c_1 \word{else} c_2 \mid \word{if} \star \; c_1 \word{else} c_2 \\
  & \mid c_1 \oplus_{p} c_2 \mid c_1; c_2 \mid \word{while} e \; c \mid \word{call} P \\
\word{bop} & := {+} \mid {-} \mid {*} \mid \word{div} \mid \word{mod} \mid {==} \mid {<>} \mid {>} \mid {<} \mid {<=} \mid {>=} \mid {\&} \mid {\mid} \\
R & \sim \dist{R} \; (\text{probability distribution})    
\end{array}
\end{displaymath}
} 
\vspace{-2ex}
\caption{Abstract syntax of the probabilistic language.}
\vspace{-3ex}
\label{fig:syntax}
\end{figure}


To represent a complete program we use a pair $(c,\mathcal{D})$, where
$c \in C$ is the body of the $\word{main}$ procedure and
$\mathcal{D}: \word{PID} \rightarrow C$ is a map from procedure
identifiers to their bodies. 
A command with no procedure calls is called \emph{closed} command.
%
The program $\progname{race}$ in Figure
\ref{fig:race} shows a complete example of a probabilistic program that
is adapted from~\cite{Chakarov13}. It models a race between a hare
(variable $h$) and a tortoise (variable $t$). The race ends when the
hare is ahead of the tortoise (exit condition $h > t$). After each unit of time (expressed using the $\word{tick}$
command, line $6$), the tortoise goes one step forward (line $3$). With probability
$\frac{1}{2}$ (line $5$) the hare remains at its position and with the same
probability it advances a random number of steps between $0$ and $10$ (line $4$). 
\vspace{-1.5ex}
\section{Expected resource bound analysis}
\label{sec:overview}
In this section, we show informally how automatic amortized resource
analysis (AARA)~\cite{Jost03,HoffmannW15,CarbonneauxHRS17} can be
generalized to compute upper bounds on the expected resource
consumption of probabilistic programs.  We first illustrate with a
classic analysis of a one-dimensional random walk how developers
currently analyze the expected resource usage. We then recap AARA for imperative
programs~\cite{CarbonneauxHZ15,CarbonneauxHRS17}. Finally, we explain
the new concept of the \emph{expected potential method} that we
develop in this work and apply it to our random walk and more challenging examples. 
\vspace{-1.5ex}
\subsection{Manual analysis of a simple random walk}
\label{sec:manualana}

Consider the following implementation of a random walk.
\begin{lstlisting}[basicstyle=\scriptsize,escapeinside={@}{@}]
while (x > 0) { x = x - 1 @$\oplus_{\frac 3 4}$@ x = x + 1; tick(1); }
\end{lstlisting}
The traditional way to analyze this program is using recurrence relations.
Let $T(n)$ be the expected runtime of the program when $x$ is initially $n$. By expected runtime we mean the expected number of $\word{tick}(1)$ statements executed.
Then, for all $n \ge 1$, $T(n)$ satisfies the following equation
$$
  T(n) = 1 + \frac 3 4 T(n-1) + \frac 1 4 T(n+1)
$$
This is a recurrence relation of degree 3, but it can be solved more easily by defining $D(n)$ to be $T(n) - T(n-1)$ and rewriting the equation as
$
3 D(n) = 4 + D(n+1)
$. 
One systematic way to find solutions for this equation is to use the generating function $G(z) = \sum_{n \ge 1} D(n)z^n$.
Writing $D$ for $D(1)$, we get 
$$
  3G(z) = 4 \sum_{n\ge 1} z^n + \frac 1 z G(z) - D = \frac{4z}{1-z} + \frac 1 z G(z) - D
$$
And thus, using algebra and 
generating functions, we have 
\begin{equation*}
\label{eq:Gn}
  G(z) 
  = \sum_{n\ge 1} (2 - D\cdot3^{n-1} + 2\cdot 3^{n-1}) z^n
\end{equation*}

To finish the reasoning, we have to find the constant $D$ using a boundary condition. 
It is clear that $T(0) = 0$, because the loop is never entered when the program is started with $x = 0$.
To find $T(1)$, we observe that the program has to first reach the state $x = n-1$ to terminate with $x = n$.  
Additionally, because each coin flip is independent of the others, reaching $n-1$ from $n$ should take the same expected time as reaching 0 from 1.
Thus, $T(n) = n\cdot T(1) = n \cdot D$ and consequently $D(n) = T(n) - T(n-1) = D$ for $n \ge 1$.
Therefore, we have $D = 2$ and $T(n) = 2\cdot n$.

There are several reasons why this classic method make it hard to automate.
First, inferring the recurrence relations is not always simple.
For example, more complex iteration patterns complicate this process.
Additionally, it is difficult to formally prove a correspondence between the program and the recurrence relation.
Second, the method for solving recurrence relations is fragile. 
If the decrement is $x = x - 2$ instead of $x = x - 1$, then the above boundary condition does not work anymore. 
Moreover, recurrence relations become of higher degrees and more difficult to solve with the use of bigger constants and multiple variables. 
Finally, the classic method is not compositional.
When programs become larger, it does not provide a principled way to reason independently on smaller components.

\vspace{-1.5ex}
\subsection{The potential method} 
It has been shown in the past decade that the \emph{potential method} of amortized analysis
provides an interesting alternative to classic resource analysis with recurrence relations. 

Assume that a program $c$ executes with initial state $\sigma$ to 
final state $\sigma'$ and
consumes $n$ resource units as defined by the tick commands, denoted $(c, \sigma) \Downarrow_n \sigma'$.
The idea of amortized analysis is to define a \emph{potential function}
$\Phi : \Sigma \to \Qplus$ that maps program states to non-negative
numbers and to show that $\Phi(\sigma) \geq n$ for all $\sigma$ such
that $(c, \sigma) \Downarrow_n \sigma'$.

To reason compositionally, we have to take into account the state
resulting from a program execution. We thus use two potential
functions $\Phi$ and $\Phi'$ that specify the available potential
before and after the execution, respectively. The functions must
satisfy the constraint $\Phi(\sigma) \ge n + \Phi'(\sigma')$ for all
states $\sigma$ and $\sigma'$ such that
$(c, \sigma) \Downarrow_n \sigma'$.  Intuitively, $\Phi(\sigma)$ must
be sufficient to pay for the resource cost of the computation and
for the potential $\Phi'(\sigma')$ on the resulting state
$\sigma'$. Thus, if $(\sigma, c_1) \Downarrow_n \sigma'$ and
$(\sigma', c_2) \Downarrow_m \sigma''$, we have
$\Phi(\sigma) \ge n + \Phi'(\sigma')$ and
$\Phi'(\sigma') \ge m + \Phi''(\sigma'')$ and therefore
$\Phi(\sigma) \ge (n + m) + \Phi''(\sigma'')$.  Note that the initial
potential function $\Phi$ provides an upper bound on the resource
consumption of the whole program.  What we have observed is that, if
we define $\{\Phi\} {c}\{\Phi'\}$ to mean
$
\forall \sigma\, n\, \sigma'.\, (\sigma, c) \Downarrow_n \sigma' \implies \Phi(\sigma) \ge
n + \Phi'(\sigma')
$, then the following familiar inference rule is valid.
$$
\vspace{-.5ex}
\Rule{QSeq}
{\{\Phi\}{c_1}\{\Phi'\} \\ \{\Phi'\}{c_2}\{\Phi''\}}
{\{\Phi\} {c_1; c_2}  \{\Phi''\}}
\vspace{-.5ex}
$$
%
Other language constructs lead to rules for the potential
functions that look very similar to Hoare logic or effect system
rules.  These rules enable reasoning about resource usage
in a flexible and compositional way, which, as a side effect, produces
a certificate for the derived resource bound.


The derivation of a resource bound using potential functions is best
explained by example. In the following deterministic example, the
worst-case cost can be bounded by $|[x,y]| = \max(0,y{-}x)$. 
\begin{lstlisting}[basicstyle=\footnotesize,escapeinside={@}{@}]
while (x<y) { x=x+1; tick(1); }
\end{lstlisting}
To derive this bound, we start with the initial potential $\Phi_0 =
|[x,y]|$, which
we also use as the loop invariant. For the loop body we have (like in
Hoare logic) to derive a triple
$\{\Phi_0\}{x=x+1;\word{tick}(1)}\{\Phi_0\}$. We can only do so
if we utilize the fact that $x<y$ at the beginning of the loop body.
The reasoning then works as follows. We start with the potential
$|[x,y]|$ and the fact that $|[x,y]| > 0$ before the assignment. If
we denote the updated version of $x$ after the assignment by $x'$ then
the relation $|[x,y]| = |[x',y]| + 1$ between the potentials before and
after the assignment holds. This means that we have the
potential $|[x,y]| + 1$ before $\word{tick}(1)$ which consumes $1$ resource unit. 
We end up with potential $|[x,y]|$ after the loop body and have established the loop invariant.

\vspace{-1.5ex}
\subsection{The expected potential method}
Maybe surprisingly, the potential method of AARA can be adapted to
reason about upper bounds on the expected cost of probabilistic
programs.

Like in the deterministic case, we would like to work with triples of
the form $\{\Phi\}c\{\Phi'\}$ for potential functions
$\Phi,\Phi' : \Sigma \to \Qplus$ to ensure compositional reasoning.
However, in the case of probabilistic programs we need to take into
account the expected value of the potential function $\Phi'$ w.r.t 
the distribution of final states resulting from the
program execution. 
More precisely, the two functions must satisfy the constraint
%
$$\Phi(\sigma) \geq
\expt{\sem{c,\mathcal{D}}(\sigma)}{\text{cost}}
+ \expt{\sem{c,\mathcal{D}}(\sigma)}{\Phi'}$$
for all program states $\sigma$. Here we write
$\sem{c,\mathcal{D}}(\sigma)$ to denote the probability distribution
over the final states as specified by the program $(c,\mathcal{D})$
and initial state $\sigma$. Finally, $\expt{\sem{c,\mathcal{D}}(\sigma)}{\text{cost}}$ is the expected resource usage of executing $c$ from the initial state $\sigma$.  
The intuitive meaning is that the
potential $\Phi(\sigma)$ is sufficient to pay for the expected
resource consumption of the execution from $\sigma$ and the expected
potential with respect to the probability distribution over the final
states. Let $\Phi(\sigma) \geq
\expt{\sem{c,\mathcal{D}}(\sigma)}{\text{cost}}
+ \expt{\sem{c,\mathcal{D}}(\sigma)}{\Phi'}$
and
$\Phi'(\sigma'_i) \geq
\expt{\sem{c',\mathcal{D}}(\sigma'_i)}{\text{cost}}
+ \expt{\sem{c',\mathcal{D}}(\sigma'_i)}{\Phi''}$
for all sample states $\sigma'_i$ from $\sem{c,\mathcal{D}}(\sigma)$. 
Then we have for all states $\sigma$
$$
\Phi(\sigma) 
\geq \expt{\sem{c;c',\mathcal{D}}(\sigma)}{\text{cost}} + \expt{\sem{c;c',\mathcal{D}}(\sigma)}{\Phi''}
$$
Hence, the initial potential $\Phi$ gives an upper-bound on the
expected value of resource consumption of the sequence $c;c'$ like in
the sequential version of potential-based reasoning.  If we write
$\{\Phi\} c \{\Phi'\}$ to mean
$\Phi(\sigma) \geq
\expt{\sem{c,\mathcal{D}}(\sigma)}{\text{cost}}
+ \expt{\sem{c,\mathcal{D}}(\sigma)}{\Phi'}$
for all program states $\sigma$ then we have again the familiar rule \textsc{QSeq}
for compositional reasoning above. 

Note that expected and worst-case resource consumption are identical
for deterministic programs.  Therefore, the expected potential method
derives worst-case bounds for deterministic programs.

\begin{figure}[t!]
\vspace{-2ex}
\centering
\begin{minipage}[b]{0.25\textwidth}
\centering
\begin{lstlisting}[basicstyle=\scriptsize,escapeinside={@}{@}]
@$\devcomment{{q}{:=}{\frac{T}{{pK1}{-}({1}{-}{p}){K2}}}}$@
@$\devcomment{{.};{0}{+}{q}{\cdot}{\interval{{x}{,}{n}{+}{K1}{-}{1}}}}$@
while (x<n) {
  @$\devcomment{{x}{<}{n};{T}{-}{q}{\cdot}{K1}{+}{q}{\cdot}{\interval{{x}{,}{n}{+}{K1}{-}{1}}}}$@
  x=x+K1
  @$\devcomment{{x}{\leq}{n}{+}{K1}{-}{1};{T}{+}{q}{\cdot}{\interval{{x}{,}{n}{+}{K1}{-}{1}}}}$@
  @$\oplus_{p}$@
  @$\devcomment{{x}{<}{n};{T}{+}{q}{\cdot}{K2}{+}{q}{\cdot}{\interval{{x}{,}{n}{+}{K1}{-}{1}}}}$@
  x=x-K2;
  @$\devcomment{{x}{\leq}{n}{+}{K1}{-}{1};{T}{+}{q}{\cdot}{\interval{{x}{,}{n}{+}{K1}{-}{1}}}}$@
  tick(T);
  @$\devcomment{{x}{\leq}{n}{+}{K1}{-}{1};{0}{+}{q}{\cdot}{\interval{{x}{,}{n}{+}{K1}{-}{1}}}}$@
}
\end{lstlisting}
\vspace{-1.5ex}
\footnotesize{$\progname{rdwalk}$\\
${\frac{T}{{pK1}{-}({1}{-}{p}){K2}}}{\cdot}{\interval{{x}{,}{n}{+}{K1}{-}{1}}}$}
\end{minipage}%
\begin{minipage}[b]{0.25\textwidth}
\centering
\begin{lstlisting}[basicstyle=\scriptsize,escapeinside={@}{@}]
@$\devcomment{{.};{\frac{2}{3}}{\cdot}{\interval{{x}{,}{n}}}{+}{2}{\cdot}{\interval{{y}{,}{m}}}}$@
while (x+3<=n) {
  @$\devcomment{{x}{+}{3}{\leq}{n};{\frac{2}{3}}{\cdot}{\interval{{x}{,}{n}}}{+}{2}{\cdot}{\interval{{y}{,}{m}}}}$@
  if (y<m)
    @$\devcomment{{y}{<}{m};{\frac{2}{3}}{\cdot}{\interval{{x}{,}{n}}}{+}{2}{\cdot}{\interval{{y}{,}{m}}}}$@
    y=y+unif(0,1);
    @$\devcomment{{y}{\leq}{m};{2}{\cdot}{\frac{1}{2}}{+}{\frac{2}{3}}{\cdot}{\interval{{x}{,}{n}}}{+}{2}{\cdot}{\interval{{y}{,}{m}}}}$@
  else
    @$\devcomment{{x}{+}{3}{\leq}{n};{\frac{2}{3}}{\cdot}{\interval{{x}{,}{n}}}{+}{2}{\cdot}{\interval{{y}{,}{m}}}}$@
    x=x+unif(0,3);
    @$\devcomment{{x}{\leq}{n};{\frac{2}{3}}{\cdot}{\frac{3}{2}}{+}{\frac{2}{3}}{\cdot}{\interval{{x}{,}{n}}}{+}{2}{\cdot}{\interval{{y}{,}{m}}}}$@
  tick(1);
  @$\devcomment{{x}{\leq}{n};{0}{+}{\frac{2}{3}}{\cdot}{\interval{{x}{,}{n}}}{+}{2}{\cdot}{\interval{{y}{,}{m}}}}$@
}
\end{lstlisting}
\vspace{-1.5ex}
\footnotesize{$\progname{rdspeed}$\\
${\frac{2}{3}}{\cdot}{\interval{{x}{,}{n}}}{+}{2}{\cdot}{\interval{{y}{,}{m}}}$}
\end{minipage}
\vspace{-2ex}
\caption{Derivations of bounds for single loop programs.}
\label{fig:singleloops}
\vspace{-3ex}
\end{figure}

\begin{figure*}[t!]
\vspace{-2ex}
\centering
\begin{minipage}[b]{0.27\textwidth}
\centering
\begin{lstlisting}[basicstyle=\scriptsize,escapeinside={@}{@}]
@$\devcomment{{.};{0}{+}{\frac{33}{20}}{\cdot}{\interval{{y}{,}{z}}}{+}{\frac{3}{20}}{\cdot}{\interval{{0}{,}{y}}}}$@
while (z-y>2) {
  @$\devcomment{{y}{+}{2}{<}{z};{0}{+}{\frac{33}{20}}{\cdot}{\interval{{y}{,}{z}}}{+}{\frac{3}{20}}{\cdot}{\interval{{0}{,}{y}}}}$@
  y=y+bin(3,2,3);
  @$\devcomment{{y}{\leq}{z};{3}{+}{\frac{33}{20}}{\cdot}{\interval{{y}{,}{z}}}{+}{\frac{3}{20}}{\cdot}{\interval{{0}{,}{y}}}}$@
  tick(3);
  @$\devcomment{{y}{\leq}{z};{0}{+}{\frac{33}{20}}{\cdot}{\interval{{y}{,}{z}}}{+}{\frac{3}{20}}{\cdot}{\interval{{0}{,}{y}}}}$@
}
@$\devcomment{{y}{+}{2}{\geq}{z};{0}{+}{\frac{33}{20}}{\cdot}{\interval{{y}{,}{z}}}{+}{\frac{3}{20}}{\cdot}{\interval{{0}{,}{y}}}}$@
while (y>9) {
  @$\devcomment{{y}{>}{9};{\frac{-1}{2}}{+}{\frac{33}{20}}{\cdot}{\interval{{y}{,}{z}}}{+}{\frac{3}{20}}{\cdot}{\interval{{0}{,}{y}}}}$@
  y=y-10
  @$\devcomment{{y}{\geq}{0};{1}{+}{\frac{33}{20}}{\cdot}{\interval{{y}{,}{z}}}{+}{\frac{3}{20}}{\cdot}{\interval{{0}{,}{y}}}}$@
  @$\oplus_{\frac{2}{3}}$@
  @$\devcomment{{y}{>}{9};{1}{+}{\frac{33}{20}}{\cdot}{\interval{{y}{,}{z}}}{+}{\frac{3}{20}}{\cdot}{\interval{{0}{,}{y}}}}$@
  skip;
  @$\devcomment{{y}{\geq}{0};{1}{+}{\frac{33}{20}}{\cdot}{\interval{{y}{,}{z}}}{+}{\frac{3}{20}}{\cdot}{\interval{{0}{,}{y}}}}$@
  tick(1);
  @$\devcomment{{y}{\geq}{0};{0}{+}{\frac{33}{20}}{\cdot}{\interval{{y}{,}{z}}}{+}{\frac{3}{20}}{\cdot}{\interval{{0}{,}{y}}}}$@
}
\end{lstlisting}
\vspace{-1.5ex}
\footnotesize{$\progname{prseq}$\\
${(}{\frac{3}{2}}{+}{\frac{3}{20}}{)}{\cdot}{\interval{{y}{,}{z}}}{+}{\frac{3}{20}}{\cdot}{\interval{{0}{,}{y}}}$}
\end{minipage}%
\begin{minipage}[b]{0.40\textwidth}
\centering
\begin{lstlisting}[basicstyle=\scriptsize,escapeinside={@}{@}]
@$\devcomment{{q}{:=}{\interval{{0}{,}{s_{\min}}}} {+}{\interval{{s_{\min}}{,}{s}}}}$@
@$\devcomment{\invrwalk{}{:=}{10}{\interval{s_{\min},s}}{\cdot}{\interval{0,s_{\min}}}{+}{10\binom{\interval{s_{\min},s}}{2}}{+}{10}{\interval{s_{\min},s}}}$@
trader(int @$s_{\min}$@,int s) {
   @$\devcomment{{.};{0}{+}\invrwalk}$@
   assume (@$s_{\min}$@>=0);
   @$\devcomment{{s_{\min}}{\geq}{0};{0}{+}{\invrwalk}}$@
   while (@$s>s_{\min}$@) {
      @$\devcomment{{s_{\min}}{\geq}{0}{\land}{s}{>}{s_{\min}};{15}{+}{15}{\cdot}{q} {+} {\invrwalk{}}}$@
      s=s+1
      @$\devcomment{{s_{\min}}{\geq}{0}{\land}{s}{\geq}{s_{\min}};{5}{\cdot}{q} {+} {\invrwalk{}}}$@
      @$\oplus_{\frac{1}{4}}$@
      @$\devcomment{{s_{\min}}{\geq}{0}{\land}{s}{>}{s_{\min}};{-}{5}{-}{5}{\cdot}{q} {+} {\invrwalk{}}}$@
      s=s-1;
      @$\devcomment{{s_{\min}}{\geq}{0}{\land}{s}{\geq}{s_{\min}};{5}{\cdot}{q} {+} {\invrwalk{}}}$@
      trade(s); 
      @$\devcomment{{s_{\min}}{\geq}{0};{0}{+}\invrwalk}$@
   }
   @$\devcomment{{s_{\min}}{\geq}{0}{\land}{s}{\leq}{s_{\min}};{0}{+}\invrwalk}$@
}
\end{lstlisting}
\vspace{-1.5ex}
\footnotesize{$\progname{trader}$\\
$10 \interval{s_{\min},s} {\cdot} \interval{0,s_{\min}} + 10\binom{\interval{s_{\min},s}}{2} + 10\interval{s_{\min},s}$}
\end{minipage}%
\begin{minipage}[b]{0.32\textwidth}
\centering
\begin{lstlisting}[basicstyle=\scriptsize,escapeinside={@}{@}]
@$\devcomment{{.};{\frac{10}{9}}{\cdot}{(}{\frac{1000}{19}}{+}{9}{)}{\cdot}{\interval{{n}{,}{0}}}{+}{\frac{1}{19}}{\cdot}{\interval{{0}{,}{y}}}}$@
while (n<0) {
  @$\devcomment{{n}{<}{0};{\frac{-1}{9}}{\cdot}{(}{\frac{1000}{19}}{+}{9}{)}{+}{\frac{10}{9}}{\cdot}{(}{\frac{1000}{19}}{+}{9}{)}{\cdot}{\interval{{n}{,}{0}}}{+}{\frac{1}{19}}{\cdot}{\interval{{0}{,}{y}}}}$@
  n=n+1
  @$\devcomment{{n}{\leq}{0};{(}{\frac{1000}{19}}{+}{9}{)}{+}{\frac{10}{9}}{\cdot}{(}{\frac{1000}{19}}{+}{9}{)}{\cdot}{\interval{{n}{,}{0}}}{+}{\frac{1}{19}}{\cdot}{\interval{{0}{,}{y}}}}$@
  @$\oplus_{\frac{9}{10}}$@
  @$\devcomment{{n}{<}{0};{(}{\frac{1000}{19}}{+}{9}{)}{+}{\frac{10}{9}}{\cdot}{(}{\frac{1000}{19}}{+}{9}{)}{\cdot}{\interval{{n}{,}{0}}}{+}{\frac{1}{19}}{\cdot}{\interval{{0}{,}{y}}}}$@
  skip;
  @$\devcomment{{n}{\leq}{0};{(}{\frac{1000}{19}}{+}{9}{)}{+}{\frac{10}{9}}{\cdot}{(}{\frac{1000}{19}}{+}{9}{)}{\cdot}{\interval{{n}{,}{0}}}{+}{\frac{1}{19}}{\cdot}{\interval{{0}{,}{y}}}}$@
  y=y+1000;
  @$\devcomment{{n}{\leq}{0};{9}{+}{\frac{10}{9}}{\cdot}{(}{\frac{1000}{19}}{+}{9}{)}{\cdot}{\interval{{n}{,}{0}}}{+}{\frac{1}{19}}{\cdot}{\interval{{0}{,}{y}}}}$@
  while (y>=100 && @$\star$@) {
    @$\devcomment{{y}{\geq}{100};{\frac{-5}{9}}{+}{9}{+}{\frac{10}{9}}{\cdot}{(}{\frac{1000}{19}}{+}{9}{)}{\cdot}{\interval{{n}{,}{0}}}{+}{\frac{1}{19}}{\cdot}{\interval{{0}{,}{y}}}}$@
    y=y-100
    @$\devcomment{{y}{\geq}{0};{9}{+}{5}{+}{\frac{10}{9}}{\cdot}{(}{\frac{1000}{19}}{+}{9}{)}{\cdot}{\interval{{n}{,}{0}}}{+}{\frac{1}{19}}{\cdot}{\interval{{0}{,}{y}}}}$@
    @$\oplus_{\frac{1}{2}}$@
    @$\devcomment{{y}{\geq}{100};{\frac{5}{9}}{+}{9}{+}{\frac{10}{9}}{\cdot}{(}{\frac{1000}{19}}{+}{9}{)}{\cdot}{\interval{{n}{,}{0}}}{+}{\frac{1}{19}}{\cdot}{\interval{{0}{,}{y}}}}$@
    y=y-90;
    @$\devcomment{{y}{\geq}{0};{9}{+}{5}{+}{\frac{10}{9}}{\cdot}{(}{\frac{1000}{19}}{+}{9}{)}{\cdot}{\interval{{n}{,}{0}}}{+}{\frac{1}{19}}{\cdot}{\interval{{0}{,}{y}}}}$@
    tick(5);
    @$\devcomment{{y}{\geq}{0};{9}{+}{\frac{10}{9}}{\cdot}{(}{\frac{1000}{19}}{+}{9}{)}{\cdot}{\interval{{n}{,}{0}}}{+}{\frac{1}{19}}{\cdot}{\interval{{0}{,}{y}}}}$@
  }
  tick(9);
  @$\devcomment{{y}{<}{100};{\frac{10}{9}}{\cdot}{(}{\frac{1000}{19}}{+}{9}{)}{\cdot}{\interval{{n}{,}{0}}}{+}{\frac{1}{19}}{\cdot}{\interval{{0}{,}{y}}}}$@
}
\end{lstlisting}

\vspace{-1.5ex}
\footnotesize{$\progname{prnes}$\\
${\frac{10}{9}}{\cdot}{(}{\frac{1000}{19}}{+}{9}{)}{\cdot}{\interval{{n}{,}{0}}}{+}{\frac{1}{19}}{\cdot}{\interval{{0}{,}{y}}}$}
\end{minipage}%
\vspace{-2ex}
\caption{Derivations of bounds on the expected value of ticks for
  probabilistic programs. 
  Example $\progname{prseq}$ contains a sequential loop so that the iterations of the second loop
  depend on the the first one and $\progname{prnes}$ contains an interacting nested loop. 
  In the derivation of $\progname{trader}$, we derive the non-linear bound $\invrwalk{}$ on the 
  global variable \code{cost} from the stock price example.}
\label{fig:challengingexamples}
\vspace{-3ex}
\end{figure*}

\vspace{-.5ex}
\paragraph{Analyzing a random walk.} 
Using the expected potential method simplifies the reasoning significantly and, as we show in this
article, can be automated using a template-based approach and LP
solving.

Consider the simple random walk again whose expected
resource usage is $\Phi = 2|[0,x]|$. This is the potential
that we have available before the loop and it will also serve as a
loop invariant. We have to prove that the
potential right after the probabilistic branching should be $2|[0,x]|+1$,
to pay for the cost of the tick statement and to restore the loop
invariant.  What remains to justify is how the probabilistic
branching turns the potential $2|[0,x]|$ into $2|[0,x]|+1$.
To do so, we reason backwards independently on the two
branches.  For each branch, what is the initial potential required
to ensure an exit potential of $2|[0,x]| +1$?
\emph{(i)} The assignment $x = x - 1$ needs initial potential
  $\Phi_1(x) = 2|[0,x]| - 1$. Indeed if we
  write $x'$ for the value of $x$ after the assignment then
  $2|[0,x]| - 1 = 2|[0,x'+1]| - 1 = 2|[0,x']| + 1$.
\emph{(ii)} Similarly, the second branch needs the initial potential $\Phi_2(x) = 2|[0,x]| + 3$.
Intuitively, since we enter the first and second branches with probabilities 
$\frac{3}{4}$ and $\frac{1}{4}$, respectively, thus the initial potential
for the probabilistic branching should be the weighted sum
$\frac{3}{4}\Phi_1(x) + \frac{1}{4}\Phi_2(x) = 2|[0,x]|$.
This reasoning would restore the loop invariant and prove
the desired bound.

The beauty of the potential method is that this reasoning is in fact sound!
But note that this is not a trivial result. For instance,
replacing the probabilistic branching by its ``average''
action $x = x - \frac 1 2$ would be unsound in general.
\vspace{-.5ex}
\paragraph{General random walk.}
Now consider the generalized version $\progname{rdwalk}$ in Figure~\ref{fig:singleloops}. It 
simulates a general \emph{one-dimensional random
  walk}~\cite{Grimmett92} with arbitrary positive constants
$K_1, K_2, T$, and $p$. The expected number of loop iterations, and
thus the expected cost modeled by the command $\word{tick}(T)$, is
bounded iff $(\star)$ ${pK_1}{-}{(}{1}{-}{p}{)}{K_2} {>} {0}$. In
this case, the expected distance for ``forward walking'' is bigger
than the expected distance for backward walking.

For the annotations in the figure, we use semicolons to separate 
the logical assertions from the potential functions.
The analysis for this example is very similar to the simple one.
It is only valid when the condition $(\star)$ is satisfied
since the initial potential function would be negative otherwise.
In the implementation, the automatic analysis reports that no bound
can be found if the program does not satisfy this requirement.
Note that the classic method would be a lot more complex in
this more general case.  Indeed, the degree of the recurrence
relation to solve would be ${K_1}{+}{K_2}{+}{1}$ and if ${K_1}{>}{1}$ the boundary
condition argument we gave in Section~\ref{sec:manualana}
would not be valid anymore.
\vspace{-1.5ex}

\subsection{Bound derivations for challenging examples}

We show how the expected potential method can derive
polynomial bounds for challenging probabilistic programs with single,
nested, and sequential loops, as well as procedure calls 
(more examples are given in Appendix~\ref{app:morederivation}). 
All the
presented bounds are derived automatically by our tool
\toolname{} whose implementation and inference algorithm are discussed in
Sections~\ref{sec:constraint}
and~\ref{sec:experiment}. 
%
\vspace{-.5ex}
\paragraph{Single loops.} 
Examples $\progname{rdwalk}$ and $\progname{rdspeed}$
in Figure~\ref{fig:singleloops} show that our expected potential
method can handle \emph{tricky iteration patterns}.
Example $\progname{rdwalk}$ has already been discussed.
Example $\progname{rdspeed}$ is randomized version of the one from papers on worst-case bound
analysis~\cite{GulwaniMC09,CarbonneauxHZ15}. The iteration first increases $y$ by $0$
or $1$ until it reaches $m$ randomly according to the uniform distribution. Then $x$ is increased by $k \in [0,3]$
which is sampled uniformly. \toolname{} derives the tight bound
${\frac{2}{3}}{\cdot}{\interval{{x}{,}{n}}}{+}{2}{\cdot}{\interval{{y}{,}{m}}}$.
The derivation is similar to the derivation of the bound for
$\progname{rdwalk}$.
To reason about the sampling construct, we consider the effect of all
possible samples on the potential function. With the same reasoning for probabilistic 
branching, we then 
compute the weighted sum on the resulting initial potentials where the
weights are assigned following the distribution. 
In some cases (like in this one), it is sound to just replace the distribution with
the expected outcome. However, this does not work in general since the
resource consumption could depend on the sample in a non-uniform way.

\vspace{-.5ex}
\paragraph{Nested and sequential loops.} 
The examples in Figure~\ref{fig:challengingexamples} also show how the expected
potential method can effectively handle \emph{interacting nested} and
\emph{sequential loops}. Example $\progname{prnes}$
contains a nested loop, in which for each iteration of the outer loop,
the variable $n$
is increased by $1$
or remains unchanged with respective probabilities $\frac{9}{10}$
and $\frac{1}{10}$.
In each iteration, the variable $y$
is incremented by $1000$.
We use the condition \lstinline[mathescape]{y>=100 && $\star$}
to express that we can non-deterministically exit the inner loop. Here,
the variable $y$
is randomly decremented by $100$
or $90$
with probability $\frac{1}{2}$.
The analysis discovers that only the expected runtime
${\frac{5}{95}}{\cdot}{\interval{{0}{,}{y}}} =
{\frac{1}{19}}{\cdot}{\interval{{0}{,}{y}}}$ of the first execution of
the inner loop depends on the initial value of $y$.
For the other executions of the loop, the expected number of ticks is
${\frac{1}{19}}{\cdot}{1000}$.
The derivation infers a tight bound in which
${\frac{10}{9}}{\cdot}{\interval{{n}{,}{0}}}$
is the expected number of iterations of the outer loop. The example
$\progname{prseq}$
shows the capability of the expected potential method to take into
account the interactions between sequential loops by deriving the
expected value of the size changes of the variables. In the first
loop, the variable $y$
is incremented by sampling from a binomial distribution with
parameters $n = 3$
and $p = \frac{2}{3}$.
In the second loop, $y$
is decreased by $10$
with probability $\frac{2}{3}$
or left unchanged otherwise. It accurately derives
the expected value of the size change of $y$
by transferring the potential $\interval{{y}{,}{z}}$
to $\interval{{0}{,}{y}}$. 
\vspace{-.5ex}
\paragraph{Non-linear bounds.} 
Example $\progname{trader}$ has non-linear bound
that demonstrates that our expected potential method can handle
programs with nested loops and procedure calls which often have a
super-linear expected resource consumption. 
In the derivation $\invrwalk{}$
acts as a loop invariant. We assume that we already established the bound
$5 \interval{0,s} = 5(\interval{0,s_{\min}} +\interval{s_{\min},s})$
for the procedure \code{trade(s)}. 
The crucial point
for understanding the derivation is the reasoning about the
probabilistic branching. First, observe that the weighted sum of the
two pre-potentials of the branches is equal to $\invrwalk{}$
if the weights are $\frac{1}{4}$
and $\frac{3}{4}$.
Second, verify that the post potential is equal to
the potential in the respective pre-potential if we substitute $s+1$
(or $s-1$) for $s$.
\vspace{-1.5ex}
\subsection{Limitations} 
Deriving symbolic resource bounds is an
undecidable problem, thus our technique does not work if loops and
recursive functions have complex control flow. Furthermore,
our implementation currently only derives polynomial bounds and 
supports sampling from discrete distributions with a finite domain.

In the technical development,
we do not cover local variables and function arguments. However, the
extension of our analysis is straightforward and local variables are implemented in
\toolname{}. While we conjecture that the technique also works for
non-monotone resources like memory that can become available during
the evaluation, our current meta-theory only covers monotone resources
like time.


\vspace{-1.5ex}
\section{Derivation system for probabilistic quantitative analysis}
\label{sec:inferencerule}
In this section, we describe the inference system used by our analysis .
It is presented like a program logic and enables compositional reasoning. As we explain in Section~\ref{sec:constraint}, the inference of a derivation can be reduced to LP solving.
\vspace{-1.5ex}
\subsection{Potential functions}
\label{sec:potfuncs}
The main idea to automate resource analysis using the potential method is to fix the shape of the potential functions.
More formally, potential functions are taken to be linear combinations of more elementary \emph{base functions}.
Finding the suitable base functions for a given program is discussed in Section~\ref{sec:experiment}. 
For now we assume a given list of $N$ base functions.
For convenience, they are represented as a vector $B = (b_1, \dots, b_N)$, 
where each $b_i : \State \to \Q$ maps program states to rational numbers. 
Building on base functions, a potential function is defined by $N$ coefficients $q_1, \dots, q_N \in \mathbb Q$ as
$
  \Phi(\sigma) = \sum_{i = 1}^N q_i \cdot b_i(\sigma)
$. 
For presentation purposes the coefficients are stored as a vector $Q = (q_1, \dots, q_N)$, called \emph{potential annotation}.
Each potential annotation corresponds to a potential function $\Phi_Q$ that we can concisely express as the dot product $\langle Q \cdot B \rangle$.

Note that potential annotations form a vector space.
Additionally, using the bi-linearity of the dot product, operations on the vector space of annotations correspond directly to operations on potential functions, that is 
$
  \Phi_{\lambda Q + \mu Q'} = \langle \lambda Q + \mu Q' \cdot B \rangle =
  \lambda \langle Q \cdot B \rangle + \mu \langle Q' \cdot B \rangle = \lambda \Phi_{Q} + \mu \Phi_{Q'}
$. 
In the following, we assume that the constant function $\mathbf{1}$ defined by $\lambda\sigma.1$ is in the list of base functions~$B$.
This way, the constant potential function $\lambda\sigma.k$ can be represented with a potential annotation where the coefficient of $\mathbf{1}$ is $k$ and all other coefficients are $0$.
\vspace{-1.5ex}
\subsection{Judgements}
The main judgement of our inference system defines the validity of a triple $\vdash \{\context, Q\} c \{\context', Q'\}$.
In the triple, $c$ is a command, $\{\context, Q\}$ is the precondition, and $\{\context', Q'\}$ is the postcondition. $\context$ is a \emph{logical context} and $Q$ is a potential annotation.
The logical context $\context \in \mathcal P(\State)$ is a predicate on program states inferred by our implementation using abstract interpretation. It describes a set of permitted states at a given program point.

Leaving the logical contexts aside---they have the same semantics as in Hoare logic---the meaning of a triple $\{\cdot, Q\} c \{\cdot, Q'\}$ is that, for any continuation command $c'$ that has its expected cost bounded by $\Phi_{Q'}$, the command $c;c'$ has its expected cost bounded by $\Phi_Q$.
When looking for the expected cost of the command $c$, one can simply use \word{skip} as the command $c'$ and derive a triple where $\Phi_{Q'} = \mathbf{0}$. In that case, $\Phi_Q$ is a bound on the expected cost of the command $c$.
To handle procedure calls, the judgement for a triple uses a \emph{specification context} $\Delta$.
This context assigns \emph{specifications} to the procedures of the program and permits a compositional analysis that also handles recursive procedures.
A specification is a valid pair of pre- and post-conditions for the procedure body, denoted $\Delta \vdash \{\context;Q\}\,\mathcal D (P)\,\{\context';Q'\}$. 
The judgement $\vdash \Delta$ defined by the rule \rul{ValidCtx} means that all the procedure specifications in the context $\Delta$ are valid, that is, the specifications are correct pre- and post-conditions for the procedure bodies.
Note that a context $\Delta$ can contain multiple specifications for the same procedure.
This enables a context-sensitive analysis.
%
\vspace{-.5ex}
\paragraph{Notations and conventions.}
For a program state $\sigma \in \Sigma$ (e.g., a map
from variable identifiers to integers), we write $\sem{e}_{\sigma}$ to denote the value of the expression $e$
in $\sigma$ and $\sigma[v/x]$ for the program state $\sigma$ extended with the mapping of $x$ to $v$. 
For a probability distribution $\dist{}$, we use $\sem{\dist{} : v}$ to indicate the probability that $\dist{}$ takes value $v$. We use $\Sigma$ to
denote the set of program states. 
The entailment relation on logical contexts $\context \models \context'$ means that $\context$ is stronger than $\context'$.
We write $\sigma \models \context$ when $\sigma \in \context$. For a proposition $p$, we write $\context \models p$ to mean that any state $\sigma$ such that $\sigma \models \context$ satisfies $p$.
For an expression $e$ and a variable $x$, $\context \land e$ stands for the logical context $\{ \sigma ~|~ \sigma \models \context \land \sem e_\sigma \neq 0 \}$ and $\context[e/x]$ stands for the logical context $\{ \sigma ~|~ \sigma[e/x] \models \context \}$.

For potential annotations $Q$, $Q'$ and $\diamond \in \{<,=,\ldots\}$,
the relation $Q \diamond Q'$ means that their components are constrained point-wise, that is, $\bigwedge_i q_i \diamond q'_i$.
Additionally, we write $Q \pm c$ where $c \in \mathbb Q$ to denote the annotation $Q'$ obtained from $Q$ by setting the coefficient $q'_i$ of the base function $\mathbf{1}$ to $q_i \pm c$ and leaving the other coefficients unchanged.

Finally, because potential functions always have to be non-negative, any rule that derives a triple $\{\context;Q\}c\{\context';Q'\}$ has two extra implicit assumptions: $Q \succeq_\context 0$ and $Q' \succeq_{\context'} 0$.
The fact that these assumptions imply the non-negativity of the potential functions becomes clear when we explain the meaning of $\succeq_.$ in the next section.
\vspace{-1.5ex}
\subsection{Inference rules}
\label{sec:inferencerules}
\begin{figure}[!t]
\vspace{-2ex}
\centering
\small {



\begin{mathpar}
\RuleToplabel{\rul{ValidCtx}}
{
  P : (\context;Q,\context';Q') \in \Delta \Rightarrow \Delta \vdash \{\context;Q\}\,\mathcal D (P)\,\{\context';Q'\}
}
{ 
  \vdash \Delta 
}
\hspace{0.1cm}\hfill
\RuleToplabel{\textcolor{ACMRed}{Q:Skip}}
{
}
{
\vdash \{\context;Q\} \word{skip} \{\context;Q\}
}

\RuleToplabel{\textcolor{ACMRed}{Q:Assert}}
{
}
{
\vdash \{\context;Q\} \word{assert} e\, \{\context \wedge e;Q\}
}
\hfill
\RuleToplabel{\textcolor{ACMRed}{Q:NonDet}}
{
\vdash \{\context;Q\} c_1 \{\context';Q'\} \\
\vdash \{\context;Q\} c_2 \{\context';Q'\}
}
{
\vdash \{\context;Q\} \word{if} \star \; c_1 \word{else} c_2\, \{\context';Q'\}
}

\RuleToplabel{\textcolor{ACMRed}{Q:If}}
{
\vdash \{\context \wedge e;Q\} c_1 \{\context';Q'\} \\\\
\vdash \{\context \wedge \neg e;Q\} c_2 \{\context';Q'\}
}
{
\vdash \{\context;Q\} \word{if} e \; c_1 \word{else} c_2\, \{\context';Q'\}
}
\hfill
\RuleToplabel{\textcolor{ACMRed}{Q:Loop}}
{
\vdash \{\context \wedge e;Q\} c \{\context;Q\}
}
{
\vdash \{\context;Q\} \word{while} e \; c\, \{\context \wedge \neg e;Q\}
}

\RuleToplabel{\textcolor{ACMRed}{Q:PIf}}
{
Q = p{\cdot}Q_1 + (1-p){\cdot}Q_2 \\
\vdash \{\context;Q_1\} c_1 \{\context';Q'\} \\
\vdash \{\context;Q_2\} c_2 \{\context';Q'\}
}
{
\vdash \{\context;Q\} c_1 \oplus_p c_2 \{\context';Q'\}
}
\hfill
\RuleToplabel{\textcolor{ACMRed}{Q:Seq}}
{
\vdash \{\context;Q\} c_1 \{\context';Q'\} \\\\
\vdash \{\context';Q'\} c_2 \{\context'';Q''\}
}
{
\vdash \{\context;Q\} c_1; c_2 \{\context'';Q''\}
}

\RuleToplabel{\textcolor{ACMRed}{Q:Sample}}
{
	\context \models R \in [a,b] \\
	\forall v_i \in [a,b]. \sem{\dist{R} : v_i} = p_i \\
	\forall v_i. \vdash \{\context;Q_i\} x = e \word{bop} v_i \{\context';Q'\} \\
	{\textstyle Q = \sum_{i}p_i{\cdot}Q_i}
}
{
	\vdash \{\context;Q\} x = e \word{bop} R \{\context';Q'\}
}
\and
\RuleToplabel{\textcolor{ACMRed}{Q:Assign}}
{
  A = (a_{i,j}) \\
  {\textstyle \forall j \in \mathcal S_{x=e}.\, b_j[e/x] = \sum_i a_{i,j} \cdot b_i} \\
  {\textstyle \forall j\not\in \mathcal S_{x=e}.\, a_{i,j} = 0} \\
  {\textstyle \forall j\not\in \mathcal S_{x=e}.\, q'_j = 0} \\
  Q = AQ' \\
}
{
\vdash \{\context[e/x];Q\} x = e \{\context;Q'\}
}
\and
\RuleToplabel{\textcolor{ACMRed}{Q:Call}}
{
P : (\context;Q,\context';Q') \in \Delta \\
x \in \Qplus
}
{
\vdash \{\context;Q+x\} \word{call} P\, \{\context';Q'+x\}
}
\hfill
\RuleToplabel{\textcolor{ACMRed}{Q:Tick}}
{
}
{
\vdash \{\context;Q\} \word{tick}(q)\, \{\context;Q - q\}
}

\RuleToplabel{\textcolor{ACMRed}{Relax}}
{
  F = (F_1, \dots, F_N) \\
  \vec{u} = (u_1, \dots, u_N)^\intercal \\\\
  \forall i. \context \models \Phi_{F_i} \ge 0 \\
  \textstyle \forall i. u_i \ge 0 \\
  Q' = Q - F\vec{u} \\
}
{
Q \succeq_\context Q'
}
\hspace{0.1cm}\hfill
\RuleToplabel{\textcolor{ACMRed}{Q:Abort}}
{
}
{
\vdash \{\context;0\} \word{abort} \{\context';Q'\}
}

\RuleToplabel{\textcolor{ACMRed}{Q:Weaken}}
{
\context \models \context_0 \\
Q \succeq_{\context} Q_0 \\
\vdash \{\context_0;Q_0\} c \{\context_0';Q_0'\}  \\
\context_0' \models \context' \\
Q_0' \succeq_{\context'_0} Q'
}
{
\vdash \{\context;Q\} c \{\context';Q'\}
}
\end{mathpar}
\vspace{-2ex}
} 
\vspace{-2ex}
\caption{Inference rules of the derivation system.}
\vspace{-3ex}
\label{fig:inferencerules3}
\end{figure}


%
Figure~\ref{fig:inferencerules3} gives the complete set of rules.
We informally describe some important rules and justify their validity. Section~\ref{sec:soundness} gives details about the formal soundness proof.

The rule \rul{Q:Tick} is the only one that accounts for the effect of consuming resources.
The \word{tick} command does not change the program state, so we require the logical context $\context$ in the pre- and postcondition to be the same.
Let $c'$ be a command with an expected resource bound $\Phi_{Q'} = \Phi_Q - q$. 
Because the cost of $\word{tick}(q)$ is exactly $q$ resource units, the expected resource bound of $\word{tick}(q);c'$ is exactly $\Phi_{Q'} + q = \Phi_Q$. 

The rule \rul{Q:PIf} accounts for probabilistic branching. 
Let $c'$ be a continuation command with an expected resource bound $\Phi_{Q'}$, $T_1$ and $T_2$ be the resource usage of executing $c_1$ and $c_2$, respectively. 
Then by the linearity of the expectations, the expected resource bound of the command $(c_1 \oplus_p c_2); c'$ is $T_c = p\cdot (T_1 + \Phi_{Q'} + (1-p)\cdot (T_2 + \Phi_{Q'})$. Using the hypothesis triples for $c_1$ and $c_2$, we have $T_c \leq p \cdot \Phi_{Q_1} + (1-p) \Phi_{Q_2} = \Phi_{Q}$. 

The second probabilistic rule \rul{Q:Sample} deals with sampling assignments.
Recall that $R$ is a random variable following a distribution $\dist R$.
The essence of the rule is that, since we assumed that $R$ is bounded, we can treat a sampling assignment as a (nested) probabilistic branching.
Each of the branches contains an assignment $x = x \word{bop} v$ with $v \in \mathbb Z$ and is executed with probability $p$, the probability of the event $R = v$.
The preconditions of each of the branches are combined like in the \rul{Q:PIf} rule.

The rules \rul{Q:Assign} and \rul{Q:Weaken} are similar to the ones of a previous implementation of AARA for the analysis of non-probabilistic programs~\cite{CarbonneauxHRS17} but have been adapted to our structured probabilistic language.
In the rule \rul{Q:Assign} for $x = e$, we represent the state transformation as a linear transformation on potential functions.
If $\Phi_{Q'}$ is a bound on the expected cost of $c'$, then $\Phi_{Q'}[e/x]$ is the expected resource bound of $x=e; c'$
%
To encode this constraint as a linear program (this is necessary to enable automation using LP solving), we find all the \emph{stable} base functions, denoted 
$\mathcal S_{x=e}$, that is, all the base functions $b_j$ for which there exists $(a_{i,j})_i \in \mathbb Q$ such that $b_j[e/x] = \sum_i a_{i,j} \cdot b_i$. That means 
a function is stable if its extending with the mapping of $x$ to $e$ can be represented by the set of based functions. 
Finally, to ensure that the transformation w.r.t the assignment $x = e$ on the potential function $\Phi_{Q'}$ is linear, we require that all the base functions that are not stable have their coefficients set to $0$ in $Q'$. 
With these constraints, we have that $\Phi_{Q'}[e/x] = \Phi_{AQ'} = \Phi_{Q}$ where $A$ is the $({N}{\times}{N})$ matrix with coefficients $(a_{i,j})$, hence justifying the validity. 

The essence of the \rul{Q:Weaken} is that it is always safe to add potential in the precondition and remove potential in the postcondition.
This concept of more (or less) potential is made precise by the predicate $Q \succeq_\context Q'$.
Semantically, $Q \succeq_\context Q'$ encodes---using linear constraints---the fact that in all states $\sigma \models \context$, we have $\Phi_{Q'}(\sigma) \ge \Phi_{Q}(\sigma)$ (see Lemma~\ref{lem:relax} in Appendix~\ref{app:soundnessproof}). 
The \rul{Relax} rule uses \emph{rewrite functions} $(F_i)_i$, as introduced in~\cite{CarbonneauxHRS17}.
Rewrite functions are linear combinations of base functions that can be proved non-negative in the logical context $\Gamma$.
Using rewrite functions, the idea of the judgement $Q \succeq_\context Q'$ is that, to obtain $Q'$, one has to subtract a non-negative quantity from $Q$.

The rule \rul{Q:Call} handles procedure calls.
The pre- and postcondition for the procedure $P$ are fetched from the specification context $\Delta$.
Then, a non-negative \emph{frame} $x \in \Qplus$ is added to the procedure specification.
This frame allows to pass some constant potential through the procedure call and is required for the analysis of most non-tail-recursive algorithms.
The idea of this framing is that if the triple $\{.;Q\}c\{.;Q'\}$ is valid, we can also take $Q'{+}x$ as postcondition and the need for the extra cost $x$ required by the continuation command can be threaded up to the precondition that becomes $Q{+}x$.
In the soundness proof, this framing boils down to the ``propagation of constants'' property on the calculus~\cite{OlmedoKKM16} used in our formal soundness proof.

\vspace{-1.5ex}
\section{Automatic constraint generation and solving using LP solvers}
\label{sec:constraint}
The automatic bound derivation is split in two phases.
First, derivation templates and constraints are generated by inductively applying the inference rules to the input program.
During this first phase, the coefficients of the potential annotations are left as symbolic names
and the inequalities are collected as constraints.
Each symbolic name corresponds to a variable in a linear program.
Second, we feed the linear program to an off-the-shelf LP solver\footnote{We use Coin-Or's CLP.}.
If the LP solver returns a solution, we obtain a valid derivation and
extract the expected resource bound. 
Otherwise, an error is reported. 
\begin{figure*}[th!]
\vspace{-2ex}
\begin{minipage}[b]{0.60\textwidth}
\begin{mathpar}
\scriptsize {
\inferrule*[right=\scriptsize{\textcolor{ACMRed}{Q:Loop}}]
{
      \inferrule*[right=\scriptsize{\textcolor{ACMRed}{Q:Seq}}]
      {
        \inferrule*[right=\scriptsize{\textcolor{ACMRed}{Q:PIf}},vdots=0.6cm,rightskip=5.5cm]
        {
          \inferrule*[right=\scriptsize{\textcolor{ACMRed}{$\mbox{\sc Q:Weaken}_1$}}]
          {
            \inferrule*[right=\scriptsize{\textcolor{ACMRed}{$\mbox{\sc Q:Assign}_1$}},rightskip=1cm]
            {
            }
            {
              \vdash \{x \geq 2;Q^{d1}\} x = x - 1 \{.;P^{d1}\}
            }
          }
          {
             \vdash \{x \geq 2;Q^{w1}\} x = x - 1 \{.;P^{w1}\}
          }
          \hspace{0.2cm}
          \inferrule*[right=\scriptsize{\textcolor{ACMRed}{$\mbox{\sc Q:Weaken}_2$}}]
          {
            \inferrule*[right=\scriptsize{\textcolor{ACMRed}{$\mbox{\sc Q:Assign}_2$}},rightskip=1cm]
            {
            }
            {
              \vdash \{x \geq 2;Q^{d2}\} x = x - 2 \{.;P^{d2}\}
            }
          }
          {
             \vdash \{x \geq 2;Q^{w2}\} x = x - 2 \{.;P^{w2}\}
          }
        }
        {
          \vdash \{x \geq 2;Q^{pi}\} x = x  - 1 \oplus_{\frac{1}{3}} x = x - 2 \{.;P^{pi}\} \hspace{4cm}
        }
        \inferrule*[right=\scriptsize{\textcolor{ACMRed}{Q:Tick}}]
        {
          \\
        }
        {
          \vdash \{.;Q^{ti}\} \word{ tick}(1) \{.;P^{ti}\}
        } 
      }
      {
        \vdash \{x \geq 2;Q^{sq}\} \; x = x  - 1 \oplus_{\frac{1}{3}} x = x - 2; \word{tick}(1) \{.;P^{sq}\}
      } 
}
{
  \vdash \{.;Q\} ~\word{while} (x >= 2) \{ x = x  - 1 \oplus_{\frac{1}{3}} x = x - 2; \word{tick}(1) \}~ \{x < 2;P\}  
} 
}
\end{mathpar}
\end{minipage}%
\hfill
\begin{minipage}[b]{0.34\textwidth}
\scriptsize {
\begin{tabular*}{1.0\textwidth}{@{\extracolsep{\fill}} l l}
Constraints & Rules \\
\hline
$Q = Q^{sq} = P^{sq} = P$ & \textcolor{ACMRed}{\textsc{Q:Loop}} \\
$Q^{sq} = Q^{pi} \land P^{pi} = Q^{ti} \land P^{ti} = P^{sq}$ & \textcolor{ACMRed}{\textsc{Q:Seq}} \\
$Q^{pi} = \frac{1}{3}{\cdot}Q^{w1} + \frac{2}{3}{\cdot}Q^{w2} \land P^{pi} = P^{w1} = P^{w2}$ & \textcolor{ACMRed}{\textsc{Q:PIf}} \\
$Q^{w1} \succeq_{(x \geq 2)} Q^{d1} \land P^{d1} \succeq_{(x \geq 2)} P^{w1}$ & \textcolor{ACMRed}{$\mbox{\sc Q:Weaken}_1$} \\
$ q^{d1}_1 = p^{d1}_1 \land q^{d1}_{x0} = 0 \land q^{d1}_{x1} = p^{d1}_{x0} \land$ & \textcolor{ACMRed}{$\mbox{\sc Q:Assign}_1$} \\
$ q^{d1}_{x2} = p^{d1}_{x1} \land p^{d1}_{x2} = 0$ & \\
$Q^{w2} \succeq_{(x \geq 2)} Q^{d2}; P^{d2} \succeq_{(x \geq 2)} P^{w2}$ & \textcolor{ACMRed}{$\mbox{\sc Q:Weaken}_2$} \\
$ q^{d2}_1 = p^{d2}_1 \land q^{d2}_{x0} = 0 \land q^{d2}_{x1} = 0 \land$ & \textcolor{ACMRed}{$\mbox{\sc Q:Assign}_2$} \\
$ q^{d2}_{x2} = p^{d2}_{x0} \land p^{d2}_{x1} = 0 \land p^{d1}_{x2} = 0$  & \\
$Q^{ti} = P^{ti} + 1$ & \textcolor{ACMRed}{\textsc{Q:Tick}}
\end{tabular*}
}
\end{minipage}%
\vspace{-2ex}
\caption{Inference of a derivation using linear constraint solving.}
\vspace{-3ex}
\label{fig:derivationexample}
\end{figure*}
\vspace{-.5ex}
\paragraph{Generating linear constraints.} 
A detailed example of this process is shown in Figure~\ref{fig:derivationexample}.
Note that \rul{Q:Weaken} is applied twice. 
Since this rule is not syntax-directed, it can be applied at any point during the derivation.
In our implementation, we apply it around all assignments. This proved sufficient in practice and limits the number of constraints generated.
In the figure, the potential annotations are represented by an upper-case letter $P$ or $Q$ with an optional superscript.
For example, $Q$ represents the potential function
\begin{equation*}
\small{q_{1} \cdot \mathbf{1} + q_{x0} \cdot |[0, x]| + q_{x1} \cdot |[1, x]| + q_{x2} \cdot |[2, x]|}
\end{equation*}
The set of base functions is $\mathbf{1}$
and $|[i, x]|$ for $i \in \{ 0, 1, 2 \}$. 
We will see that they are sufficient to infer a bound.
Details of how to select base functions are given in Section~\ref{sec:experiment}.
To apply weakening, we need rewrite functions, we pick
$$
\small{
\begin{array}{c}
  F_0 = 1;
  F_1 = -1 + |[0, x]| - |[1, x]|; 
  F_2 = -2 + |[0, x]| - |[2, x]|
\end{array}
}
$$
$F_1$ is applicable (i.e., non-negative) iff $x \ge 1$.
Similarly, $F_2$ is applicable iff $x \ge 2$.
This means that both rewrite functions can be used at the beginning of the loop body, when $x \ge 2$ can be proved because of the loop condition.

The constraints given in the table in Figure~\ref{fig:derivationexample} use shorthand notations to constrain all the coefficients of two annotations.
For instance $Q = Q^{sq}$ should be expanded into $q_1 = q^{sq}_1 \land q_{x0} = q^{sq}_{x0} \land q_{x1} = q^{sq}_{x1} \land q_{x2} = q^{sq}_{x2}$.
The most interesting rules are the probabilistic branching, the two weakenings, and the two assignments.
For the probabilistic branching, following \rul{Q:PIf}, the preconditions of the two branches are linearly combined using the weights $\frac 1 3$ and $1 - \frac 1 3 = \frac 2 3$.

We now discuss the first weakening. The second one generates an identical set of constraints---but the LP solver will give it a different solution.
The most interesting constraints are the ones for $Q^{w1} \succeq_{(x \ge 2)} Q^{d1}$.
This relation is defined by the rule \rul{Relax} in Figure~\ref{fig:inferencerules3} and involves finding all the applicable rewrite functions in the logical state $x \ge 2$.
As discussed above, $F_0$, $F_1$, and $F_2$ are all applicable, and the following system of constraints, written in matrix notation, is generated.
\begin{equation*}
\label{eq:relax1}
\footnotesize{
\left(
\begin{array}{c}
 q^{d1}_1 \\
 q^{d1}_{x0} \\
 q^{d1}_{x1} \\
 q^{d1}_{x2}
\end{array}
\right)
=
\left(
\begin{array}{c}
 q^{w1}_1 \\
 q^{w1}_{x0} \\
 q^{w1}_{x1} \\
 q^{w1}_{x2}
\end{array}
\right)
-
\left(
\begin{array}{l}
 1\;{-}{1}\;{-}{2} \\
 0\;\;\;1\;\;\;1  \\
 0\;{-}{1}\;\;\;0  \\
 0\;\;\;0\;{-}{1} \\
\end{array}
\right)
\left(
\begin{array}{c}
 u_0 \\
 u_1 \\
 u_2
\end{array}
\right)
\land
\left(
\begin{array}{c}
 u_0 \\
 u_1 \\
 u_2
\end{array}
\right)
\ge
\left(
\begin{array}{c}
 0 \\
 0 \\
 0
\end{array}
\right)
}
\end{equation*}
The columns of the $({4}{\times}{3})$ matrix correspond, in order, to $F_0$, $F_1$, and $F_2$.
The coefficients $(u_i)$ are fresh names that are local to this weakening.

For the first assignment \textcolor{ACMRed}{$\mbox{\sc Q:Assign}_1$},
the stable set discussed in Section~\ref{sec:inferencerules} is $\mathcal S_{x = x - 1} = \{ 1, |[0, x]|, |[1, x]| \}$.
Indeed, only  $|[2, x]|$ is unstable since it becomes $|[3, x]|$ after the assignment $x = x - 1$.
Since the assignment leaves $\mathbf{1}$ unchanged and changes $|[0, x]|$ into $|[1, x]|$ and $|[1, x]|$ into $|[2, x]|$, the system of constraints generated is
\begin{equation*}
\label{eq:assign1}
\footnotesize{
\left(
\begin{array}{c}
 q^{d1}_1 \\
 q^{d1}_{x0} \\
 q^{d1}_{x1} \\
 q^{d1}_{x2}
\end{array}
\right)
=
\left(
\begin{array}{l}
 1\;\;0\;\;0\;\;0 \\
 0\;\;0\;\;0\;\;0 \\
 0\;\;1\;\;0\;\;0 \\
 0\;\;0\;\;1\;\;0 \\
\end{array}
\right)
\left(
\begin{array}{c}
 p^{d1}_1 \\
 p^{d1}_{x0} \\
 p^{d1}_{x1} \\
 p^{d1}_{x2}
\end{array}
\right)
\hspace{0.4cm}\land\hspace{0.4cm}
p^{d1}_{x2} = 0
}
\end{equation*}
or $q^{d1}_1 = p^{d1}_1 \land q^{d1}_{x1} = p^{d1}_{x0} \land q^{d1}_{x2} = p^{d1}_{x1} \land p^{d1}_{x2} = 0$.

\vspace{-.5ex}
\paragraph{Solving the constraints.}
The LP solver does not only find a solution that satisfies the constraints, it also optimizes a linear objective function.
In our case, we would like to find the tightest---i.e, smallest---upper bound on the expected resource consumption.
In the implementation, we use an iterative scheme that takes full advantage of the incremental solving capabilities of modern LP solvers.
Starting at the maximum degree $d$, we ask the LP solver to minimize the coefficients $(q^d_i)_i$ of all the base functions of degree $d$.
If a solution $(k^d_i)_i$ is returned, we add the constraints $\bigwedge_i q^d_i = k^d_i$ to the linear program.
Then, the same scheme is iterated for base functions of degree $d-1, d-2, \dots, 1$.
For our running example, the first objective function for the linear coefficients is
$
  20 \cdot q_{x0} + 10 \cdot q_{x1} + 1 \cdot q_{x2}
$.
The weight of the coefficients are set to signify facts about the base functions to the LP solver.
For instance, $q_{x0}$ gets a smaller weight than $q_{x1}$ because $|[0, x]| \ge |[1, x]|$ for all $x$.
The final solution returned by the LP solver is $q_{x0} = \frac 3 5$ and $q_\star = 0$ otherwise.
Thus the derived bound is $\frac 3 5 |[0, x]|$.


\vspace{-1.5ex}
\section{Soundness of the analysis}
\label{sec:soundness}
The soundness of the analysis is proved with respect to an operational semantics based on Markov decision processes (see Appendix~\ref{app:opeationalsemantics}). It leverages previous work on probabilistic programs by relying on the soundness of a weakest pre-expectation (WP) calculus~\cite{Kamin16,OlmedoKKM16}. The weakest pre-expectation $\word{ert}[c,\mathcal{D}](f)(\sigma)$ is the (exact) expected amount
of resources consumed by a program $(c,\mathcal{D})$ started in state $\sigma$ if it is followed by a computation that has an expected resource consumption given by a function $f: \Sigma \rightarrow \Rplus \cup \{\infty\}$. See  Appendix~\ref{app:wpcalculus} for more details and a formal definition.

We first interpret the pre- and postconditions of the triples as expectations.
This interpretation is a function $\mathcal T$ that maps $\{\context;Q\}$ to the assertion $\mathcal T(\context;Q)$ defined as
$
\mathcal{T}(\context;Q)(\sigma) := \max(\context(\sigma), \Phi_{Q}(\sigma))
$, where $\Phi_{Q}$ is the potential function associated with the quantitative annotation $Q$ and $\context$ is lifted as a function on states such that $\context(\sigma)$ is 0 if $\sigma \models \context$ and $\infty$ otherwise.
The soundness of the automatic analysis can now be stated formally w.r.t the WP calculus.
\vspace{-.5ex}
\begin{theorem}[Soundness of the automatic analysis]
\label{theo:soundness}
Let $c$ be a command in a larger program $(\_, \mathcal D)$.
If 
$\vdash \{\context;Q\} c \{\context';Q'\}$ 
is derivable, then $\forall \sigma \in \Sigma$, the following holds 
$$
\mathcal{T}(\context;Q)(\sigma) \geq \word{ert}[c,\mathcal{D}](\mathcal{T}(\context';Q'))(\sigma)
$$
\vspace{-1.0ex}
\end{theorem}
\vspace{-1.5ex}
\begin{proof}
The proof is done by induction on the program structure and the derivation using the inference rules.
See Appendix~\ref{app:soundnessproof} for details.
\end{proof}
\vspace{-1.5ex}


\vspace{-1.5ex}
\section{Implementation and experiments}
\label{sec:experiment}
In this section, we first describe the implementation of the automatic
analysis in the tool \toolname{}. Then, we evaluate the performance
of our tool on a set of challenging examples.\footnote{The source code of the examples, \toolname{}, the experiments, and the simulation-based comparison have been submitted as auxiliary material.}
\vspace{-1.5ex}
\subsection{Implementation}
\label{sec:implementation}
\toolname{} is implemented in OCaml and consists of about $5000$ LOC. 
The tool currently works on imperative integer programs written in a Python-like syntax that supports recursive procedures.
It also has a C interface based on LLVM. Currently, \toolname{} supports four common distributions: Bernoulli, binomial, hyper-geometric, and uniform. 
However, there are no limitations to the distributions that can be supported as long as they have a finite domain.
\vspace{-.5ex}
\paragraph{Potential functions.} 
To discovery the bounds on expected resource usage automatically, in this work, we focus on inferring polynomial potential functions that are \emph{linear combinations} of \emph{base functions} picked among the monomials. Formally, they are defined by the following syntax.
\begin{displaymath}
\vspace{-1.0ex}
\begin{array}{lll}
M & :=  \mathbf{1} \mid x \mid M_1{\cdot}M_2 \mid \word{max}(0,P) 	& x \in \word{VID} \\
P & := k{\cdot}M \mid P_1{+}P_2     									& k \in \mathbb{Q}
\end{array}
\vspace{-1.0ex}
\end{displaymath}
\vspace{-.5ex}
\paragraph{Generating base and rewrite functions.} 
Our analysis can work with every set of base functions.
While it would be possible to to fix a set of functions once and for all as in previous work on resource analysis~\cite{HoffmannAH10,CarbonneauxHZ15}, we found that it is more effective to select the base functions for each program using a heuristic~\cite{CarbonneauxHRS17}.
The abstract interpretation (AI) used in \toolname{} to infer logical contexts derives linear inequalities between program variables
and uses a Presburger decision procedure.
%
Our implementation uses these inequalities to heuristically generate a set of base and rewrite functions.
For example, if $n > x$ at one program point, the heuristic will add the monomial $\max(0,n-x)$ as a base function.
Higher-degree base functions can be constructed by considering successive powers and products of simpler base functions. 
One can use a more complex and powerful AI such as the Apron library~\cite{Jeannet09}.
In practice, we found that our simple AI is sufficient to infer many bounds and provides good performance.

When the heuristic adds a new base function, a set of rewrite functions is enriched to allow transfers of potential to and from the new base function.
For instance, for the base function $\max(0, n-x)$ we add the rewrite function $F = \max(0,n-x) - \max(0,n-x-1) - 1$.
$F$ can be used for an assignment $x = x + 1$ when $n-x > 0$. 
So if $\max(0,n-x)$ is the potential before the assignment, $F$ can be used to turn it into $\max(0,n-(x+1))  + 1$ which, after the assignment, becomes $\max(0, n-x) + 1$, effectively extracting one unit of constant potential.

\vspace{-.5ex}
\paragraph{User interaction.} 
Occasionally, when a program requires a complex potential transformation, our heuristic might not be sophisticated enough to identify 
an appropriate set of rewrite functions. 
In this case, the user can manually specify a set
of rewrite functions as hints to be used by the analysis.
These hints, in contrast with typical assertions, 
have no runtime effect and do not compromise soundness. In
particular, before using a rewrite function, the
analyzer checks that its non-negativity condition is
satisfied.

\vspace{-1.5ex}
\subsection{Experimental evaluation}
\vspace{-.5ex}
\paragraph{Evaluation setup.} 
To evaluated the practicality of our framework, we have designed and
collected $\numexamp$
challenging examples with different looping and recursion patterns
that depend on probabilistic branching and sampling assignments.  In total,
the benchmark consists of more than $1000$ LOC.

The programs \progname{bayesian}~\cite{GordonHNR14}, 
\progname{filling}~\cite{Sankaranarayanan13}, \progname{race}, \progname{2drwalk}, \progname{robot},
\progname{roulette}~\cite{Chakarov13}, and \progname{sampling}~\cite{Kamin16}
have been described in the literature on probabilistic programs, in which  
their expected resource consumption has been analyzed manually.  
The programs \progname{C4B\_$*$},
\progname{prseq}, \progname{prseq}, \progname{preseq\_bin},  
\progname{prspeed}, \progname{rdseql},
\progname{rdspeed}, and \progname{recursive} are probabilistic versions of deterministic
examples from previous
work~\cite{CarbonneauxHRS17,CarbonneauxHZ15,GulwaniMC09}. The other
examples are either adaptations of classic randomized algorithms or
handcrafted new programs that demonstrate particular capabilities of
our analysis. Section~\ref{sec:overview} contains some representative
listings.

To measure the expected resource usage of all examples by simulation,
we uniformly chose the range of inputs to be $1000$ to $5000$ and
allowed only $1$ input variable to vary while choosing fixed random
values for other inputs.\footnote{We reduced the input ranges of
  polynomial programs by an order of magnitude because their simulation
  runtime is very long.}  We sampled the resource usage
$10000$ times for each input 
using the GSL-GNU scientific library~\cite{gsl}. We then compared the results
to our statically computed bounds. The simulation is implemented in
C$^{++}$ and consists of more than $5000$ LOC.

The experiments were run on a machine with an Intel Core i5 2.4 GHz
processor and 8GB of RAM under macOS 10.13.1. The LP solver we use is CoinOr CLP~\cite{CoinOrCLP}.
\vspace{-.5ex}
\paragraph{Results.} 
\begin{table}
	\vspace{-2ex}
  \begin{center}
{\footnotesize
\begin{tabular*}{0.5\textwidth}{@{\extracolsep{\fill}} l l l l}
\multicolumn{4}{c}{Linear programs} \\
\hline
Program & Expected bound & Error($\%$) & Time(s)\\
\hline
\noalign{\vskip 1mm}
\word{2drwalk}
& ${2}{\cdot}{\interval{d,n+1}}$
& $0.170$
& 2.278 \\
\word{bayesian}
& ${5}{\cdot}{\interval{0,n}}$
& $0$
& 0.272 \\
\word{ber}
& ${2}{\cdot}{\interval{x,n}}$
& $0.026$
& 0.008 \\
\word{bin}
& ${0.2}{\cdot}{\interval{0,n+9}}$
& $0.290$
& 0.281 \\
\word{C4B\_t09}
& ${8.27273}{\cdot}{\interval{0,x}}$
& $5.362$
& 0.061 \\
\word{C4B\_t13}
& ${1.25}{\cdot}{\interval{0,x}} + {\interval{0,y}}$
& $0.009$
& 0.045 \\
\word{C4B\_t15}
& ${2}{\cdot}{\interval{0,x}}$
& A.S
& 0.044 \\
\word{C4B\_t19}
& ${\interval{0,k+i+51}} + {2}{\cdot}{\interval{100,i}}$
& $2.711$
& 0.058 \\
\word{C4B\_t30}
& ${0.5}{\cdot}{\interval{0,x+2}} + {0.5}{\cdot}{\interval{0,y+2}}$
& W.C
& 0.032 \\
\word{C4B\_t61}
& ${0.060606}{\cdot}{\interval{0,l-1}} + {\interval{0,l}}$
& $0.754$
& 0.028 \\
\word{condand}
& ${\interval{0,m}} + {\interval{0,n}}$
& A.S
& 0.010 \\
\word{cooling}
& ${0.42}{\cdot}{\interval{0,t+5}} + {\interval{\text{st},\text{mt}}}$
& $0.192$
& 0.079 \\
\word{fcall}
& ${2}{\cdot}{\interval{x,n}}$
& $0.025$
& 0.008 \\
\word{filling}
& ${0.037037}{\cdot}{\interval{0,\text{vol}+2}} +$
& $0.713$
& 0.615 \\
& ${0.333333}{\cdot}{\interval{0,\text{vol}+10}} +$
&
& \\
& ${0.296296}{\cdot}{\interval{0,\text{vol}+11}}$
&
&\\
\word{hyper}
& ${5}{\cdot}{\interval{x,n}}$
& $0.061$
& 0.013 \\
\word{linear01}
& ${0.6}{\cdot}{\interval{0,x}}$
& $0.036$
& 0.016 \\
\word{miner}
& ${7.5}{\cdot}{\interval{0,n}}$
& $0.071$
& 0.077 \\
\word{prdwalk}
& ${1.14286}{\cdot}{\interval{x,n+4}}$
& $0.128$
& 0.052 \\
\word{prnes}
& ${68.4795}{\cdot}{\interval{0,-n}} + {0.052631}{\cdot}{\interval{0,y}}$
& $0.122$
& 0.057 \\
\word{prseq}
& ${1.65}{\cdot}{\interval{y,x}} + {0.15}{\cdot}{\interval{0,y}}$
& $0.144$
& 0.057 \\
\word{prseq\_bin}
& ${1.65}{\cdot}{\interval{y,x}} + {0.15}{\cdot}{\interval{0,y}}$
& $0.150$
& 0.082 \\
\word{prspeed}
& ${2}{\cdot}{\interval{y,m}} + {0.666667}{\cdot}{\interval{x,n}}$
& $0.039$
& 0.057 \\
\word{race}
& ${0.666667}{\cdot}{\interval{h,t+9}}$
& $0.294$
& 0.245 \\
\word{rdseql}
& ${2.25}{\cdot}{\interval{0,x}} + {\interval{0,y}}$
& $0.007$
& 0.025 \\
\word{rdspeed}
& ${2}{\cdot}{\interval{y,m}} + {0.666667}{\cdot}{\interval{x,n}}$
& $0.039$
& 0.040 \\
\word{rdwalk}
& ${2}{\cdot}{\interval{x,n+1}}$
& $0.075$
& 0.012 \\
\word{robot}
& ${0.384615}{\cdot}{\interval{0,n+6}}$
& R.D
& 2.658 \\
\word{roulette}
& ${4.93333}{\cdot}{\interval{n,10010}}$
& $0.282$
& 1.216 \\
\word{sampling}
& ${2}{\cdot}{\interval{0,n}}$
& $0.026$
& 3.347 \\
\word{sprdwalk}
& ${2}{\cdot}{\interval{x,n}}$
& $0.032$
& 0.017 \\
\noalign{\vskip 1mm}
\hline
\multicolumn{4}{c}{Polynomial programs} \\
\hline
\noalign{\vskip 1mm}
\word{complex}
& ${6}{\cdot}{\interval{0,m}}{\cdot}{\interval{0,n}} + {3}{\cdot}{\interval{0,n}}{+}{\interval{0,y}}$
& $0.118$
& 3.415 \\
\word{multirace}
& ${2}{\cdot}{\interval{0,m}}{\cdot}{\interval{0,n}} + {4}{\cdot}{\interval{0,n}}$
& $0.703$
& 9.034 \\
\word{pol04}
& ${4.5}{\cdot}{\interval{0,x}}^{2} + {7.5}{\cdot}{\interval{0,x}}$
& $0.779$
& 0.585 \\
\word{pol05}
& ${\interval{0,x}}^{2} + {\interval{0,x}}$
& $0.431$
& 0.353 \\
\word{pol06}
& ${0.625}{\cdot}{\interval{\text{min},s}} +$ 
& A.S
& 7.066 \\
& ${2}{\cdot}{\interval{\text{min},s}}{\cdot}{\interval{0,\text{min}}} + {0.625}{\cdot}{\interval{\text{min},s}}^{2}$
&\\
\word{pol07}
& ${1.5}{\cdot}{\interval{0,n-2}}{\cdot}{\interval{0,n-1}}$ 
& $0.008$
& 4.534 \\
\word{rdbub}
& ${3}{\cdot}{\interval{0,n}}^{2}$
& $0.106$
& 0.190 \\
\word{recursive}
& ${0.25}{\cdot}{\interval{l,h}}^{2}{+}{1.75}{\cdot}{\interval{l,h}}$
& $0.281$
& 3.791 \\
\word{trader}
& ${5}{\cdot}{\interval{s_{\text{min}},s}}^{2} + {5}{\cdot}{\interval{s_{\text{min}},s}} +$
& $0.251$
& 7.262 \\
& ${10}{\cdot}{\interval{s_{\text{min}},s}}{\cdot}{\interval{0,s_{\text{min}}}}$
&
&\\
\end{tabular*}}
\end{center}


  \caption{Automatically-derived bounds on the expected number of \word{ticks} with \toolname{}.}
  \vspace{-7ex}
  \label{tab:evalation}
\end{table}
The results of the evaluation are compiled in
Table~\ref{tab:evalation}. The table is split into linear and
non-linear bounds. It contains the inferred bounds, the total time
taken by \toolname{}, and the means in percentage of the absolute errors 
between the measured expected values and the inferred bounds.
%
In general, the analysis finds bounds quickly: All the examples are
processed in less than 10 seconds. The analysis time mainly depends on
three factors: the number of variables in the program, the number of
base functions, and the size of the distribution's domain in the sampling
commands. 
The user can specify a maximal degree of
the bounds to control the number of base functions under consideration.  
Our inference rule for the sampling commands is very
precise but the price we pay for the precision is a linear constraint
set whose size is proportional to the range of the sampling distribution.
%

As shown in the \emph{Error}
column, the derived bounds are often not only asymptotically tight but
also contain very precise constant factors.
Figure~\ref{fig:measuredplot} shows representative plots of
comparisons of the inferred bounds and measured cost samples. Our experiments indicate that the computed bounds are
close to the measured expected numbers of ticks. 
%
Appendix~\ref{app:simulation} contains plots for the other benchmarks.

%
However, there is no guarantee that \toolname{} infers asymptotically tight bounds 
and there are many classes of bounds that \toolname{} cannot derive. For example, 
for the programs whose errors are denoted by $A.S$ in Table~\ref{tab:evalation} we did
not compute asymptotically tight bounds. \progname{C4B\_t15} has logarithmic 
expected cost, thus the best bound that \toolname{} can derive is a linear bound. 
Similarly, $\interval{{0}{,}{n}}{+}\interval{{0}{,}{m}}$ is the best bound that can be inferred 
for \progname{condand} whose expected cost is
$2{\cdot}\text{min}\{\interval{{0}{,}{n}},\interval{{0}{,}{m}}\}$. Another source of imprecise constant 
factors in the bounds is rounding. The program \progname{robot} has an imprecise constant 
factor, denoted $R.D$ in the table, because it contains a deep nesting of probabilistic choices. 

Since we do not assume a particular distribution of the inputs,
the bounds on the expected cost have to consider the worst case inputs.
If a program does not contain probabilistic constructs then we preform in fact a worst-case analysis. 
Thus, comparing with the sampled expected cost on the worst-case inputs gives us a very small error even the derived bound is not asymptoticallt tight. 
For instance, \toolname{} derives the untight bound ${0.5}{\cdot}{\interval{{0}{,}{x+2}}}{+}{0.5}{\cdot}{\interval{{0}{,}{y+2}}}$ for \progname{C4B\_t30} whose expected cost is ${0.5}{\cdot}{\interval{{0}{,}{2}{\cdot}(\text{min}\{x,y\}{+}{2})}}$. If we compare the derived bound with the sampled exected cost on the worst-case inputs (e.g., values of $x$ and $y$ such that $x=y$), then we obtain a very small error. 
We mark the error with W.C in this case. 
%
%
\begin{figure*}[th!]
\vspace{-2ex}
\centering
\begin{minipage}[b]{0.33\textwidth}
\centering
\includegraphics[width=1.0\textwidth]{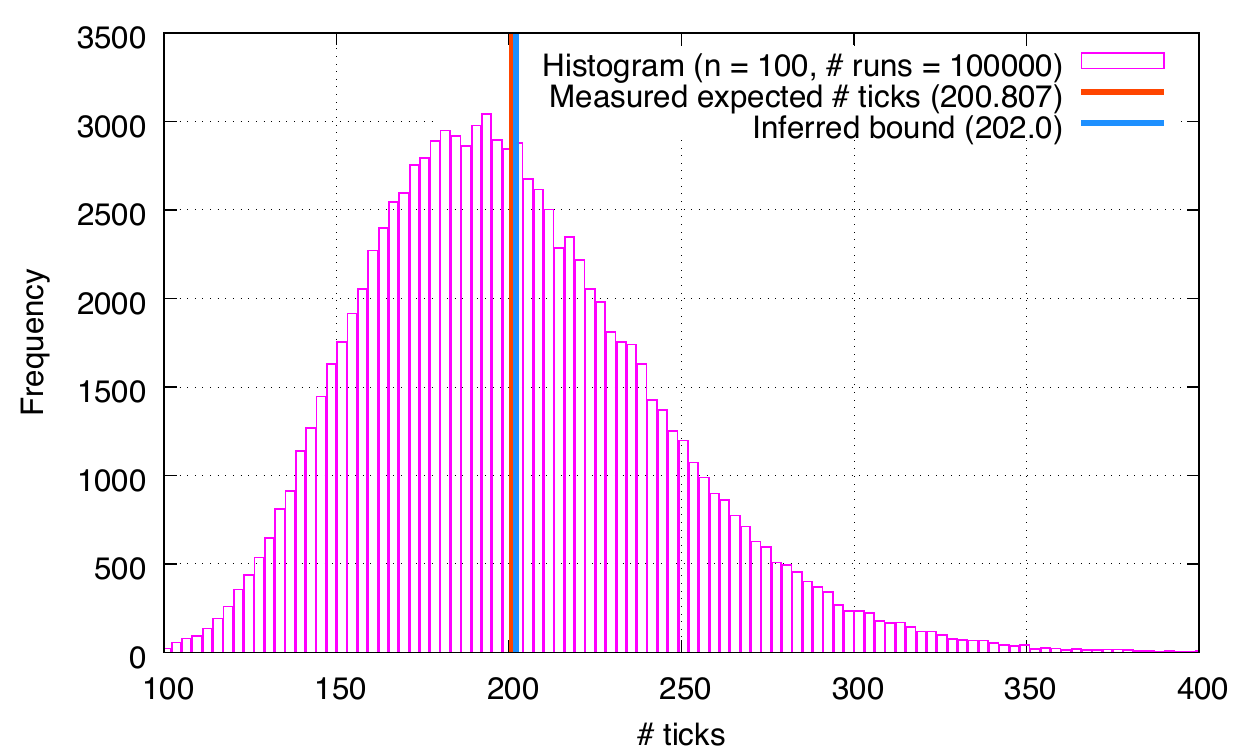}
\end{minipage}%
\begin{minipage}[b]{0.33\textwidth}
\centering
\includegraphics[width=1.1\textwidth]{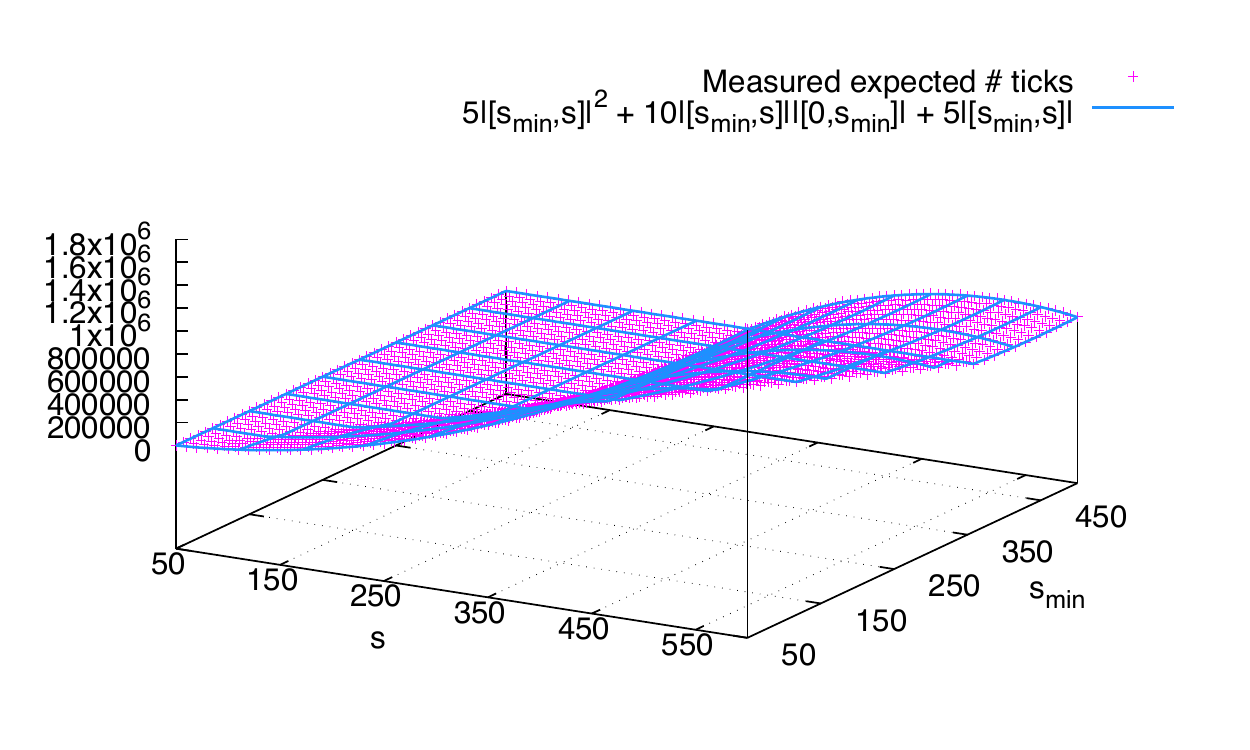}
\end{minipage}%
\begin{minipage}[b]{0.33\textwidth}
\centering
\includegraphics[width=1.0\textwidth]{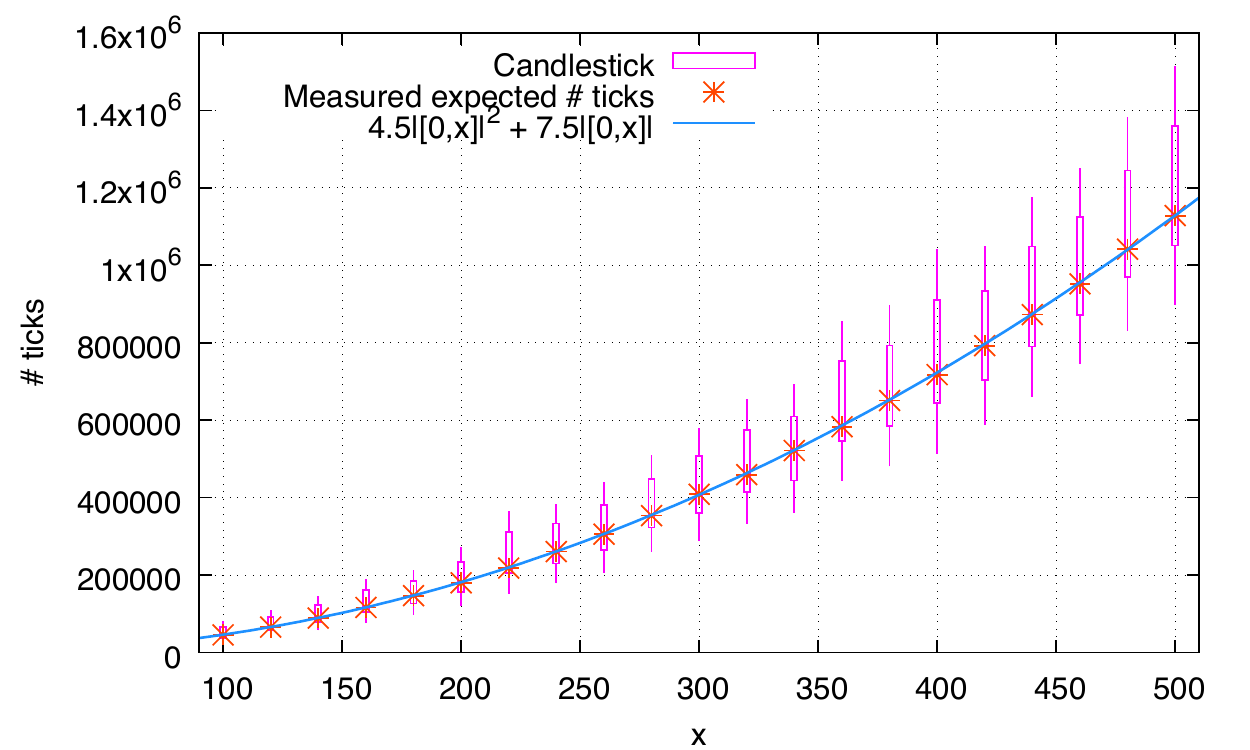}
\end{minipage}%
\vspace{-2ex}
\caption{Comparison of automatically derived bounds with measured
  cost samples. On the left: histogram of the distribution of \#ticks for
  \progname{rdwalk} with $n=100$. On the right: The inferred bound on
  the expected \#ticks (blue lines) compared to the measured expected
  values for various input sizes (red crosses) for \progname{trader}
  (at the center) and \progname{pol04} (on the right). In the latter,
  the candlesticks represent the highest and lowest sampled
  values and the second and third quartile.}
\vspace{-3ex}
\label{fig:measuredplot}
\end{figure*}


\vspace{-1.5ex}
\section{Related work}
\label{sec:relatedwork}

Our work is a confluence of ideas from automatic resource bound
analysis and analysis of probabilistic programs. They have
been extensively studied but developed independently. In spite
of abundant related research, we are not aware of existing techniques
that can automatically derive symbolic bounds on the expected runtime
of probabilistic programs.
\vspace{-.5ex}
\paragraph{Resource bound analysis.} 
Most closely related to
our work is prior work on AARA for deterministic programs. 
AARA has been introduced in~\cite{Jost03} for automatically deriving linear worst-case bounds
for first-order functional programs. The technique has been
generalized to derive polynomial
bounds~\cite{HoffmannH10,HoffmannAH12,HoffmannS13,HofmannM14,HofmannM15},
lower bounds~\cite{NgoDFH16}, and to handle (strictly evaluated)
programs with arrays and references~\cite{LichtmanH17}, higher-order
functions~\cite{Jost10,HoffmannW15}, lazy functional
programs~\cite{SimoesVFJH12,VasconcelosJFH15}, object-oriented
programs~\cite{Jost06,HofmannR13}, and user defined data
types~\cite{Jost09,HoffmannW15}. It also has been integrated into separation logic~\cite{Atkey10} and 
proof assistants~\cite{Nipkow15,ChargueraudP15}. 
A distinctive common theme of sharing 
is compositionality and automatic bound inference via LP solving.

In contrast to our work,
all prior research on AARA targets deterministic programs and derives
worst-case bounds rather than bounds on the expected resource usage. 
In our formulation of AARA for probabilistic programs, we build on
prior work that integrated AARA into Hoare logic to derive bounds for
imperative code~\cite{Nielson87,CarbonneauxHRZ13,CarbonneauxHZ15,CarbonneauxHRS17}, 
a new technique for deriving polynomial bounds on the expected
resource usage of programs with probabilistic sampling and branching.

Beyond AARA there exists many other approaches to automatic worst-case
resource bound analysis for deterministic programs. They are based on
sized types~\cite{Vasconcelos08}, linear dependent
types~\cite{LagoG11,LagoP13}, refinement
types~\cite{CicekGA15,CicekBGGH16}, annotated type
systems~\cite{Crary00,Danielsson08},
defunctionalization~\cite{AvanziniLM15}, recurrence
relations~\cite{Grobauer01,Benzinger04,AlonsoG12,FloresH14,AlbertFR15,DannerLR15,KincaidBBR2017},
abstract interpretation~\cite{GulwaniMC09,BlancHHK10,Zuleger11,SinnZV14,CernyHKRZ15},
template based assume-guarantee reasoning~\cite{MadhavanKK17}, measure
functions~\cite{ChatterjeeFG17}, and techniques from term
rewriting~\cite{AvanziniM13,NoschinskiEG13,BrockschmidtEFFG14,FrohnNHBG16}. 
These techniques do not apply to
probabilistic programs and do not derive bounds on expected resource
usage. 

The decision to base our analysis on AARA is mainly motivated
by the strong connection to existing techniques for (manually)
analyzing expected runtime (see next paragraph) and the general
advantages of AARA, including compositionality, tracking of
amortization effects, flexible cost models, and efficient bound
inference using LP solving.

We are only aware of few works that study the analysis of expected
resource usage of probabilistic programs. Chatterjee et
al.~\cite{ChatterjeeFM17} propose a technique for solving recurrence
relations that arise in the analysis of expected runtime cost. Their
technique can derive bounds of the form $O(\log n)$,
$O(n)$,
and $O(n \log n)$.
Similarly, Flajolet et al.~\cite{FlSaZi91} describe an automatic for
average-case analysis that is based on generating functions and that
can be seen as a method for solving recurrences.
While these techniques apply to recurrences that describe the resource
usage of randomized algorithms, the works do not propose a technique
for deriving recurrences from a program. It is therefore not a
push-button analysis for probabilistic programs but complementary to
our work since they can derive logarithmic bounds. 
\vspace{-.5ex}
\paragraph{Analysis of probabilistic programs.} 
Considering work on analyzing probabilistic programs, most closely
related is a recent line of work by Kaminski et
al.~\cite{Kamin16,OlmedoKKM16}. The goal of this
work is to characterize the expected runtime of
probabilistic programs. However, they use a WP 
calculus to derive pre-expectations and do not consider any 
automation. The technique can be seen as a generalization of
quantitative Hoare logic~\cite{CarbonneauxHRZ13,CarbonneauxHZ15} for
AARA to the probabilistic setting but does not provide support for
automatic reasoning. In fact, when attempting to generalize AARA to
probabilistic programs we were first unaware of the existing work and
rediscovered some of the proof rules. Our contributions are new
specialized proof rules that allow for automation using LP solving and
a prototype implementation of the new technique. While our soundness
proof is original, it leverages the proof by Kaminski et al. by
relying on the soundness of the rules for weakest preconditions.

The use of pre-expectations for reasoning about probabilistic programs
dates back to the pioneering work of Kozen and
others~\cite{Kozen81,McIver04,CelikuM05}. It has been automated using
constraint generation~\cite{KatoenMMM10} and abstract
interpretation~\cite{Chakarov14} to derive quantitative invariants.
However, it is unclear how to use them to automatically derive symbolic (polynomial) bounds like in our work. 
%
%

Another body of research relies probabilistic pushdown automata and
martingale theory to analyze the termination time~\cite{BrazdilKKV15} 
and the expected number of steps~\cite{EsparzaKM05}. 
The use of martingale theory to automatically analyze probabilistic
programs has been pioneered in~\cite{Chakarov13}. While their technique also relies on linear
constraints, it is proving almost-sure termination instead of resource
bounds and relies on Farka's lemma. More general methods
\cite{ChatterjeeFG16} are able to synthesize polynomial
ranking-supermartingales for proving termination.

Abstract interpretation has also been applied to probabilistic
programs~\cite{Monniaux01,Monniaux05,CousotM12} but we are not aware
of its application to derive bounds on the expected resource usage. 
%
Another approach to automatically analyze probabilistic programs is based on 
symbolic inference~\cite{GehrMV16} and analyzing execution paths with 
statistical techniques~\cite{Sankaranarayanan13,GeldenhuysDV12,BorgesFdPV14}. 
%
In the context of analyzing differential privacy, there are works with limited automation that
focus on deriving bounds on the privacy budget for probabilistic
programs~\cite{Haeberlen2011,Barthe2014}. 

\vspace{-1.5ex}
\section{Conclusion}
\label{sec:conclusion}

We have introduced a new technique for automatically inferring
polynomial bounds on the expected resource consumption of
probabilistic programs. The technique is a combination of existing
manual quantitative reasoning for probabilistic programs and an
automatic worst-case bound analysis for deterministic programs.
The effectiveness of the technique is demonstrated with an
implementation and the automatic analysis of challenging examples from
previous work.

In the future, we plan to study how to build on the introduced technique to
automatically derive tail bounds, that is, worst-case bounds that hold
with high probability. 
We are also working on a more direct soundness
argument that also works for non-monotone resources. Finally, we plan
to build on Resource Aware ML~\cite{HoffmannW15} to apply the expected
potential method to (higher-order) functional programs.

\bibliography{lit,publications}

\appendix
\clearpage
\vspace{-1.5ex}
\section{Operational cost semantics}
\label{app:opeationalsemantics}
Following existing work~\cite{OlmedoKKM16}, we provide an operational
semantics for probabilistic programs using pushdown MDPs extended with
a reward function. Interested readers can find more details about MDPs
in the literature~\cite{EsparzaKM05,Baier08}. 

A program state $\sigma: \word{VID} \rightarrow \mathbb{Z}$ is a map
from variable identifiers to integer values. We write
$\sem{e}_{\sigma}$ to denote the value of the expression $e$
in the program state $\sigma$. We write $\sigma[v/x]$ for the program state
$\sigma$ extended with the mapping of $x$ to $v$. For a probability
distribution $\dist{}$, we use $\sem{\dist{} : v}$ to indicate the
probability that $\dist{}$ assigns to value $v$. We use $\Sigma$ to
denote the set of program states.
%
\begin{figure*}[!th]
\centering
\small{
\begin{mathpar}
\RuleToplabel{\textcolor{ACMRed}{S:Term}}
{
}
{
(\downarrow,\sigma) \xrightarrow{\tau,\epsilon,\epsilon,1} (\word{Term},\sigma)
}

\RuleToplabel{\textcolor{ACMRed}{S:Skip}}
{
\word{cmd}(\ell) = \word{skip} \\
\word{fcmd}(\ell) = \ell'
}
{
(\ell,\sigma) \xrightarrow{\tau,\gamma,\gamma,1} (\ell',\sigma)
}

\RuleToplabel{\textcolor{ACMRed}{S:Abort}}
{
  \word{cmd}(\ell) = \word{abort} \\
}
{
  (\ell,\sigma) \xrightarrow{\tau,\gamma,\gamma,1} (\ell,\sigma)
}

\RuleToplabel{\textcolor{ACMRed}{S:Return}}
{
}
{
  (\downarrow,\sigma) \xrightarrow{\tau,\ell',\epsilon,1} (\ell',\sigma)
}

\RuleToplabel{\textcolor{ACMRed}{S:Assert}}
{
  \sem{e}_{\sigma} = \word{true} \\
  \word{cmd}(\ell) = \word{assert} e \\
  \word{fcmd}(\ell) = \ell'
}
{
  (\ell,\sigma) \xrightarrow{\tau,\gamma,\gamma,1} (\ell',\sigma)
}


\RuleToplabel{\textcolor{ACMRed}{S:Tick}}
{
  \word{cmd}(\ell) = \word{tick} \\
  \word{fcmd}(\ell) = \ell'
}
{
(\ell,\sigma) \xrightarrow{\tau,\gamma,\gamma,1} (\ell',\sigma)
}

\RuleToplabel{\textcolor{ACMRed}{S:LoopB}}
{
  \word{cmd}(\ell) = \word{while} e \; c\\\\
  \sem{e}_{\sigma} = \word{true}\\
  \word{fcmd}(\ell) = \ell'
}
{
  (\ell,\sigma) \xrightarrow{\tau,\gamma,\gamma,1} (\ell',\sigma)
}

\RuleToplabel{\textcolor{ACMRed}{S:PIfL}}
{
  \word{cmd}(\ell) = c_1 \oplus_{p} c_2 \\
  \word{fcmd}(\ell) = \ell'
}
{
  (\ell,\sigma) \xrightarrow{\tau,\gamma,\gamma,p} (\ell',\sigma)
}

\RuleToplabel{\textcolor{ACMRed}{S:PIfR}}
{
  \word{cmd}(\ell) = c_1 \oplus_{p} c_2 \\
  \word{scmd}(\ell) = \ell'
}
{
  (\ell,\sigma) \xrightarrow{\tau,\gamma,\gamma,{1}{-}{p}} (\ell',\sigma)
}

\RuleToplabel{\textcolor{ACMRed}{S:LoopE}}
{
  \word{cmd}(\ell) = \word{while} e \; c\\\\
  \sem{e}_{\sigma} = \word{false}\\
  \word{scmd}(\ell) = \ell'
}
{
  (\ell,\sigma) \xrightarrow{\tau,\gamma,\gamma,1} (\ell',\sigma)
}

\RuleToplabel{\textcolor{ACMRed}{S:Assign}}
{
  \word{cmd}(\ell) = \word{id} = e \\
  \word{fcmd}(\ell) = \ell' \\
  \sigma' = \sigma[\sem{e}_{\sigma}/\word{id}]
}
{
  (\ell,\sigma) \xrightarrow{\tau,\gamma,\gamma,1} (\ell',\sigma')
}

\RuleToplabel{\textcolor{ACMRed}{S:Sample}}
{
  \word{cmd}(\ell) = \word{id} = e \word{bop} R \\
  \word{fcmd}(\ell) = \ell' \\\\
  \sem{\dist{R} : v} = p > 0 \\
  \sigma' = \sigma[\sem{e}_{\sigma}\word{bop}v/\word{id}]
}
{
  (\ell,\sigma) \xrightarrow{\tau,\gamma,\gamma,p} (\ell',\sigma')
}

\RuleToplabel{\textcolor{ACMRed}{S:NonDetL}}
{
  \word{cmd}(\ell) = \word{if} \star \; c_1 \word{else} c_2 \\
  \word{fcmd}(\ell) = \ell'
}
{
  (\ell,\sigma) \xrightarrow{\word{Th},\gamma,\gamma,1} (\ell',\sigma)
}

\RuleToplabel{\textcolor{ACMRed}{S:NonDetR}}
{
  \word{cmd}(\ell) = \word{if} \star \; c_1 \word{else} c_2 \\
  \word{scmd}(\ell) = \ell'
}
{
  (\ell,\sigma) \xrightarrow{\word{El},\gamma,\gamma,1} (\ell',\sigma)
}

\RuleToplabel{\textcolor{ACMRed}{S:Call}}
{
  \word{cmd}(\ell) = \word{call} P\\
  \word{fcmd}(\ell) = \ell'
}
{
  (\ell,\sigma) \xrightarrow{\tau,\gamma,\gamma.\ell',1} (\word{init}(\mathcal{D}),\sigma)
}

\RuleToplabel{\textcolor{ACMRed}{S:IfT}}
{
  \sem{e}_{\sigma} = \word{true} \\
  \word{cmd}(\ell) = \word{if} e \; c_1 \word{else} c_2 \\
  \word{fcmd}(\ell) = \ell'
}
{
  (\ell,\sigma) \xrightarrow{\tau,\gamma,\gamma,1} (\ell',\sigma)
}

\RuleToplabel{\textcolor{ACMRed}{S:IfF}}
{
  \sem{e}_{\sigma} = \word{false} \\
  \word{cmd}(\ell) = \word{if} e \; c_1 \word{else} c_2 \\
  \word{scmd}(\ell) = \ell'
}
{
  (\ell,\sigma) \xrightarrow{\tau,\gamma,\gamma,1} (\ell',\sigma)
}
\end{mathpar}}
\vspace{-1.5ex}
\caption{Rules of the probabilistic pushdown transition relation.}
\vspace{-1.5ex}
\label{fig:semantics}
\end{figure*}

Given a program $(c,\mathcal{D})$,
let $L$
be the finite set of all program locations. Let $\ell_0 \in L$
be the initial location of $c$,
$\downarrow$
be the special symbol indicating a termination of a procedure, and let
$\word{Term}$ be a special symbol
for the termination of the whole program. For simplicity, we
assume the existence of auxiliary functions that can be defined
inductively on commands. The function
$\word{init}: C \rightarrow L$
maps procedure bodies to the initial program locations,
the functions $\word{cmd}: L \rightarrow C$
maps locations to their corresponding commands, and
$\word{fcmd}: L \rightarrow L \cup \{\downarrow\}$
and $\word{scmd}: L \rightarrow L \cup \{\downarrow\}$
map locations to their first and second successors, respectively. If a
location $\ell$
has no such successor, then $\word{fcmd}(\ell) = \downarrow$
and $\word{scmd}(\ell) = \downarrow$.

\vspace{-.5ex}
\paragraph{Probabilistic transitions.}

A program configuration is of the form $(\ell, \sigma)$,
where $\sigma \in \Sigma$
is the current program state and
$\ell \in L \cup \{\downarrow, \word{Term}\}$
is a program location indicating the current command or a special
symbol. We view the configurations as states of a pushdown MDP (the
general definition follows) with transitions of the form

\vspace{-2.5ex}
$$(\ell,\sigma) \xrightarrow{\alpha,\gamma,\overline \gamma,p} (\ell',\sigma')$$
\vspace{-2.5ex}

Pushdown MDPs operate on stack of locations that act as return
addresses. In a such a transition, $(\ell,\sigma)$
is the current configuration, $\alpha \in \word{Act}$
is an action, $\gamma$
is the program location on top of the return stack (or $\epsilon$
to indicate an empty stack), $\overline \gamma$
is a finite sequence of program locations to be pushed on the return
stack, $p$
is the probability of the transition, and $(\ell',\sigma')$
is the configuration after the transition. Like in a standard MPD, at
each program configuration, the sum of probabilities of the outgoing
edges labeled with a fixed action is either $0$ or $1$.
The transition rules of our pushdown MDP is given in
Figure~\ref{fig:semantics}. In the rules, we identify a singleton
symbol $\gamma$
with a one element sequence. 

The set of actions is given by
$\word{Act} \defineas \{\word{Th}, \tau, \word{El}\}$
where $\word{Th}$
is the action for the \emph{then} branch, $\word{El}$
is the action for the \emph{else} branch of a non-deterministic choice command, and $\tau$
is the standard action for other transitions. 
In the rule \rul{S:Call} for procedure calls, the location of the
command after the call command is pushed on the location stack and the
control is moved to the first location of the body of the callee procedure. When a
procedure terminates it reaches a state of the form
$(\downarrow,\sigma)$.
If the location stack is non-empty then the rule \rul{S:Return} is
applicable and a return location $\ell'$
is popped from the stack. If the stack is empty then the rule
\rul{S:Term} is applicable and the program terminates.
\vspace{-.5ex}
\paragraph{Resource consumption.} To complete our pushdown MDP
semantics, we have to define the resource consumption of a
computation. We do so by introducing a \emph{reward function} that
assigns a reward to each configuration. The resource consumption of an
execution is then the sum of the rewards of the visited configurations.

For simplicity, we assume that the cost of the program is defined
exclusively by the $\word{tick}(q)$
command, which consumes $q \geq 0$
resource units. Thus we assign the reward $q$
to configurations with locations $\ell$
that contain the command $\word{tick}(q)$
and a reward of $0$ to other configurations.

To facilitate composition, it is handy to define the reward function
with respect to a function
$f: \Sigma \rightarrow \mathbb{R} \cup \{\infty\}$
that defines the reward of a continuation after the program
terminates. So we define the reward of the configuration
$(\word{Term},\sigma)$ to be $f(\sigma)$.
\vspace{-.5ex}
\paragraph{Pushdown MDPs.}

In summary, the semantics of a probabilistic program $(c,\mathcal{D})$
with the initial state $\sigma_0$
is defined by the pushdown MDP
$\mathfrak{M}^{f}_{\sigma_0}\sem{c,\mathcal{D}} =
(S,s_0,\word{Act},P,\mathcal{E},\epsilon,\mathfrak{Re})$, where
\begin{itemize}
  \item $S \defineas \{(\ell,\sigma) \mid \ell \in L \cup \{\downarrow, \word{Term}\}, \sigma \in \Sigma\}$,
  \item $s_0 \defineas (\ell_0, \sigma_0)$, 
  \item $\word{Act} \defineas \{\word{Th}, \tau, \word{El}\}$,
  \item the transition probability relation $P$ is defined by the rules in Figure \ref{fig:semantics},
  \item $\mathcal{E} \defineas L \cup \{\epsilon\}$
  \item $\epsilon$ is the bottom–of–stack symbol, and
  \item $\mathfrak{Re} : S \rightarrow \mathbb{R}^+_0$ is the reward function defined as follows.
  $$
  \mathfrak{Re}(s) \defineas
  \begin{cases}
    f(\sigma)   & \quad \text{if } s = (\word{Term},\sigma) \\
    q           & \quad \text{if } s = (\ell,\sigma) \wedge \word{cmd}(\ell) = \word{tick}(q) \\
    0           & \quad \text{otherwise}
  \end{cases}
  $$
\end{itemize}
\vspace{-.5ex}
\paragraph{Expected resource usage.} 
The expected resource usage of the program $(c,\mathcal{D})$ is the expected reward collected when the associated pushdown MDP 
$\mathfrak{M}^{f}_{\sigma}\sem{(c,\mathcal{D})}$ eventually reaches the set of terminated 
states $(\word{Term},\_)$
from the starting state $s_0 = (\ell_0, \sigma)$, denoted ${\word{ExpRew}}^{\mathfrak{M}^{f}_{\sigma}\sem{c,\mathcal{D}}}(\word{Term})$. Formally, it is defined as 
$$
\begin{cases}
\infty	\quad \text{if } \word{inf}_{\mathfrak{S}} \sum_{\widehat{\pi} \in \Pi(s_0,\word{Term})} \mathbb{P}^{\mathfrak{M}^{f}_{\mathfrak{S},\sigma}\sem{c,\mathcal{D}}}(\widehat{\pi}) < 1 & \\
\word{sup}_{\mathfrak{S}} \sum_{\widehat{\pi} \in \Pi(s_0,\word{Term})} \mathbb{P}^{\mathfrak{M}^{f}_{\mathfrak{S},\sigma}\sem{c,\mathcal{D}}}(\widehat{\pi}) \cdot \mathfrak{Re}(\widehat{\pi}) \quad \text{otherwise} &
\end{cases}
$$
where 
$\mathfrak{S}$ is a \emph{scheduler} for the pushdown MDP mapping a finite sequence of states to an action. Intuitively, it resolves the non-determinism by giving an action given a sequence of states that has been visited. Thus, $\mathfrak{S}$ induces a Markov chain, denoted $\mathfrak{M}^{f}_{\mathfrak{S},\sigma}\sem{c,\mathcal{D}}$, from $\mathfrak{M}^{f}_{\sigma}\sem{c,\mathcal{D}}$. 
$\Pi(s_0,\word{Term})$ denotes the set of all finite paths from $s_0$ to some state $(\word{Term},\sigma')$ in $\mathfrak{M}^{f}_{\mathfrak{S},\sigma}\sem{c,\mathcal{D}}$, $\mathbb{P}^{\mathfrak{M}^{f}_{\sigma}\sem{c,\mathcal{D}}}(\widehat{\pi})$ is the probability of the finite path $\widehat{\pi}$. And $\mathfrak{Re}(\widehat{\pi})$ is the \emph{cumulative reward} collected along $\widehat{\pi}$.
 \vspace{-1.5ex}
\section{Weakest pre-expectation transformer}
\label{app:wpcalculus}
A weakest pre-expectation (WP) calculus~\cite{McIver04,Gretz14} expresses the resource usage of program $(c,\mathcal{D})$ using an \emph{expected runtime transformer} given in continuation-passing style.
The transformer used to analyze our language is defined in Table~\ref{tab:ert};
it operates on the set of functions $\contcost \defineas \{f \mid f : \Sigma \rightarrow \Rplus \cup \{\infty\}\}$, usually called \emph{expectations}. 
In our case, it is a good intuition to think of expectations as mere potential functions.
More precisely, the transformer $\word{ert}[c,\mathcal{D}](f)(\sigma)$ computes the expected number of ticks consumed by the program $(c,\mathcal{D})$ from the input state $\sigma$ and followed by a computation that has an expected tick consumption given by $f$.
Because $f$ is evaluated in the final state and $\word{ert}[c,\mathcal{D}](f)$ is evaluated in the initial state, they are called the \emph{pre-} and \emph{post-expectation}, respectively.
If one chooses the post-expectation $f$ to be the constantly zero function then $\word{ert}[c,\mathcal{D}](\mathbf{0})$ is the expected number of ticks for the program. 
\vspace{-.5ex}
\paragraph{Definition of the expected cost transformer.}
The rules defining the expected cost transformer follow the structure of the command $c$.
We describe the intuition behind a few rules of the transformer.
If $c$ is $\word{tick}(q)$, the expected cost for $c$ followed by a computation of expected cost $f$ is $\mathbf{q} + f$, because the (deterministic) cost of the $\word{tick}(q)$ command is precisely $q$.
For a sequence statement, $\word{ert}[c_1;c_2,\mathcal{D}]$ is defined as the application of $\word{ert}[c_1,\mathcal{D}]$ to the expected value obtained from $\word{ert}[c_2,\mathcal{D}]$; this is the usual continuation-passing style definition.
For a conditional statement, $\word{ert}[\word{if} e \; c_1 \word{else} c_2,\mathcal{D}]$ is defined as the expected cost of the branch that will be executed.
For a non-deterministic choice, $\word{ert}[\word{if} \star \; c_1 \word{else} c_2,\mathcal{D}]$ is the maximum between the expected costs of two branches.
For a probabilistic branching, $\word{ert}[c_1 \oplus_{p} c_2,\mathcal{D}]$ is the weighted sum of the expected costs of two branches.
Similarly, the expected cost of a sampling assignment is defined by considering all the outcomes weighted according the random variable's distribution.
Lastly, the expected cost of loops and procedure calls are expressed using least fixed points.
For procedure calls, an auxiliary cost transformer $\word{ert}[{\cdot}]^{\sharp}_{X}(f)$ is needed (See Section~\ref{subsec:charrecursive}). 
It is parameterized over another expected cost transformer $X:\contcost \to \contcost$.
Its definition is almost identical to the one of the regular expected cost transformer except for procedure calls where $\word{ert}[\word{call} P]^{\sharp}_{X}(f) = X(f)$.
The justification that the fixed points used in the definition exist can be found in the previous work~\cite{McIver04,Gretz14}.
%
\begin{table*}[!th]
\centering
\begin{tabular*}{1.0\textwidth}{@{\extracolsep{\fill}} l l}
\hline
$c$ 												& ${\word{ert}}[c,\mathcal{D}](f)$ \\
\hline
$\word{abort}$										& $\mathbf{0}$ \\									
$\word{skip}$	& $f$ \\
$\word{tick}(q)$ 									& $\mathbf{q} + f$ \\
$\word{assert} e$	& $\sem{e :\word{true}} {\cdot} f$ \\
$\word{id} = e$										& $f[e/\word{id}]$ \\
$\word{id} = e \word{bop} R$					    & $\lambda\sigma.\expt{\dist{R}}{\lambda v.f(\sigma[e \word{bop} v/\word{id}])}$ \\
$\word{if} e \; c_1 \word{else} c_2$                & $\sem{e:\word{true}}{\cdot}{\word{ert}}[c_1,\mathcal{D}](f) + \sem{e:\word{false}}{\cdot}{\word{ert}}[c_2,\mathcal{D}](f)$ \\
$\word{if} \star \; c_1 \word{else} c_2$		    & $\word{max}\{{\word{ert}}[c_1,\mathcal{D}](f), {\word{ert}}[c_2,\mathcal{D}](f)\}$ \\
$c_1 \oplus_{p} c_2$								& $p{\cdot}{\word{ert}}[c_1,\mathcal{D}](f) + (1-p){\cdot}{\word{ert}}[c_2,\mathcal{D}](f)$\\
$c_1; c_2$										    & ${\word{ert}}[c_1,\mathcal{D}]({\word{ert}}[c_2,\mathcal{D}](f))$ \\
$\word{while} e \; c$								& ${\word{lfp}} X.(\sem{e:{\word{true}}}{\cdot}{\word{ert}}[c,\mathcal{D}](X) + \sem{e:\word{false}}{\cdot}f)$\\
$\word{call} P$										& ${\word{lfp}} X.({\word{ert}}[\mathcal{D}(P)]^{\sharp}_{X})(f)$
\end{tabular*}
\caption{Definition of the expected cost transformer {\word{ert}}. $\mathbb{E}_{\dist{R}}[h] \defineas \sum_{v}\prob{R = v}{\cdot}h(v)$ represents the expected value of the random variable $h$ w.r.t the distribution $\dist{R}$. $\word{max}\{f_1,f_2\} \defineas \lambda \sigma.\word{max}\{f_1(\sigma),f_2(\sigma)\}$. ${\word{lfp}} X.F(X)$ is the least fixed point of the function $F$.}
\vspace{-2ex}
\label{tab:ert}
\end{table*}
\vspace{-.5ex}
\paragraph{Example use of the transformer.} 
For example, let $c$ be the body of the example loop $\progname{rdwalk1}$ from Section~\ref{sec:manualana}, and let $f$ be the expectation function $2x$.
Then the expected cost transformer $\word{ert}[c,\mathcal{D}](f)$ is computed as follows.
\begin{align*}
&\word{ert}[x=x-1 \oplus_{3{/}4} x=x+1; \word{tick}(1), {\mathcal{D}}]~(2x) \\
&= \word{ert}[x=x-1 \oplus_{3{/}4} x=x+1, {\mathcal{D}}]~(\word{ert}[\word{tick}(1), \mathcal D]~(2x)) \\
&= \word{ert}[x=x-1 \oplus_{3{/}4} x=x+1, {\mathcal{D}}]~(1 + 2x) \\
&= \textstyle \frac 3 4 \word{ert}[x=x-1, {\mathcal{D}}]~(1 + 2x) + \\
&\;\;\;\;\textstyle \frac 1 4 \word{ert}[x=x+1, {\mathcal{D}}]~(1 + 2x) \\
&= \textstyle \frac 3 4 (1 + 2x - 2) + \frac 1 4 (1 + 2x + 2) \\
&= 2x = f
\end{align*}
In fact, this computation has established that $f$ is an invariant for the body of the loop.
Because the expected cost of loops is defined as a fixed point, finding invariants is critical for the analysis of programs with loops.
In general, however, finding exact invariants like the one we just found is a hard problem.
Instead, one can find a so-called upper invariant, and those provide upper bounds on the loop's cost.
Inferring such upper invariants is precisely the role of the rule \rul{Q:Loop} in our system.
\vspace{-.5ex}
\paragraph{Soundness of the transformer.}

The following theorem states the soundness of the expected cost transformer with respect to the MDP-based semantics.
\begin{theorem}[Soundness of the transformer]
\label{theo:ert}
Let $(c,\mathcal{D})$ be a probabilistic program and $f \in \contcost$ be an expectation.
Then, for every program state $\sigma \in \State$, the following holds 
$$
{\word{ExpRew}}^{\mathfrak{M}^{f}_{\sigma}\sem{c,\mathcal{D}}}(\word{Term}) = \word{ert}[c,\mathcal{D}](f)(\sigma)
$$
\end{theorem}
\begin{proof}
By induction on the command $c$.
The details can be found in~\cite{Kamin16,OlmedoKKM16}.
\end{proof}

\vspace{-1.5ex}
\section{Bounded loops and recursive procedure calls}
\subsection{Bounded loops}
\label{subsec:boundedloops}
The expected cost transformer for \word{while} loops is defined using the fixed point techniques. Alternatively, we can express the expected cost transformer for loops using bounded loops. A bounded loop is obtained by successively unrolling the loop up to a finite number of executions of the loop body.
\begin{lemma}
\label{lem:loopif}
Let $c$ be a command w.r.t a declaration $\mathcal{D}$. Then
$$
{\word{ert}}[\word{while} e \; c,\mathcal{D}] = {\word{ert}}[\word{if} e \; \{c; \word{while} e \; c\} \word{else} \word{skip},\mathcal{D}] 
$$
\end{lemma}
\begin{proof}
For all $f \in \contcost$, consider the following characteristic function
$$
F_{f}(X) = \sem{e:{\word{true}}}{\cdot}{\word{ert}}[c,\mathcal{D}](X) + \sem{e:\word{false}}{\cdot}f
$$
We reason as follows
$$
\begin{array}{ll}
& {\word{ert}}[\word{while} e \; c,\mathcal{D}](f) \\
   & \proofcomment{\text{Table } \ref{tab:ert}} \\
 = & \word{lfp} F_f \\
   & \proofcomment{\text{Definition of fixed point}} \\
 = & F_f(\word{lfp} F_f) \\
 = & \sem{e:{\word{true}}}{\cdot}{\word{ert}}[c,\mathcal{D}](\word{lfp} F_f) + \sem{e:\word{false}}{\cdot}f\\
   & \proofcomment{\text{Table } \ref{tab:ert}} \\
 = & \sem{e:{\word{true}}}{\cdot}{\word{ert}}[c,\mathcal{D}]({\word{ert}}[\word{while} e \; c,\mathcal{D}](f)) + \\
 & \sem{e:\word{false}}{\cdot}f\\
   & \proofcomment{\text{Table } \ref{tab:ert}} \\
 = & \sem{e:{\word{true}}}{\cdot}{\word{ert}}[c; \word{while} e \; c,\mathcal{D}](f) + \sem{e:\word{false}}{\cdot}f\\
   & \proofcomment{\word{ert}[\word{skip},\mathcal{D}](f) = f} \\
 = & \sem{e:{\word{true}}}{\cdot}{\word{ert}}[c; \word{while} e \; c,\mathcal{D}](f) + \\
 & \sem{e:\word{false}}{\cdot}{\word{ert}}[\word{skip},\mathcal{D}](f)\\
   & \proofcomment{\text{Table } \ref{tab:ert}} \\
 = & {\word{ert}}[\word{if} e \; \{c; \word{while} e \; c\} \word{else} \word{skip},\mathcal{D}](f)
\end{array}
$$
\end{proof}
The \emph{$n^{th}$ bounded} execution of a \word{while} loop command, denoted $\word{while}^{k} e \; c$, is defined as follows.
\begin{align*}
{\word{while}}^{0} e \; c \defineas & \word{abort} \\
{\word{while}}^{n+1} e \; c \defineas & \word{if} e \; \{c; {\word{while}}^{n} e \; c\} \word{else} \word{skip}
\end{align*}
The following theorem states that the expected cost transformer can be expressed via the supremum of a sequence of bounded executions.
\begin{theorem}
\label{theo:loopsup}
Let $\mathcal{D}$ be a declaration, the following holds for all $n \in \mathbb{N}$ and $f \in \contcost$.
\begin{align*}
{\word{sup}}_{n}{\word{ert}}[{\word{while}}^{n} \; e \; c, \mathcal{D}](f) = & {\word{lfp}} X.(\sem{e:{\word{true}}}{\cdot}{\word{ert}}[c,\mathcal{D}](X) \\
& + \sem{e:\word{false}}{\cdot}f)
\end{align*}
\end{theorem}
\begin{proof}
For all $f \in \contcost$, consider the following characteristic function
$$
F_{f}(X) = \sem{e:{\word{true}}}{\cdot}{\word{ert}}[c,\mathcal{D}](X) + \sem{e:\word{false}}{\cdot}f
$$
Let $C_n \defineas {\word{while}}^{n} e \; c$, we first prove that for every $n \in \mathbb{N}$, $\word{ert}[C_n,\mathcal{D}](f) = F^{n}_{f}(\mathbf{0})$, where $F^{0}_{f} \defineas \word{id}$ and $F^{k+1}_f \defineas F_f \circ F^{k}_f$. The proof is done by induction on $n$.
\begin{itemize}
\item \emph{Base case.} It is intermediately satisfied because 
$$
\word{ert}[\word{abort},\mathcal{D}](f) = \mathbf{0} = F^{0}_{f}(\mathbf{0})
$$
\item \emph{Induction case.} Assume that $\word{ert}[C_n,\mathcal{D}](f) = F^{n}_{f}(\mathbf{0})$, we reason as follows
$$
\begin{array}{ll}
& \word{ert}[C_{n+1},\mathcal{D}](f) \\
   & \proofcomment{\text{Definition of bounded loops}} \\
 = & \word{ert}[\word{if} e \; \{c; {\word{while}}^{n} e \; c\} \word{else} \word{skip},\mathcal{D}](f) \\
   & \proofcomment{\text{Table } \ref{tab:ert}} \\
 = & \sem{e:{\word{true}}}{\cdot}{\word{ert}}[c; {\word{while}}^{n} e \; c,\mathcal{D}](f) + \sem{e:\word{false}}{\cdot}f\\
   & \proofcomment{\text{Table } \ref{tab:ert}} \\
 = & \sem{e:{\word{true}}}{\cdot}{\word{ert}}[c,\mathcal{D}](\word{ert}[{\word{while}}^{n} e \; c,\mathcal{D}](f)) + \\
 & \sem{e:\word{false}}{\cdot}f\\
   & \proofcomment{\text{By I.H}} \\
 = & \sem{e:{\word{true}}}{\cdot}{\word{ert}}[c,\mathcal{D}](F^{n}_{f}(\mathbf{0})) + \sem{e:\word{false}}{\cdot}f\\
   & \proofcomment{\text{Definition of } F^{k+1}_f} \\
 = & F^{n+1}_f(\mathbf{0})
\end{array}
$$
\end{itemize}
$F_f$ is monotone because of the monotonicity of \word{ert}. Therefore, using Kleene’s Fixed Point Theorem, if holds that 
$$
\begin{array}{ll}
{\word{sup}}_{n}{\word{ert}}[C_n, \mathcal{D}](f) & = {\word{sup}}_{n}F^{n}_{f}(\mathbf{0}) = \word{lfp}X. F_f(X) \\
& = \word{ert}[\word{while} e \; c,\mathcal{D}](f)
\end{array}
$$
\end{proof}
\vspace{-1.0ex}
\subsection{Characterization of procedure call}
\label{subsec:charrecursive}
The detailed definition of the characteristic function for (recursive) procedure call is given in Table \ref{tab:charfunctioncall}. 
\begin{table*}[!th]
\centering
\begin{tabular*}{1.0\textwidth}{@{\extracolsep{\fill}} l l}
\hline
$c$ 																& $\word{ert}[c]^{\sharp}_{X}(f)$ \\
\hline
$\word{skip}$, $\word{weaken}$						& $f$ \\
$\word{abort}$										& $\mathbf{0}$ \\									
$\word{tick}(q)$ 									& $\mathbf{q} + f$ \\
$\word{assert} e$									& $\sem{e:\word{true}}{\cdot}f$ \\
$\word{id} = e$										& $f[e/\word{id}]$ \\
$\word{id} = e \word{bop} R$					    & $\lambda\sigma.\mathbb{E}_{\dist{R}}[\lambda v.f(\sigma[e \word{bop} v/\word{id}])]$ \\
$\word{if} e \; c_1 \word{else} c_2$                & $\sem{e:\word{true}}{\cdot}\word{ert}[c_1]^{\sharp}_{X}(f) + \sem{e:\word{false}}{\cdot}\word{ert}[c_2]^{\sharp}_{X}(f)$ \\
$\word{if} \star \; c_1 \word{else} c_2$		    & $\word{max}\{\word{ert}[c_1]^{\sharp}_{X}(f), \word{ert}[c_2]^{\sharp}_{X}(f)\}$ \\
$c_1 \oplus_{p} c_2$								& $p{\cdot}\word{ert}[c_1]^{\sharp}_{X}(f) + (1-p){\cdot}\word{ert}[c_2]^{\sharp}_{X}(f)$\\
$c_1; c_2$										    & $\word{ert}[c_1]^{\sharp}_{X}(\word{ert}[c_2]^{\sharp}_{X}(f))$ \\
$\word{while} e \; c$								& $\word{lfp} Y.(\sem{e:\word{true}}{\cdot}\word{ert}[c]^{\sharp}_{X}(Y) + \sem{e:\word{false}}{\cdot}f)$\\
$\word{call} P$										& $X(f)$
\end{tabular*}
\caption{Characterization of procedure call.}
\label{tab:charfunctioncall}
\end{table*}
The following lemma says that if a procedure $P$ has a closed body (e.g., there is no procedure calls) then the characteristic function w.r.t the expected cost transformer of the body of $P$ gives exactly the expected cost transformer of the command $\word{call} P$. 
\begin{lemma}
\label{lem:charfunctioncall}
For every command $c$ and closed command $c'$, the following holds where the declaration $\mathcal{D}(P) = c'$.
$$
{\word{ert}}[c]^{\sharp}_{{\word{ert}}[c',\mathcal{D}]} = {\word{ert}}[c,\mathcal{D}]
$$
\end{lemma}
\begin{proof}
The proof is done by induction on the structure of $c$. If $c$ has the form that is different from $\word{call} P$, then by I.H, the definition of the expected cost transformer, and the characteristic function, it follows directly. For instance, we illustrate the proof with the conditional and loop commands. If $c$ is of the form $\word{if} e \; c_1 \word{else} c_2$. For all $f \in \contcost$, we reason as follows.
$$
\begin{array}{ll}
& {\word{ert}}[c]^{\sharp}_{{\word{ert}}[c',\mathcal{D}]}(f) \\
   & \proofcomment{\text{Table } \ref{tab:charfunctioncall}} \\
 = & \sem{e:\word{true}}{\cdot}{\word{ert}}[c_1]^{\sharp}_{{\word{ert}}[c',\mathcal{D}]}(f) + \\
 & \sem{e:\word{false}}{\cdot}{\word{ert}}[c_2]^{\sharp}_{{\word{ert}}[c',\mathcal{D}]}(f) \\
   & \proofcomment{\text{By I.H for } {\word{ert}}[c_1]^{\sharp}_{{\word{ert}}[c',\mathcal{D}]} \text{ and } {\word{ert}}[c_2]^{\sharp}_{{\word{ert}}[c',\mathcal{D}]}} \\
 = & \sem{e:\word{true}}{\cdot}{\word{ert}}[c_1,\mathcal{D}](f) + \sem{e:\word{false}}{\cdot}{\word{ert}}[c_2,\mathcal{D}](f) \\
   & \proofcomment{\text{Table } \ref{tab:ert}} \\
 = & {\word{ert}}[c,\mathcal{D}](f)
\end{array}
$$
If $c$ is of the form $\word{while} e \; c_1$, then we have the following for all $f \in \contcost$.
$$
\begin{array}{ll}
& {\word{ert}}[c]^{\sharp}_{{\word{ert}}[c',\mathcal{D}]}(f) \\
   & \proofcomment{\text{Table } \ref{tab:charfunctioncall}} \\
 = & {\word{lfp}} Y.(\sem{e:\word{true}}{\cdot}{\word{ert}}[c_1]^{\sharp}_{{\word{ert}}[c',\mathcal{D}]}(Y) + \sem{e:\word{false}}{\cdot}f) \\
  & \proofcomment{\text{By I.H for } {\word{ert}}[c_1]^{\sharp}_{{\word{ert}}[c',\mathcal{D}]}(Y)} \\
 = & {\word{lfp}} Y.(\sem{e:\word{true}}{\cdot}{\word{ert}}[c_1,\mathcal{D}](Y) + \sem{e:\word{false}}{\cdot}f) \\
  & \proofcomment{\text{Table } \ref{tab:ert}} \\
 = & {\word{ert}}[c,\mathcal{D}](f)
\end{array}
$$
We consider the case that $c$ is of the form $\word{call} P$. For all $n \geq 1$, ${\word{call}}^{\mathcal{D}}_{n}P = c'[{{\word{call}}^{\mathcal{D}}_{n-1}P}{/}{\word{call}} P] = c'$ since $c'$ is closed. Hence, for all $f \in \contcost$, we have the following.
$$
\begin{array}{ll}
& {\word{ert}}[\word{call}P,\mathcal{D}](f) \\
  & \proofcomment{\text{Theorem } \ref{theo:limitapproximation}} \\
 = & {\word{sup}}_{n}\word{ert}[{\word{call}}^{\mathcal{D}}_{n}P](f) \\
  & \proofcomment{{\word{ert}}[{\word{call}}^{\mathcal{D}}_{0}P](f) = \mathbf{0}; {\word{call}}^{\mathcal{D}}_{n+1}P = c'}\\
 = & {\word{sup}}_{n}\word{ert}[{\word{call}}^{\mathcal{D}}_{n+1}P](f) = {\word{sup}}_{n}\word{ert}[c',\mathcal{D}](f) \\
  & \proofcomment{\text{The supremum of constant sequence}} \\
 = & {\word{ert}}[c',\mathcal{D}](f) \\
  & \proofcomment{\text{Table } \ref{tab:charfunctioncall}} \\
 = & {\word{ert}}[{\word{call}}P]^{\sharp}_{{\word{ert}}[c',\mathcal{D}]}(f)
\end{array}
$$
\end{proof}
\vspace{-1.0ex}
\subsection{Syntactic replacement of procedure call}
\label{subsec:replacement}
Table \ref{tab:functioncall} gives the formal inductive definition of the syntactic replacement of procedure calls $c[{c'}{/}{\word{call}} P]$ on the structure of the command $c$. 
\begin{table*}[!th]
\centering
\begin{tabular*}{1.0\textwidth}{@{\extracolsep{\fill}} l l}
\hline
$c$ 																& $c[{c'}{/}{\word{call}} P]$ \\
\hline
$\word{skip}$, $\word{abort}$, $\word{assert} e$, $\word{weaken}$, 	& $c$ \\
$\word{tick}(q)$, $\word{id} = e$, $\word{id} = e \word{bop} R$     &     \\							
$\word{call} P$														& $c'$ \\
$\word{if} e \; c_1 \word{else} c_2$                				& $\word{if} e \; c_1[{c'}{/}{\word{call}} P] \word{else} c_2[{c'}{/}{\word{call}} P]$ \\
$\word{if} \star \; c_1 \word{else} c_2$		    				& $\word{if} \star \; c_1[{c'}{/}{\word{call}} P] \word{else} c_2[{c'}{/}{\word{call}} P]$ \\
$c_1 \oplus_{p} c_2$												& $c_1[{c'}{/}{\word{call}} P] \oplus_{p} c_2[{c'}{/}{\word{call}} P]$ \\
$c_1; c_2$										    				& $c_1[{c'}{/}{\word{call}} P]; c_2[{c'}{/}{\word{call}} P]$\\
$\word{while} e \; c$												& $\word{while} e \; c[{c'}{/}{\word{call}} P]$
\end{tabular*}
\caption{Syntactic replacement of procedure call.}
\vspace{-2.0ex}
\label{tab:functioncall}
\end{table*}
The following lemma says the property of the replacement by a closed command w.r.t the expected cost transformer \word{ert}, where the declaration $\mathcal{D}(P) = c'$.
\begin{lemma}
\label{lem:functioncall}
For every command $c$ and closed command $c'$, the following holds
$$
\word{ert}[c[{c'}{/}{\word{call}} P],\mathcal{D}] = \word{ert}[c,\mathcal{D}]
$$
\end{lemma}
\begin{proof}
The proof is done by induction on the structure of $c$. If $c$ has the form that is different from $\word{call} P$, then by I.H, the definition of the expected cost transformer, and the syntactic replacement of procedure calls, it follows directly. For instance, we illustrate the proof with the conditional and loop commands. If $c$ is of the form $\word{if} e \; c_1 \word{else} c_2$. We reason as follows.
$$
\begin{array}{ll}
& {\word{ert}}[c[{c'}{/}{\word{call}} P],\mathcal{D}] \\
  & \proofcomment{\text{Table } \ref{tab:functioncall}} \\
 = & {\word{ert}}[{\word{if}} e \; c_1[{c'}{/}{\word{call}} P] \word{else} c_2[{c'}{/}{\word{call}} P],\mathcal{D}] \\
  & \proofcomment{\text{Table } \ref{tab:ert}} \\
 = & \sem{e:\word{true}}{\cdot}{\word{ert}}[c_1[{c'}{/}{\word{call}} P],\mathcal{D}] + \\
 & \sem{e:\word{false}}{\cdot}{\word{ert}}[c_2[{c'}{/}{\word{call}} P],\mathcal{D}] \\
  & \proofcomment{\text{By I.H for } {\word{ert}}[c_1[{c'}{/}{\word{call}} P],\mathcal{D}], {\word{ert}}[c_2[{c'}{/}{\word{call}} P],\mathcal{D}]} \\
 = & \sem{e:\word{true}}{\cdot}{\word{ert}}[c_1,\mathcal{D}] + \sem{e:\word{false}}{\cdot}{\word{ert}}[c_2,\mathcal{D}] \\
  & \proofcomment{\text{Table } \ref{tab:ert}} \\
 = & {\word{ert}}[c,\mathcal{D}]
\end{array}
$$
If $c$ is of the form $\word{while} e \; c_1$, then for all $f \in \contcost$, we have the following.
$$
\begin{array}{ll}
& {\word{ert}}[c[{c'}{/}{\word{call}} P],\mathcal{D}](f) \\
  & \proofcomment{\text{Table } \ref{tab:functioncall}} \\
 = & {\word{ert}}[{\word{while}} e \; c_1[{c'}{/}{\word{call}} P]](f) \\
  & \proofcomment{\text{Table } \ref{tab:ert}} \\
 = & {\word{lfp}} X.(\sem{e:\word{true}}{\cdot}\word{ert}[c_1[{c'}{/}{\word{call}} P],\mathcal{D}](X) + \\
 & \sem{e:\word{false}}{\cdot}f) \\
  & \proofcomment{\text{By I.H for } {\word{ert}}[c_1[{c'}{/}{\word{call}} P],\mathcal{D}]} \\
 = & {\word{lfp}} X.(\sem{e:\word{true}}{\cdot}\word{ert}[c_1,\mathcal{D}](X) + \sem{e:\word{false}}{\cdot}f) \\
  & \proofcomment{\text{Table } \ref{tab:ert}} \\
 = & {\word{ert}}[c,\mathcal{D}](f)
\end{array}
$$

We consider the case that $c$ is of the form $\word{call} P$. For all $n \geq 1$, ${\word{call}}^{\mathcal{D}}_{n}P = c'[{{\word{call}}^{\mathcal{D}}_{n-1}P}{/}{\word{call}} P] = c'$ since $c'$ is closed. Hence, for all $f \in \contcost$, we have the following.
$$
\begin{array}{ll}
& {\word{ert}}[\word{call}P,\mathcal{D}](f) \\
  & \proofcomment{\text{Theorem } \ref{theo:limitapproximation}} \\
 = & {\word{sup}}_{n}\word{ert}[{\word{call}}^{\mathcal{D}}_{n}P](f) \\
  & \proofcomment{{\word{ert}}[{\word{call}}^{\mathcal{D}}_{0}P](f) = \mathbf{0}; {\word{call}}^{\mathcal{D}}_{n+1}P = c'}\\
 = & {\word{sup}}_{n}\word{ert}[{\word{call}}^{\mathcal{D}}_{n+1}P](f) = {\word{sup}}_{n}\word{ert}[c',\mathcal{D}](f) \\
  & \proofcomment{\text{The supremum of constant sequence}} \\
 = & {\word{ert}}[c',\mathcal{D}](f) \\
  & \proofcomment{{\word{call}}P[{c'}{/}{\word{call}P}] = c'} \\
 = & {\word{ert}}[\word{call}P[{c'}{/}{\word{call}P}],\mathcal{D}](f)
\end{array}
$$
\end{proof}
\vspace{-1.0ex}
\subsection{Finite approximation for procedure call}
\label{subsec:supfixpoint}
The expected cost transformer for (recursive) procedure calls is defined using the fixed point techniques. Alternatively, we borrow the approach in~\cite{OlmedoKKM16} to express the expected cost  transformer for procedure calls as the limit of its finite approximations, or truncations, and show that they are equivalent. Let $P$ be a procedure, the \emph{$n^{th}$ inlining} call of $P$ w.r.t a declaration $\mathcal{D}$, denoted ${\word{call}}^{\mathcal{D}}_{n}P$, is defined as follows.
\begin{align*}
{\word{call}}^{\mathcal{D}}_{0}P 	\defineas & \word{abort} \\
{\word{call}}^{\mathcal{D}}_{n+1}P 	\defineas & \mathcal{D}(P)[{\word{call}}^{\mathcal{D}}_{n}P{/}{\word{call}P}]
\end{align*}
where $c[{c'}{/}{\word{call}} P]$ is defined in Table \ref{tab:functioncall}. Intuitively, ${\word{call}}^{\mathcal{D}}_{n}P$ is a sequence of approximations of $\word{call}P$ where the ''worst'' approximation ${\word{call}}^{\mathcal{D}}_{0}P$, while the approximation gets more precise when $n$ increases. 

The expected cost transformer for procedure calls using the limit of its finite approximation is equivalent to the transformer defined using the fixed point techniques as stated by the following theorem.
\begin{theorem}[Limit of finite approximations]
\label{theo:limitapproximation}
Let $P$ be a procedure w.r.t a declaration $\mathcal{D}$, the following holds for all $n \in \mathbb{N}$.
\begin{align*}
{\word{sup}}_{n}{\word{ert}}[{\word{call}}^{\mathcal{D}}_{n}P] = & \word{lfp} X.(\word{ert}[\mathcal{D}(P)]^{\sharp}_{X})
\end{align*}
\end{theorem}
\begin{proof}
For all $f \in \contcost$, consider the following characteristic function
$$
F_f(X) = \word{ert}[\mathcal{D}(P)]^{\sharp}_{X}(f)
$$
where $\word{ert}[c]^{\sharp}_{X}(f)$ is defined in Table \ref{tab:charfunctioncall}. We first prove that for every $n \in \mathbb{N}$, $\word{ert}[{\word{call}}^{\mathcal{D}}_{n}P](f) = F^{n}_f(\mathbf{0})$, where $F^{0}_f \defineas \word{id}$ and $F^{k+1}_f \defineas F_f \circ F^{k}_f$. The proof is done by induction on $n$.
\begin{itemize}
\item \emph{Base case.} It is intermediately satisfied because 
$$
\word{ert}[\word{abort},\mathcal{D}](f) = \mathbf{0} = F^{0}_{f}(\mathbf{0})
$$
\item \emph{Induction case.} Assume that $\word{ert}[{\word{call}}^{\mathcal{D}}_{n}P](f) = F^{n}_f(\mathbf{0})$, we reason as follows
$$
\begin{array}{ll}
& \word{ert}[{\word{call}}^{\mathcal{D}}_{n+1}P](f) \\
   & \proofcomment{\text{Definition of inlining}} \\
 = & \word{ert}[\mathcal{D}(P)[{\word{call}}^{\mathcal{D}}_{n}P{/}\word{call}P]](f) \\
   & \proofcomment{\text{Lemma \ref{lem:functioncall} for closed command } {\word{call}}^{\mathcal{D}}_{n}P} \\
 = & \word{ert}[\mathcal{D}(P), \mathcal{D}](f) \\
   & \proofcomment{\text{Lemma \ref{lem:charfunctioncall} for closed command } {\word{call}}^{\mathcal{D}}_{n}P} \\
 = & {\word{ert}}[\mathcal{D}(P)]^{\sharp}_{{\word{ert}}[{\word{call}}^{\mathcal{D}}_{n}P]}(f) \\
   & \proofcomment{\text{By I.H}} \\
 = & {\word{ert}}[\mathcal{D}(P)]^{\sharp}_{F^{n}_f(\mathbf{0})}(f) \\
   & \proofcomment{\text{Definition of } F_f} \\
 = & F_f(F^{n}_f(\mathbf{0})) \\
   & \proofcomment{\text{Definition of } F^{n+1}_f} \\
 = & F^{n+1}_f(\mathbf{0})
\end{array}
$$
\end{itemize}
$F_f$ is monotone because of the monotonicity of ${\word{ert}}^{\sharp}_{X}[\cdot]$. Therefore, using Kleene’s Fixed Point Theorem, if holds that 
$$
{\word{sup}}_{n}{\word{ert}}[{\word{call}}^{\mathcal{D}}_{n}P] = \word{lfp} X.(\word{ert}[\mathcal{D}(P)]^{\sharp}_{X}) = \word{ert}[\word{call} P, \mathcal{D}]
$$
\end{proof}
 \vspace{-1.5ex}
\section{Proof of the soundness}
\label{app:soundnessproof}
\subsection{Potential relax}
\label{app:relax}
\begin{lemma}[Potential relax]
\label{lem:relax}
For each program state $\sigma$ such that $\sigma \models \context$, if $Q \succeq_{\context} Q'$, then $\Phi_{Q}(\sigma) \geq \Phi_{Q'}(\sigma)$.
\end{lemma}
\begin{proof}
Lemma \ref{lem:relax} is proved using the rule \textcolor{ACMRed}{\textsc{Q:Relax}} as follows. Let $B = (b_1,\cdots,b_n)^{\intercal}$ be the set of all based functions, then we have 
$$
\begin{array}{lclcl}
\Phi_Q(\sigma) 		& = & \langle Q{\cdot}B \rangle (\sigma) 	& = & \sum^{n}_{i = 1}q_i{\cdot}b_i(\sigma) \\
\Phi_{Q'}(\sigma) 	& = & \langle Q'{\cdot}B \rangle (\sigma) 	& = & \sum^{n}_{i = 1}q'_i{\cdot}b_i(\sigma) \\
\Phi_{F_k}(\sigma) 	& = & \langle F_k{\cdot}B \rangle (\sigma) 	& = & \sum^{n}_{i = 1}f^k_i{\cdot}b_i(\sigma)
\end{array}
$$
Consider any program state $\sigma$ such that $\sigma \models \context$, we reason as follows
$$
\begin{array}{ll}
& (q'_1,\cdots,q'_n) \\
   & \proofcomment{Q' = Q - F\vec{u} \text{ and matrix multiplication}} \\
 = & (q_1,\cdots,q_n) - (\sum^{N}_{k=1}f^k_1{\cdot}u_k,\cdots,\sum^{N}_{k=1}f^k_n{\cdot}u_k) \\
   & \proofcomment{\text{Vector subtraction}} \\
 = & ((q_1 - \sum^{N}_{k=1}f^k_1{\cdot}u_k),\cdots,(q_n - \sum^{N}_{k=1}f^k_n{\cdot}u_k))
\end{array}
$$
$$
\begin{array}{ll}
& \Phi_{Q'}(\sigma) \\
   & \proofcomment{\text{Observation above}} \\
 = & \sum^{n}_{i=1}(q_i - \sum^{N}_{k=1}f^k_i{\cdot}u_k){\cdot}b_i(\sigma) \\
   & \proofcomment{\text{Algebra}} \\
 = & \sum^{n}_{i=1}(q_i{\cdot}b_i(\sigma) - \sum^{N}_{k=1}f^k_i{\cdot}u_k{\cdot}b_i(\sigma)) \\
 = & \sum^{n}_{i=1}q_i{\cdot}b_i(\sigma) - \sum^{n}_{i=1}\sum^{N}_{k=1}f^k_i{\cdot}u_k{\cdot}b_i(\sigma) \\
 = & \sum^{n}_{i=1}q_i{\cdot}b_i(\sigma) - \sum^{N}_{k=1}(\sum^{n}_{i=1}f^k_i{\cdot}u_k{\cdot}b_i(\sigma)) \\
   & \proofcomment{\forall i.\sigma. \Phi_{F_i}(\sigma) \geq 0; \forall i. u_i \geq 0} \\
 \leq & \sum^{n}_{i=1}q_i{\cdot}b_i(\sigma) = \Phi_{Q}(\sigma)
\end{array}
$$
\end{proof}
\vspace{-1.0ex}
\subsection{Soundness of the automatic analysis}
\label{app:soundness}
Let $(c,\mathcal{D})$ be a program, the proof is done by induction on the program structure and the derivation using the inference rules.

\vspace{-.5ex}
\paragraph{Skip} 
For all program states $\sigma \in \Sigma$, we have 
$$
\begin{array}{ll}
& \word{ert}[\word{skip},\mathcal{D}](\mathcal{T}(\context;Q))(\sigma) \\
   & \proofcomment{\text{Table } \ref{tab:ert}} \\
 = & \mathcal{T}(\context;Q)(\sigma) \leq \mathcal{T}(\context;Q)(\sigma)
\end{array}
$$ 

\vspace{-.5ex}
\paragraph{Abort}
For all program states $\sigma \in \Sigma$, we have 
$$
\begin{array}{ll}
& \word{ert}[\word{abort},\mathcal{D}](\mathcal{T}(\context;Q))(\sigma) \\
   & \proofcomment{\text{Table } \ref{tab:ert}} \\
 = & \mathbf{0}(\sigma) \leq \mathcal{T}(\context;Q)(\sigma)
\end{array}
$$ 

\vspace{-.5ex}
\paragraph{Assert} 
Consider any program state $\sigma \in \Sigma$, if $\sigma \not \models \context$, then it follows 
$$
\begin{array}{ll}
& \word{ert}[\word{assert} e,\mathcal{D}](\mathcal{T}(\context \wedge e;Q))(\sigma) \\
   & \proofcomment{\text{Definition of translation function}} \\ 
 \leq & \infty = \mathcal{T}(\context;Q)
\end{array}
$$
If $\sigma \models \context$, we have 
$$
\begin{array}{ll}
& \word{ert}[\word{assert} e,\mathcal{D}](\mathcal{T}(\context \wedge e;Q))(\sigma) \\
   & \proofcomment{\text{Table } \ref{tab:ert}} \\
 = & \sem{e:\word{true}}{\cdot}\mathcal{T}(\context;Q)(\sigma) \\
 \leq & \mathcal{T}(\context;Q)(\sigma)
\end{array}
$$

\vspace{-.5ex}
\paragraph{Weaken} 
Consider any state $\sigma \in \Sigma$, if $\sigma \models \context'_2$ then $\sigma \models \context'_1$ because $\context'_2 \models \context'_1$. We have $Q'_2 \succeq_{\context'_2} Q'_1$, it holds that 
$$
\begin{array}{ll}
& \mathcal{T}(\context'_2;Q'_2)(\sigma) \\
   & \proofcomment{\text{Definition of translation function}} \\
 = & \Phi_{Q'_2}(\sigma) \\
   & \proofcomment{\text{Lemma \ref{lem:relax}}} \\
 \geq & \Phi_{Q'_1}(\sigma) = \mathcal{T}(\context'_1;Q'_1)(\sigma)
\end{array}
$$ 
If $\sigma \not \models \context'_2$, then we get 
$$
\begin{array}{ll}
& \mathcal{T}(\context'_2;Q'_2)(\sigma) \\
   & \proofcomment{\text{Definition of translation function}} \\
 = & \infty \geq \mathcal{T}(\context'_1;Q'_1)(\sigma)
\end{array}
$$
Therefore, we have $\mathcal{T}(\context'_2;Q'_2)(\sigma) \geq \mathcal{T}(\context'_1;Q'_1)(\sigma)$ for all $\sigma \in \Sigma$. Similarly, it follows $\mathcal{T}(\context_1;Q_1)(\sigma) \geq \mathcal{T}(\context_2;Q_2)(\sigma)$ for all $\sigma \in \Sigma$. For all states $\sigma \in \Sigma$, we get 
$$
\begin{array}{ll}
& \mathcal{T}(\context_1;Q_1)(\sigma) \\
      & \proofcomment{\text{Observation above}} \\
 \geq & \mathcal{T}(\context_2;Q_2)(\sigma) \\
      & \proofcomment{\text{By I.H for } C} \\
 \geq & \word{ert}[c,\mathcal{D}](\mathcal{T}(\context'_2;Q'_2))(\sigma) \\
      & \proofcomment{\text{Observation above and monotonicity of \word{ert}}} \\
 \geq & \word{ert}[c,\mathcal{D}](\mathcal{T}(\context'_1;Q'_1))(\sigma)
\end{array}
$$

\vspace{-.5ex}
\paragraph{Tick}
Consider any state $\sigma \in \Sigma$ such that $\sigma \not \models \context$, then 
$$
\begin{array}{ll}
& \mathcal{T}(\context;Q)(\sigma) \\
   & \proofcomment{\text{Definition of translation function}} \\
 = & \infty \geq \word{ert}[\word{tick}(q),\mathcal{D}](\mathcal{T}(\context;Q - q))(\sigma)
\end{array}
$$ 
For all states $\sigma \in \Sigma$ such that $\sigma \models \context$, we have 
$$
\begin{array}{ll}
& \word{ert}[\word{tick}(q),\mathcal{D}](\mathcal{T}(\context;Q - q))(\sigma) \\
   & \proofcomment{\text{Table } \ref{tab:ert}} \\
 = 		& q + \mathcal{T}(\context;Q - q)(\sigma) \\
   & \proofcomment{\text{Definition of translation function}} \\
 = 		& q + \Phi_{Q}(\sigma) - q \leq \mathcal{T}(\context;Q)(\sigma)
\end{array}
$$ 

\vspace{-.5ex}
\paragraph{Assignment}
Suppose $c$ is of the form $x = e$. Hence, the automatic analysis derivation ends with an application of the rule \textcolor{ACMRed}{\textsc{Q:Assign}}. Consider any program state $\sigma$ such that $\sigma \not \models \context[e{/}x]$, 
$$
\begin{array}{ll}
& \mathcal{T}(\context[e{/}x];Q)(\sigma) \\
   & \proofcomment{\text{Definition of translation function}} \\
 = & \infty \geq \word{ert}[c,\mathcal{D}](\mathcal{T}(\context;Q'))(\sigma)
\end{array}
$$
If $\sigma \models \context[e{/}x]$, then we reason as follows
$$
\begin{array}{ll}
& {\word{ert}}[x = e, \mathcal{D}](\mathcal{T}(\context;Q'))(\sigma) \\
   & \proofcomment{\text{Table } \ref{tab:ert}} \\
 = & \mathcal{T}(\context;Q')(\sigma[e{/}x]) \\
   & \proofcomment{\text{Definition of translation function}} \\
 = & \word{max}(\context[e{/}x](\sigma),\Phi_{Q'}(\sigma[e{/}x])) = \Phi_{Q'}(\sigma[e{/}x]) \\
   & \proofcomment{\text{Definition of potential function and rule \textcolor{ACMRed}{\textsc{Q:Assign}}}} \\
 = & \sum_{j \in \mathcal S_{x=e}}q'_j{\cdot}b_j[e{/}x](\sigma) + \sum_{j \not \in \mathcal S_{x=e}}0{\cdot}b_j[e{/}x](\sigma) \\
   & \proofcomment{\text{Rule \textcolor{ACMRed}{\textsc{Q:Assign}}}} \\
 = & \sum_{j \in \mathcal S_{x=e}}q'_j{\cdot}\sum_{i}a_{i,j}{\cdot}b_i(\sigma) = \Phi_{Q}(\sigma) \\
   & \proofcomment{\text{Definition of translation function}} \\
 = & \mathcal{T}(\context[e{/}x];Q)
\end{array}
$$
\vspace{-.5ex}
\paragraph{Sampling} 
Suppose $c$ is of the form $x = e \word{bop} R$, where $R$ is a random variable distributed by the probability distribution $\dist{R}$ and $R \in [a, b]$. Hence, the automatic analysis derivation ends with an application of the rule \textcolor{ACMRed}{\textsc{Q:Sample}}. Let $\context^{i} = \context'[e \word{bop} v_i{/}x]$, then for all $i$, we have $\context \models \context_i$. Consider any program state $\sigma$ such that $\sigma \not \models \context$, 
$$
\begin{array}{ll}
& \mathcal{T}(\context;Q)(\sigma) \\
   & \proofcomment{\text{Definition of translation function}} \\
 = & \infty \geq \word{ert}[c,\mathcal{D}](\mathcal{T}(\context';Q'))(\sigma)
\end{array}
$$
If $\sigma \models \context$, then for all $i$, we have $\sigma \models \context^{i}$ and the following hold
$$
\begin{array}{ll}
& \mathcal{T}(\context;Q_i)(\sigma) \\
   & \proofcomment{\text{Definition of translation function}} \\
 = & \Phi_{Q_i}(\sigma) \\
   & \proofcomment{\text{By I.H for assignment}} \\
 \geq & {\word{ert}}[x = e \word{bop} v_i,\mathcal{D}](\mathcal{T}(\context';Q'))(\sigma) \\
   & \proofcomment{\text{Table } \ref{tab:ert}} \\
 = & \mathcal{T}(\context';Q')(\sigma[e \word{bop} v_i{/}x]) \\
   & \proofcomment{\text{Definition of translation function}} \\
 = & \word{max}(\context^{i}(\sigma), \Phi_{Q'}(\sigma[e \word{bop} v_i{/}x])) \\
 = & \Phi_{Q'}(\sigma[e \word{bop} v_i{/}x]) \\
\end{array}
$$
$$
\begin{array}{ll}
& {\word{ert}}[x = e \word{bop} R, \mathcal{D}](\mathcal{T}(\context';Q'))(\sigma) \\
   & \proofcomment{\text{Table } \ref{tab:ert}} \\
 = & \lambda\sigma.\expt{\dist{R}}{\lambda v.\mathcal{T}(\context';Q')(\sigma[e \word{bop} v/x])} \\
   & \proofcomment{\text{Definition of expectation}} \\
 = & \sum_{i}\sem{\dist{R} : v_i}{\cdot}(\word{max}(\context^{i}(\sigma),\Phi_{Q'}(\sigma[e \word{bop} v_i/x]))) \\
   & \proofcomment{\text{Definition of translation function}} \\
 = & \sum_{i}p_i{\cdot}\Phi_{Q'}(\sigma[e \word{bop} v_i/x]) \\
   & \proofcomment{\text{Observation above}} \\
 \leq & \sum_{i}p_i{\cdot}\Phi_{Q_i}(\sigma) \\
   & \proofcomment{\text{By rule \textcolor{ACMRed}{\textsc{Q:Sample}}}} \\
 = & \Phi_Q(\sigma) = \mathcal{T}(\context;Q)(\sigma)
\end{array}
$$
\vspace{-.5ex}
\paragraph{If}
Suppose $c$ is of the form $\word{if} e \; c_1 \word{else} c_2$, thus the automatic analysis derivation ends with an application of the rule \textcolor{ACMRed}{\textsc{Q:If}}. 
Consider any program state $\sigma$ such that $\sigma \not \models \context$, 
$$
\begin{array}{ll}
& \mathcal{T}(\context;Q)(\sigma) \\
   & \proofcomment{\text{Definition of translation function}} \\
 = & \infty \geq \word{ert}[c,\mathcal{D}](\mathcal{T}(\context';Q'))(\sigma)
\end{array}
$$
If $\sigma \models \context \wedge e$, then we have 
$$
\begin{array}{ll}
& \mathcal{T}(\context;Q)(\sigma) \\
   & \proofcomment{\text{Definition of translation function}} \\
 =  	 & \mathcal{T}(\context \wedge e;Q)(\sigma) \\
& \proofcomment{\text{By I.H for } c_1} \\
 \geq & \word{ert}[c_1,\mathcal{D}](\mathcal{T}(\context';Q'))(\sigma) \\
& \proofcomment{\sem{e:\word{false}}(\sigma) = 0; \sem{e:\word{true}}(\sigma) = 1} \\
 = 	 & \sem{e:\word{true}}(\sigma){\cdot}\word{ert}[c_1,\mathcal{D}](\mathcal{T}(\context';Q'))(\sigma) + \\
&\sem{e:\word{false}}(\sigma){\cdot}\word{ert}[c_2,\mathcal{D}](\mathcal{T}(\context';Q'))(\sigma) \\
   & \proofcomment{\text{Table } \ref{tab:ert}} \\
 = 	 & \word{ert}[c,\mathcal{D}](\mathcal{T}(\context';Q'))(\sigma)
\end{array}
$$
If $\sigma \models \context \wedge \neg e$, then we get 
$$
\begin{array}{ll}
& \mathcal{T}(\context;Q)(\sigma) \\
   & \proofcomment{\text{Definition of translation function}} \\
 =  & \mathcal{T}(\context \wedge \neg e;Q)(\sigma) \\
& \proofcomment{\text{By I.H for } c_2} \\
\geq 	& \word{ert}[c_2,\mathcal{D}](\mathcal{T}(\context';Q'))(\sigma) \\
& \proofcomment{\sem{e:\word{false}}(\sigma) = 1; \sem{e:\word{true}}(\sigma) = 0} \\
 = 	& \sem{e:\word{true}}(\sigma){\cdot}\word{ert}[c_1,\mathcal{D}](\mathcal{T}(\context';Q'))(\sigma) + \\
& \sem{e:\word{false}}(\sigma){\cdot}\word{ert}[c_2,\mathcal{D}](\mathcal{T}(\context';Q'))(\sigma) \\
   & \proofcomment{\text{Table } \ref{tab:ert}} \\
 = 		& \word{ert}[c,\mathcal{D}](\mathcal{T}(\context';Q'))(\sigma)
\end{array}
$$

\vspace{-.5ex}
\paragraph{Nondeterministic branching} 
Suppose $c$ is of the form $\word{if} \star \; c_1 \word{else} c_2$, thus the automatic analysis derivation ends with an application of the rule \textcolor{ACMRed}{\textsc{Q:NonDet}}. 
Consider any program state $\sigma \in \Sigma$, we get 
$$
\begin{array}{ll}
& \mathcal{T}(\context;Q)(\sigma) \\
& \proofcomment{\text{By I.H for } c_1} \\
 \geq & \word{ert}[c_1,\mathcal{D}](\mathcal{T}(\context';Q'))(\sigma) \\
& \proofcomment{\text{By I.H for } c_2} \\
 \geq & \word{ert}[c_2,\mathcal{D}](\mathcal{T}(\context';Q'))(\sigma) \\
& \proofcomment{\text{Algebra}} \\
 \geq & \word{max}\{\word{ert}[c_1,\mathcal{D}](\mathcal{T}(\context';Q'))(\sigma), \\
 & \word{ert}[c_2,\mathcal{D}](\mathcal{T}(\context';Q'))(\sigma)\} \\
   & \proofcomment{\text{Table } \ref{tab:ert}} \\
 = & \word{ert}[c,\mathcal{D}](\mathcal{T}(\context';Q'))(\sigma)
\end{array}
$$

\vspace{-.5ex}
\paragraph{Probabilistic branching} 
Suppose $c$ is of the form $c_1 \oplus_{p} c_2$, thus the automatic analysis ends with an application of the rule \textcolor{ACMRed}{\textsc{Q:PIf}}. Consider any program state $\sigma$ such that $\sigma \not \models \context$, we have
$$
\begin{array}{ll}
& \mathcal{T}(\context;Q)(\sigma) \\
   & \proofcomment{\text{Definition of translation function}} \\
 = & \infty \geq \word{ert}[c,\mathcal{D}](\mathcal{T}(\context';Q'))(\sigma)
\end{array}
$$ 
If $\sigma \models \context$, then it follows 
$$
\begin{array}{ll}
& \mathcal{T}(\context;Q)(\sigma) \\
   & \proofcomment{\text{Definition of translation function and rule \textcolor{ACMRed}{\textsc{Q:PIf}}}} \\
 = & \Phi_{Q}(\sigma) = {p}{\cdot}\Phi_{Q_1}(\sigma) + {(1}{-}{p)}{\cdot}\Phi_{Q_2}(\sigma) \\
 = & {p}{\cdot}\mathcal{T}(\context;Q_1)(\sigma) + {(1}{-}{p)}{\cdot}\mathcal{T}(\context;Q_2)(\sigma)\\
& \proofcomment{\text{By I.H for } c_1 \text{ and } c_2} \\
 \geq & {p}{\cdot}\word{ert}[c_1,\mathcal{D}](\mathcal{T}(\context';Q'))(\sigma) + \\
 & {(1}{-}{p)}{\cdot}\word{ert}[c_2,\mathcal{D}](\mathcal{T}(\context';Q'))(\sigma)\\
   & \proofcomment{\text{Table } \ref{tab:ert}} \\
 = & \word{ert}[c,\mathcal{D}](\mathcal{T}(\context';Q'))(\sigma)
\end{array}
$$

\vspace{-.5ex}
\paragraph{Sequence} 
Suppose $c$ is of the form $c_1; c_2$, thus the automatic analysis derivation ends with an application of the rule \textcolor{ACMRed}{\textsc{Q:Seq}}. 
Consider any program state $\sigma \in \Sigma$, we get 
$$
\begin{array}{ll}
& \mathcal{T}(\context;Q)(\sigma) \\
& \proofcomment{\text{By I.H for } c_1} \\
 \geq & \word{ert}[c_1,\mathcal{D}](\mathcal{T}(\context';Q'))(\sigma)\\
& \proofcomment{\text{By I.H for } c_2 \text{ and monotonicity of \word{ert}}} \\
 \geq & \word{ert}[c_1,\mathcal{D}](\word{ert}[c_2,\mathcal{D}](\mathcal{T}(\context'';Q'')))(\sigma)\\
   & \proofcomment{\text{Table } \ref{tab:ert}} \\
 = & \word{ert}[c,\mathcal{D}](\mathcal{T}(\context'';Q''))(\sigma)
\end{array}
$$

\vspace{-.5ex}
\subsubsection*{Loop}
Suppose $c$ is of the form $\word{while} e \; c_1$, thus the automatic analysis ends with an application of the rule \textcolor{ACMRed}{\textsc{Q:Loop}}. For any program state $\sigma$, we consider the characteristic function $F_{\mathcal{T}(\context \wedge \neg e;Q)}(\mathcal{T}(\context;Q))$ defined as 
$$
\sem{e:\word{true}}{\cdot}\word{ert}[c_1,\mathcal{D}](\mathcal{T}(\context;Q)) + \sem{e:\word{false}}{\cdot}\mathcal{T}(\context \wedge \neg e;Q)
$$ 
If $\sigma \not \models \context$, then 
$$
\begin{array}{ll}
& F_{\mathcal{T}(\context \wedge \neg e;Q)}(\mathcal{T}(\context;Q))(\sigma) \\
   & \proofcomment{\text{Definition of translation function}} \\
 \leq & \infty = \mathcal{T}(\context;Q)(\sigma)
\end{array}
$$ 
If $\sigma \models \context \wedge \neg e$, then we have
$$
\begin{array}{ll}
& F_{\mathcal{T}(\context \wedge \neg e;Q)}(\mathcal{T}(\context;Q))(\sigma) \\
 = & \sem{e:\word{true}}(\sigma){\cdot}\word{ert}[c_1,\mathcal{D}](\mathcal{T}(\context;Q))(\sigma) + \\
& \sem{e:\word{false}}(\sigma){\cdot}\mathcal{T}(\context \wedge \neg e;Q)(\sigma) \\
& \proofcomment{\sem{e:\word{false}}(\sigma) = 1; \sem{e:\word{true}}(\sigma) = 0} \\
 = & \mathcal{T}(\context \wedge \neg e;Q)(\sigma) \\
   & \proofcomment{\text{Definition of translation function}} \\
 = & \Phi_{Q}(\sigma) = \mathcal{T}(\context;Q)(\sigma)
\end{array}
$$
If $\sigma \models \context \wedge e$, we get
$$
\begin{array}{ll}
& \mathcal{T}(\context \wedge e;Q))(\sigma) \\
   & \proofcomment{\text{Definition of translation function}} \\
 = & \mathcal{T}(\context ;Q))(\sigma) = \Phi_{Q}(\sigma)\\
& \proofcomment{\text{By I.H for } c_1} \\
 \geq & \word{ert}[c_1,\mathcal{D}](\mathcal{T}(\context;Q))(\sigma)\\
& \proofcomment{\sem{e:\word{false}}(\sigma) = 0; \sem{e:\word{true}}(\sigma) = 1} \\
 = & F_{\mathcal{T}(\context \wedge \neg e;Q)}(\mathcal{T}(\context;Q))(\sigma)
\end{array}
$$
Thus, for all program states $\sigma \in \Sigma$, we have 
$$
\begin{array}{rrl}
& F_{\mathcal{T}(\context \wedge \neg e;Q)}(\mathcal{T}(\context;Q))(\sigma) & \leq \mathcal{T}(\context;Q)(\sigma)\\

& \proofcomment{\text{Definition of } \sqsubseteq} & \\
\Leftrightarrow & F_{\mathcal{T}(\context \wedge \neg e;Q)}(\mathcal{T}(\context;Q)) & \sqsubseteq \mathcal{T}(\context;Q) \\

& \proofcomment{\text{Park's Theorem\footnotemark~\cite{Wechler92}}} & \\
\Rightarrow & \word{lfp}F_{\mathcal{T}(\context \wedge \neg e;Q)} & \sqsubseteq \mathcal{T}(\context;Q) \\
& \proofcomment{\text{Table } \ref{tab:ert}} & \\
\Leftrightarrow & \word{ert}[c,\mathcal{D}](\mathcal{T}(\context \wedge \neg e;Q)) & \sqsubseteq \mathcal{T}(\context;Q) \\
& \proofcomment{\text{Definition of } \sqsubseteq} & \\
\Leftrightarrow & \word{ert}[c,\mathcal{D}](\mathcal{T}(\context \wedge \neg e;Q'))(\sigma) & \leq \mathcal{T}(\context;Q)(\sigma)
\end{array}
\footnotetext{If $H: \mathcal{D} \rightarrow \mathcal{D}$ is a continuous function over an $\omega$-cpo $(\mathcal{D},\sqsubseteq)$ with bottom element, then $H(d) \sqsubseteq d$ implies $\word{lfp}H \sqsubseteq d$ for all $d \in \mathcal{D}$.}
$$
Following Theorem \ref{theo:loopsup}, the expected cost transformer for loops can be expressed as 
$$
\word{ert}[\word{while} e \; c_1,\mathcal{D}] = {\word{sup}}_{n}{\word{ert}}[{\word{while}}^{n} \; e \; c_1, \mathcal{D}]
$$
Alternatively, we can prove the soundness of the rule \textcolor{ACMRed}{\textsc{Q:Loop}} by showing that for all program states $\sigma$
$$
\forall n. {\word{ert}}[{\word{while}}^{n} \; e \; c_1, \mathcal{D}](\mathcal{T}(\context \wedge \neg e;Q))(\sigma) \leq \mathcal{T}(\context;Q)(\sigma)
$$

\vspace{-.5ex}
\subsubsection*{Procedure call}
Suppose $c$ is of the form $\word{call} P$, thus the automatic analysis ends with applications of the rules \textcolor{ACMRed}{\textsc{Q:Call}} and \textcolor{ACMRed}{\textsc{ValidCtx}}. Since by Theorem \ref{theo:limitapproximation}, $\word{ert}[{\word{call}P,\mathcal{D}}] = {\word{sup}}_{n}{\word{ert}}[{\word{call}}^{\mathcal{D}}_{n}P]$, where the $n^{th}$-inlining of procedure call ${\word{call}}^{\mathcal{D}}_{n}P$ is defined in Section \ref{subsec:supfixpoint}. To prove that 
$$\word{ert}[{\word{call}P,\mathcal{D}}](\mathcal{T}(\context';Q'+c))(\sigma) \leq \mathcal{T}(\context;Q+c)(\sigma)$$ for all program states $\sigma$, it is sufficient to show that 
$$\forall n. {\word{ert}}[{\word{call}}^{\mathcal{D}}_{n}P](\mathcal{T}(\context';Q'+c))(\sigma) \leq \mathcal{T}(\context;Q+c)(\sigma)
$$
The proof is done by induction on the natural value $n$. 
\begin{itemize}
\item \emph{Base case.} It is intermediately satisfied because 
$$
\begin{array}{ll}
& {\word{ert}}[{\word{call}}^{\mathcal{D}}_{0}P](\mathcal{T}(\context';Q'+c))(\sigma) \\
& \proofcomment{\text{Definition of } {\word{call}}^{\mathcal{D}}_{n}P} \\
 = & {\word{ert}}[\word{abort}](\mathcal{T}(\context';Q'+c))(\sigma) \\
& \proofcomment{\text{Table } \ref{tab:ert}} \\
 = & \mathbf{0}(\sigma) \leq \mathcal{T}(\context;Q+c)(\sigma)
\end{array}
$$ 
\item \emph{Induction case.} We reason as follows
$$
\begin{array}{ll}
 & \proofcomment{\text{Rules \textcolor{ACMRed}{\textsc{Q:Call}}, \textcolor{ACMRed}{\textsc{ValidCtx}}, and by I.H for } \mathcal{D}(P)} \\
 & \proofcomment{\text{where the current body of } P \text{ is } {\word{call}}^{\mathcal{D}}_{n}P}  \\
 & {\word{ert}}[\mathcal{D}(P),\mathcal{D}](\mathcal{T}(\context';Q'))(\sigma) \leq \mathcal{T}(\context;Q)(\sigma) \\

 & \proofcomment{\text{Algebra}} \\
\Leftrightarrow & {\word{ert}}[\mathcal{D}(P),\mathcal{D}](\mathcal{T}(\context';Q'))(\sigma) + x \leq \mathcal{T}(\context;Q)(\sigma) + x \\

& \proofcomment{\text{Definition of translation function}}  \\
\Leftrightarrow & {\word{ert}}[\mathcal{D}(P),\mathcal{D}](\mathcal{T}(\context';Q'))(\sigma) + x \leq \mathcal{T}(\context;Q+x)(\sigma) \\

& \proofcomment{\text{Property of \word{ert}}} \\
\Rightarrow & {\word{ert}}[\mathcal{D}(P),\mathcal{D}](\mathcal{T}(\context';Q')+x)(\sigma) \leq \mathcal{T}(\context;Q+x)(\sigma)\\

& \proofcomment{\text{Definition of translation function}} \\
\Leftrightarrow & {\word{ert}}[\mathcal{D}(P),\mathcal{D}](\mathcal{T}(\context';Q'+x))(\sigma)  \leq \mathcal{T}(\context;Q+x)(\sigma) \\




& \proofcomment{\text{Lemma \ref{lem:functioncall} where } {\word{call}}^{\mathcal{D}}_{n}P \text{ is closed}} \\
\Leftrightarrow & {\word{ert}}[\mathcal{D}(P)[{\word{call}}^{\mathcal{D}}_{n}P{/}{\word{call}P}],\mathcal{D}](\mathcal{T}(\context';Q'+c))(\sigma) \\
& \leq \mathcal{T}(\context;Q+c)(\sigma) \\

& \proofcomment{\text{Definition of replacement in Table \ref{tab:functioncall}}}  \\
\Leftrightarrow & {\word{ert}}[{\word{call}}^{\mathcal{D}}_{n+1}P](\mathcal{T}(\context';Q'+c))(\sigma) \leq \mathcal{T}(\context;Q+c)(\sigma) 
\end{array}
$$
\end{itemize}

\vspace{-1.5ex}
\section{Basic properties of the expected cost transformer}
We write $\sqsubseteq$ to denote the point-wise order relation between expectations. 
Formally, let $f, g \in \contcost$, $f \sqsubseteq g$ if and only if $f(\sigma) \leq g(\sigma)$ for all 
program states $\sigma \in \Sigma$. 
\vspace{-1.0ex}
\subsection{The $\omega$-complete partial order}
\begin{lemma}[$\omega$-cpo]
$(\contcost, \sqsubseteq)$ is an $\omega$-complete partial order (cpo) with bottom element $\lambda\sigma.\mathbf{0}$ and top element 
$\lambda\sigma.\mathbf{\infty}$.
\end{lemma}
\begin{proof}
To prove that $(\contcost, \sqsubseteq)$ is an $\omega$-cpo, we need to show that $(\contcost, \sqsubseteq)$ is a partial order set and then 
to show that every $\omega$-chain $S \subseteq \contcost$ has a supremum. Such a supremum can be constructed by taking the point-wise supremum 
$
\word{sup}S = \lambda\sigma.\word{sup}\{f(\sigma) \mid f \in S\}
$
which always exists because every subset of $\mathbb{R}^{+}_{0}$ has a supremum. Hence, it remains to prove $(\contcost, \sqsubseteq)$ is a partial order set.
\paragraph{Antisymmetry} $f \sqsubseteq g \wedge g \sqsubseteq f$
$$
\begin{array}{ll}
& \proofcomment{\text{Definition of } \sqsubseteq} \\
\Leftrightarrow & \forall \sigma. f(\sigma) \leq g(\sigma) \wedge g(\sigma) \leq f(\sigma) \\

\Leftrightarrow & \forall \sigma. f(\sigma) \leq g(\sigma) \leq f(\sigma) \\
\Rightarrow & \forall \sigma. f(\sigma) = g(\sigma) \\
\Leftrightarrow & f = g
\end{array}
$$ 
\paragraph{Reflexivity} $\forall \sigma. f(\sigma) = f(\sigma)$
$$
\begin{array}{ll}
\Rightarrow & \forall \sigma. f(\sigma) \leq f(\sigma) \\
& \proofcomment{\text{Definition of } \sqsubseteq} \\
\Leftrightarrow & f \sqsubseteq f\\
\end{array}
$$ 
\paragraph{Transitivity} $f \sqsubseteq g \wedge g \sqsubseteq h$
$$
\begin{array}{ll}
& \proofcomment{\text{Definition of } \sqsubseteq} \\
\Leftrightarrow & \forall \sigma. f(\sigma) \leq g(\sigma) \wedge g(\sigma) \leq h(\sigma) \\

\Leftrightarrow & \forall \sigma. f(\sigma) \leq g(\sigma) \leq h(\sigma) \\
\Rightarrow & \forall \sigma. f(\sigma) \leq h(\sigma) \\

& \proofcomment{\text{Definition of } \sqsubseteq} \\
\Leftrightarrow & f \sqsubseteq h
\end{array}
$$ 
\end{proof}
\vspace{-1.0ex}
\subsection{Basic algebraic properties}
The transformer $\word{ert}$ enjoys many algebraic properties such as 
\emph{monotonicity}, \emph{propagation of constants}, \emph{scaling}, 
and \emph{continuity} properties which we summarize in the following lemma.
\begin{lemma}[Algebraic properties of \word{ert}]
For every program $(c,\mathcal{D})$, every $f, g \in \contcost$, every increasing $\omega$-chain 
$f_0 \sqsubseteq f_1 \sqsubseteq \cdots$, every constant expectation $\mathbf{k} = \lambda\sigma.k, k \in \mathbb{R}^{+}_{0}$, and every constant $r \in \mathbb{R}^{+}_{0}$, 
the following properties hold.
\begin{itemize}
\item \text{Continuity:}
$$
{\word{sup}}_{n}\transformer{c}{\mathcal{D}}{f_n} = \transformer{c}{\mathcal{D}}{{\word{sup}}_{n}f_n}
$$
\item \text{Monotonicity:}
$$
f \sqsubseteq g \Rightarrow \transformer{c}{\mathcal{D}}{f} \sqsubseteq \transformer{c}{\mathcal{D}}{g}
$$
\item \text{Propagation of constants:}
$$
\transformer{c}{\mathcal{D}}{\mathbf{k} + f} \sqsubseteq \mathbf{k} + \transformer{c}{\mathcal{D}}{f}
$$
\item \text{Sub-additivity:}
$$
\begin{array}{l}
\transformer{c}{\mathcal{D}}{f + g} \sqsubseteq \transformer{c}{\mathcal{D}}{f} + \transformer{c}{\mathcal{D}}{g}\\
\text{if } c \text{ is non-determinism-free}
\end{array}
$$
\item \text{Scaling:}
$$
\begin{array}{l}
\word{min}(1,r){\cdot}\transformer{c}{\mathcal{D}}{f} \sqsubseteq \transformer{c}{\mathcal{D}}{r{\cdot}f}\\
\transformer{c}{\mathcal{D}}{r{\cdot}f} \sqsubseteq \word{max}(1,r){\cdot}\transformer{c}{\mathcal{D}}{f}
\end{array}
$$
\item \text{Preservation of $\infty$:}
$$
\transformer{c}{\mathcal{D}}{\infty} = \infty \quad \text{if } c \text{ is abort-free}
$$
\end{itemize}
\end{lemma}

\vspace{-.5ex}
\subsubsection{Continuity}
Let $f_0 \sqsubseteq f_1 \sqsubseteq \cdots$ be an increasing $\omega$-chain. The proof is done by induction on the structure of the command $c$. 

\vspace{-.5ex}
\paragraph{Skip}
$$
\begin{array}{ll}
& {\word{sup}}_{n}\transformer{\word{skip}}{\mathcal{D}}{f_n}\\
& \proofcomment{\text{Table } \ref{tab:ert}} \\
= & {\word{sup}}_{n}f_n \\
& \proofcomment{\text{Table } \ref{tab:ert}} \\
= & \transformer{\word{skip}}{\mathcal{D}}{{\word{sup}}_{n}f_n}
\end{array}
$$ 

\vspace{-.5ex}
\paragraph{Abort} 
$$
\begin{array}{ll}
& {\word{sup}}_{n}\transformer{\word{abort}}{\mathcal{D}}{f_n}\\
& \proofcomment{\text{Table } \ref{tab:ert} \text{ and algebra}} \\
= & \mathbf{0} \\
= & \transformer{\word{abort}}{\mathcal{D}}{{\word{sup}}_{n}f_n}
\end{array}
$$ 

\vspace{-.5ex}
\paragraph{Assert}  
$$
\begin{array}{ll}
& {\word{sup}}_{n}\transformer{\word{assert} e}{\mathcal{D}}{f_n}\\
& \proofcomment{\text{Table } \ref{tab:ert} \text{ and algebra}} \\
= & \sem{e:\word{true}}{\cdot}({\word{sup}}_{n}f_n) \\
= & \transformer{\word{assert} e}{\mathcal{D}}{{\word{sup}}_{n}f_n}
\end{array}
$$ 

\vspace{-.5ex}
\paragraph{Weaken} 
$$
\begin{array}{ll}
& {\word{sup}}_{n}\transformer{\word{weaken}}{\mathcal{D}}{f_n}\\
& \proofcomment{\text{Table } \ref{tab:ert}} \\
= & {\word{sup}}_{n}f_n \\
& \proofcomment{\text{Table } \ref{tab:ert}} \\
= & \transformer{\word{weaken}}{\mathcal{D}}{{\word{sup}}_{n}f_n}
\end{array}
$$ 

\vspace{-.5ex}
\paragraph{Tick} 
$$
\begin{array}{ll}
& {\word{sup}}_{n}\transformer{\word{tick}(q)}{\mathcal{D}}{f_n}\\
& \proofcomment{\text{Table } \ref{tab:ert}} \\
= & {\word{sup}}_{n}(\mathbf{q} + f_n) \\
& \proofcomment{\text{Algebra}} \\
= & \mathbf{q} + {\word{sup}}_{n}f_n \\
& \proofcomment{\text{Table } \ref{tab:ert}} \\
= & \transformer{\word{tick}(q)}{\mathcal{D}}{{\word{sup}}_{n}f_n}
\end{array}
$$ 

\vspace{-.5ex}
\paragraph{Assignment}
$$
\begin{array}{ll}
& {\word{sup}}_{n}\transformer{x = e}{\mathcal{D}}{f_n}\\
& \proofcomment{\text{Table } \ref{tab:ert}} \\
= & {\word{sup}}_{n}(f_n[e/x]) = ({\word{sup}}_{n}{f_n})[e/x] \\ 
& \proofcomment{\text{Table } \ref{tab:ert}} \\
= & \transformer{x = e}{\mathcal{D}}{{\word{sup}}_{n}f_n}
\end{array}
$$ 

\vspace{-.5ex}
\paragraph{Sampling} The proof relies on the Lebesgue's Monotone Convergence Theorem (LMCT). 
$$
\begin{array}{ll}
& {\word{sup}}_{n}\transformer{x = e \word{bop} R}{\mathcal{D}}{f_n}\\
& \proofcomment{\text{Table } \ref{tab:ert}} \\
= & {\word{sup}}_{n}(\lambda\sigma.\expt{\dist{R}}{\lambda v.f_n(\sigma[e \word{bop} v/x])}) \\ 
& \proofcomment{\text{LMCT}} \\
= & \lambda\sigma.\expt{\dist{R}}{{\word{sup}}_{n}\lambda v.f_n(\sigma[e \word{bop} v/x])} \\
= & \lambda\sigma.\expt{\dist{R}}{\lambda v.{\word{sup}}_{n}f_n(\sigma[e \word{bop} v/x])} \\
& \proofcomment{\text{Table } \ref{tab:ert}} \\
= & \transformer{x = e \word{bop} R}{\mathcal{D}}{{\word{sup}}_{n}f_n}
\end{array}
$$ 

\vspace{-.5ex}
\paragraph{If}
The proof relies on the Monotone Sequence Theorem (MST), that is, if $\langle a_n \rangle$ is a monotonic sequence 
in $\mathbb{R}^{+}_{0}$ then ${\word{sup}}_na_n = {\word{lim}}_{n\rightarrow \infty}a_n$. 
$$
\begin{array}{ll} 
& {\word{sup}}_{n}\transformer{\word{if} e \; c_1 \word{else} c_2}{\mathcal{D}}{f_n}\\
& \proofcomment{\text{Table } \ref{tab:ert}} \\
= & {\word{sup}}_{n}(\sem{e:\word{true}}{\cdot}\transformer{c_1}{\mathcal{D}}{f_n} + \\
& \sem{e:\word{false}}{\cdot}\transformer{c_2}{\mathcal{D}}{f_n}) \\

& \proofcomment{\text{MST}} \\
= & {\word{lim}}_{n\rightarrow \infty}(\sem{e:\word{true}}{\cdot}\transformer{c_1}{\mathcal{D}}{f_n} + \\
& \sem{e:\word{false}}{\cdot}\transformer{c_2}{\mathcal{D}}{f_n}) \\

= & \sem{e:\word{true}}{\cdot}{\word{lim}}_{n\rightarrow \infty}\transformer{c_1}{\mathcal{D}}{f_n} + \\
& \sem{e:\word{false}}{\cdot}{\word{lim}}_{n\rightarrow \infty}\transformer{c_2}{\mathcal{D}}{f_n} \\

& \proofcomment{\text{MST}} \\
= & \sem{e:\word{true}}{\cdot}{\word{sup}}_{n}\transformer{c_1}{\mathcal{D}}{f_n} + \\
& \sem{e:\word{false}}{\cdot}{\word{sup}}_{n}\transformer{c_2}{\mathcal{D}}{f_n} \\

& \proofcomment{\text{By I.H on } c_1 \text{ and } c_2} \\
= & \sem{e:\word{true}}{\cdot}\transformer{c_1}{\mathcal{D}}{{\word{sup}}_nf_n} + \\
& \sem{e:\word{false}}{\cdot}\transformer{c_2}{\mathcal{D}}{{\word{sup}}_nf_n} \\

& \proofcomment{\text{Table } \ref{tab:ert}} \\
= & {\transformer{\word{if} e \; c_1 \word{else} c_2}{\mathcal{D}}{\word{sup}}_{n}f_n} 
\end{array}
$$ 

\vspace{-.5ex}
\paragraph{Nondeterministic branching}
$$
\begin{array}{ll} 
& {\word{sup}}_{n}\transformer{\word{if} \star \; c_1 \word{else} c_2}{\mathcal{D}}{f_n}\\
   & \proofcomment{\text{Table } \ref{tab:ert}} \\
 = & {\word{sup}}_{n}({\word{max}}\{\transformer{c_1}{\mathcal{D}}{f_n},\transformer{c_2}{\mathcal{D}}{f_n}\}) \\
 \sqsupseteq & {\word{max}}\{{\word{sup}}_{n}\transformer{c_1}{\mathcal{D}}{f_n},{\word{sup}}_{n}\transformer{c_2}{\mathcal{D}}{f_n}\}\} \\

   & \proofcomment{\text{By I.H on } c_1 \text{ and } c_2} \\
 = & {\word{max}}\{\transformer{c_1}{\mathcal{D}}{{\word{sup}}_{n}f_n},{\transformer{c_2}{\mathcal{D}}{\word{sup}}_{n}f_n}\}\} \\

   & \proofcomment{\text{Table } \ref{tab:ert}} \\
 = & \transformer{\word{if} \star \; c_1 \word{else} c_2}{\mathcal{D}}{{\word{sup}}_{n}f_n} 
\end{array}
$$ 
Let $A = {\word{max}}\{{\word{sup}}_{n}\transformer{c_1}{\mathcal{D}}{f_n},{\word{sup}}_{n}\transformer{c_2}{\mathcal{D}}{f_n}\}$.  
Assume that 
$$
A \sqsubset {\word{sup}}_{n}({\word{max}}\{\transformer{c_1}{\mathcal{D}}{f_n},\transformer{c_2}{\mathcal{D}}{f_n}\})
$$
Then there exists $m \in \mathbb{N}$ such that 
$$
\begin{array}{rll}
A 
& \sqsubset 	& {\word{max}}\{\transformer{c_1}{\mathcal{D}}{f_m},\transformer{c_2}{\mathcal{D}}{f_m}\} \\

&   & \proofcomment{\text{Definition of supremum}} \\
& \sqsubseteq 	& {\word{max}}\{{\word{sup}}_{n}\transformer{c_1}{\mathcal{D}}{f_n},{\word{sup}}_{n}\transformer{c_2}{\mathcal{D}}{f_n}\} \\
& = & A
\end{array}
$$
Therefore, we get 
$$
\begin{array}{rll}
A 
& = & {\word{max}}\{{\word{sup}}_{n}\transformer{c_1}{\mathcal{D}}{f_n},{\word{sup}}_{n}\transformer{c_2}{\mathcal{D}}{f_n}\} \\
& = & \transformer{\word{if} \star \; c_1 \word{else} c_2}{\mathcal{D}}{{\word{sup}}_{n}f_n} \\
& \sqsupseteq & {\word{sup}}_{n}({\word{max}}\{\transformer{c_1}{\mathcal{D}}{f_n},\transformer{c_2}{\mathcal{D}}{f_n}\}) \\
& = & {\word{sup}}_{n}\transformer{\word{if} \star \; c_1 \word{else} c_2}{\mathcal{D}}{f_n} 
\end{array}
$$
By observation above, it follows.

\vspace{-.5ex}
\paragraph{Probabilistic branching} 
The proof relies on the Monotone Sequence Theorem (MST), that is, if $\langle a_n \rangle$ is a monotonic sequence 
in $\mathbb{R}^{+}_{0}$ then ${\word{sup}}_na_n = {\word{lim}}_{n\rightarrow \infty}a_n$. 
$$
\begin{array}{ll}
& {\word{sup}}_{n}\transformer{c_1 \oplus_{p} c_2}{\mathcal{D}}{f_n} \\
   & \proofcomment{\text{Table } \ref{tab:ert}} \\
 = & {\word{sup}}_{n}(p{\cdot}\transformer{c_1}{\mathcal{D}}{f_n} + (1-p){\cdot}\transformer{c_2}{\mathcal{D}}{f_n}) \\

   & \proofcomment{\text{MST}} \\
 = & {\word{lim}}_{n\rightarrow \infty}(p{\cdot}\transformer{c_1}{\mathcal{D}}{f_n} + (1-p){\cdot}\transformer{c_2}{\mathcal{D}}{f_n}) \\

 = & p{\cdot}{\word{lim}}_{n\rightarrow \infty}\transformer{c_1}{\mathcal{D}}{f_n} + (1-p){\cdot}{\word{lim}}_{n\rightarrow \infty}\transformer{c_2}{\mathcal{D}}{f_n} \\

   & \proofcomment{\text{MST}} \\
 = & p{\cdot}{\word{sup}}_{n}\transformer{c_1}{\mathcal{D}}{f_n} + (1-p){\cdot}{\word{sup}}_{n}\transformer{c_2}{\mathcal{D}}{f_n} \\

   & \proofcomment{\text{By I.H on } c_1 \text{ and } c_2} \\
 = & p{\cdot}\transformer{c_1}{\mathcal{D}}{{\word{sup}}_nf_n} + (1-p){\cdot}\transformer{c_2}{\mathcal{D}}{{\word{sup}}_nf_n} \\

   & \proofcomment{\text{Table } \ref{tab:ert}} \\
 = & {\transformer{c_1 \oplus_{p} c_2}{\mathcal{D}}{\word{sup}}_{n}f_n} 
\end{array}
$$

\vspace{-.5ex}
\paragraph{Sequence} 
$$
\begin{array}{ll}
& {\word{sup}}_{n}\transformer{c_1; c_2}{\mathcal{D}}{f_n} \\
   & \proofcomment{\text{Table } \ref{tab:ert}} \\
 = & {\word{sup}}_{n}(\transformer{c_1}{\mathcal{D}}{\transformer{c_2}{\mathcal{D}}{f_n}}) \\

   & \proofcomment{\text{By I.H on } c_1} \\
 = & \transformer{c_1}{\mathcal{D}}{{\word{sup}}_{n}\transformer{c_2}{\mathcal{D}}{f_n}} \\

   & \proofcomment{\text{By I.H on } c_2} \\
 = & \transformer{c_1}{\mathcal{D}}{\transformer{c_2}{\mathcal{D}}{{\word{sup}}_{n}f_n}} \\

   & \proofcomment{\text{Table } \ref{tab:ert}} \\
 = & \transformer{c_1; c_2}{\mathcal{D}}{{\word{sup}}_{n}f_n}
\end{array}
$$ 

\vspace{-.5ex}
\paragraph{Loop}
By I.H for ${\word{while}}^{k} e \; c$ (defined in Section \ref{subsec:boundedloops}), we get
$$
{\word{sup}}_{n}\transformer{{\word{while}}^{k} e \; c}{\mathcal{D}}{f_n} = \transformer{{\word{while}}^{k} e \; c}{\mathcal{D}}{{\word{sup}}_{n}f_n}
$$
On the other hand, by Theorem \ref{theo:loopsup}, $\word{ert}[\word{while} e \; c,\mathcal{D}] = {\word{sup}}_{k}{\word{ert}}[{\word{while}}^{k} e \; c,\mathcal{D}]$, where ${\word{while}}^{k} e \; c$ is the \emph{$k^{th}$ bounded} execution of a \word{while} loop command, we have 
$$
\begin{array}{ll}
& \transformer{\word{while} e \; c}{\mathcal{D}}{{\word{sup}}_{n}f_n} \\
 = & {\word{sup}}_{k}\transformer{{\word{while}}^{k} e \; c}{\mathcal{D}}{{\word{sup}}_{n}f_n} \\

   & \proofcomment{\text{Observation above}} \\
 = & {\word{sup}}_{k}({\word{sup}}_{n}\transformer{{\word{while}}^{k} e \; c}{\mathcal{D}}{f_n}) \\
 = & {\word{sup}}_{n}({\word{sup}}_{k}\transformer{{\word{while}}^{k} e \; c}{\mathcal{D}}{f_n}) \\

   & \proofcomment{\text{By Theorem } \ref{theo:loopsup}} \\
 = & {\word{sup}}_{n}(\transformer{\word{while} e \; c}{\mathcal{D}}{f_n})
\end{array}
$$

\vspace{-.5ex}
\paragraph{Procedure call}
Because ${\word{call}}^{\mathcal{D}}_{k}P$ (defined in Section \ref{subsec:supfixpoint}) is closed command for all $k \in \mathbb{N}$, by I.H we get
$$
{\word{sup}}_{n}{\word{ert}}[{\word{call}}^{\mathcal{D}}_{k}P](f_n) = {\word{ert}}[{\word{call}}^{\mathcal{D}}_{k}P]({{\word{sup}}_{n}f_n})
$$
On the other hand, by Theorem \ref{theo:limitapproximation}, we have 
$$
\word{ert}[{\word{call}P,\mathcal{D}}] = {\word{sup}}_{k}{\word{ert}}[{\word{call}}^{\mathcal{D}}_{k}P]
$$ 
where ${\word{call}}^{\mathcal{D}}_{k}P$ is the $k^{th}$-inlining of procedure call, we have 
$$
\begin{array}{ll}
& \transformer{\word{call}P}{\mathcal{D}}{{\word{sup}}_{n}f_n} \\
 = & {\word{sup}}_{k}{\word{ert}}[{\word{call}}^{\mathcal{D}}_{k}P]({{\word{sup}}_{n}f_n}) \\

   & \proofcomment{\text{Observation above}} \\
 = & {\word{sup}}_{k}({\word{sup}}_{n}{\word{ert}}[{\word{call}}^{\mathcal{D}}_{k}P](f_n)) \\
 = & {\word{sup}}_{n}({\word{sup}}_{k}{\word{ert}}[{\word{call}}^{\mathcal{D}}_{k}P](f_n)) \\

   & \proofcomment{\text{By Theorem } \ref{theo:limitapproximation}} \\
 = & {\word{sup}}_{n}(\transformer{{\word{call}P}}{\mathcal{D}}{f_n})
\end{array}
$$
\vspace{-.5ex}
\subsubsection{Monotonicity}
The monotonicity follows from the continuity of $\word{ert}$ as follows.
$$
\begin{array}{rll}
\transformer{c}{\mathcal{D}}{g} 
&   & \proofcomment{f \sqsubseteq g} \\
& = & \transformer{c}{\mathcal{D}}{{\word{sup}}\{f,g\}} \\

&   & \proofcomment{\text{Continuity}} \\
& = & {\word{sup}}\{\transformer{c}{\mathcal{D}}{f},\transformer{c}{\mathcal{D}}{g}\} \\

&   & \proofcomment{\text{Definition of supremum}} \\
& \sqsupseteq & \transformer{c}{\mathcal{D}}{f}
\end{array}
$$
\vspace{-.5ex}
\subsubsection{Propagation of constants}
The proof is done by induction on the structure of the command $c$. 

\vspace{-.5ex}
\paragraph{Skip} 
$$
\begin{array}{ll}
& {\transformer{\word{skip}}{\mathcal{D}}{\mathbf{k} + f}} \\
   & \proofcomment{\text{Table } \ref{tab:ert}} \\
 = & \mathbf{k} + f \\

   & \proofcomment{\text{Table } \ref{tab:ert}} \\
 = & \mathbf{k} + \transformer{\word{skip}}{\mathcal{D}}{f}
\end{array}
$$ 

\vspace{-.5ex}
\paragraph{Abort} 
$$
\begin{array}{ll}
& {\transformer{\word{abort}}{\mathcal{D}}{\mathbf{k} + f}} \\
   & \proofcomment{\text{Table } \ref{tab:ert} \text{ and algebra}} \\
 = & \mathbf{0} \\
 \sqsubseteq & \mathbf{k} + \transformer{\word{abort}}{\mathcal{D}}{f}
\end{array}
$$ 

\vspace{-.5ex}
\paragraph{Assert} 
$$
\begin{array}{ll}
& \transformer{\word{assert} e}{\mathcal{D}}{\mathbf{k} + f} \\
   & \proofcomment{\text{Table } \ref{tab:ert} \text{ and algebra}} \\
 = & \sem{e:\word{true}}{\cdot}(\mathbf{k} + f)  = \sem{e:\word{true}}{\cdot}\mathbf{k} + \sem{e:\word{true}}{\cdot}f\\
   & \proofcomment{\sem{e:\word{true}} \leq 1} \\
 \sqsubseteq & \mathbf{k} + \transformer{\word{assert} e}{\mathcal{D}}{f} 
\end{array}
$$ 

\vspace{-.5ex}
\paragraph{Weaken} 
$$
\begin{array}{ll}
& \transformer{\word{weaken}}{\mathcal{D}}{\mathbf{k} + f} \\
   & \proofcomment{\text{Table } \ref{tab:ert}} \\
 = & \mathbf{k} + f \\
   & \proofcomment{\text{Table } \ref{tab:ert}} \\
 = & \mathbf{k} + \transformer{\word{weaken}}{\mathcal{D}}{f} 
\end{array}
$$ 

\vspace{-.5ex}
\paragraph{Tick}
$$
\begin{array}{ll}
& \transformer{\word{tick}(q)}{\mathcal{D}}{\mathbf{k} + f} \\
   & \proofcomment{\text{Table } \ref{tab:ert}} \\
 = & \mathbf{q} + \mathbf{k} + f \\
   & \proofcomment{\text{Table } \ref{tab:ert}} \\
 = & \mathbf{k} + \transformer{\word{tick}(q)}{\mathcal{D}}{f} 
\end{array}
$$ 

\vspace{-.5ex}
\paragraph{Assignment}
$$
\begin{array}{ll}
& \transformer{x = e}{\mathcal{D}}{\mathbf{k} + f} \\
   & \proofcomment{\text{Table } \ref{tab:ert}} \\
 = & \mathbf{k} + f[e/x] \\ 
   & \proofcomment{\text{Table } \ref{tab:ert}} \\
 = & \mathbf{k} + \transformer{x = e}{\mathcal{D}}{f} 
\end{array}
$$ 

\vspace{-.5ex}
\paragraph{Sampling} 
The proof relies on the linearity property of expectations (LPE). 
$$
\begin{array}{ll}
& \transformer{x = e \word{bop} R}{\mathcal{D}}{\mathbf{k} + f} \\
   & \proofcomment{\text{Table } \ref{tab:ert}} \\
 = & \lambda\sigma.\expt{\dist{R}}{\lambda v.(\mathbf{k} + f)(\sigma[e \word{bop} v/x])} \\ 

   & \proofcomment{\text{LPE and } \mathbf{k}(\sigma[e \word{bop} v/x]) = \mathbf{k}(\sigma)} \\
 = & \mathbf{k} + \lambda\sigma.\expt{\dist{R}}{\lambda v.f(\sigma[e \word{bop} v/x])} \\
   & \proofcomment{\text{Table } \ref{tab:ert}} \\
 = & \mathbf{k} + \transformer{x = e \word{bop} R}{\mathcal{D}}{f} 
\end{array}
$$ 

\vspace{-.5ex}
\paragraph{If}
$$
\begin{array}{ll}
& \transformer{\word{if} e \; c_1 \word{else} c_2}{\mathcal{D}}{\mathbf{k} + f} \\
   & \proofcomment{\text{Table } \ref{tab:ert}} \\
 = & \sem{e:\word{true}}{\cdot}\transformer{c_1}{\mathcal{D}}{\mathbf{k} + f} + \\
   & \sem{e:\word{false}}{\cdot}\transformer{c_2}{\mathcal{D}}{\mathbf{k} + f} \\

   & \proofcomment{\text{By I.H on } c_1 \text{ and } c_2} \\
 \sqsubseteq & \sem{e:\word{true}}{\cdot}(\mathbf{k} + \transformer{c_1}{\mathcal{D}}{f}) + \\
   & \sem{e:\word{false}}{\cdot}(\mathbf{k} + \transformer{c_2}{\mathcal{D}}{f}) \\

   & \proofcomment{\text{Algebra}} \\
 = & \mathbf{k} + \sem{e:\word{true}}{\cdot}\transformer{c_1}{\mathcal{D}}{f} + \\
   & \sem{e:\word{false}}{\cdot}\transformer{c_2}{\mathcal{D}}{f} \\

   & \proofcomment{\text{Table } \ref{tab:ert}} \\
 = & \mathbf{k} + \transformer{\word{if} e \; c_1 \word{else} c_2}{\mathcal{D}}{f} 
\end{array}
$$ 

\vspace{-.5ex}
\paragraph{Nondeterministic branching} 
$$
\begin{array}{ll}
& \transformer{\word{if} \star \; c_1 \word{else} c_2}{\mathcal{D}}{\mathbf{k} + f} \\
   & \proofcomment{\text{Table } \ref{tab:ert}} \\
 = & {\word{max}}\{\transformer{c_1}{\mathcal{D}}{\mathbf{k} + f},\transformer{c_2}{\mathcal{D}}{\mathbf{k} + f}\} \\

   & \proofcomment{\text{By I.H on } c_1 \text{ and } c_2} \\
 \sqsubseteq & {\word{max}}\{\mathbf{k} + \transformer{c_1}{\mathcal{D}}{f},\mathbf{k} + \transformer{c_2}{\mathcal{D}}{f}\} \\

   & \proofcomment{\text{Algebra}} \\
 = & \mathbf{k} + {\word{max}}\{\transformer{c_1}{\mathcal{D}}{f},\transformer{c_2}{\mathcal{D}}{f}\} \\

   & \proofcomment{\text{Table } \ref{tab:ert}} \\
 = & \mathbf{k} + \transformer{\word{if} \star \; c_1 \word{else} c_2}{\mathcal{D}}{f} 
\end{array}
$$ 

\vspace{-.5ex}
\paragraph{Probabilistic branching} 
$$
\begin{array}{ll}
& \transformer{c_1 \oplus_{p} c_2}{\mathcal{D}}{\mathbf{k} + f} \\
   & \proofcomment{\text{Table } \ref{tab:ert}} \\
 = & p{\cdot}\transformer{c_1}{\mathcal{D}}{\mathbf{k} + f} + (1-p){\cdot}\transformer{c_2}{\mathcal{D}}{\mathbf{k} + f} \\

   & \proofcomment{\text{By I.H on } c_1 \text{ and } c_2} \\
 \sqsubseteq & p{\cdot}(\mathbf{k} + \transformer{c_1}{\mathcal{D}}{f}) + (1-p){\cdot}(\mathbf{k} + \transformer{c_2}{\mathcal{D}}{f}) \\

   & \proofcomment{\text{Algebra}} \\
 = & \mathbf{k} + p{\cdot}\transformer{c_1}{\mathcal{D}}{f} + (1-p){\cdot}\transformer{c_2}{\mathcal{D}}{f} \\

   & \proofcomment{\text{Table } \ref{tab:ert}} \\
 = & \mathbf{k} + \transformer{c_1 \oplus_{p} c_2}{\mathcal{D}}{f} 
\end{array}
$$ 

\vspace{-.5ex}
\paragraph{Sequence} 
$$
\begin{array}{ll}
& \transformer{c_1; c_2}{\mathcal{D}}{\mathbf{k} + f} \\
   & \proofcomment{\text{Table } \ref{tab:ert}} \\
 = & \transformer{c_1}{\mathcal{D}}{\transformer{c_2}{\mathcal{D}}{\mathbf{k} + f}} \\

   & \proofcomment{\text{By I.H on } c_2} \\
 \sqsubseteq & \transformer{c_1}{\mathcal{D}}{\mathbf{k} + \transformer{c_2}{\mathcal{D}}{f}} \\

   & \proofcomment{\text{By I.H on } c_1} \\
 \sqsubseteq & \mathbf{k} + \transformer{c_1}{\mathcal{D}}{\transformer{c_2}{\mathcal{D}}{f}} \\

   & \proofcomment{\text{Table } \ref{tab:ert}} \\
 = & \mathbf{k} + \transformer{c_1; c_2}{\mathcal{D}}{f} 
\end{array}
$$ 

\vspace{-.5ex}
\paragraph{Loop}
Consider the characteristic function w.r.t the expectation $f$
$$
F_{f} \defineas \sem{e:{\word{true}}}{\cdot}\transformer{c}{\mathcal{D}}{X} + \sem{e:\word{false}}{\cdot}f
$$
We first need to show that $F_{\mathbf{k}+f}(\mathbf{k} + \word{lfp}F_{f}) \sqsubseteq \mathbf{k} + \word{lfp}F_{f}$. Then following Park's Theorem\footnote{If $H: \mathcal{D} \rightarrow \mathcal{D}$ is a continuous function over an $\omega$-cpo $(\mathcal{D},\sqsubseteq)$ with bottom element, then $H(d) \sqsubseteq d$ implies $\word{lfp}H \sqsubseteq d$ for all $d \in \mathcal{D}$.}~\cite{Wechler92}, we get $\word{lfp}F_{\mathbf{k}+f} \sqsubseteq \mathbf{k} + \word{lfp}F_{f}$. 
$$
\begin{array}{ll}
& F_{\mathbf{k}+f}(\mathbf{k} + \word{lfp}F_{f}) \\
   & \proofcomment{\text{Definition of } F_{\mathbf{k}+f}} \\
 = & \sem{e:{\word{true}}}{\cdot}\transformer{c}{\mathcal{D}}{\mathbf{k} + \word{lfp}F_{f}} + \sem{e:\word{false}}{\cdot}(\mathbf{k}+f) \\

   & \proofcomment{\text{By I.H on } c} \\
 \sqsubseteq & \sem{e:{\word{true}}}{\cdot}(\mathbf{k} + \transformer{c}{\mathcal{D}}{\word{lfp}F_{f}}) + \sem{e:\word{false}}{\cdot}(\mathbf{k}+f) \\

   & \proofcomment{\text{Algebra}} \\
 = & \mathbf{k} + \sem{e:{\word{true}}}{\cdot}\transformer{c}{\mathcal{D}}{\word{lfp}F_{f}} + \sem{e:\word{false}}{\cdot}f \\

   & \proofcomment{\text{Definition of } F_{f}} \\
 = & \mathbf{k} + F_f(\word{lfp}F_{f}) \\

   & \proofcomment{\text{Definition of \word{lfp}}} \\
 = & \mathbf{k} + \word{lfp}F_{f} \\
\end{array}
$$ 	 

\vspace{-.5ex}
\paragraph{Procedure call}
Because ${\word{call}}^{\mathcal{D}}_{n}P$ (defined in Section \ref{subsec:supfixpoint}) is closed command for all $n \in \mathbb{N}$, by I.H we get
$$
{\word{ert}}[{\word{call}}^{\mathcal{D}}_{n}P](\mathbf{k} + f) \sqsubseteq \mathbf{k} + {\word{ert}}[{\word{call}}^{\mathcal{D}}_{n}P](f)
$$
On the other hand, by Theorem \ref{theo:limitapproximation}, we have 
$$
\word{ert}[{\word{call}P,\mathcal{D}}] = {\word{sup}}_{n}{\word{ert}}[{\word{call}}^{\mathcal{D}}_{n}P]
$$
where ${\word{call}}^{\mathcal{D}}_{n}P$ is the $n^{th}$-inlining of procedure call, we have 
$$
\begin{array}{ll}
& \transformer{\word{call}P}{\mathcal{D}}{\mathbf{k} + f} \\
 = & {\word{sup}}_{n}{\word{ert}}[{\word{call}}^{\mathcal{D}}_{n}P]({\mathbf{k} + f}) \\

   & \proofcomment{\text{Observation above}} \\
 \sqsubseteq & {\word{sup}}_{n}(\mathbf{k} + {\word{ert}}[{\word{call}}^{\mathcal{D}}_{n}P](f)) \\
 = & \mathbf{k} + {\word{sup}}_{n}{\word{ert}}[{\word{call}}^{\mathcal{D}}_{n}P](f) \\

   & \proofcomment{\text{By Theorem } \ref{theo:limitapproximation}} \\
 = & \mathbf{k} + \transformer{\word{call}P}{\mathcal{D}}{f} 
\end{array}
$$
\vspace{-.5ex}
\subsubsection{Sub-additivity}
The proof is done by induction on the structure of the command $c$. Note that $c$ is \emph{non-determinism} free, that is, $c$ contains no \word{non-deterministic} commands.

\vspace{-.5ex}
\paragraph{Skip} 
$$
\begin{array}{ll}
& {\transformer{\word{skip}}{\mathcal{D}}{f + g}} \\
   & \proofcomment{\text{Table } \ref{tab:ert}} \\
 = & f + g \\

   & \proofcomment{\text{Table } \ref{tab:ert}} \\
 = & \transformer{\word{skip}}{\mathcal{D}}{f} + \transformer{\word{skip}}{\mathcal{D}}{g}
\end{array}
$$ 

\vspace{-.5ex}
\paragraph{Abort} 
$$
\begin{array}{ll}
& {\transformer{\word{abort}}{\mathcal{D}}{f + g}} \\
   & \proofcomment{\text{Table } \ref{tab:ert} \text{ and algebra}} \\
 = & \mathbf{0} \\
 = & \transformer{\word{abort}}{\mathcal{D}}{f} + \transformer{\word{abort}}{\mathcal{D}}{g}
\end{array}
$$ 

\vspace{-.5ex}
\paragraph{Assert} 
$$
\begin{array}{ll}
& \transformer{\word{assert} e}{\mathcal{D}}{f + g} \\
   & \proofcomment{\text{Table } \ref{tab:ert} \text{ and algebra}} \\
 = & \sem{e:\word{true}}{\cdot}(f + g)  = \sem{e:\word{true}}{\cdot}f + \sem{e:\word{true}}{\cdot}g\\

   & \proofcomment{\text{Table } \ref{tab:ert}} \\
 = & \transformer{\word{assert} e}{\mathcal{D}}{f} + \transformer{\word{assert} e}{\mathcal{D}}{g}  
\end{array}
$$ 

\vspace{-.5ex}
\paragraph{Weaken} 
$$
\begin{array}{ll}
& \transformer{\word{weaken}}{\mathcal{D}}{f + g} \\
   & \proofcomment{\text{Table } \ref{tab:ert}} \\
 = & f + g \\
   & \proofcomment{\text{Table } \ref{tab:ert}} \\
 = & \transformer{\word{weaken}}{\mathcal{D}}{f} + \transformer{\word{weaken}}{\mathcal{D}}{g}  
\end{array}
$$ 

\vspace{-.5ex}
\paragraph{Tick}
$$
\begin{array}{ll}
& \transformer{\word{tick}(q)}{\mathcal{D}}{f + g} \\
   & \proofcomment{\text{Table } \ref{tab:ert}} \\
 = & \mathbf{q} + f + g\\

   & \proofcomment{\text{Algebra}} \\
 \sqsubseteq &  \mathbf{q} + f + \mathbf{q} + g \\

   & \proofcomment{\text{Table } \ref{tab:ert}} \\
 = & \transformer{\word{tick}(q)}{\mathcal{D}}{f} + \transformer{\word{tick}(q)}{\mathcal{D}}{g}  
\end{array}
$$ 

\vspace{-.5ex}
\paragraph{Assignment}
$$
\begin{array}{ll}
& \transformer{x = e}{\mathcal{D}}{f + g} \\
   & \proofcomment{\text{Table } \ref{tab:ert}} \\
 = & f[e/x] + g[e/x] \\ 

   & \proofcomment{\text{Table } \ref{tab:ert}} \\
 = & \transformer{x = e}{\mathcal{D}}{f} + \transformer{x = e}{\mathcal{D}}{g}
\end{array}
$$ 

\vspace{-.5ex}
\paragraph{Sampling} 
The proof relies on the linearity property of expectations (LPE). 
$$
\begin{array}{ll}
& \transformer{x = e \word{bop} R}{\mathcal{D}}{f + g} \\
   & \proofcomment{\text{Table } \ref{tab:ert}} \\
 = & \lambda\sigma.\expt{\dist{R}}{\lambda v.(f + g)(\sigma[e \word{bop} v/x])} \\ 

   & \proofcomment{\text{Linearity of expectations}} \\
 = & \lambda\sigma.\expt{\dist{R}}{\lambda v.f(\sigma[e \word{bop} v/x])} + \\
 & \lambda\sigma.\expt{\dist{R}}{\lambda v.g(\sigma[e \word{bop} v/x])} \\
   & \proofcomment{\text{Table } \ref{tab:ert}} \\
 = & \transformer{x = e \word{bop} R}{\mathcal{D}}{f} + \transformer{x = e \word{bop} R}{\mathcal{D}}{g}  
\end{array}
$$ 

\vspace{-.5ex}
\paragraph{If}
$$
\begin{array}{ll}
& \transformer{\word{if} e \; c_1 \word{else} c_2}{\mathcal{D}}{f + g} \\
   & \proofcomment{\text{Table } \ref{tab:ert}} \\
 = & \sem{e:\word{true}}{\cdot}\transformer{c_1}{\mathcal{D}}{f + g} + \\
   & \sem{e:\word{false}}{\cdot}\transformer{c_2}{\mathcal{D}}{f + g} \\

   & \proofcomment{\text{By I.H on } c_1 \text{ and } c_2} \\
 \sqsubseteq & \sem{e:\word{true}}{\cdot}(\transformer{c_1}{\mathcal{D}}{f} + \transformer{c_1}{\mathcal{D}}{g}) + \\
   & \sem{e:\word{false}}{\cdot}(\transformer{c_2}{\mathcal{D}}{f} + \transformer{c_2}{\mathcal{D}}{g}) \\

   & \proofcomment{\text{Table } \ref{tab:ert} \text{ and algebra}} \\
 = & \transformer{\word{if} e \; c_1 \word{else} c_2}{\mathcal{D}}{f} + \transformer{\word{if} e \; c_1 \word{else} c_2}{\mathcal{D}}{g}   
\end{array}
$$ 

\vspace{-.5ex}
\paragraph{Probabilistic branching} 
$$
\begin{array}{ll}
& \transformer{c_1 \oplus_{p} c_2}{\mathcal{D}}{f + g} \\
   & \proofcomment{\text{Table } \ref{tab:ert}} \\
 = & p{\cdot}\transformer{c_1}{\mathcal{D}}{f + g} + (1-p){\cdot}\transformer{c_2}{\mathcal{D}}{f + g} \\

   & \proofcomment{\text{By I.H on } c_1 \text{ and } c_2} \\
 \sqsubseteq & p{\cdot}(\transformer{c_1}{\mathcal{D}}{f} + \transformer{c_1}{\mathcal{D}}{g}) + \\
   & (1-p){\cdot}(\transformer{c_2}{\mathcal{D}}{f} + \transformer{c_2}{\mathcal{D}}{g}) \\

   & \proofcomment{\text{Table } \ref{tab:ert} \text{ and algebra}} \\
 = & \transformer{c_1 \oplus_{p} c_2}{\mathcal{D}}{f} + \transformer{c_1 \oplus_{p} c_2}{\mathcal{D}}{g}  
\end{array}
$$ 

\vspace{-.5ex}
\paragraph{Sequence} 
$$
\begin{array}{ll}
& \transformer{c_1; c_2}{\mathcal{D}}{f + g} \\
   & \proofcomment{\text{Table } \ref{tab:ert}} \\
 = & \transformer{c_1}{\mathcal{D}}{\transformer{c_2}{\mathcal{D}}{f + g}} \\

   & \proofcomment{\text{By I.H on } c_2} \\
 \sqsubseteq & \transformer{c_1}{\mathcal{D}}{\transformer{c_2}{\mathcal{D}}{f} + \transformer{c_2}{\mathcal{D}}{g}} \\

   & \proofcomment{\text{By I.H on } c_1} \\
 \sqsubseteq & \transformer{c_1}{\mathcal{D}}{\transformer{c_2}{\mathcal{D}}{f}} + \transformer{c_1}{\mathcal{D}}{\transformer{c_2}{\mathcal{D}}{g}}\\

   & \proofcomment{\text{Table } \ref{tab:ert}} \\
 = & \transformer{c_1; c_2}{\mathcal{D}}{f} + \transformer{c_1; c_2}{\mathcal{D}}{g}  
\end{array}
$$ 

\vspace{-.5ex}
\paragraph{Loop}
Consider the characteristic function w.r.t the expectation $f$
$$
F_{f} \defineas \sem{e:{\word{true}}}{\cdot}\transformer{c}{\mathcal{D}}{X} + \sem{e:\word{false}}{\cdot}f
$$
We first need to show that $F_{\mathbf{k}+f}(\word{lfp}F_{f} + \word{lfp}F_{g}) \sqsubseteq \word{lfp}F_{f} + \word{lfp}F_{g}$. Then following Park's Theorem~\cite{Wechler92}, we get $\word{lfp}F_{\mathbf{k}+f} \sqsubseteq \word{lfp}F_{f} + \word{lfp}F_{g}$. 
$$
\begin{array}{ll}
& F_{\mathbf{k}+f}(\word{lfp}F_{f} + \word{lfp}F_{g}) \\
   & \proofcomment{\text{Definition of } F_{\mathbf{k}+f}} \\
 = & \sem{e:{\word{true}}}{\cdot}\transformer{c}{\mathcal{D}}{\word{lfp}F_{f} + \word{lfp}F_{g}} + \\
 & \sem{e:\word{false}}{\cdot}(\word{lfp}F_{f} + \word{lfp}F_{g}) \\

   & \proofcomment{\text{By I.H on } c} \\
 \sqsubseteq & \sem{e:{\word{true}}}{\cdot}(\transformer{c}{\mathcal{D}}{\word{lfp}F_{f}} + \transformer{c}{\mathcal{D}}{\word{lfp}F_{g}}) + \\
   & \sem{e:\word{false}}{\cdot}(\word{lfp}F_{f} + \word{lfp}F_{g}) \\

   & \proofcomment{\text{Algebra}} \\
 = & \sem{e:{\word{true}}}{\cdot}\transformer{c}{\mathcal{D}}{\word{lfp}F_{f}} + \sem{e:\word{false}}{\cdot}\word{lfp}F_{f} + \\
   & \sem{e:{\word{true}}}{\cdot}\transformer{c}{\mathcal{D}}{\word{lfp}F_{g}} + \sem{e:\word{false}}{\cdot}\word{lfp}F_{g} \\

   & \proofcomment{\text{Definition of } F_{f} \text{ and } F_{g}} \\
 = & F_f(\word{lfp}F_{f}) + F_g(\word{lfp}F_{g})\\

   & \proofcomment{\text{Definition of \word{lfp}}} \\
 = & \word{lfp}F_{f} + \word{lfp}F_{g}\\
\end{array}
$$ 	 

\vspace{-.5ex}
\paragraph{Procedure call}
Because ${\word{call}}^{\mathcal{D}}_{n}P$ (defined in Section \ref{subsec:supfixpoint}) is closed command for all $n \in \mathbb{N}$, by I.H we get
$$
\word{ert}[{\word{call}}^{\mathcal{D}}_{n}P](f + g) \sqsubseteq \word{ert}[{\word{call}}^{\mathcal{D}}_{n}P](f) + \word{ert}[{\word{call}}^{\mathcal{D}}_{n}P](g)
$$
On the other hand, by Theorem \ref{theo:limitapproximation}, we get
$$
\word{ert}[{\word{call}P,\mathcal{D}}] = {\word{sup}}_{n}{\word{ert}}[{\word{call}}^{\mathcal{D}}_{n}P]
$$
where ${\word{call}}^{\mathcal{D}}_{n}P$ is the $n^{th}$-inlining of procedure call, we have 
$$
\begin{array}{ll}
& \transformer{\word{call}P}{\mathcal{D}}{f + g} \\
 = & {\word{sup}}_{n}{\word{ert}}[{\word{call}}^{\mathcal{D}}_{n}P](f + g) \\

   & \proofcomment{\text{Observation above}} \\
 \sqsubseteq & {\word{sup}}_{n}(\word{ert}[{\word{call}}^{\mathcal{D}}_{n}P](f) + \word{ert}[{\word{call}}^{\mathcal{D}}_{n}P](g)) \\
 = & {\word{sup}}_{n}\word{ert}[{\word{call}}^{\mathcal{D}}_{n}P](f) + {\word{sup}}_{n}\word{ert}[{\word{call}}^{\mathcal{D}}_{n}P](g) \\

   & \proofcomment{\text{By Theorem } \ref{theo:limitapproximation}} \\
 = & \transformer{\word{call}P}{\mathcal{D}}{f} + \transformer{\word{call}P}{\mathcal{D}}{g}  
\end{array}
$$
\vspace{-.5ex}
\subsubsection{Scaling}
The proof is done by induction on the structure of the command $c$. 

\vspace{-.5ex}
\paragraph{Skip} 
$$
\begin{array}{ll}
   & \word{min}(1,r){\cdot}f \sqsubseteq r{\cdot}f \sqsubseteq \word{max}(1,r){\cdot}f \\

   & \proofcomment{\text{Table } \ref{tab:ert}} \\
 \Leftrightarrow & \word{min}(1,r){\cdot}f \sqsubseteq \transformer{\word{skip}}{\mathcal{D}}{r{\cdot}f} \sqsubseteq \word{max}(1,r){\cdot}f \\

   & \proofcomment{\text{Table } \ref{tab:ert}} \\
 \Leftrightarrow & \word{min}(1,r){\cdot}\transformer{\word{skip}}{\mathcal{D}}{f} \sqsubseteq \transformer{\word{skip}}{\mathcal{D}}{r{\cdot}f} \\
 & \sqsubseteq \word{max}(1,r){\cdot}\transformer{\word{skip}}{\mathcal{D}}{f}
\end{array}
$$ 

\vspace{-.5ex}
\paragraph{Abort} 
$$
\begin{array}{ll}
   & \word{min}(1,r){\cdot}\mathbf{0} \sqsubseteq r{\cdot}\mathbf{0} \sqsubseteq \word{max}(1,r){\cdot}\mathbf{0} \\

   & \proofcomment{\text{Table } \ref{tab:ert}} \\
 \Leftrightarrow & \word{min}(1,r){\cdot}\transformer{\word{abort}}{\mathcal{D}}{f} \sqsubseteq \transformer{\word{abort}}{\mathcal{D}}{r{\cdot}f} \sqsubseteq \\
 & \word{max}(1,r){\cdot}\transformer{\word{abort}}{\mathcal{D}}{f}
\end{array}
$$ 

\vspace{-.5ex}
\paragraph{Assert} 
$$
\begin{array}{ll}
   & \word{min}(1,r){\cdot}\sem{e:\word{true}}{\cdot}f \sqsubseteq r{\cdot}\sem{e:\word{true}}{\cdot}f \\
   & \sqsubseteq \word{max}(1,r){\cdot}\sem{e:\word{true}}{\cdot}f \\

   & \proofcomment{\text{Table } \ref{tab:ert}} \\
 \Leftrightarrow & \word{min}(1,r){\cdot}\transformer{\word{assert} e}{\mathcal{D}}{f} \sqsubseteq \transformer{\word{assert} e}{\mathcal{D}}{r{\cdot}f} \\
 & \sqsubseteq \word{max}(1,r){\cdot}\transformer{\word{assert} e}{\mathcal{D}}{f}
\end{array}
$$ 

\vspace{-.5ex}
\paragraph{Weaken} 
$$
\begin{array}{ll}
   & \word{min}(1,r){\cdot}f \sqsubseteq r{\cdot}f \sqsubseteq \word{max}(1,r){\cdot}f \\

   & \proofcomment{\text{Table } \ref{tab:ert}} \\
 \Leftrightarrow & \word{min}(1,r){\cdot}f \sqsubseteq \transformer{\word{weaken}}{\mathcal{D}}{r{\cdot}f} \sqsubseteq \word{max}(1,r){\cdot}f \\

   & \proofcomment{\text{Table } \ref{tab:ert}} \\
 \Leftrightarrow & \word{min}(1,r){\cdot}\transformer{\word{weaken}}{\mathcal{D}}{f} \sqsubseteq \transformer{\word{weaken}}{\mathcal{D}}{r{\cdot}f} \\
 & \sqsubseteq \word{max}(1,r){\cdot}\transformer{\word{weaken}}{\mathcal{D}}{f}
\end{array}
$$ 

\vspace{-.5ex}
\paragraph{Tick}
$$
\begin{array}{ll}
   & \mathbf{q} + \word{min}(1,r){\cdot}f \sqsubseteq \mathbf{q} + r{\cdot}f \sqsubseteq \mathbf{q} + \word{max}(1,r){\cdot}f \\

 \Rightarrow & \word{min}(1,r)(\mathbf{q} + f) \sqsubseteq \mathbf{q} + r{\cdot}f \sqsubseteq \word{max}(1,r)(\mathbf{q} + f) \\

   & \proofcomment{\text{Table } \ref{tab:ert}} \\
 \Leftrightarrow & \word{min}(1,r){\cdot}\transformer{\word{tick}(q)}{\mathcal{D}}{f} \sqsubseteq \transformer{\word{tick}(q)}{\mathcal{D}}{r{\cdot}f} \\
 & \sqsubseteq \word{max}(1,r){\cdot}\transformer{\word{tick}(q)}{\mathcal{D}}{f}
\end{array}
$$ 

\vspace{-.5ex}
\paragraph{Assignment}
$$
\begin{array}{ll}
   & \word{min}(1,r){\cdot}f[e/x] \sqsubseteq r{\cdot}f[e/x] \sqsubseteq \word{max}(1,r){\cdot}f[e/x] \\

   & \proofcomment{\text{Table } \ref{tab:ert}} \\
 \Leftrightarrow & \word{min}(1,r){\cdot}\transformer{x = e}{\mathcal{D}}{f} \sqsubseteq \transformer{x = e}{\mathcal{D}}{r{\cdot}f} \\
 & \sqsubseteq \word{max}(1,r){\cdot}\transformer{x = e}{\mathcal{D}}{f}
\end{array}
$$ 

\vspace{-.5ex}
\paragraph{Sampling} 
The proof relies on the linearity property of expectations (LPE). 
$$
\begin{array}{ll}
   & \word{min}(1,r){\cdot}\lambda\sigma.\expt{\dist{R}}{\lambda v.f(\sigma[e \word{bop} v/x])} \sqsubseteq \\
   & r{\cdot}\lambda\sigma.\expt{\dist{R}}{\lambda v.f(\sigma[e \word{bop} v/x])} \sqsubseteq \\
   & \word{max}(1,r){\cdot}\lambda\sigma.\expt{\dist{R}}{\lambda v.f(\sigma[e \word{bop} v/x])} \\

 \Leftrightarrow & \word{min}(1,r){\cdot}\lambda\sigma.\expt{\dist{R}}{\lambda v.f(\sigma[e \word{bop} v/x])} \sqsubseteq \\
 & \lambda\sigma.r{\cdot}\expt{\dist{R}}{\lambda v.f(\sigma[e \word{bop} v/x])} \sqsubseteq \\
   & \word{max}(1,r){\cdot}\lambda\sigma.\expt{\dist{R}}{\lambda v.f(\sigma[e \word{bop} v/x])} \\

   & \proofcomment{\text{LPE}} \\
 \Leftrightarrow & \word{min}(1,r){\cdot}\lambda\sigma.\expt{\dist{R}}{\lambda v.f(\sigma[e \word{bop} v/x])} \sqsubseteq \\
 & \lambda\sigma.\expt{\dist{R}}{\lambda v.(r{\cdot}f)(\sigma[e \word{bop} v/x])} \sqsubseteq \\
   & \word{max}(1,r){\cdot}\lambda\sigma.\expt{\dist{R}}{\lambda v.f(\sigma[e \word{bop} v/x])} \\

   & \proofcomment{\text{Table } \ref{tab:ert}} \\
 \Leftrightarrow & \word{min}(1,r){\cdot}\transformer{x = e \word{bop} R}{\mathcal{D}}{f} \sqsubseteq \\
 & \transformer{x = e \word{bop} R}{\mathcal{D}}{r{\cdot}f} \sqsubseteq \\
   & \word{max}(1,r){\cdot}\transformer{x = e \word{bop} R}{\mathcal{D}}{f}
\end{array}
$$

\vspace{-.5ex}
\paragraph{If}
$$
\begin{array}{ll}
   & \proofcomment{\text{By I.H on } c_1} \\
  \Rightarrow & \word{min}(1,r){\cdot}\transformer{c_1}{\mathcal{D}}{f} \sqsubseteq \transformer{c_1}{\mathcal{D}}{r{\cdot}f} \sqsubseteq \\
  & \word{max}(1,r){\cdot}\transformer{c_1}{\mathcal{D}}{f} \\

   & \proofcomment{\text{Algebra}} \\
 \Leftrightarrow & \word{min}(1,r){\cdot}\sem{e:\word{true}}{\cdot}\transformer{c_1}{\mathcal{D}}{f} \sqsubseteq \\
 & \sem{e:\word{true}}{\cdot}\transformer{c_1}{\mathcal{D}}{r{\cdot}f} \sqsubseteq \\
   & \word{max}(1,r){\cdot}\sem{e:\word{true}}{\cdot}\transformer{c_1}{\mathcal{D}}{f} \\

   & \proofcomment{\text{By I.H on } c_2} \\
 \Rightarrow & \word{min}(1,r){\cdot}\sem{e:\word{true}}{\cdot}\transformer{c_1}{\mathcal{D}}{f} + \\
 & \word{min}(1,r){\cdot}\transformer{c_2}{\mathcal{D}}{f} \sqsubseteq \\
   & \sem{e:\word{true}}{\cdot}\transformer{c_1}{\mathcal{D}}{r{\cdot}f} + \transformer{c_2}{\mathcal{D}}{r{\cdot}f} \sqsubseteq \\
   & \word{max}(1,r){\cdot}\sem{e:\word{true}}{\cdot}\transformer{c_1}{\mathcal{D}}{f} + \\
   & \word{max}(1,r){\cdot}\transformer{c_2}{\mathcal{D}}{f}\\

   & \proofcomment{\text{Algebra}} \\
 \Leftrightarrow & \word{min}(1,r){\cdot}\sem{e:\word{true}}{\cdot}\transformer{c_1}{\mathcal{D}}{f} + \\
 & \word{min}(1,r){\cdot}\sem{e:\word{false}}{\cdot}\transformer{c_2}{\mathcal{D}}{f} \sqsubseteq \\
   & \sem{e:\word{true}}{\cdot}\transformer{c_1}{\mathcal{D}}{r{\cdot}f} + \sem{e:\word{false}}{\cdot}\transformer{c_2}{\mathcal{D}}{r{\cdot}f} \sqsubseteq \\
   & \word{max}(1,r){\cdot}\sem{e:\word{true}}{\cdot}\transformer{c_1}{\mathcal{D}}{f} + \\
   & \word{max}(1,r){\cdot}\sem{e:\word{false}}{\cdot}\transformer{c_2}{\mathcal{D}}{f}\\

   & \proofcomment{\text{Table } \ref{tab:ert}} \\
 \Leftrightarrow & \word{min}(1,r){\cdot}\transformer{\word{if} e \; c_1 \word{else} c_2}{\mathcal{D}}{f} \sqsubseteq \\
 & \transformer{\word{if} e \; c_1 \word{else} c_2}{\mathcal{D}}{r{\cdot}f} \sqsubseteq  \\
   & \word{max}(1,r){\cdot}\transformer{\word{if} e \; c_1 \word{else} c_2}{\mathcal{D}}{f}
\end{array}
$$ 

\vspace{-.5ex}
\paragraph{Nondeterministic branching} 
$$
\begin{array}{ll}
   & \proofcomment{\text{By I.H on } c_1} \\
  \Rightarrow & \word{min}(1,r){\cdot}\transformer{c_1}{\mathcal{D}}{f} \sqsubseteq \transformer{c_1}{\mathcal{D}}{r{\cdot}f} \sqsubseteq \\
  & \word{max}(1,r){\cdot}\transformer{c_1}{\mathcal{D}}{f} \\

   & \proofcomment{\text{By I.H on } c_2} \\
 \Rightarrow & \word{min}(1,r){\cdot}\transformer{c_2}{\mathcal{D}}{f} \sqsubseteq \transformer{c_2}{\mathcal{D}}{r{\cdot}f} \sqsubseteq \\
 & \word{max}(1,r){\cdot}\transformer{c_2}{\mathcal{D}}{f} \\

   & \proofcomment{\text{Algebra}} \\
 \Leftrightarrow & \word{max}\{\word{min}(1,r){\cdot}\transformer{c_1}{\mathcal{D}}{f}, \\
 & \word{min}(1,r){\cdot}\transformer{c_2}{\mathcal{D}}{f}\} \\ 
 & \sqsubseteq \word{max}\{\transformer{c_1}{\mathcal{D}}{r{\cdot}f}, \transformer{c_2}{\mathcal{D}}{r{\cdot}f}\} \sqsubseteq \\
   & \word{max}\{\word{max}(1,r){\cdot}\transformer{c_1}{\mathcal{D}}{f}, \\
   & \word{max}(1,r){\cdot}\transformer{c_2}{\mathcal{D}}{f}\} \\

 \Leftrightarrow & \word{min}(1,r){\cdot}\word{max}\{\transformer{c_1}{\mathcal{D}}{f}, \transformer{c_2}{\mathcal{D}}{f}\} \sqsubseteq \\
   & \word{max}\{\transformer{c_1}{\mathcal{D}}{r{\cdot}f}, \transformer{c_2}{\mathcal{D}}{r{\cdot}f}\} \sqsubseteq \\
   & \word{max}(1,r){\cdot}\word{max}\{\transformer{c_1}{\mathcal{D}}{f}, \transformer{c_2}{\mathcal{D}}{f}\} \\

   & \proofcomment{\text{Table } \ref{tab:ert}} \\
 \Leftrightarrow & \word{min}(1,r){\cdot}\transformer{\word{if} \star \; c_1 \word{else} c_2}{\mathcal{D}}{f} \sqsubseteq \\
 & \transformer{\word{if} \star \; c_1 \word{else} c_2}{\mathcal{D}}{r{\cdot}f} \sqsubseteq  \\
   & \word{max}(1,r){\cdot}\transformer{\word{if} \star \; c_1 \word{else} c_2}{\mathcal{D}}{f}
\end{array}
$$ 

\vspace{-.5ex}
\paragraph{Probabilistic branching} 
$$
\begin{array}{ll}
   & \proofcomment{\text{By I.H on } c_1} \\
  \Rightarrow & \word{min}(1,r){\cdot}\transformer{c_1}{\mathcal{D}}{f} \sqsubseteq \transformer{c_1}{\mathcal{D}}{r{\cdot}f} \sqsubseteq \\
  & \word{max}(1,r){\cdot}\transformer{c_1}{\mathcal{D}}{f} \\

   & \proofcomment{\text{Algebra}} \\
 \Leftrightarrow & \word{min}(1,r){\cdot}p{\cdot}\transformer{c_1}{\mathcal{D}}{f} \sqsubseteq p{\cdot}\transformer{c_1}{\mathcal{D}}{r{\cdot}f} \sqsubseteq \\
 & \word{max}(1,r){\cdot}p{\cdot}\transformer{c_1}{\mathcal{D}}{f} \\

   & \proofcomment{\text{By I.H on } c_2} \\
 \Rightarrow & \word{min}(1,r){\cdot}\transformer{c_2}{\mathcal{D}}{f} \sqsubseteq \transformer{c_2}{\mathcal{D}}{r{\cdot}f} \sqsubseteq \\
 & \word{max}(1,r){\cdot}\transformer{c_2}{\mathcal{D}}{f} \\

   & \proofcomment{\text{Algebra}} \\
 \Leftrightarrow & \word{min}(1,r){\cdot}(1-p){\cdot}\transformer{c_2}{\mathcal{D}}{f} \sqsubseteq \\
 & (1-p){\cdot}\transformer{c_2}{\mathcal{D}}{r{\cdot}f} \sqsubseteq \\
 & \word{max}(1,r){\cdot}(1-p){\cdot}\transformer{c_2}{\mathcal{D}}{f} \\

   & \proofcomment{\text{By composing}} \\
 \Leftrightarrow & \word{min}(1,r){\cdot}(p{\cdot}\transformer{c_1}{\mathcal{D}}{f} + (1-p){\cdot}\transformer{c_2}{\mathcal{D}}{f}) \sqsubseteq \\
   & p{\cdot}\transformer{c_1}{\mathcal{D}}{r{\cdot}f} + (1-p){\cdot}\transformer{c_2}{\mathcal{D}}{r{\cdot}f} \sqsubseteq \\
   & \word{max}(1,r){\cdot}(p{\cdot}\transformer{c_1}{\mathcal{D}}{f} + (1-p){\cdot}\transformer{c_2}{\mathcal{D}}{f})\\

   & \proofcomment{\text{Table } \ref{tab:ert}} \\
 \Leftrightarrow & \word{min}(1,r){\cdot}\transformer{c_1 \oplus_{p} c_2}{\mathcal{D}}{f} \sqsubseteq \transformer{c_1 \oplus_{p} c_2}{\mathcal{D}}{r{\cdot}f} \\
 & \sqsubseteq \word{max}(1,r){\cdot}\transformer{c_1 \oplus_{p} c_2}{\mathcal{D}}{f}
\end{array}
$$ 

\vspace{-.5ex}
\paragraph{Sequence} 
$$
\begin{array}{ll}
   & \proofcomment{\text{By I.H on } c_2} \\
 \Rightarrow & \word{min}(1,r){\cdot}\transformer{c_2}{\mathcal{D}}{f} \sqsubseteq \transformer{c_2}{\mathcal{D}}{r{\cdot}f} \sqsubseteq \\
 & \word{max}(1,r){\cdot}\transformer{c_2}{\mathcal{D}}{f} \\

   & \proofcomment{\text{Monotonicity of \word{ert}}} \\
 \Rightarrow & \transformer{c_1}{\mathcal{D}}{\word{min}(1,r){\cdot}\transformer{c_2}{\mathcal{D}}{f}} \sqsubseteq \\
 & \transformer{c_1}{\mathcal{D}}{\transformer{c_2}{\mathcal{D}}{r{\cdot}f}} \sqsubseteq \\
   & \transformer{c_1}{\mathcal{D}}{\word{max}(1,r){\cdot}\transformer{c_2}{\mathcal{D}}{f}} \\

   & \proofcomment{\text{By I.H on } c_1} \\
 \Rightarrow & \word{min}(1,r){\cdot}\transformer{c_1}{\mathcal{D}}{\transformer{c_2}{\mathcal{D}}{f}} \sqsubseteq \\
 & \transformer{c_1}{\mathcal{D}}{\transformer{c_2}{\mathcal{D}}{r{\cdot}f}} \sqsubseteq \\
   & \word{max}(1,r){\cdot}\transformer{c_1}{\mathcal{D}}{\transformer{c_2}{\mathcal{D}}{f}} \\

   & \proofcomment{\text{Table } \ref{tab:ert}} \\
 \Leftrightarrow & \word{min}(1,r){\cdot}\transformer{c_1; c_2}{\mathcal{D}}{f} \sqsubseteq \transformer{c_1; c_2}{\mathcal{D}}{r{\cdot}f} \sqsubseteq \\
 & \word{max}(1,r){\cdot}\transformer{c_1; c_2}{\mathcal{D}}{f}
\end{array}
$$ 

\vspace{-.5ex}
\paragraph{Loop}
Consider the characteristic function w.r.t the expectation $f$
$$
F_{f} \defineas \sem{e:{\word{true}}}{\cdot}\transformer{c}{\mathcal{D}}{X} + \sem{e:\word{false}}{\cdot}f
$$
For all $n \in \mathbb{N}$, we first need to show that 
$$
\word{min}(1,r){\cdot}{\word{sup}}_{n}F^{n}_{f}(\mathbf{0}) \sqsubseteq {\word{sup}}_{n}F^{n}_{r{\cdot}f}(\mathbf{0}) \sqsubseteq \word{max}(1,r){\cdot}{\word{sup}}_{n}F^{n}_{f}(\mathbf{0})
$$
where $F^{0}_{f} \defineas \word{id}$ and $F^{k+1}_f \defineas F_f \circ F^{k}_f$. Because $F_f$ and $F_{r{\cdot}f}$ are monotone because of the monotonicity of \word{ert}, then using Kleene’s Fixed Point Theorem, if holds that 
$$
\begin{array}{l}
\word{min}(1,r){\cdot}\transformer{\word{while} e \; c}{\mathcal{D}}{f} \sqsubseteq \transformer{\word{while} e \; c}{\mathcal{D}}{r{\cdot}f} \\
\sqsubseteq \word{max}(1,r){\cdot}\transformer{\word{while} e \; c}{\mathcal{D}}{f}
\end{array}
$$
To prove we need to show the following holds for all $n \in \mathbb{N}$
$$
\word{min}(1,r){\cdot}F^{n}_{f}(\mathbf{0}) \sqsubseteq F^{n}_{r{\cdot}f}(\mathbf{0}) \sqsubseteq 
\word{max}(1,r){\cdot}F^{n}_{f}(\mathbf{0})
$$
The proof is done by induction on the natural value $n$. 
\begin{itemize}
	\item \emph{Base case.} For $n = 0$, we have
$$
\begin{array}{ll}
   & \word{min}(1,r){\cdot}F^{0}_{f}(\mathbf{0}) \sqsubseteq F^{0}_{r{\cdot}f}(\mathbf{0}) \sqsubseteq \word{max}(1,r){\cdot}F^{0}_{f}(\mathbf{0}) \\

 \Leftrightarrow & \word{min}(1,r){\cdot}\mathbf{0} \sqsubseteq \mathbf{0} \sqsubseteq \word{max}(1,r){\cdot}\mathbf{0} \\
 \Leftrightarrow & \mathbf{0} \sqsubseteq \mathbf{0} \sqsubseteq \mathbf{0}
\end{array}
$$ 
	\item \emph{Induction case.} Assume that 
$$
\begin{array}{ll}
   & \proofcomment{\text{By I.H on } n} \\
   & \word{min}(1,r){\cdot}F^{n}_{f}(\mathbf{0}) \sqsubseteq F^{n}_{r{\cdot}f}(\mathbf{0}) \sqsubseteq \word{max}(1,r){\cdot}F^{n}_{f}(\mathbf{0}) \\

   & \proofcomment{\text{Monotonicity of \word{ert}}} \\
 \Rightarrow & \transformer{c}{\mathcal{D}}{\word{min}(1,r){\cdot}F^{n}_{f}(\mathbf{0})} \sqsubseteq \transformer{c}{\mathcal{D}}{F^{n}_{r{\cdot}f}(\mathbf{0})} \sqsubseteq \\
 & \transformer{c}{\mathcal{D}}{\word{max}(1,r){\cdot}F^{n}_{f}(\mathbf{0})} \\

   & \proofcomment{\text{By I.H on } c} \\
 \Rightarrow & \word{min}(1,r){\cdot}\transformer{c}{\mathcal{D}}{F^{n}_{f}(\mathbf{0})} \sqsubseteq \transformer{c}{\mathcal{D}}{F^{n}_{r{\cdot}f}(\mathbf{0})} \sqsubseteq \\
 & \word{max}(1,r){\cdot}\transformer{c}{\mathcal{D}}{F^{n}_{f}(\mathbf{0})} \\

   & \proofcomment{\text{Algebra}} \\
 \Leftrightarrow & \word{min}(1,r){\cdot}\sem{e:\word{true}}{\cdot}\transformer{c}{\mathcal{D}}{F^{n}_{f}(\mathbf{0})} \sqsubseteq \\
 & \sem{e:\word{true}}{\cdot}\transformer{c}{\mathcal{D}}{F^{n}_{r{\cdot}f}(\mathbf{0})} \sqsubseteq \\
   & \word{max}(1,r){\cdot}\sem{e:\word{true}}{\cdot}\transformer{c}{\mathcal{D}}{F^{n}_{f}(\mathbf{0})} \\

   & \proofcomment{\text{By I.H on } n} \\
 \Rightarrow & \sem{e:\word{false}}{\cdot}\word{min}(1,r){\cdot}F^{n}_{f}(\mathbf{0}) + \\
 & \word{min}(1,r){\cdot}\sem{e:\word{true}}{\cdot}\transformer{c}{\mathcal{D}}{F^{n}_{f}(\mathbf{0})} \sqsubseteq \\
   & \sem{e:\word{false}}{\cdot}F^{n}_{r{\cdot}f}(\mathbf{0}) + \sem{e:\word{true}}{\cdot}\transformer{c}{\mathcal{D}}{F^{n}_{r{\cdot}f}(\mathbf{0})} \sqsubseteq \\
   & \sem{e:\word{false}}{\cdot}\word{max}(1,r){\cdot}F^{n}_{f}(\mathbf{0}) + \\
   & \word{max}(1,r){\cdot}\sem{e:\word{true}}{\cdot}\transformer{c}{\mathcal{D}}{F^{n}_{f}(\mathbf{0})} \\

   & \proofcomment{\text{Definitions of } F_f \text{ and } F_{r{\cdot}f}} \\
 \Leftrightarrow & \word{min}(1,r){\cdot}F^{n+1}_f(\mathbf{0}) \sqsubseteq F^{n+1}_{r{\cdot}f}(\mathbf{0}) \sqsubseteq \word{max}(1,r){\cdot}F^{n+1}_f(\mathbf{0})

\end{array}
$$ 
\end{itemize}

\vspace{-.5ex}
\paragraph{Procedure call}
Because ${\word{call}}^{\mathcal{D}}_{n}P$ (defined in Section \ref{subsec:supfixpoint}) is closed command for all $n \in \mathbb{N}$, by I.H we get
$$
\begin{array}{l}
\word{min}(1,r){\cdot}{\word{ert}}[{\word{call}}^{\mathcal{D}}_{n}P](f) \sqsubseteq {\word{ert}}[{\word{call}}^{\mathcal{D}}_{n}P](r{\cdot}f) \sqsubseteq \\
\word{max}(1,r){\cdot}{\word{ert}}[{\word{call}}^{\mathcal{D}}_{n}P](f)
\end{array}
$$
On the other hand, by Theorem \ref{theo:limitapproximation}, $\word{ert}[{\word{call}P,\mathcal{D}}](r{\cdot}f) = {\word{sup}}_{n}{\word{ert}}[{\word{call}}^{\mathcal{D}}_{n}P](r{\cdot}f)$, where ${\word{call}}^{\mathcal{D}}_{n}P$ is the $n^{th}$-inlining of procedure call, we have 
$$
\begin{array}{ll}
   & {\word{sup}}_{n}(\word{min}(1,r){\cdot}{\word{ert}}[{\word{call}}^{\mathcal{D}}_{n}P](f)) \sqsubseteq \\
   & {\word{sup}}_{n}{\word{ert}}[{\word{call}}^{\mathcal{D}}_{n}P](r{\cdot}f) \sqsubseteq \\
   & {\word{sup}}_{n}(\word{max}(1,r){\cdot}{\word{ert}}[{\word{call}}^{\mathcal{D}}_{n}P](f)) \\

 \Leftrightarrow & {\word{sup}}_{n}(\word{min}(1,r){\cdot}{\word{ert}}[{\word{call}}^{\mathcal{D}}_{n}P](f)) \sqsubseteq \\
 & \word{ert}[{\word{call}P,\mathcal{D}}](r{\cdot}f) \sqsubseteq \\
   & {\word{sup}}_{n}(\word{max}(1,r){\cdot}{\word{ert}}[{\word{call}}^{\mathcal{D}}_{n}P](f)) \\

 \Leftrightarrow & \word{min}(1,r){\cdot}{\word{sup}}_{n}({\word{ert}}[{\word{call}}^{\mathcal{D}}_{n}P](f)) \sqsubseteq \\
 & \word{ert}[{\word{call}P,\mathcal{D}}](r{\cdot}f) \sqsubseteq \\
   & \word{max}(1,r){\cdot}{\word{sup}}_{n}({\word{ert}}[{\word{call}}^{\mathcal{D}}_{n}P](f)) \\

   & \proofcomment{\text{Theorem \ref{theo:limitapproximation}}} \\
 \Leftrightarrow & \word{min}(1,r){\cdot}\word{ert}[{\word{call}P,\mathcal{D}}](f) \sqsubseteq \word{ert}[{\word{call}P,\mathcal{D}}](r{\cdot}f) \sqsubseteq \\
   & \word{max}(1,r){\cdot}\word{ert}[{\word{call}P,\mathcal{D}}](f) \\
\end{array}
$$ 
\vspace{-.5ex}
\subsubsection{Preservation of infinity}
The proof is done by induction on the structure of the command $c$. Note that $c$ is \emph{abort} free, that is, $c$ contains no \word{abort} commands.

\vspace{-.5ex}
\paragraph{Skip} 
$$
\begin{array}{ll}
& \transformer{\word{skip}}{\mathcal{D}}{\infty} \\
   & \proofcomment{\text{Table } \ref{tab:ert}} \\
 = & \infty
\end{array}
$$

\vspace{-.5ex}
\paragraph{Assert} 
$$
\begin{array}{ll}
& \transformer{\word{assert} e}{\mathcal{D}}{\infty} \\
   & \proofcomment{\text{Table } \ref{tab:ert}} \\
 = & \sem{e:\word{true}}{\cdot}\infty = \infty
\end{array}
$$ 

\vspace{-.5ex}
\paragraph{Weaken} 
$$
\begin{array}{ll}
& \transformer{\word{weaken}}{\mathcal{D}}{\infty} \\
   & \proofcomment{\text{Table } \ref{tab:ert}} \\
 = & \infty
\end{array}
$$

\vspace{-.5ex}
\paragraph{Tick}
$$
\begin{array}{ll}
& \transformer{\word{tick}(q)}{\mathcal{D}}{\infty} \\
   & \proofcomment{\text{Table } \ref{tab:ert}} \\
 = & \mathbf{q} + \infty = \infty
\end{array}
$$ 

\vspace{-.5ex}
\paragraph{Assignment}
$$
\begin{array}{ll}
& \transformer{x = e}{\mathcal{D}}{\infty} \\
   & \proofcomment{\text{Table } \ref{tab:ert}} \\
 = & \infty[e/x] = \infty
\end{array}
$$ 

\vspace{-.5ex}
\paragraph{Sampling} 
$$
\begin{array}{ll}
& \transformer{x = e \word{bop} R}{\mathcal{D}}{\infty} \\
   & \proofcomment{\text{Table } \ref{tab:ert}} \\
 = & \lambda\sigma.\expt{\dist{R}}{\lambda v.(\infty)(\sigma[e \word{bop} v/x])} \\ 
 = & \expt{\dist{R}}{\infty} = \infty \\ 
\end{array}
$$ 

\vspace{-.5ex}
\paragraph{If}
$$
\begin{array}{ll}
& \transformer{\word{if} e \; c_1 \word{else} c_2}{\mathcal{D}}{\infty} \\
   & \proofcomment{\text{Table } \ref{tab:ert}} \\
 = & \sem{e:\word{true}}{\cdot}\transformer{c_1}{\mathcal{D}}{\infty} + \\
   & \sem{e:\word{false}}{\cdot}\transformer{c_2}{\mathcal{D}}{\infty} \\

   & \proofcomment{\text{By I.H on } c_1 \text{ and } c_2} \\
 = & \sem{e:\word{true}}{\cdot}\infty + \sem{e:\word{false}}{\cdot}\infty \\
 = & \infty 
\end{array}
$$ 

\vspace{-.5ex}
\paragraph{Nondeterministic branching} 
$$
\begin{array}{ll}
& \transformer{\word{if} \star \; c_1 \word{else} c_2}{\mathcal{D}}{\infty} \\
   & \proofcomment{\text{Table } \ref{tab:ert}} \\
 = & {\word{max}}\{\transformer{c_1}{\mathcal{D}}{\infty},\transformer{c_2}{\mathcal{D}}{\infty}\} \\

   & \proofcomment{\text{By I.H on } c_1 \text{ and } c_2} \\
 = & {\word{max}}\{\infty,\infty\} \\
 = & \infty 
\end{array}
$$ 

\vspace{-.5ex}
\paragraph{Probabilistic branching} 
$$
\begin{array}{ll}
& \transformer{c_1 \oplus_{p} c_2}{\mathcal{D}}{\infty} \\
   & \proofcomment{\text{Table } \ref{tab:ert}} \\
 = & p{\cdot}\transformer{c_1}{\mathcal{D}}{\infty} + (1-p){\cdot}\transformer{c_2}{\mathcal{D}}{\infty} \\

   & \proofcomment{\text{By I.H on } c_1 \text{ and } c_2} \\
 = & p{\cdot}\infty + (1-p){\cdot}\infty = 1{\cdot}\infty \\
 = & \infty 
\end{array}
$$ 

\vspace{-.5ex}
\paragraph{Sequence} 
$$
\begin{array}{ll}
& \transformer{c_1; c_2}{\mathcal{D}}{\infty} \\
   & \proofcomment{\text{Table } \ref{tab:ert}} \\
 = & \transformer{c_1}{\mathcal{D}}{\transformer{c_2}{\mathcal{D}}{\infty}} \\

   & \proofcomment{\text{By I.H on } c_2} \\
 = & \transformer{c_1}{\mathcal{D}}{\infty} \\

   & \proofcomment{\text{By I.H on } c_1} \\
 = & \infty 
\end{array}
$$ 

\vspace{-.5ex}
\paragraph{Loop}
Consider the characteristic function w.r.t the expectation $\infty$
$$
F_{\infty} \defineas \sem{e:{\word{true}}}{\cdot}\transformer{c}{\mathcal{D}}{X} + \sem{e:\word{false}}{\cdot}\infty
$$
By Theorem \ref{theo:loopsup}, we have $\transformer{\word{while} e \; c}{\mathcal{D}}{\infty} = {\word{sup}}_{n}F^{n}_{\infty}(\mathbf{0})$, thus we show ${\word{sup}}_{n}F^{n}_{\infty}(\mathbf{0}) = \infty$. The proof is done by contradiction. Assume that ${\word{sup}}_{n}F^{n}_{\infty}(\mathbf{0}) < \infty$. Hence, there exists $M < \infty$ such that $\forall n \in \mathbb{N}.\sigma \in \Sigma. F^{n}_{\infty}(\mathbf{0})(\sigma) \leq M$. By definition of $F^{n}_{\infty}(\mathbf{0})$, it is 
$$
\sem{e:{\word{true}}}(\sigma){\cdot}\transformer{c}{\mathcal{D}}{F^{n-1}_{\infty}(\mathbf{0})}(\sigma) + \sem{e:\word{false}}(\sigma){\cdot}\infty
$$
If $\sigma \not \models e$ then $\sem{e:{\word{true}}}(\sigma) = 0$ and $\sem{e:\word{false}}(\sigma) = 1$. We get $F^{n}_{\infty}(\mathbf{0})(\sigma) = \infty$. Therefore, the assumption is contradictory, or $\transformer{\word{while} e \; c}{\mathcal{D}}{\infty} = \infty$.

\vspace{-.5ex}
\paragraph{Procedure call}
Because ${\word{call}}^{\mathcal{D}}_{n}P$ (defined in Section \ref{subsec:supfixpoint}) is closed command for all $n \in \mathbb{N}$, by I.H we get
$$
{\word{ert}}[{\word{call}}^{\mathcal{D}}_{n}P](\infty) = \infty
$$
On the other hand, by Theorem \ref{theo:limitapproximation}, we have 
$$
\begin{array}{ll}
& \transformer{\word{call}P}{\mathcal{D}}{\infty} \\
 = & {\word{sup}}_{n}{\word{ert}}[{\word{call}}^{\mathcal{D}}_{n}P](\infty) \\

   & \proofcomment{\text{Observation above}} \\
 = & {\word{sup}}_{n}\infty \\
 = & \infty
\end{array}
$$

\vspace{-1.5ex}
\section{Simulation-based comparison}
\label{app:simulation}
We measured the expected numbers of ticks for the collected $\numexamp$ challenging examples 
with different looping and recursion patterns that use probabilistic branching and sampling 
assignments. Then we compared the results to our computed bounds. 
The following figures show the complete list of plots of these comparisons. In the plots, 
we draw a candlestick chart for each example as well as the measured mean value and the 
corresponding inferred bound of the number of ticks. The candlesticks represent the highest 
($h$), lowest ($l$) measured values, $75\%$, and $25\%$ of the interval $(h-l)$. 

Our experiments indicate that the computed bounds are close to the measured expected values needed 
for both linear and polynomial expected bound programs. 
However, there is no guarantee that \toolname{} infers asymptotically tight bounds 
and there is many classes of bounds that \toolname{} cannot derive. For example, 
\progname{C4B\_t15} has expected logarithmic 
expected cost, thus the best bound that \toolname{} can derive is a linear bound. 
Similarly, $\interval{{0}{,}{n}}{+}\interval{{0}{,}{m}}$ is the best bound that can be inferred 
for \progname{condand} whose expected cost defined by function 
$2{\cdot}\text{min}\{\interval{{0}{,}{n}},\interval{{0}{,}{m}}\}$. Another source of imprecise constant 
factors in the bounds is rounding. Program \progname{robot} has an imprecise constant 
factor because it uses a profound depth of nested probabilistic choice. 

Some programs whose worst-case resource usage is only triggered by some 
particular inputs, called \emph{worst-case inputs}. Thus, this makes the difference between the measured values and the inferred bounds big. For instance, the worst-case inputs of \progname{C4B\_t30} are values of $x$ and $y$ such that they are equal. While we measured the expected values with $x$ varying from $1000$ to $5000$ and $y$ is fixed to be $3000$.
\begin{figure}[th!]
\centering
\includegraphics[width=0.5\textwidth]{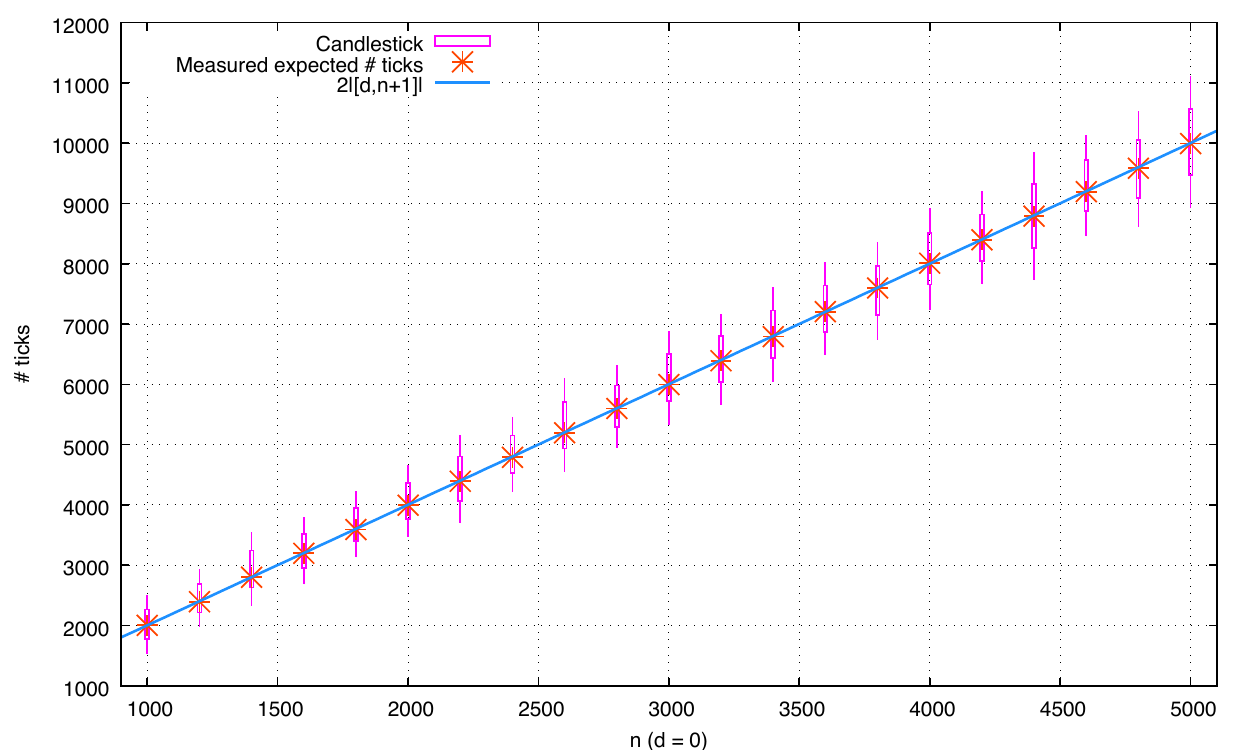}
\caption{Example \progname{2drwalk}.}
\label{fig:2drwalk}
\end{figure}
\begin{figure}[th!]
\centering
\includegraphics[width=0.5\textwidth]{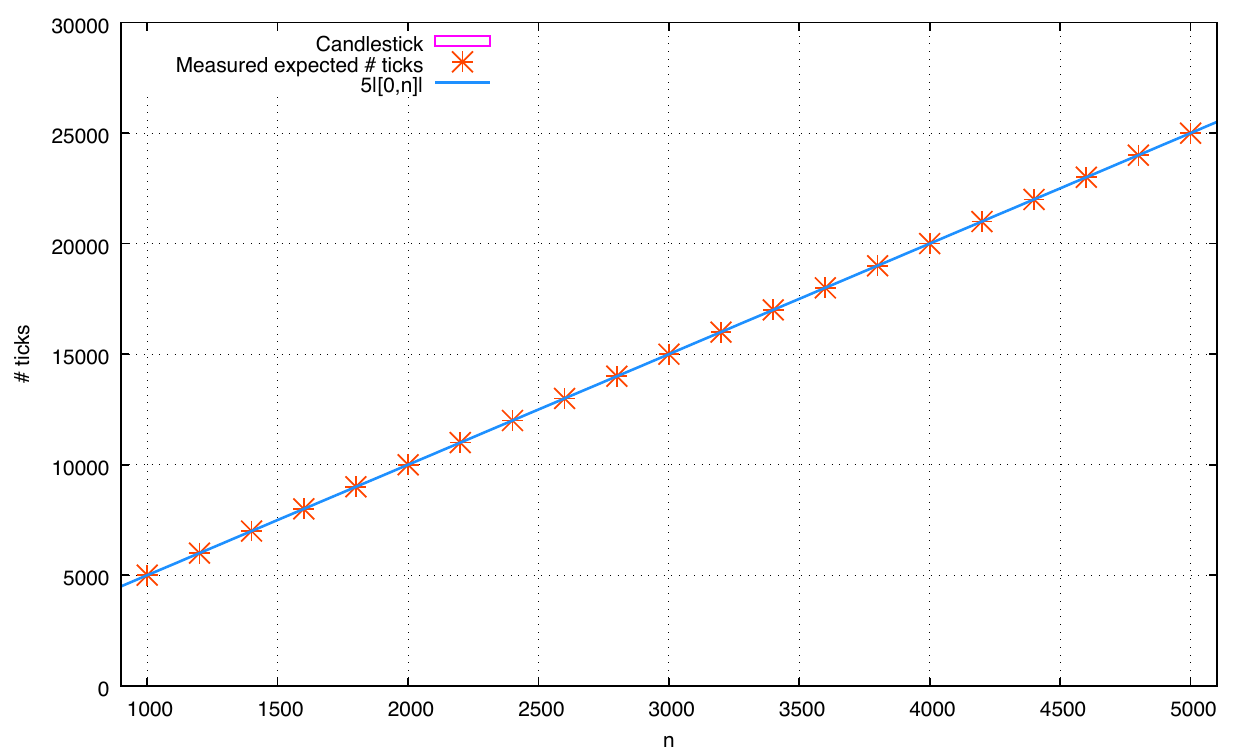}
\caption{Example \progname{bayesian}.}
\label{fig:bayesian}
\end{figure}
\begin{figure}[th!]                                                 
\centering                                                     
\includegraphics[width=0.5\textwidth]{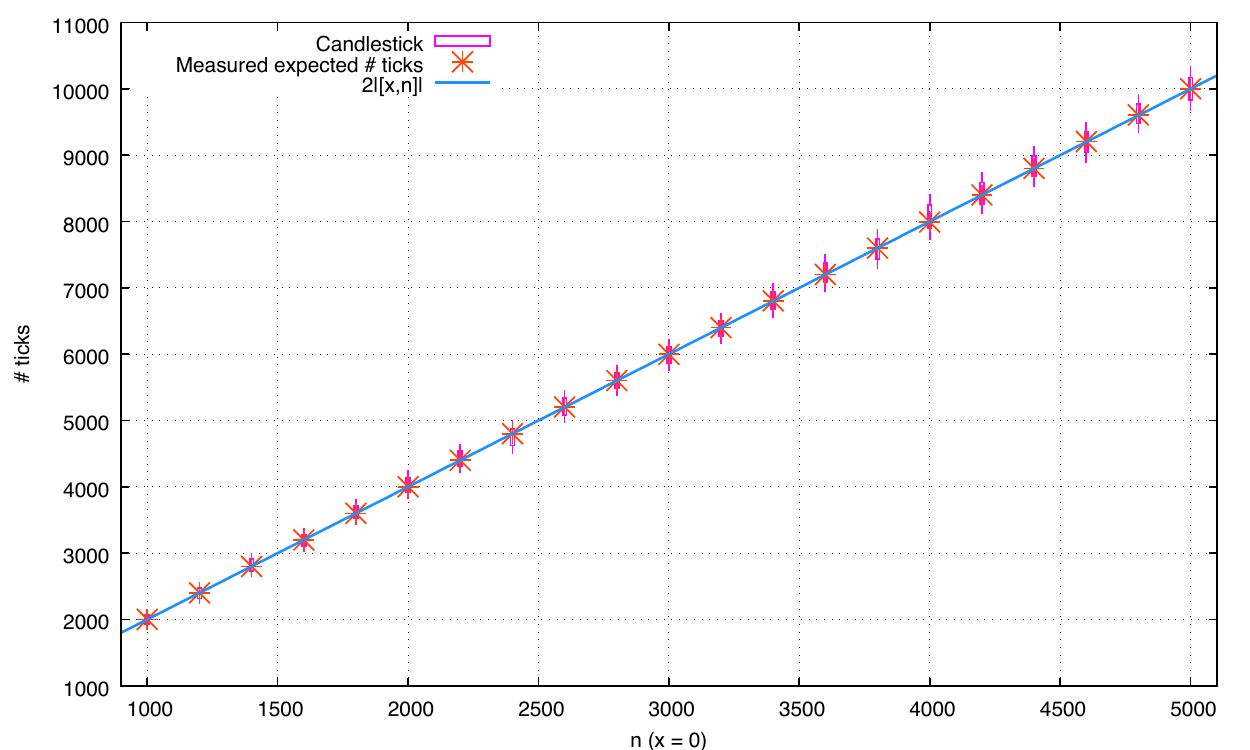}     
\caption{Example \progname{ber}.}                                                                           
\label{fig:ber}                                           
\end{figure}
\begin{figure}[th!]                                                                                              
\centering                                                     
\includegraphics[width=0.5\textwidth]{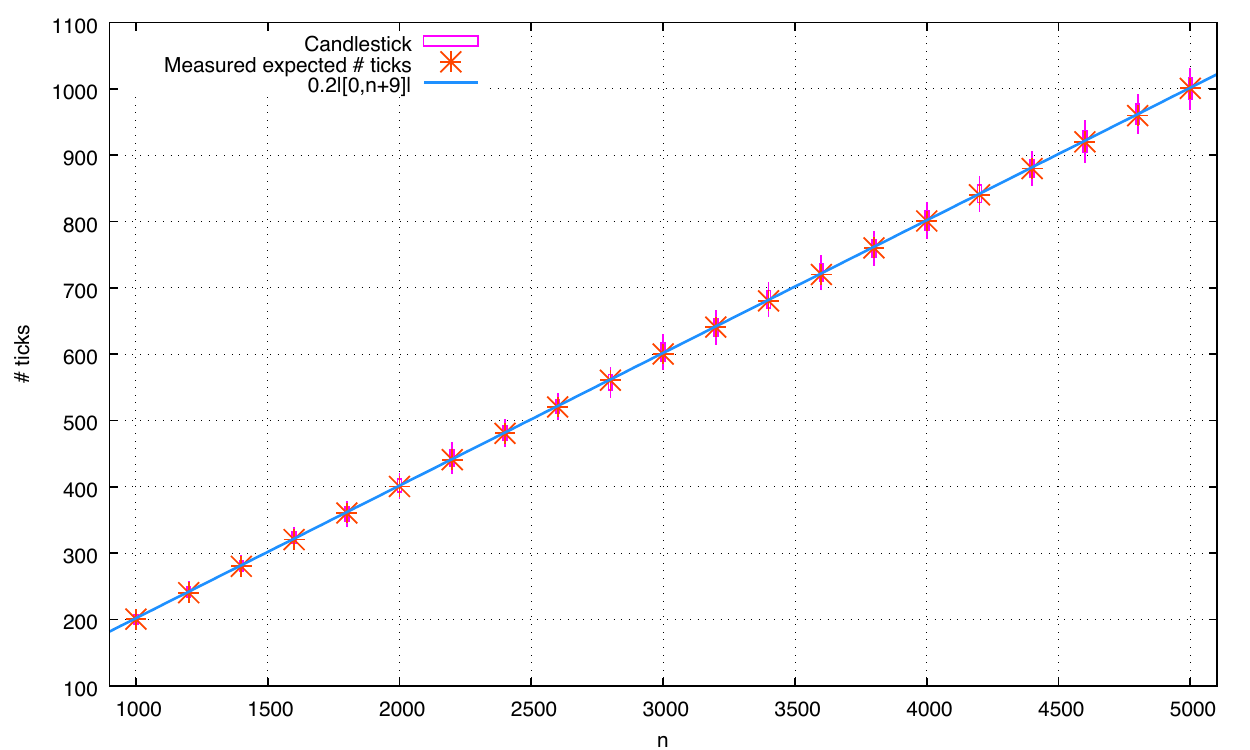}                                                       
\caption{Example \progname{bin}.}                                                                                
\label{fig:bin}                                                
\end{figure}

\clearpage

\begin{figure}[th!]                                                                                              
\centering                                                     
\includegraphics[width=0.5\textwidth]{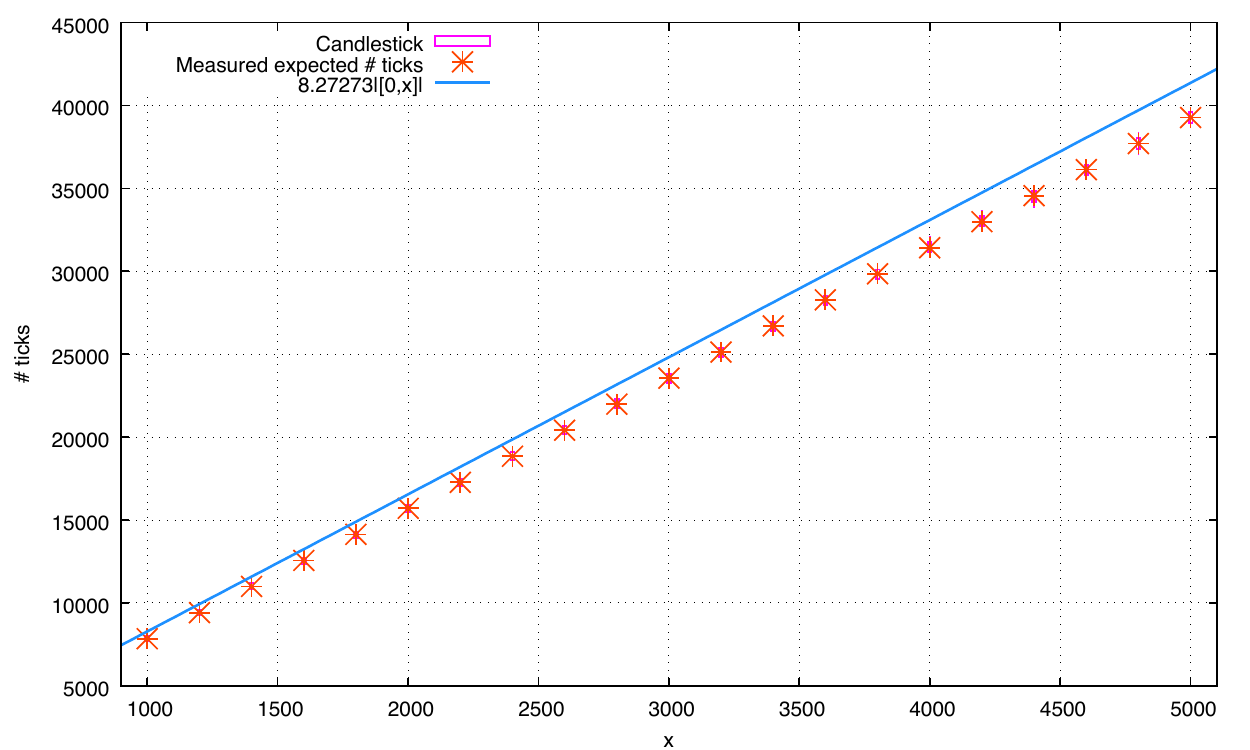}                                                       
\caption{Example \progname{C4B\_t09}.}                                                                                
\label{fig:c4bt09}                                                
\end{figure}                                                   
\begin{figure}[th!]                                            
\centering                                                     
\includegraphics[width=0.5\textwidth]{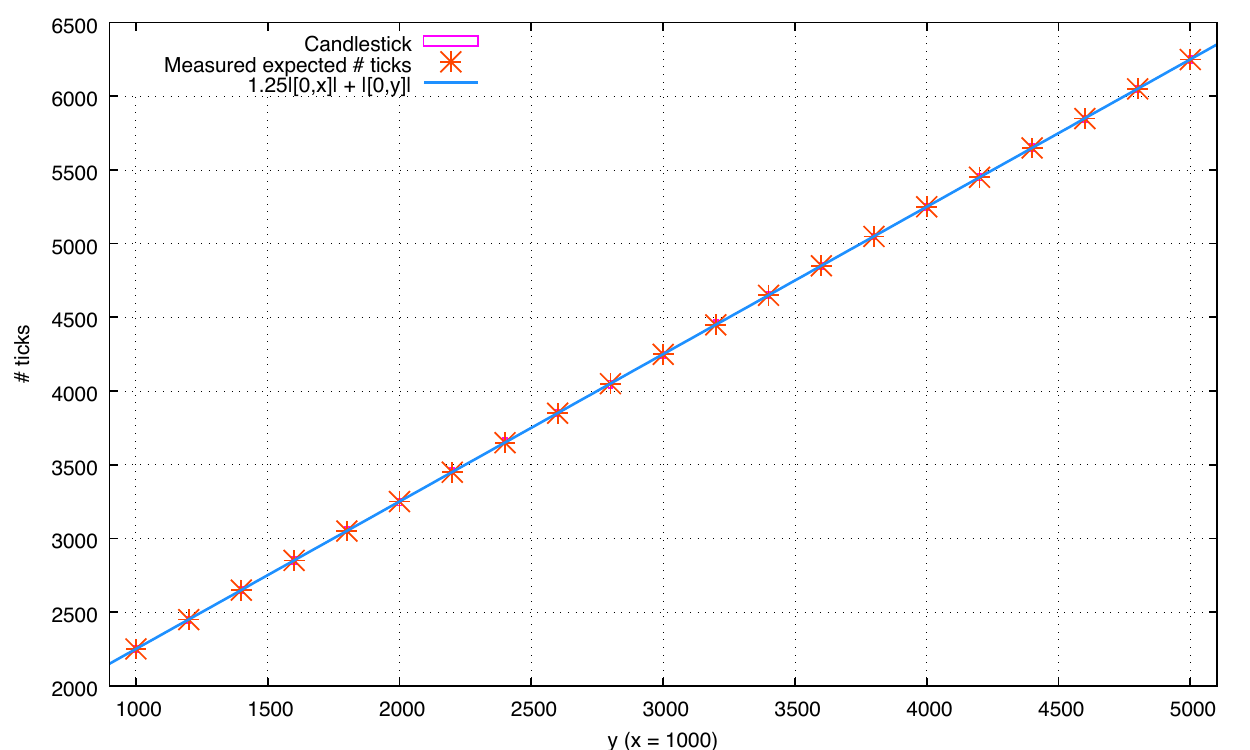} 
\caption{Example \progname{C4B\_t13}.}                         
\label{fig:c4bt13}                                             
\end{figure}
\begin{figure}[th!]                                            
\centering                                                     
\includegraphics[width=0.5\textwidth]{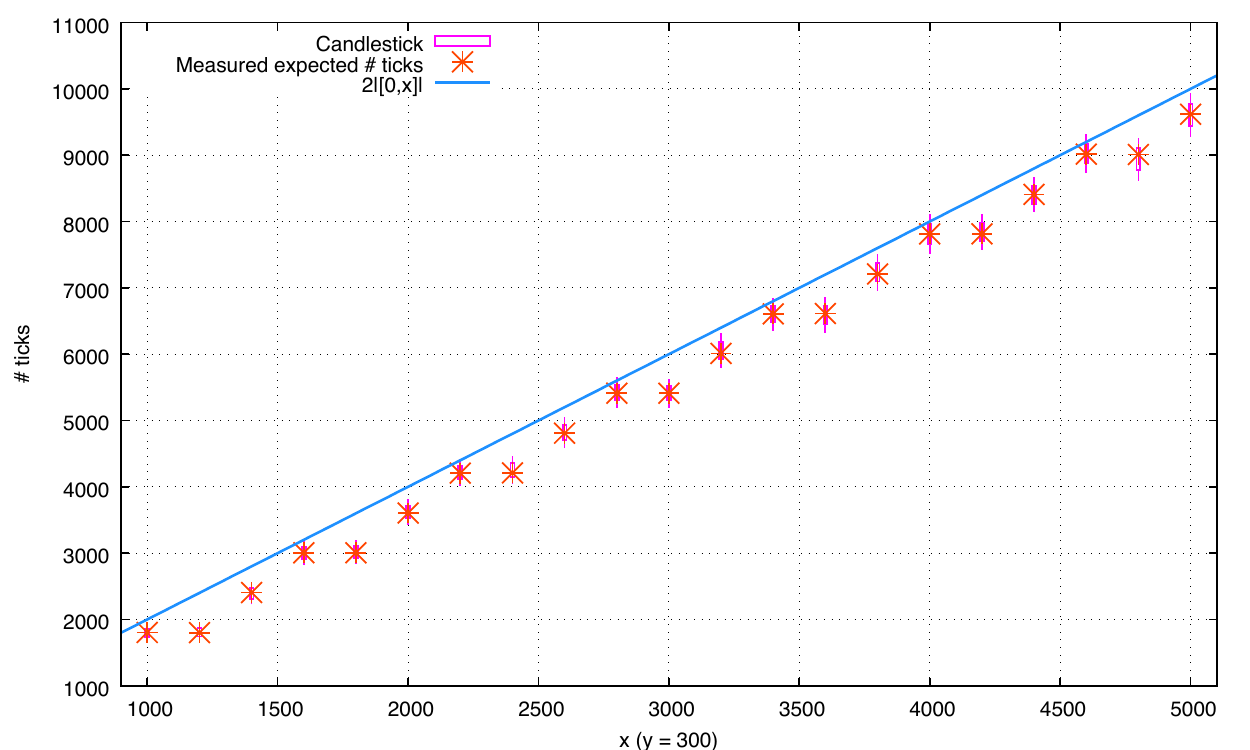} 
\caption{Example \progname{C4B\_t15}.}                         
\label{fig:c4bt15}                                             
\end{figure}
\begin{figure}[th!]                                            
\centering                                                     
\includegraphics[width=0.5\textwidth]{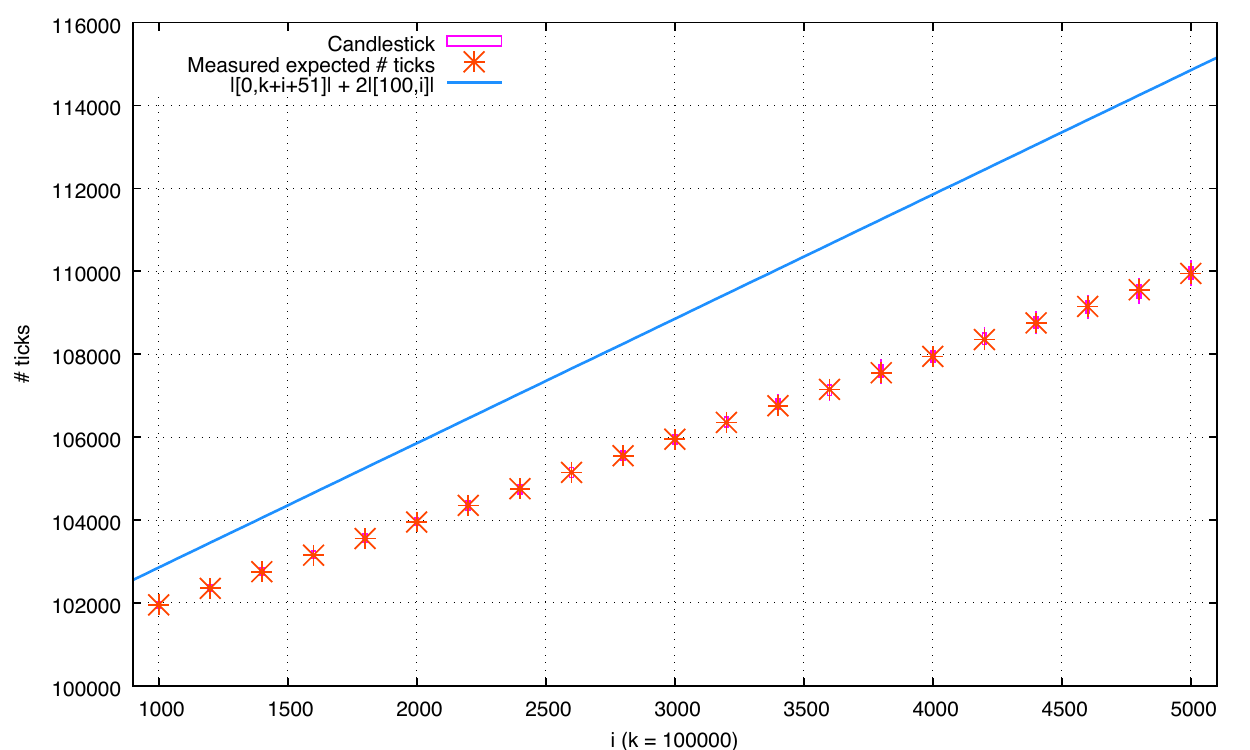} 
\caption{Example \progname{C4B\_t19}.}                         
\label{fig:c4bt19}                                             
\end{figure}
\begin{figure}[th!]                                            
\centering                                                     
\includegraphics[width=0.5\textwidth]{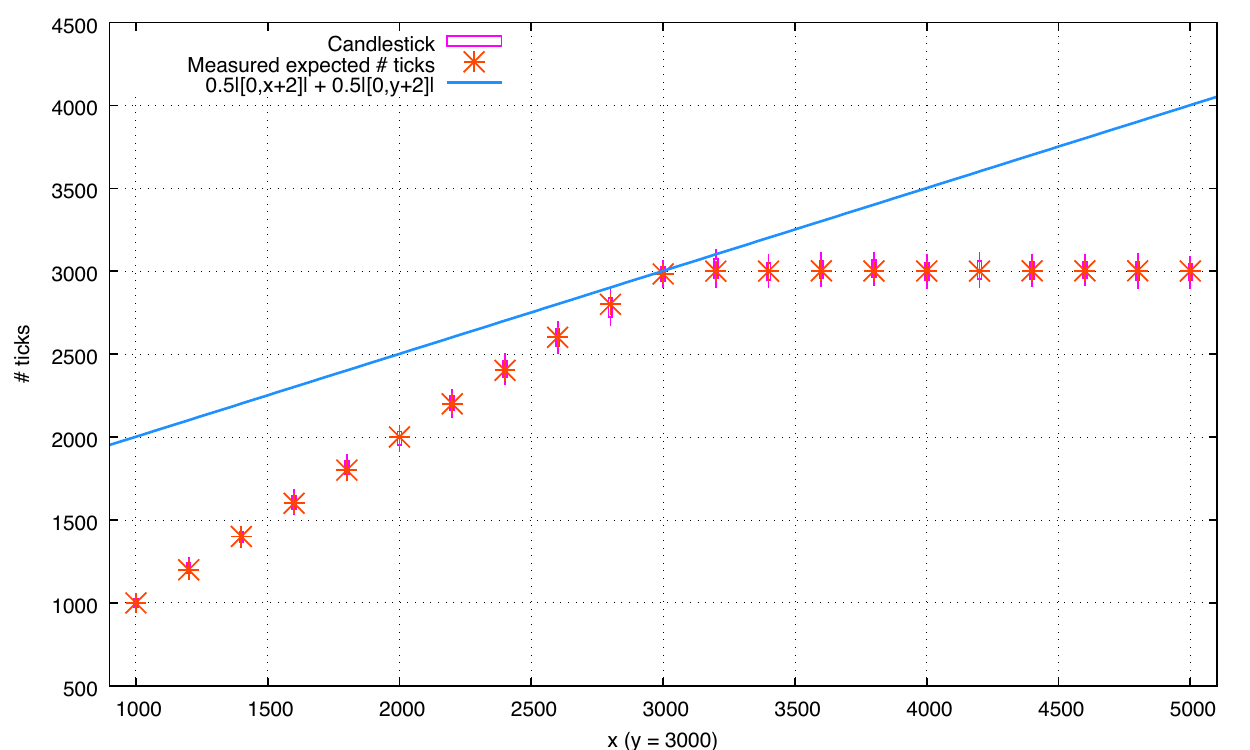} 
\caption{Example \progname{C4B\_t30}.}                         
\label{fig:c4bt30}                                             
\end{figure}
\begin{figure}[th!]                                            
\centering                                                     
\includegraphics[width=0.5\textwidth]{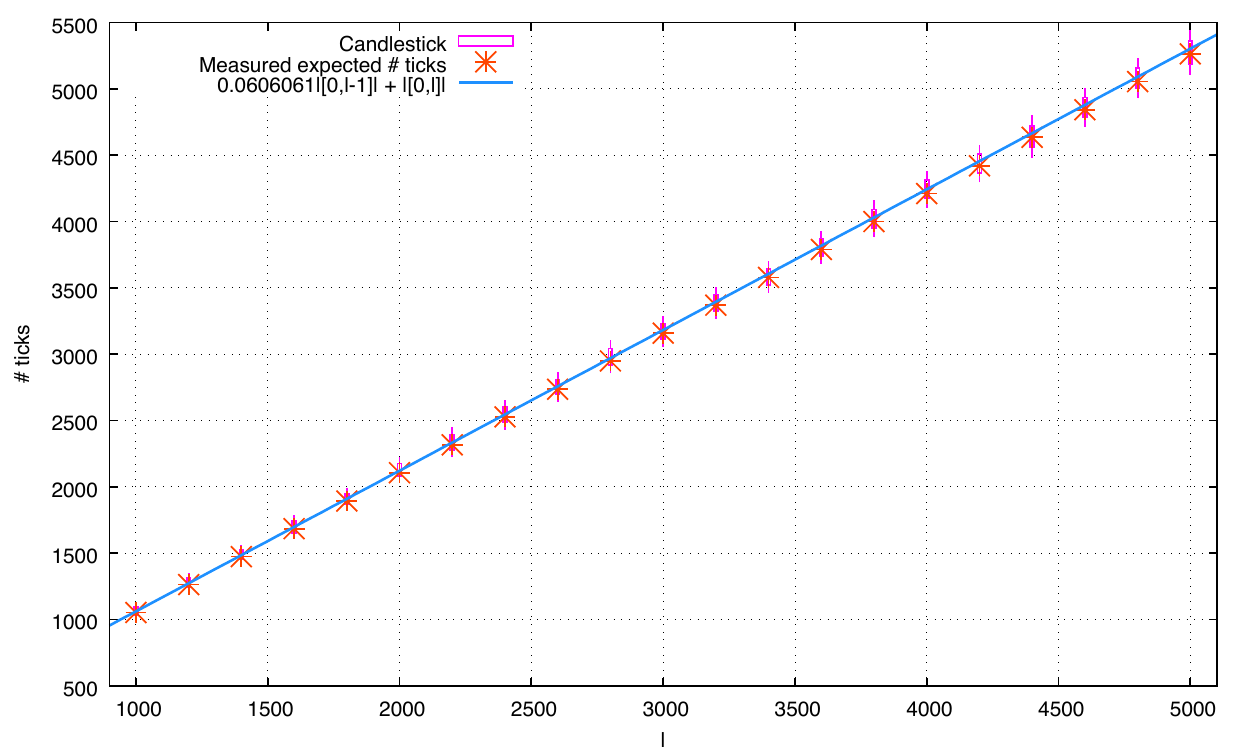} 
\caption{Example \progname{C4B\_t61}.}                         
\label{fig:c4bt61}                                             
\end{figure}

\clearpage

\begin{figure}[th!]                                            
\centering                                                     
\includegraphics[width=0.5\textwidth]{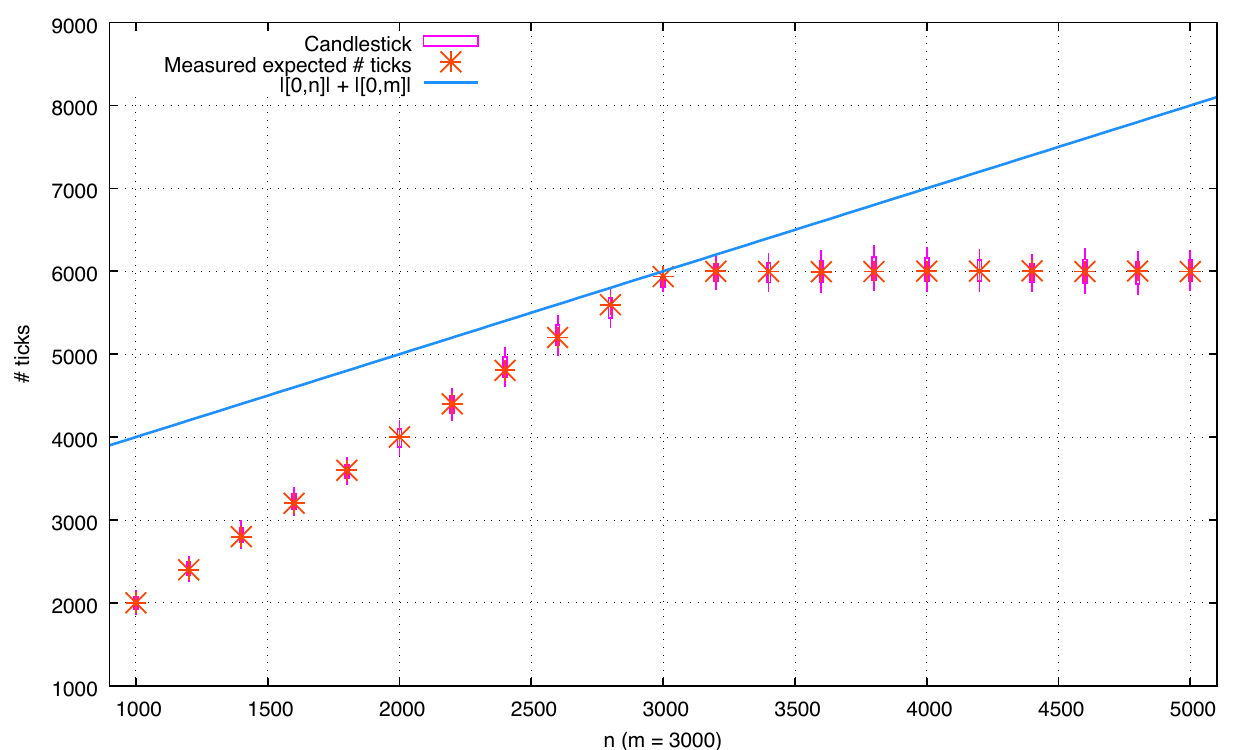} 
\caption{Example \progname{condand}.}                         
\label{fig:condand}                                             
\end{figure}
\begin{figure}[th!]                                            
\centering                                                     
\includegraphics[width=0.5\textwidth]{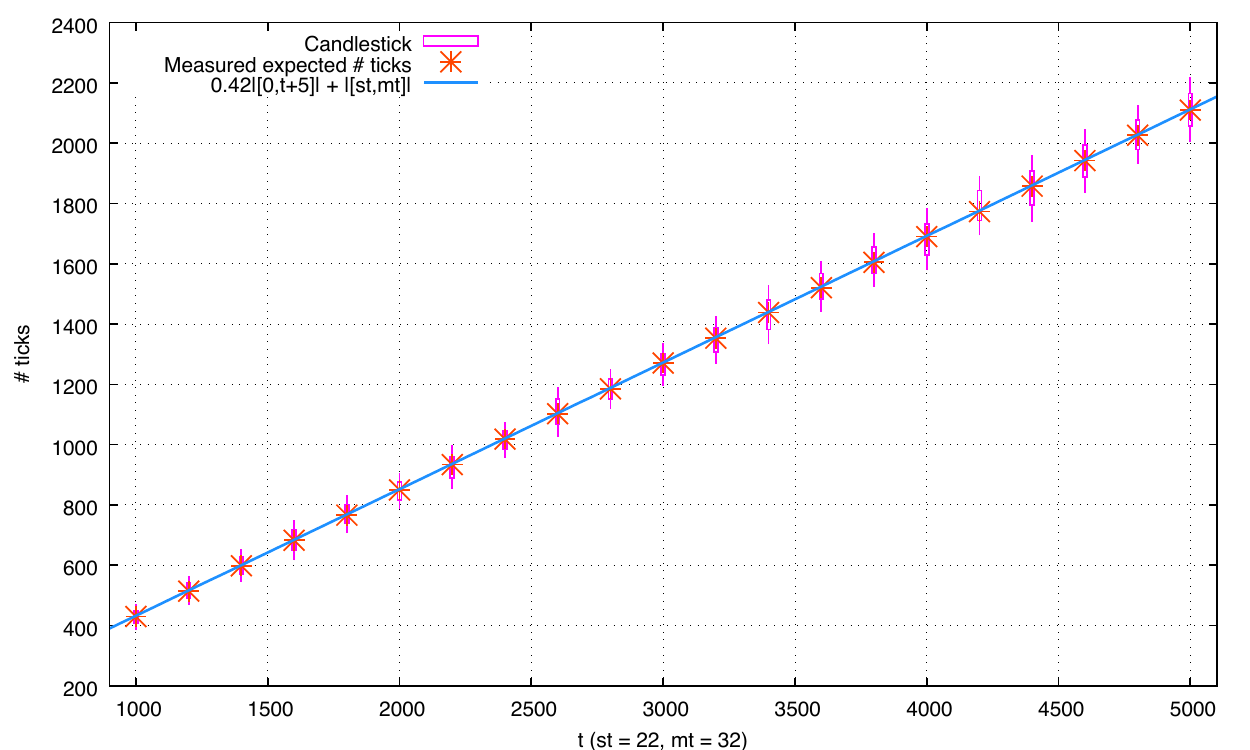} 
\caption{Example \progname{cooling}.}                         
\label{fig:cooling}                                             
\end{figure}
\begin{figure}[th!]                                            
\centering                                                     
\includegraphics[width=0.5\textwidth]{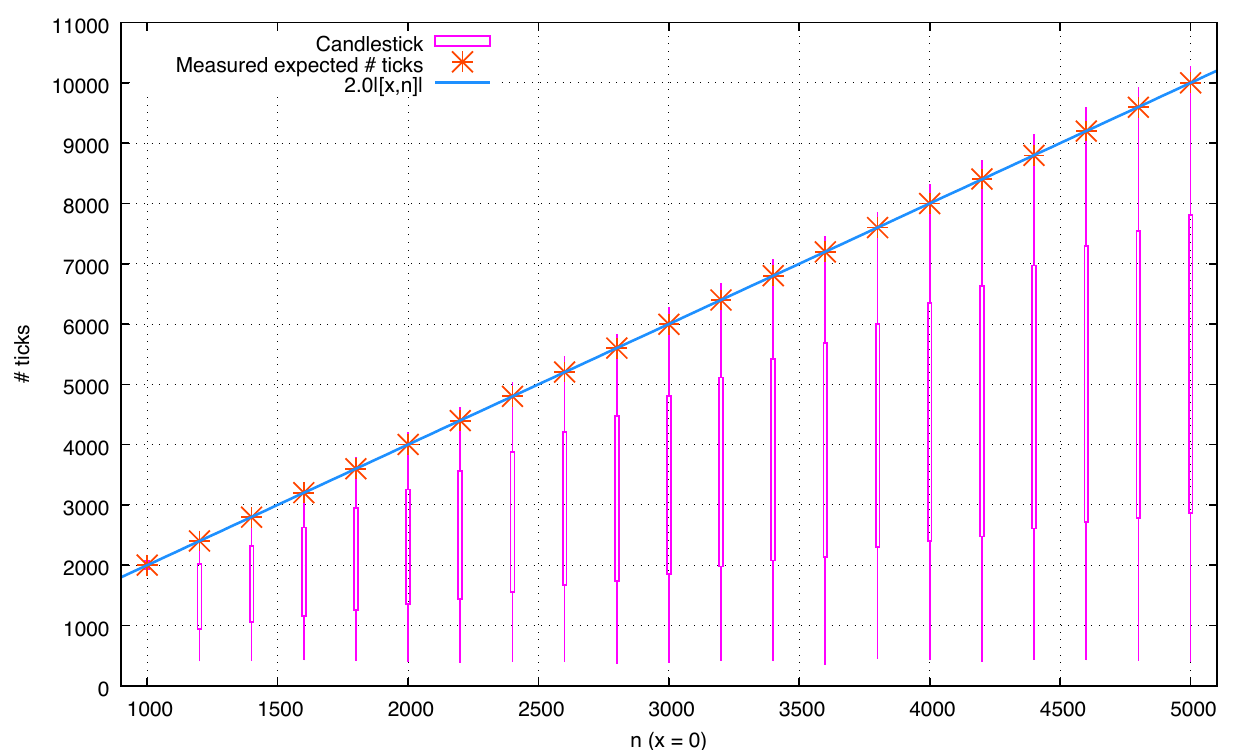} 
\caption{Example \progname{fcall}.}                         
\label{fig:fcall}                                             
\end{figure}
\begin{figure}[th!]                                            
\centering                                                     
\includegraphics[width=0.5\textwidth]{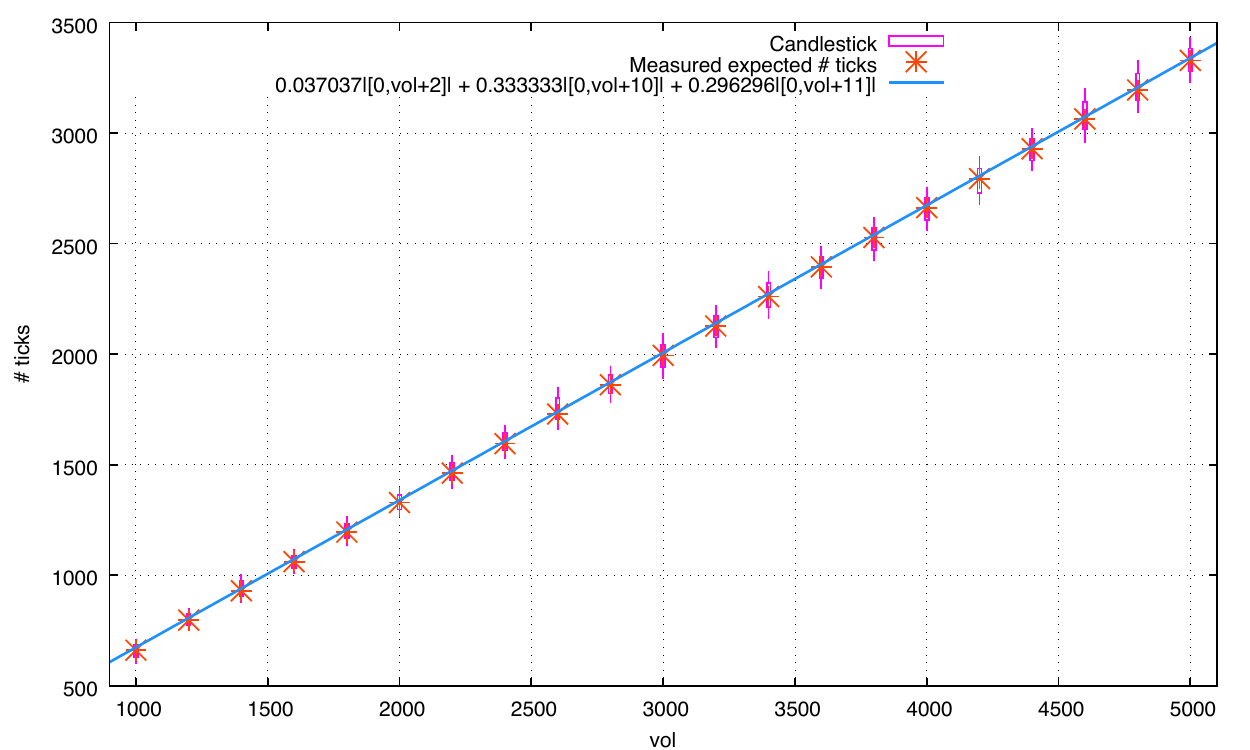} 
\caption{Example \progname{filling}.}                         
\label{fig:filling}                                             
\end{figure}
\begin{figure}[th!]                                            
\centering                                                     
\includegraphics[width=0.5\textwidth]{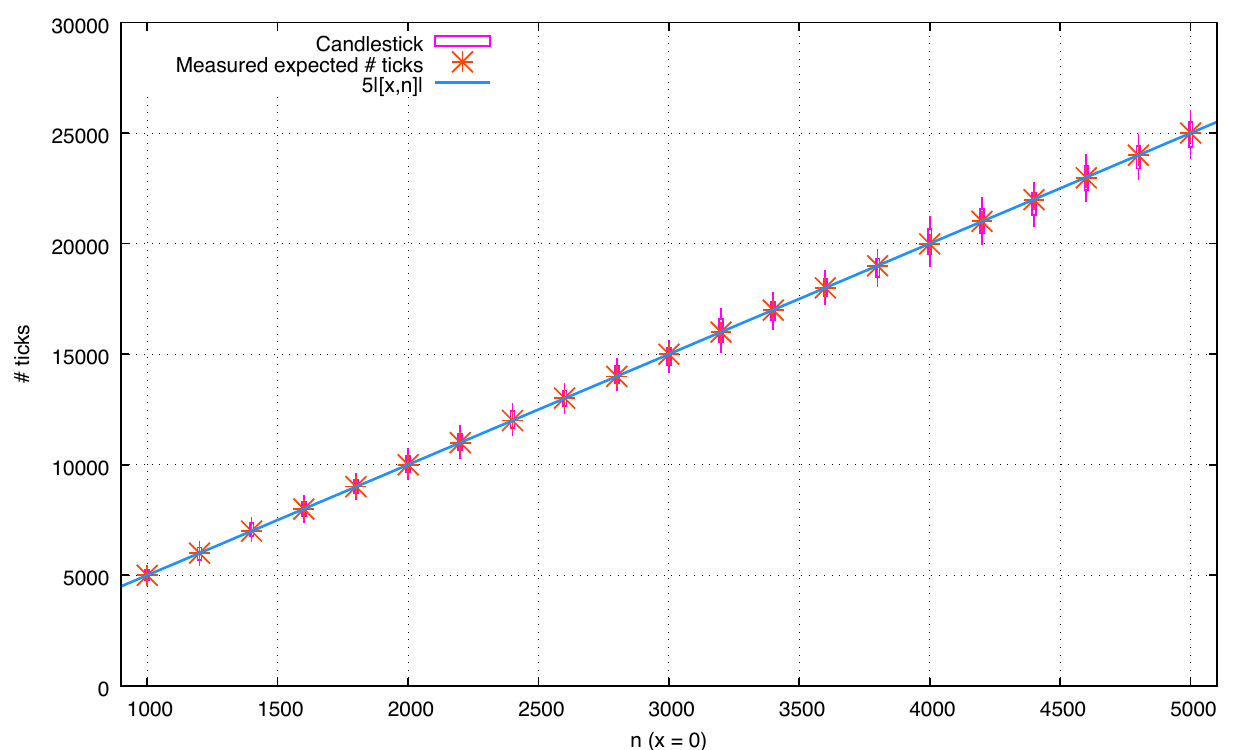} 
\caption{Example \progname{hyper}.}                         
\label{fig:hyper}                                             
\end{figure}                                                   
\begin{figure}[th!]                                            
\centering                                                     
\includegraphics[width=0.5\textwidth]{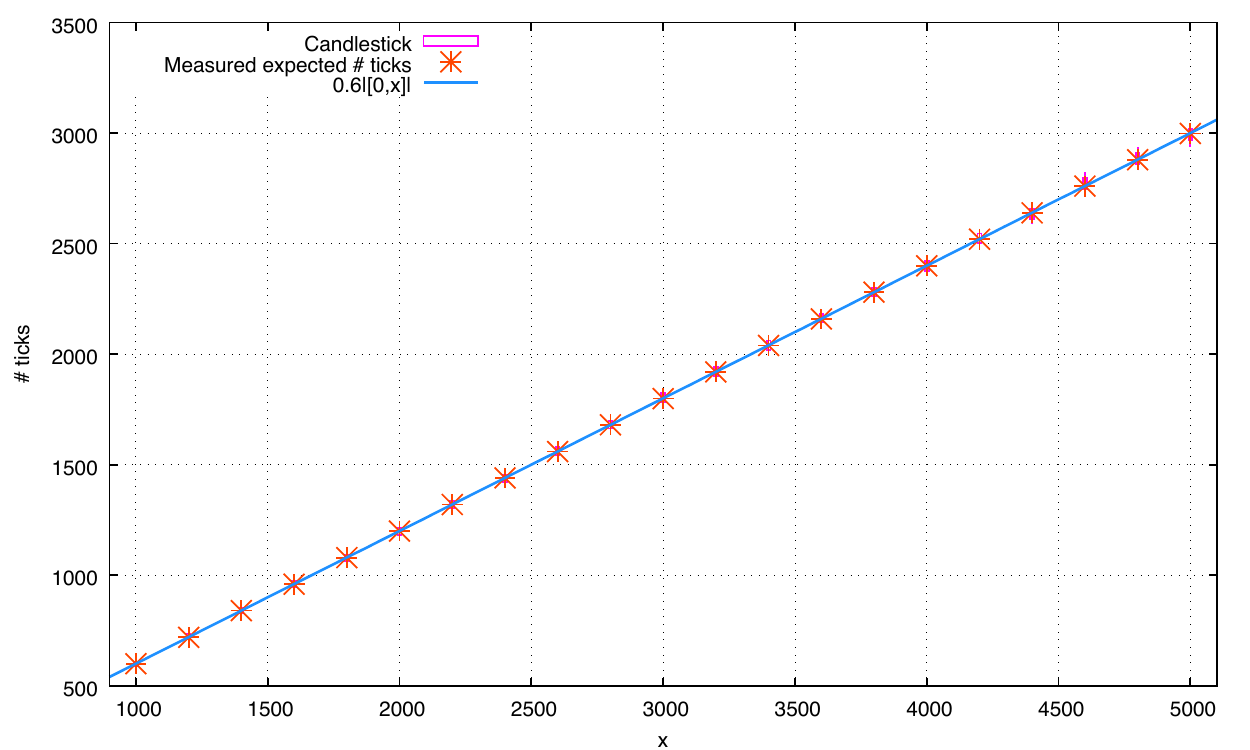}    
\caption{Example \progname{linear01}.}                            
\label{fig:linear01}                                              
\end{figure}
\clearpage
\begin{figure}[th!]                                            
\centering                                                     
\includegraphics[width=0.5\textwidth]{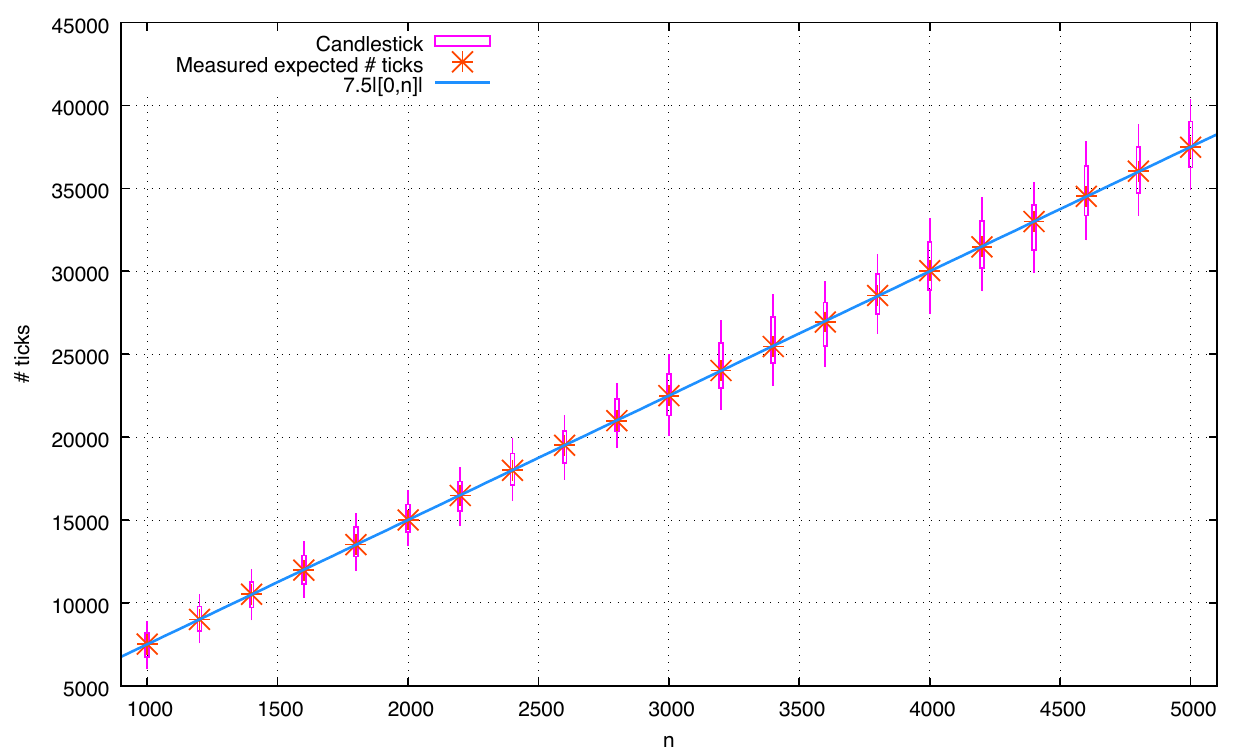}    
\caption{Example \progname{miner}.}                            
\label{fig:miner}                                              
\end{figure}
\begin{figure}[th!]                                            
\centering                                                     
\includegraphics[width=0.5\textwidth]{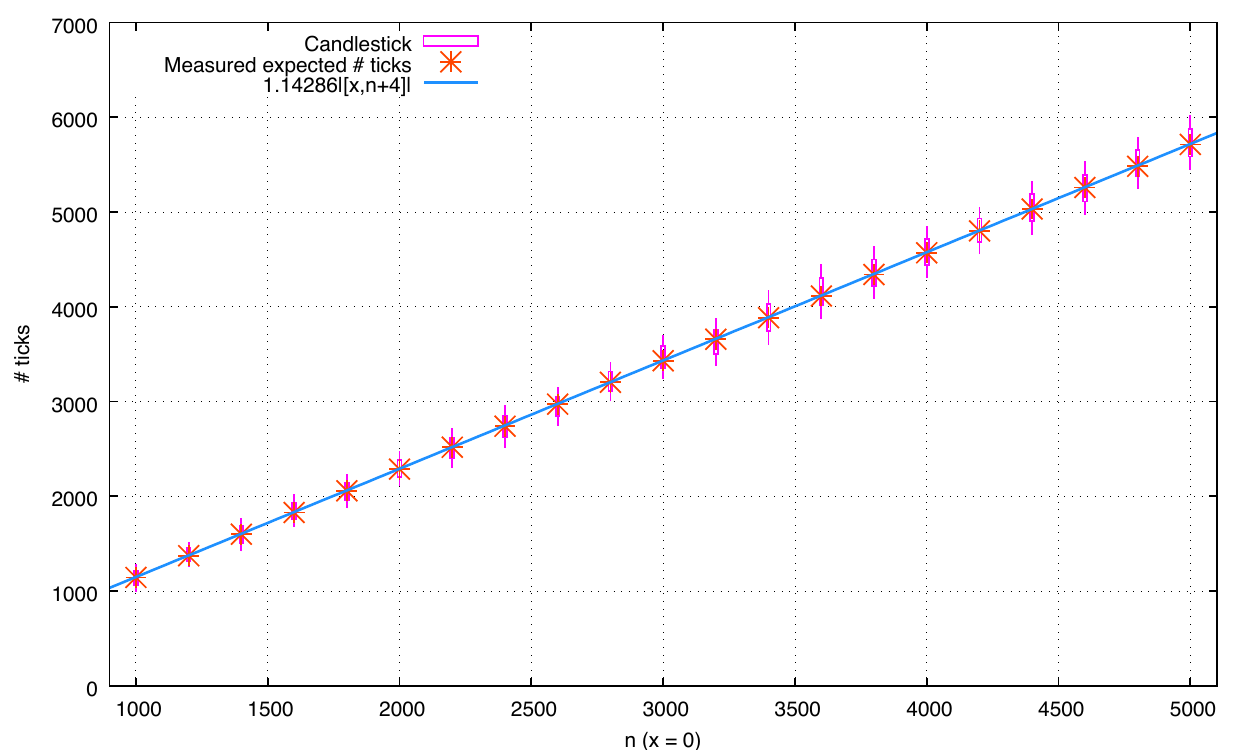}    
\caption{Example \progname{prdwalk}.}                            
\label{fig:prdwalk}                                              
\end{figure}
\begin{figure}[th!]                                            
\centering                                                     
\includegraphics[width=0.5\textwidth]{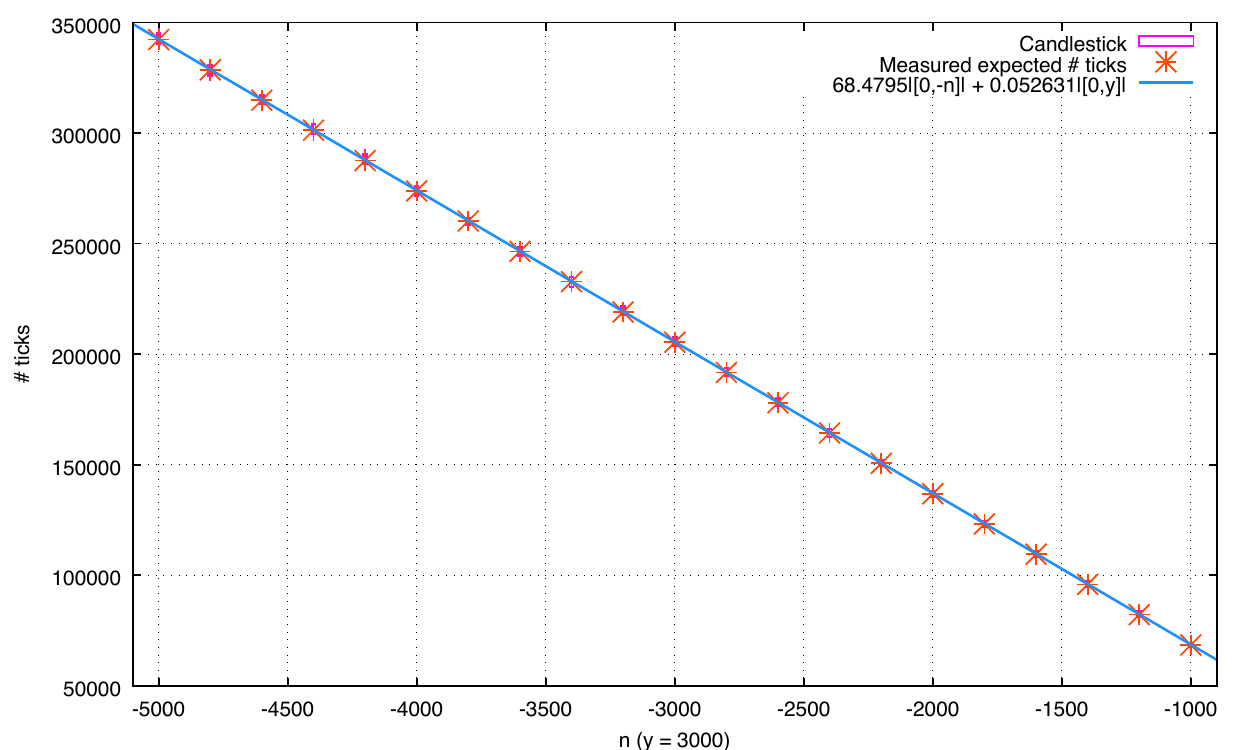}    
\caption{Example \progname{prnes}.}                            
\label{fig:prnes}                                              
\end{figure}
\begin{figure}[th!]                                            
\centering                                                     
\includegraphics[width=0.5\textwidth]{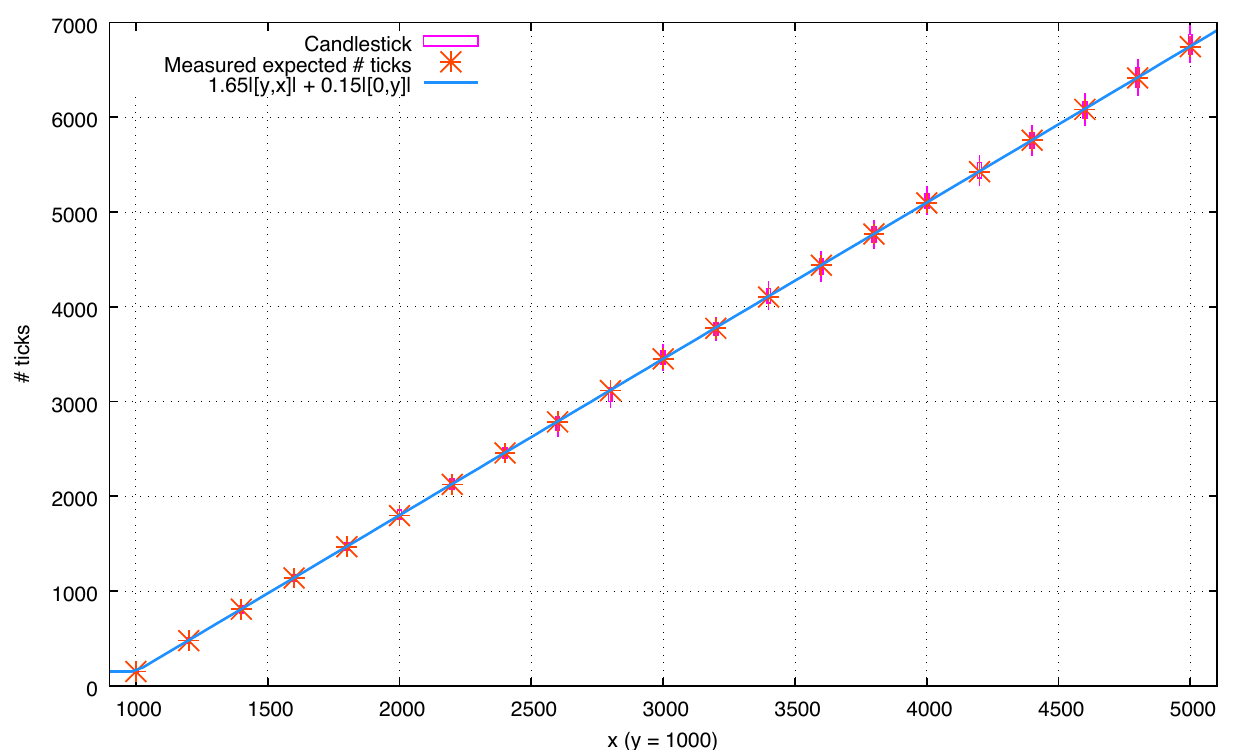}    
\caption{Example \progname{prseq}.}                            
\label{fig:prseq}                                              
\end{figure}
\begin{figure}[th!]                                            
\centering                                                     
\includegraphics[width=0.5\textwidth]{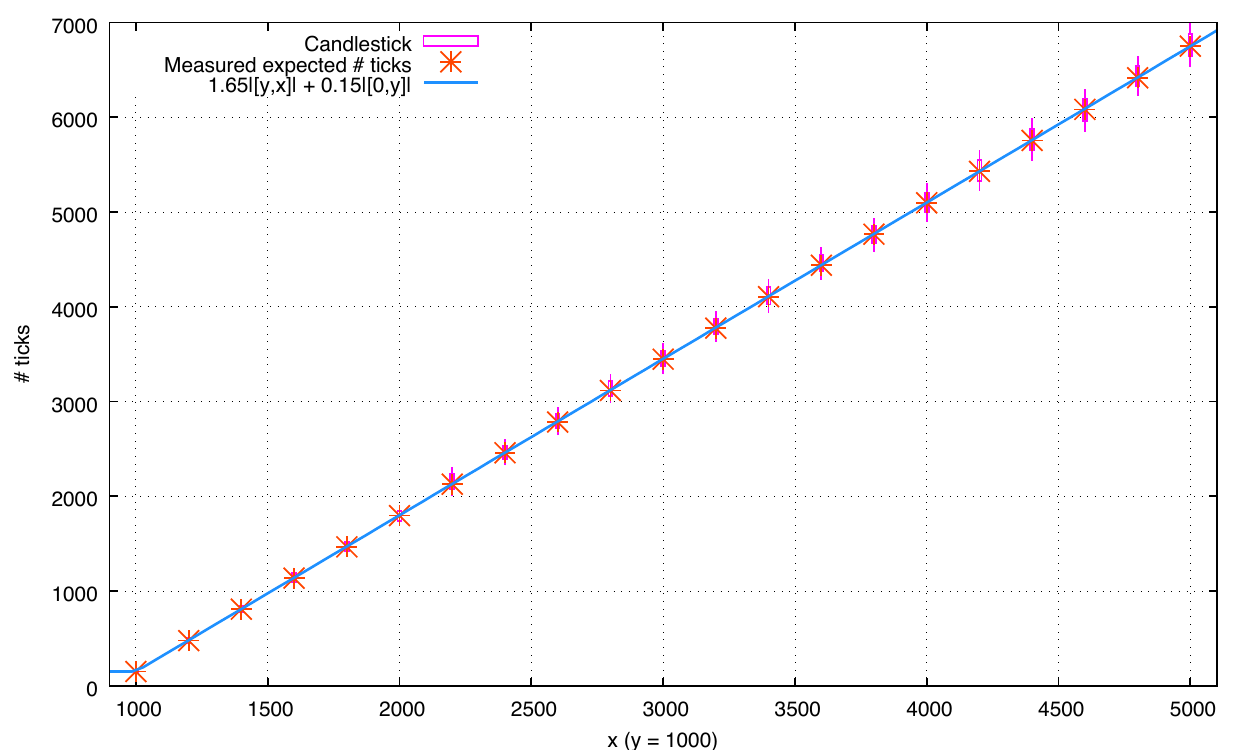}    
\caption{Example \progname{prseq\_bin}.}                            
\label{fig:prseqbin}                                              
\end{figure}
\begin{figure}[th!]                                            
\centering                                                     
\includegraphics[width=0.5\textwidth]{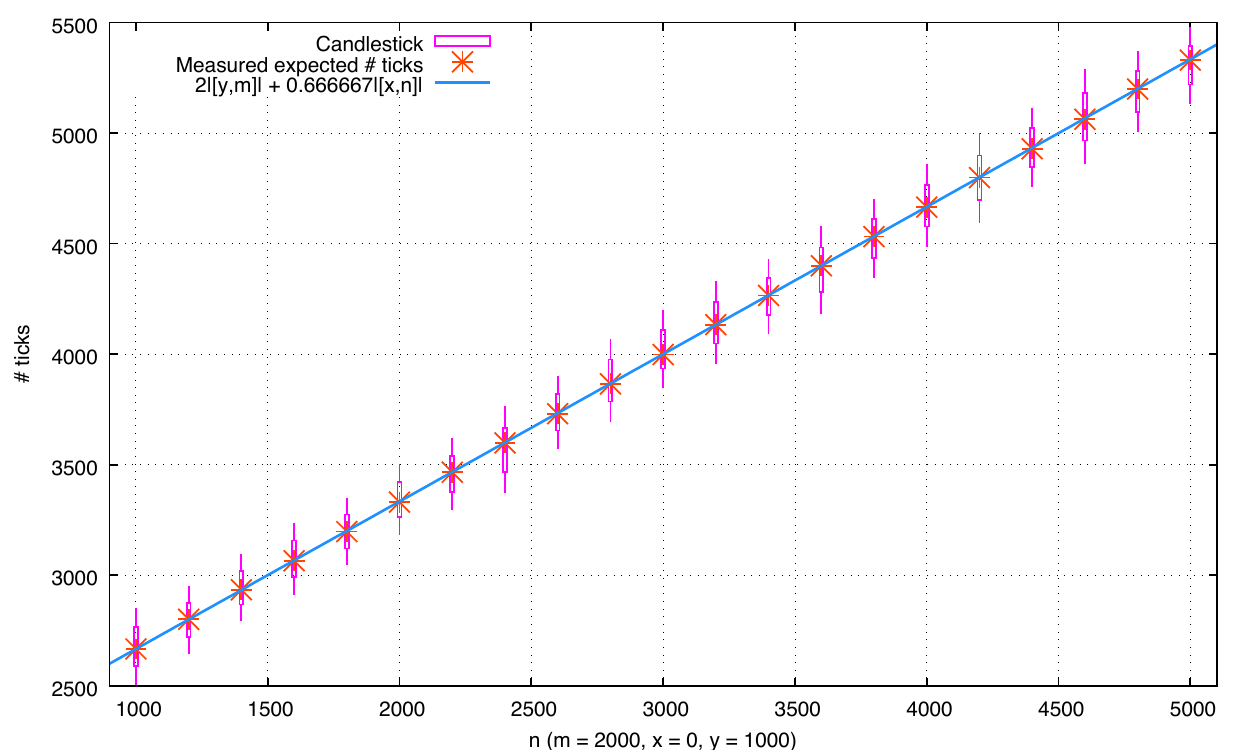}    
\caption{Example \progname{prspeed}.}                            
\label{fig:prspeed}                                              
\end{figure}
\begin{figure}[th!]                                            
\centering                                                     
\includegraphics[width=0.5\textwidth]{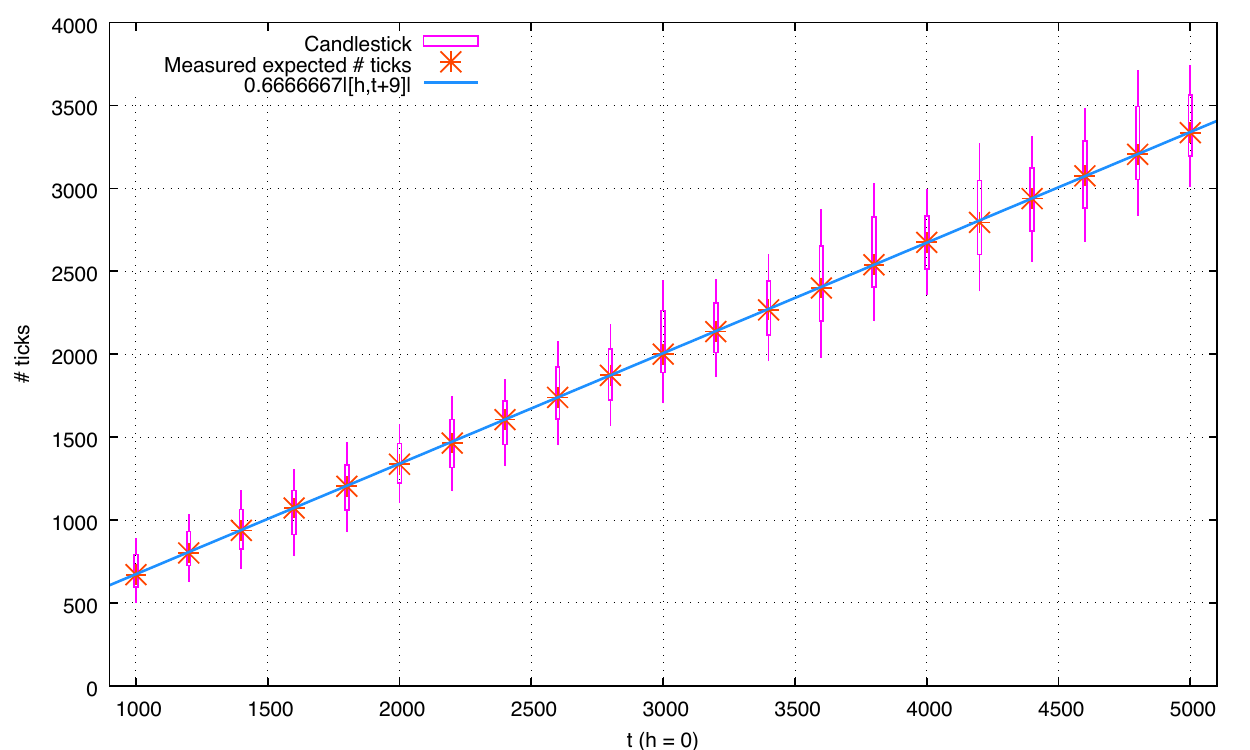}    
\caption{Example \progname{race}.}                            
\label{fig:race}                                              
\end{figure}
\begin{figure}[th!]                                            
\centering                                                     
\includegraphics[width=0.5\textwidth]{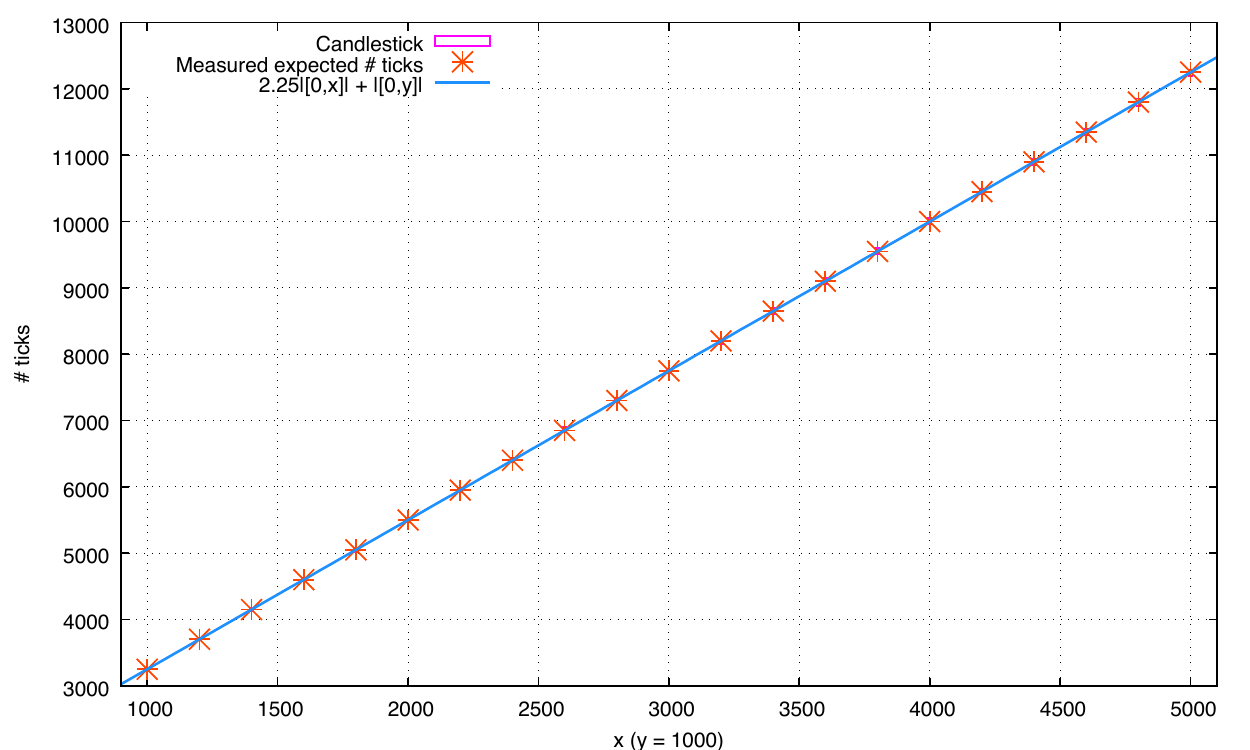}    
\caption{Example \progname{rdseql}.}                            
\label{fig:rdseql}                                              
\end{figure}
\begin{figure}[th!]                                            
\centering                                                     
\includegraphics[width=0.5\textwidth]{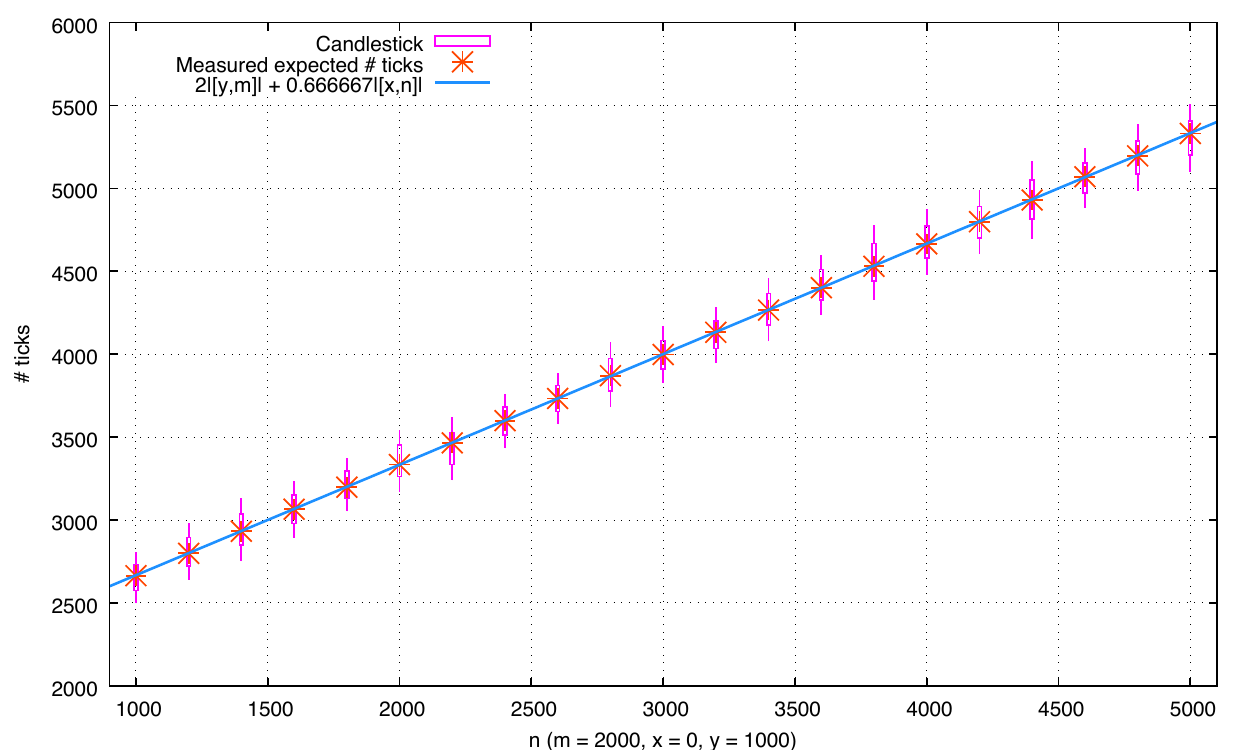}    
\caption{Example \progname{rdspeed}.}                            
\label{fig:rdspeed}                                              
\end{figure}                                                       
\begin{figure}[th!]                                            
\centering                                                     
\includegraphics[width=0.5\textwidth]{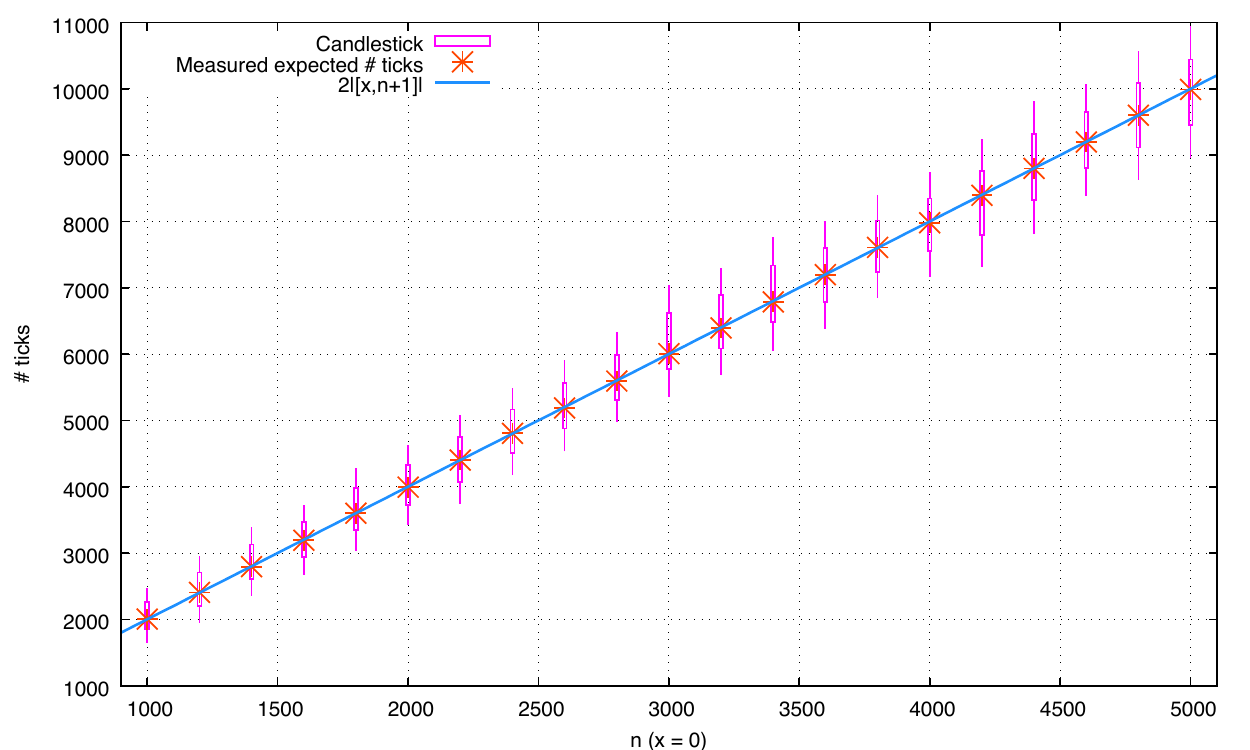}    
\caption{Example \progname{rdwalk}.}                            
\label{fig:rdwalk}                                              
\end{figure}
\begin{figure}[th!]                                            
\centering                                                     
\includegraphics[width=0.5\textwidth]{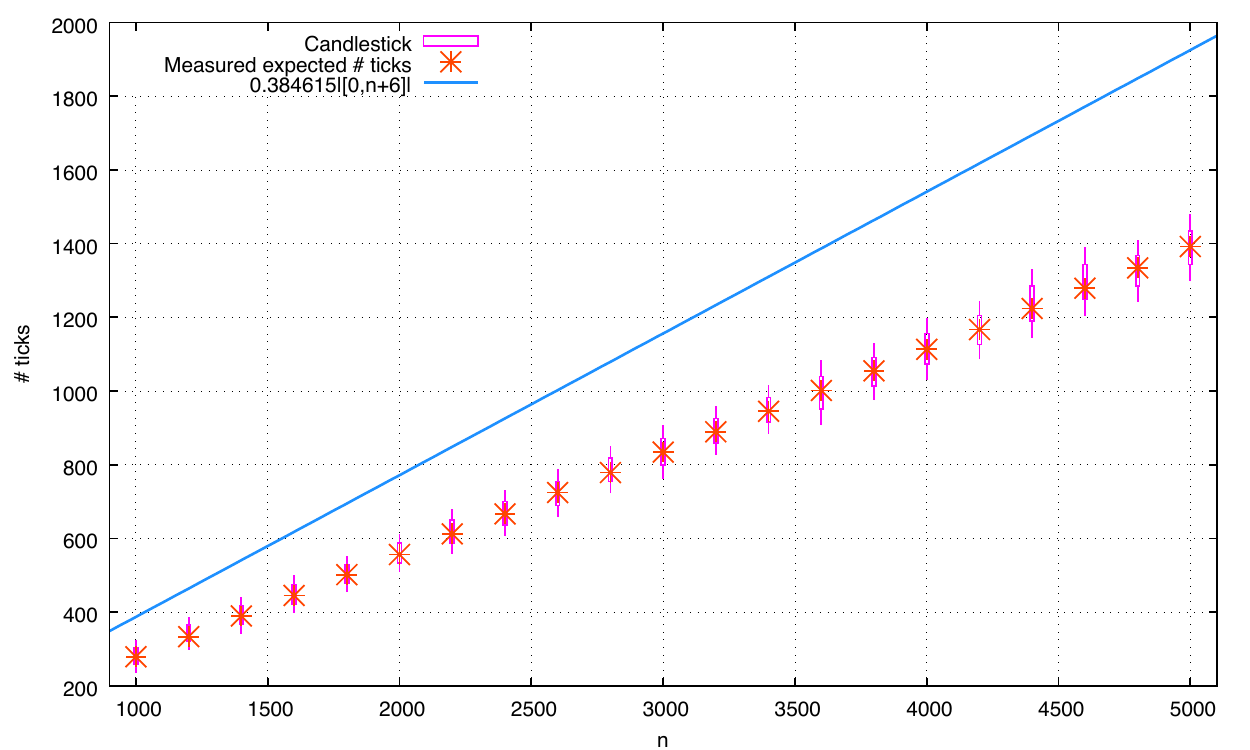}    
\caption{Example \progname{robot}.}                            
\label{fig:robot}                                              
\end{figure}
\begin{figure}[th!]                                            
\centering                                                     
\includegraphics[width=0.5\textwidth]{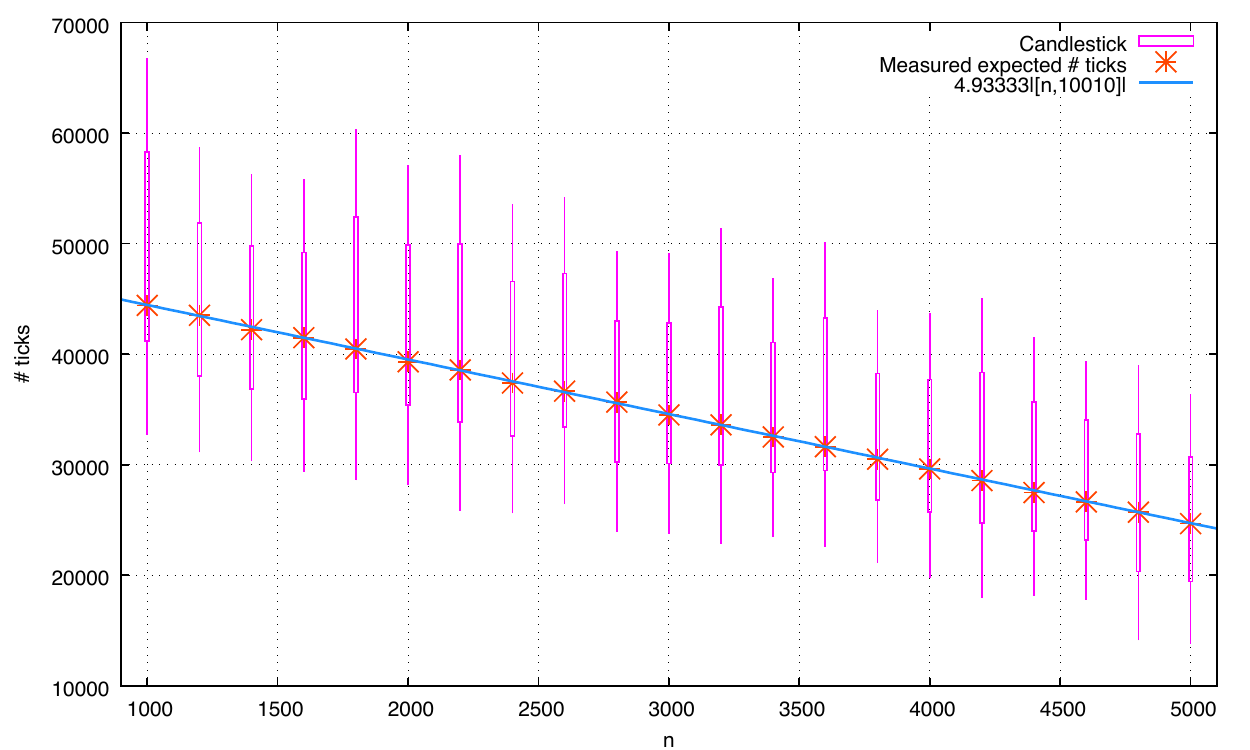}    
\caption{Example \progname{roulette}.}                            
\label{fig:roulette}                                              
\end{figure}
\begin{figure}[th!]                                            
\centering                                                     
\includegraphics[width=0.5\textwidth]{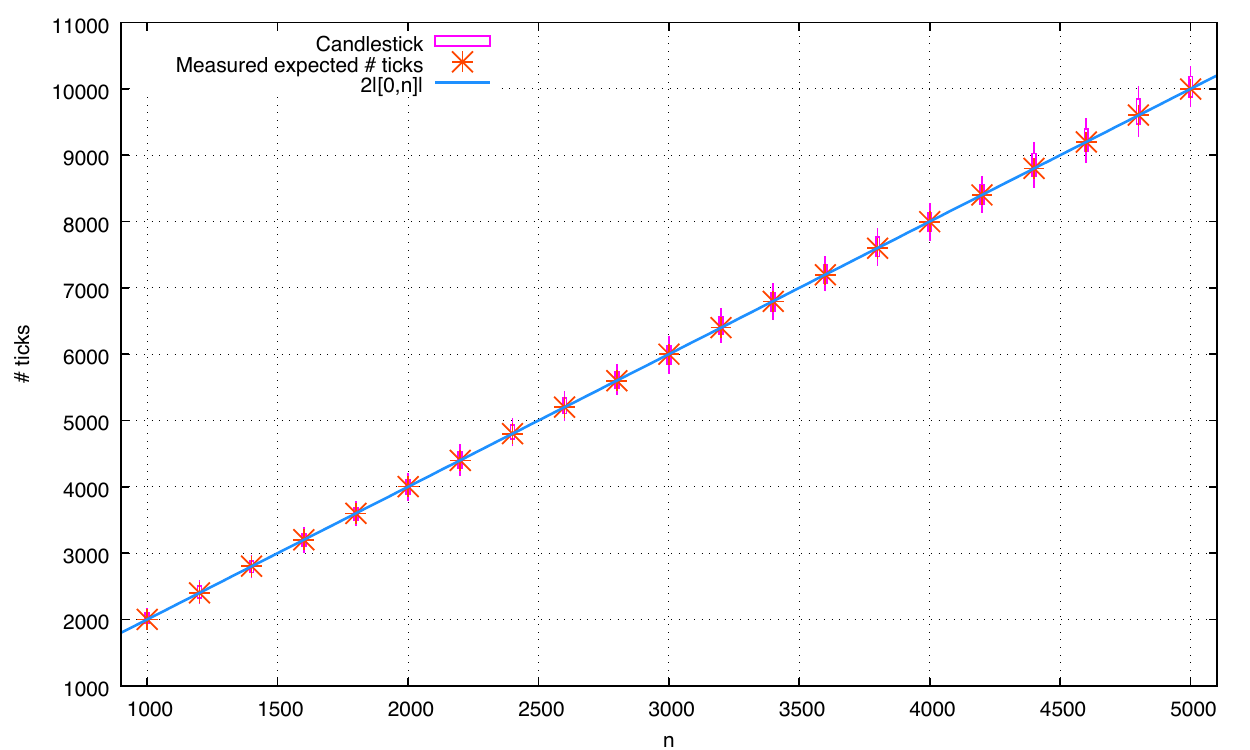}    
\caption{Example \progname{sampling}.}                            
\label{fig:sampling}                                              
\end{figure}
\begin{figure}[th!]                                            
\centering                                                     
\includegraphics[width=0.5\textwidth]{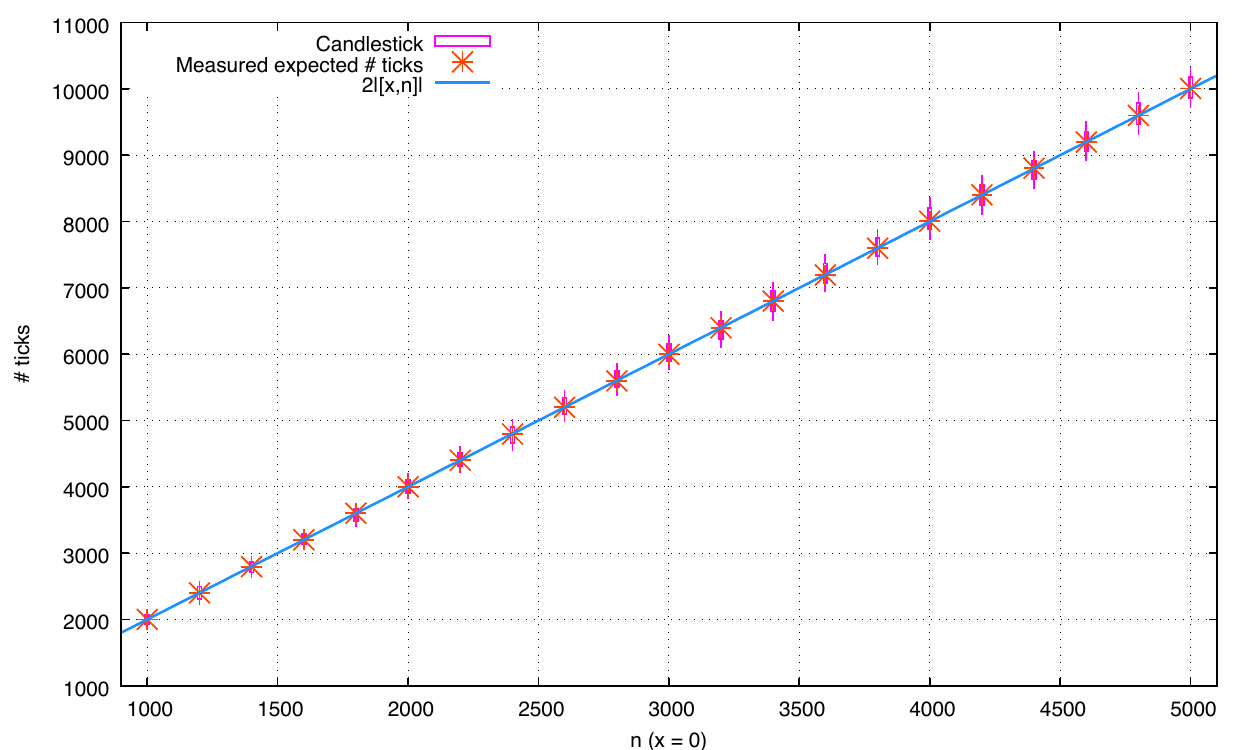}    
\caption{Example \progname{sprdwalk}.}                            
\label{fig:sprdwalk}                                              
\end{figure}                                                   
\begin{figure}[th!]                                            
\centering                                                     
\includegraphics[width=0.5\textwidth]{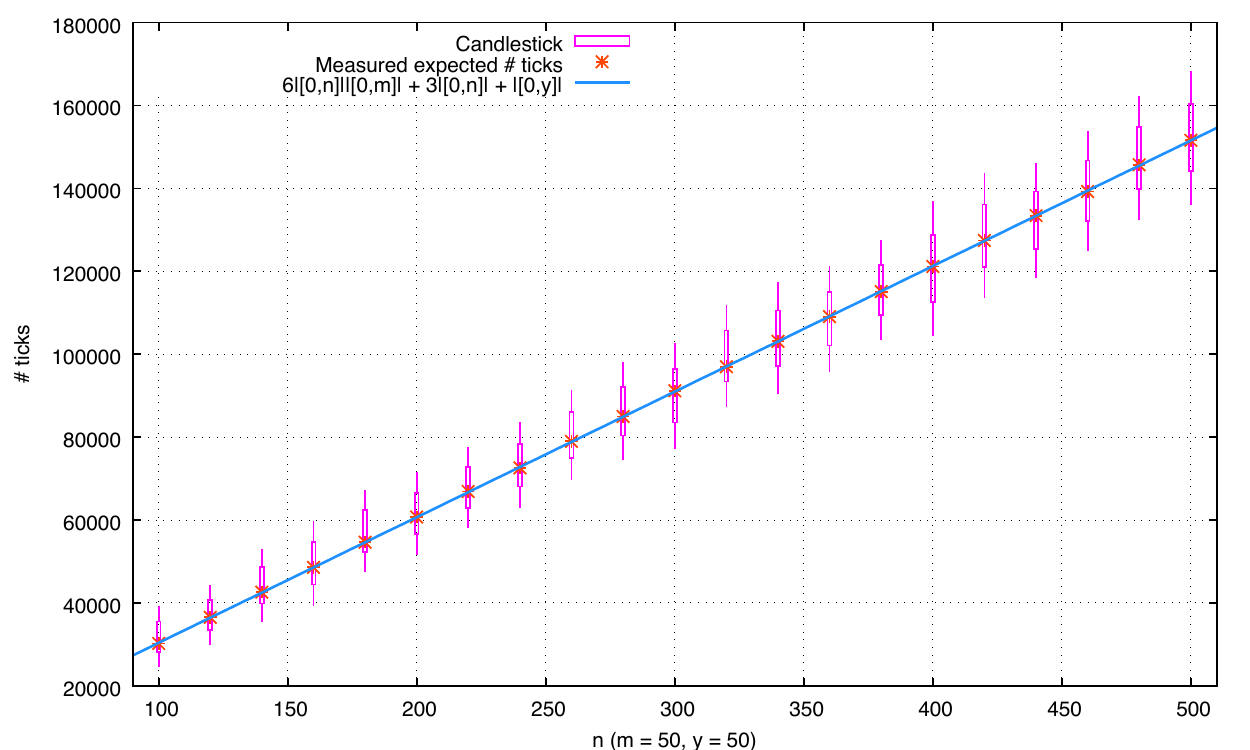}    
\caption{Example \progname{complex}.}                            
\label{fig:complex}                                              
\end{figure}
\begin{figure}[th!]                                            
\centering                                                     
\includegraphics[width=0.5\textwidth]{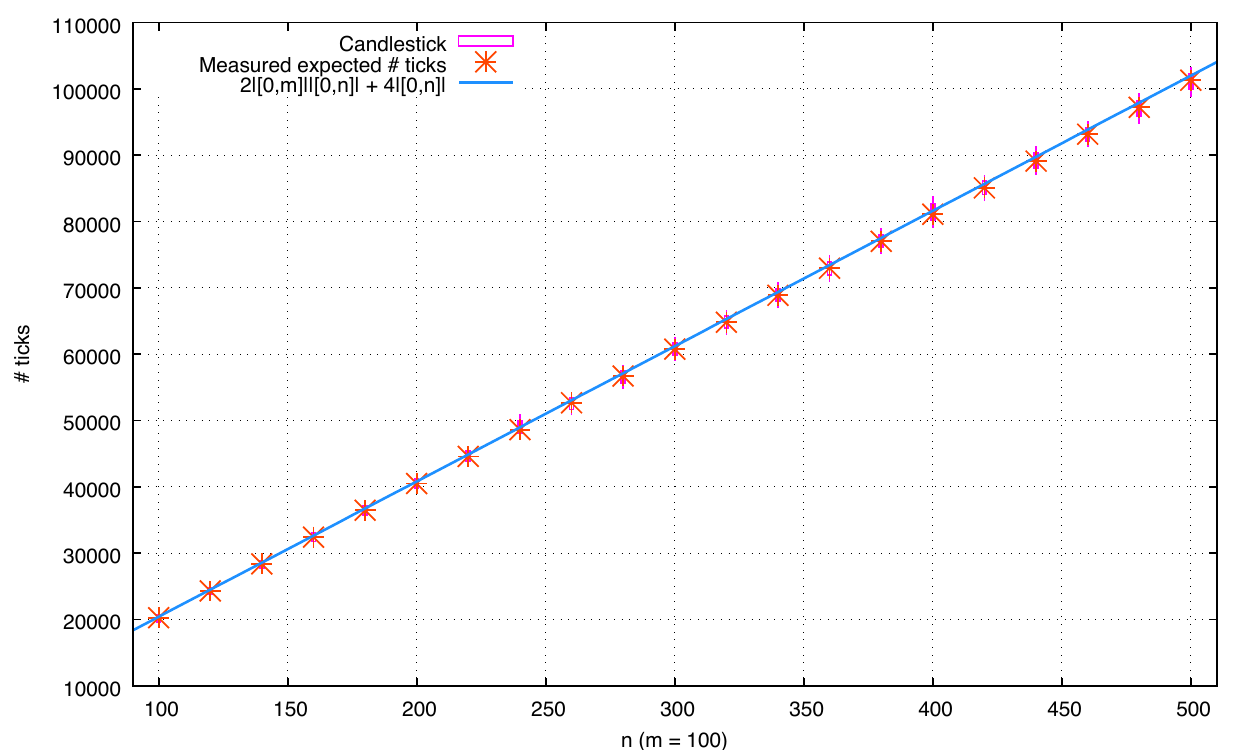}    
\caption{Example \progname{multirace}.}                            
\label{fig:multirace}                                              
\end{figure}
\begin{figure}[th!]                                            
\centering                                                     
\includegraphics[width=0.5\textwidth]{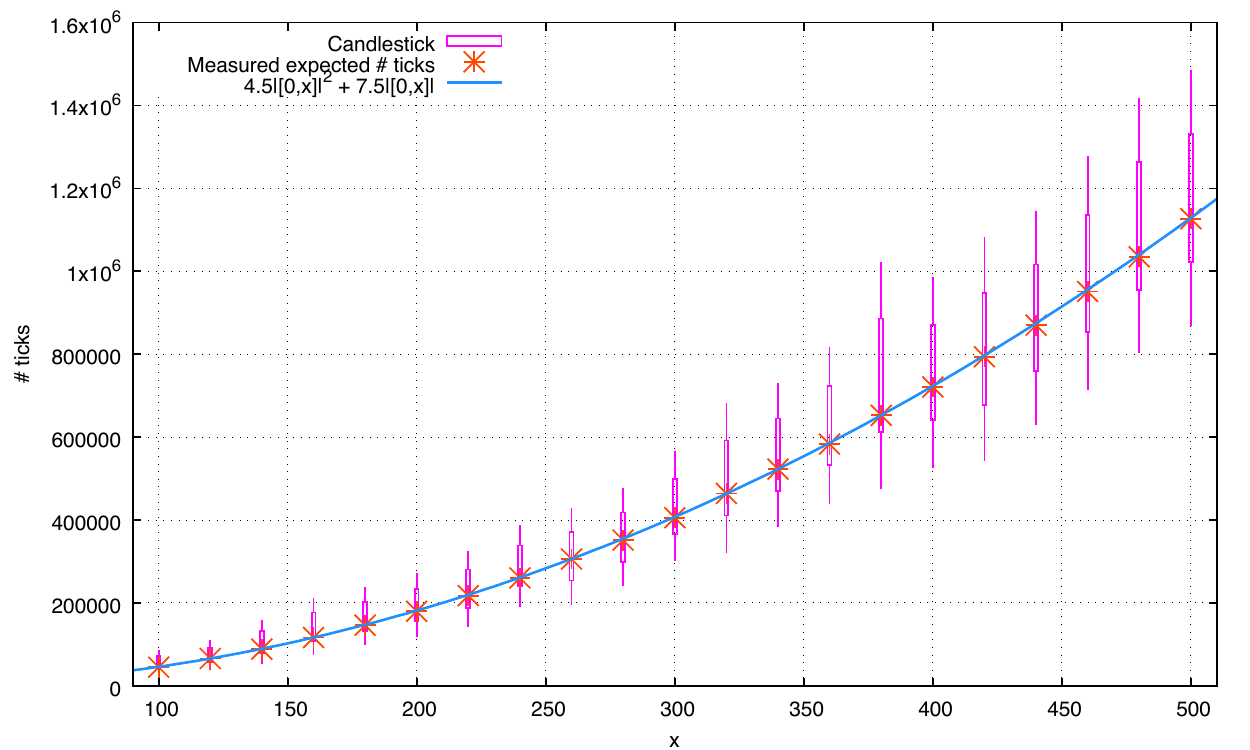}    
\caption{Example \progname{pol04}.}                            
\label{fig:pol04}                                              
\end{figure}
\begin{figure}[th!]                                            
\centering                                                     
\includegraphics[width=0.5\textwidth]{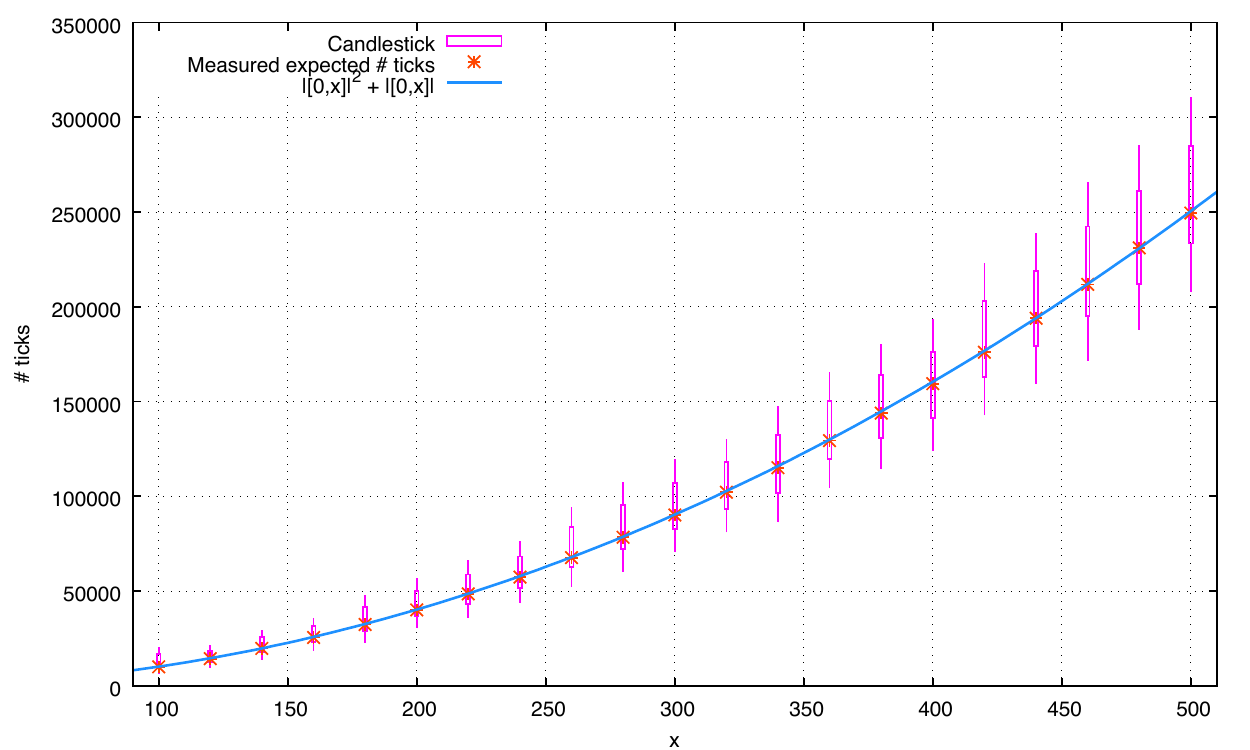}    
\caption{Example \progname{pol05}.}                            
\label{fig:pol05}                                              
\end{figure}
\clearpage
\begin{figure}[th!]                                            
\centering                                                     
\includegraphics[width=0.5\textwidth]{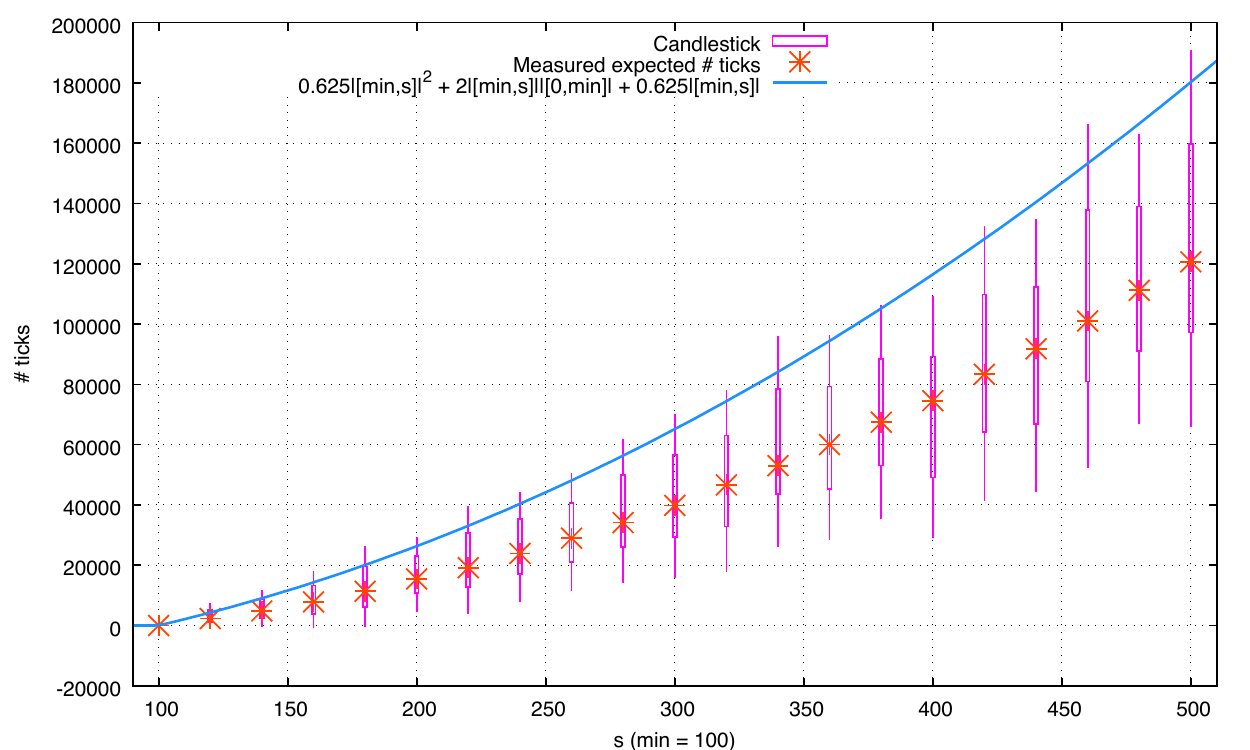}    
\caption{Example \progname{pol06}.}                            
\label{fig:pol06}                                              
\end{figure}                                                   
\begin{figure}[th!]                                            
\centering                                                     
\includegraphics[width=0.5\textwidth]{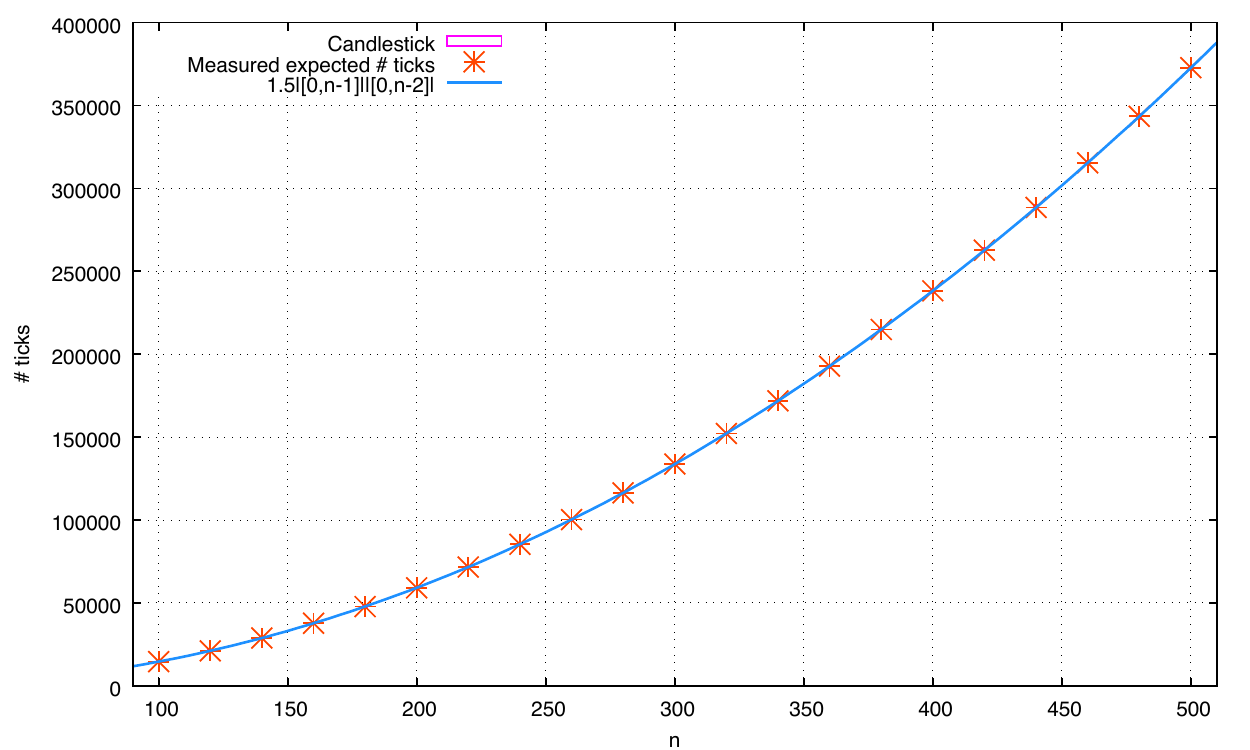}    
\caption{Example \progname{pol07}.}                            
\label{fig:pol07}                                              
\end{figure}
\begin{figure}[th!]                                            
\centering                                                     
\includegraphics[width=0.5\textwidth]{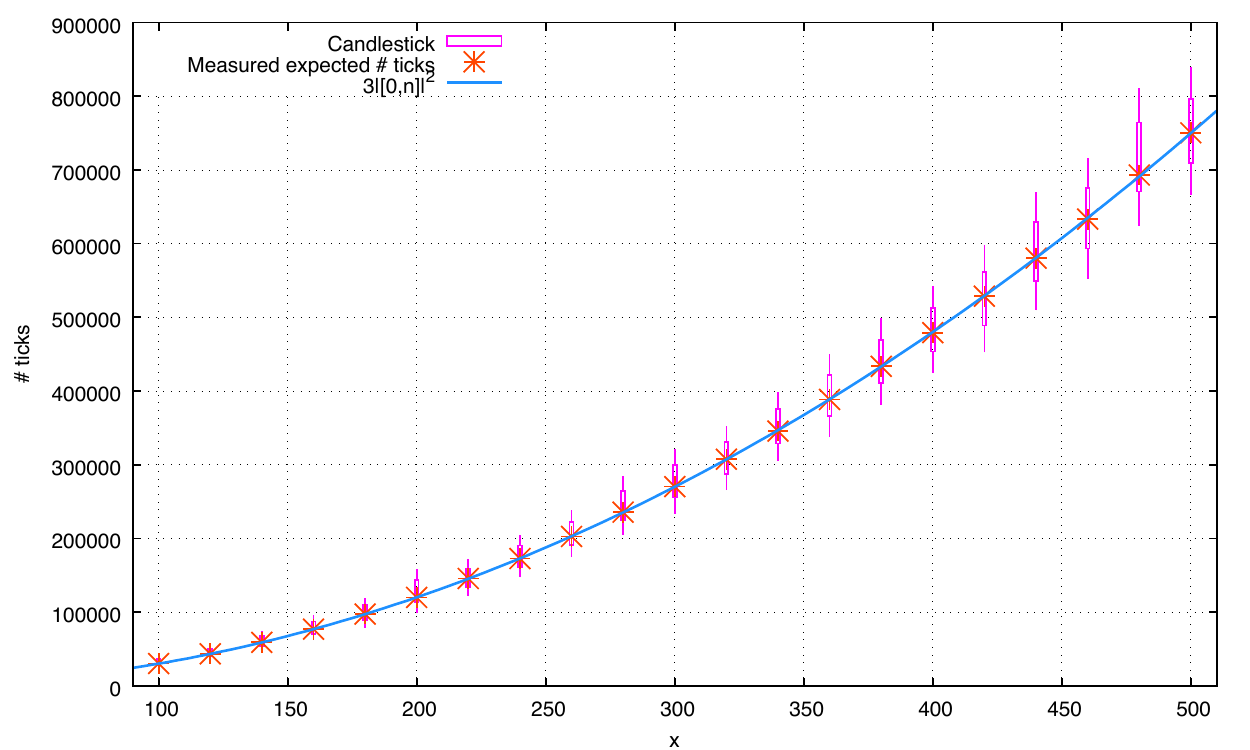}    
\caption{Example \progname{rdbub}.}                            
\label{fig:rdbub}                                              
\end{figure}
\begin{figure}[th!]                                            
\centering                                                     
\includegraphics[width=0.5\textwidth]{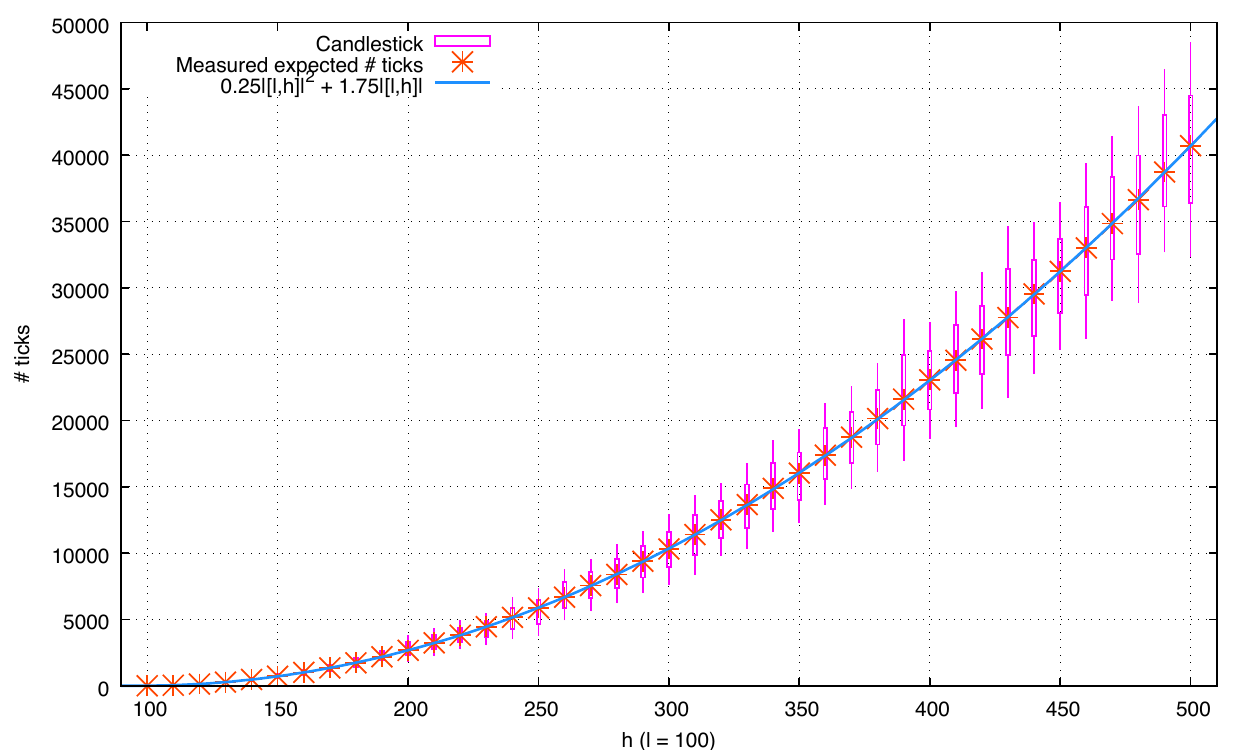}    
\caption{Example \progname{recursive}.}                            
\label{fig:recursive}                                              
\end{figure}
\begin{figure}[th!]                                            
\centering                                                     
\includegraphics[width=0.5\textwidth]{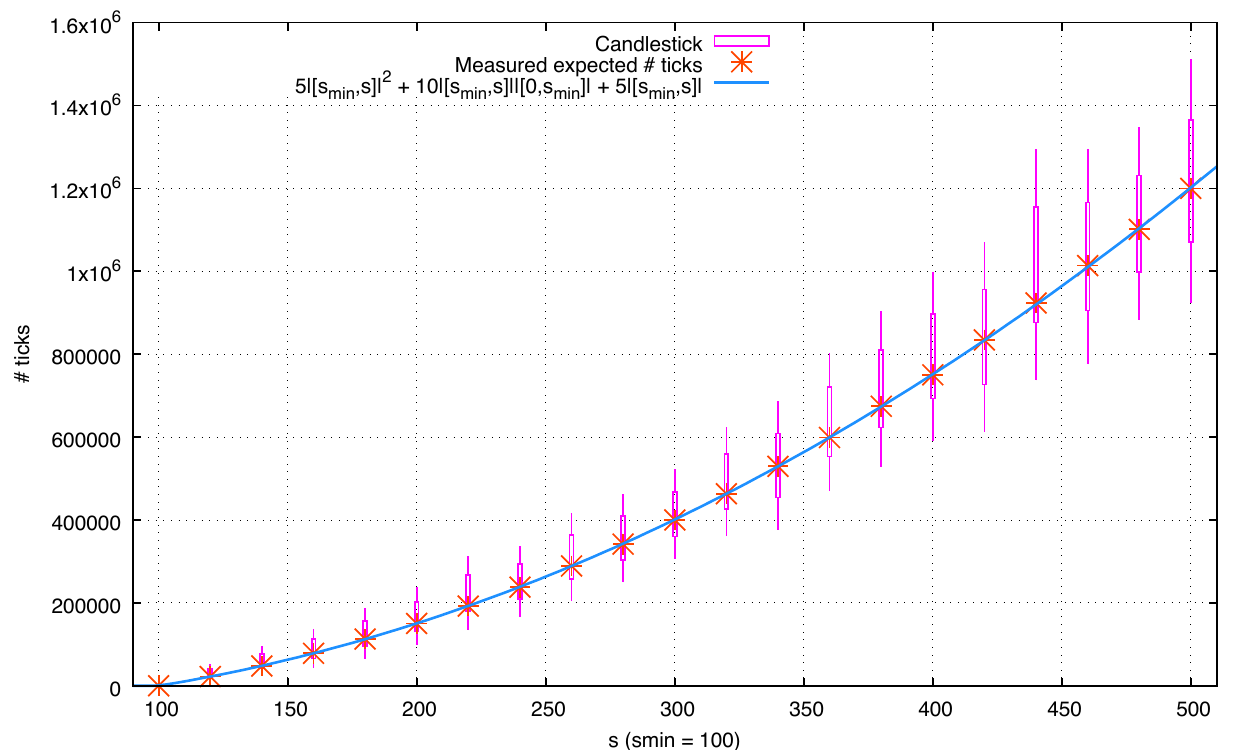}    
\caption{Example \progname{trader}.}                            
\label{fig:trader}                                              
\end{figure}                                                   

\vspace{-1.5ex}
\section{More examples of derivation}
\label{app:morederivation}
In this section, we show more examples of derivation for 
challenging probabilistic programs. We also demonstrate how 
\toolname{} handles Boolean values to get the tightest expected bounds. 
All the presented bound derivations can be derived automatically by our tool
\toolname{}. 

The derivation on the left side of Figure \ref{fig:moredevration01} infers the bound
$\Phi =
\frac{2T}{{pK_1}{+}{({1}{-}{p}){K_2}}}{\cdot}{\interval{x,{n}{+}{K_2}{-}{1}}}$
on the number of ticks for a generalized version of
$\progname{rdwalk}$,
in which both probabilistic branching and sampling commands are
used. The value of $x$
is added by a uniformly-distributed random variable from $0$
to $K_1$
with probability $p$.
With probability $({1}{-}{p})$,
it is incremented by a uniformly-distributed random variable from $0$
to $K_2$.
The reasoning is similar to the ones of previous examples, 
in which we obtain the following post-potential after the samplings.
$$
\Phi' = \frac{2T}{{pK_1}{+}{({1}{-}{p}){K_2}}}{\cdot}{\interval{x,{n}{+}{K_2}{-}{1}}}{+}{T}
$$
The corresponding pre-potentials w.r.t the samplings are given as follows.
$$
\begin{array}{l}
\Phi_1 = {(}{T} {-} {\frac{TK_1}{{pK_1}{+}(1{-}p)K_2}} {+}
{\frac{2T}{pK_1{+}(1{-}p)K_2}}{\cdot}{\interval{x,{n}{+}{K_2}{-}{1}}}{)} \\
\Phi_2 = {(}{T} {-} {\frac{TK_2}{pK_1{+}(1{-}p)K_2}} {+}
\frac{2T}{pK_1{+}(1{-}p)K_2}{\cdot}\interval{x,{n}{+}{K_2}{-}{1}}{)}
\end{array}
$$
Thus, the pre-potential before the probabilistic branching that also serves 
as a loop invariant is 
$\Phi = {p}{\cdot}{\Phi_1} {+} {(}{1}{-}{p}{)}{\cdot}{\Phi_2}$. 
\begin{figure*}[th!]
\centering
\begin{minipage}[b]{0.45\textwidth}
\centering
\begin{lstlisting}[basicstyle=\footnotesize,escapeinside={@}{@}]
@$\devcomment{.;{0}{+}{\frac{2T}{pK_1+(1-p)K_2}}{\cdot}{\interval{x,{n}{+}{K_2}{-}{1}}}}$@
while (x<n) {
  @$\devcomment{{x}{<{n}};{T}{-}{\frac{TK_1}{{pK_1}{+}{({1}{-}{p})K_2}}}{+}{\frac{2T}{pK_1+(1-p)K_2}}{\cdot}{\interval{x,{n}{+}{K_2}{-}{1}}}}$@
  x=x+unif(0,K1)
  @$\devcomment{{x}{\leq{n}{+}{K_2}{-}{1}};{T}{+}{\frac{2T}{pK_1+(1-p)K_2}}{\cdot}{\interval{x,{n}{+}{K_2}{-}{1}}}}$@
  @$\oplus_{p}$@
  @$\devcomment{{x}{<{n}};{T}{-}{\frac{TK_2}{pK_1+(1-p)K_2}}{+}{\frac{2T}{pK_1+(1-p)K_2}}{\cdot}{\interval{x,{n}{+}{K_2}{-}{1}}}}$@
  x=x+unif(0,K2);
  @$\devcomment{{x}{\leq{n}{+}{K_2}{-}{1}};{T}{+}{\frac{2T}{pK_1+(1-p)K_2}}{\cdot}{\interval{x,{n}{+}{K_2}{-}{1}}}}$@
  tick(T);
  @$\devcomment{{x}{\leq{n}{+}{K_2}{-}{1}};{0}{+}{\frac{2T}{pK_1+(1-p)K_2}}{\cdot}{\interval{x,{n}{+}{K_2}{-}{1}}}}$@
}
@$\devcomment{.;{0}{+}{\frac{2T}{pK_1+(1-p)K_2}}{\cdot}{\interval{x,{n}{+}{K_2}{-}{1}}}}$@
\end{lstlisting}
$\progname{prdwalk}$\\
${\frac{2T}{pK_1+(1-p)K_2}}{\cdot}{\interval{x,{n}{+}{K_2}{-}{1}}}$
\end{minipage}%
\begin{minipage}[b]{0.25\textwidth}
\centering
\begin{lstlisting}[basicstyle=\footnotesize,escapeinside={@}{@}]
@$\devcomment{.;{0}{+}{\frac{2}{3}}{\cdot}{\interval{h,{t}{+}{9}}}}$@
while (h<=t) {
  @$\devcomment{{h}{\leq}{t};{0}{+}{\frac{2}{3}}{\cdot}{\interval{h,{t}{+}{9}}}}$@
  t=t+1;
  @$\devcomment{{h}{<}{t};{{-}{\frac{2}{3}}{+}\frac{2}{3}}{\cdot}{\interval{h,{t}{+}{9}}}}$@
  @$\devcomment{{h}{<}{t};{{-}{\frac{7}{3}}{+}\frac{2}{3}}{\cdot}{\interval{h,{t}{+}{9}}}}$@
  h=h+unif(0,10)
  @$\devcomment{{h}{\leq}{t{+}{9}};{{1}{+}\frac{2}{3}}{\cdot}{\interval{h,{t}{+}{9}}}}$@
  @$\oplus_{\frac{1}{2}}$@
  @$\devcomment{{h}{<}{t};{{1}+\frac{2}{3}}{\cdot}{\interval{h,{t}{+}{9}}}}$@
  skip;
  @$\devcomment{{h}{\leq}{t{+}{9}};{{1}{+}\frac{2}{3}}{\cdot}{\interval{h,{t}{+}{9}}}}$@
  tick(1);
  @$\devcomment{{h}{\leq}{t};{{0}{+}\frac{2}{3}}{\cdot}{\interval{h,{t}{+}{9}}}}$@
}
@$\devcomment{.;{0}{+}{\frac{2}{3}}{\cdot}{\interval{h,{t}{+}{9}}}}$@
\end{lstlisting}
$\progname{race}$\\
${\frac{2}{3}}{\cdot}{\interval{h,{t}{+}{9}}}$
\end{minipage}%
\begin{minipage}[b]{0.30\textwidth}
\centering
\begin{lstlisting}[basicstyle=\footnotesize,escapeinside={@}{@}]
@$\devcomment{.;{0}{+}{\frac{5}{4}}{\cdot}\interval{0,x}{+}\interval{0,y}}$@
while (x>0) {
  @$\devcomment{{x}{>}{0};{0}{+}{\frac{5}{4}}{\cdot}\interval{0,x}{+}\interval{0,y}}$@
  x=x-1;
  @$\devcomment{{x}{\geq}{0};{\frac{5}{4}{+}\frac{5}{4}}{\cdot}\interval{0,x}{+}\interval{0,y}}$@
    @$\devcomment{{x}{\geq}{0};{{2}{+}\frac{5}{4}}{\cdot}\interval{0,x}{+}\interval{0,y}}$@
    y=y+1
    @$\devcomment{{x}{\geq}{0};{{1}{+}\frac{5}{4}}{\cdot}\interval{0,x}{+}\interval{0,y}}$@
  @$\oplus_{\frac{1}{4}}$@
    @$\devcomment{{x}{\geq}{0};{{1}{+}\frac{5}{4}}{\cdot}\interval{0,x}{+}\interval{0,y}}$@
    while (y>0) {
      @$\devcomment{{x}{\geq}{0}{\wedge}{y}{>}{0};{{1}{+}\frac{5}{4}}{\cdot}\interval{0,x}{+}\interval{0,y}}$@
      y=y-1;
      @$\devcomment{{x}{\geq}{0}{\wedge}{y}{\geq}{0};{{2}{+}\frac{5}{4}}{\cdot}\interval{0,x}{+}\interval{0,y}}$@
      tick(1);
      @$\devcomment{{x}{\geq}{0}{\wedge}{y}{\geq}{0};{{1}{+}\frac{5}{4}}{\cdot}\interval{0,x}{+}\interval{0,y}}$@
    }
  @$\devcomment{{x}{\geq}{0};{{1}{+}\frac{5}{4}}{\cdot}\interval{0,x}{+}\interval{0,y}}$@
  tick(1);
  @$\devcomment{{x}{\geq}{0};{{0}{+}\frac{5}{4}}{\cdot}\interval{0,x}{+}\interval{0,y}}$@
}
@$\devcomment{.;{{0}{+}\frac{5}{4}}{\cdot}\interval{0,x}{+}\interval{0,y}}$@
\end{lstlisting}
$\progname{C4B\_t13}$\\
$\frac{5}{4}{\cdot}{\interval{0,x}}{+}{\interval{0,y}}$
\end{minipage}%
\caption{More derivations of bounds on the number of ticks for loopy probabilistic programs using both probabilistic branching and random sampling assignment. The parameters $K_2 \geq K_1 > 0$, $T > 0$, and $p \geq 0$ denote concrete constants.}
\label{fig:moredevration01}
\end{figure*}

Example $\progname{race}$ introduced in Section~\ref{sec:programmodel} models a race between a hare
(variable $h$) and a tortoise (variable $t$). With the same reasoning for $\progname{prdwalk}$, we can infer the bound ${\frac{2}{3}}{\cdot}{\interval{h,{t}{+}{9}}}$ on the expected 
value of the elapsed time.

Example $\progname{C4B\_t13}$ adapted from~\cite{CarbonneauxHZ15} shows how expected potential method can be used to find linear expected bound for nested loop. The outer loop iterates $\interval{0,x}$ times. We increase $y$ by $1$ with probability $\frac{1}{4}$, otherwise the inner loop is executed until $y$ is $0$ with probability $\frac{3}{4}$. \toolname{} computes a tight bound by discovering that only one run of the inner loop depends on the initial value of $y$, that is ${\interval{0,y}}$, for the others runs, the expected value of total number of iterations is $\frac{1}{4}{\cdot}{\interval{0,x}}$. 

In Figure~\ref{fig:moredevration02}, example $\progname{miner}$ simulates a miner who is trapped in a mine. The miner is sent to the mine for $n$ times independently, for each time, with probability $\frac{1}{2}$ he is trapped and with the same probability he is safe. When he is trapped, there are $3$ doors in the mine for using. The first door leads to a tunnel that will take him to safety after $3$ hours (representing by $3$ \word{ticks}). The second door leads to a tunnel that returns him to the mine after $5$ hours. And the third door leads to a tunnel that returns him to the mine after $7$ hours. At all times, the miner is equally likely to choose any one of the doors, meaning that he chooses any door with equivalent probability $\frac{1}{3}$. 

Let's consider the manual way to compute expected time the miner reaches safety for $n$ independent times. 
Let $X_i$ be the random variable representing the miner escape time in the $i^{th}$ time and let $X = X_1 + \cdots + X_n$, we are asking for $\expt{}{X}$. We first need to use the \emph{law of total expectation} to calculate $\expt{}{X_i}$. 
Let $A$ be the event that the miner is trapped then $\expt{}{X_i} = \expt{}{X_i|A}{\cdot}\prob{A} + 0{\cdot}\prob{\neg A} = \frac{1}{2}{\cdot}\expt{}{X_i|A}$. Now use again the law of total expectation to calculate $\expt{}{X_i|A}$, 
let $B$ be the event that the miner reaches safety in the first try, then we have $\expt{}{X_i|A} = \expt{}{(X_i|A)|B}{\cdot}\prob{B} + \expt{}{(X_i|A)| \neg B}{\cdot}\prob{\neg B} = 3{\cdot}\frac{1}{3} + (5{\cdot}\frac{1}{2} + 7{\cdot}\frac{1}{2} + \expt{}{X_i|A}){\cdot}\frac{2}{3} = 5 + \frac{2}{3}{\cdot}\expt{}{X_i|A}$, thus $\expt{}{X_i|A} = 15$. Finally, by linearity of expectation, we get $\expt{}{X} = \sum^{n}_{i=1}\expt{}{X_i} = \frac{15}{2}{\cdot}n$ if $n > 0$, $0$ otherwise. 

As we have shown in the above reasoning process, the manual method involves many heuristic theorems and techniques in probability theory and algebra. As a result, it is very hard or impossible to automate. This example demonstrates how \toolname{} handles Boolean values (the variable $\word{flag}$) to get the tightest expected bound $\frac{15}{2}{\cdot}{\interval{0,n}}$. It assigns the potential ${15}{\cdot}{\interval{0,\word{flag}}}$ to the variable $\word{flag}$ for paying the cost of the inner loop, in which $\frac{15}{2}{\cdot}{\interval{0,n}}{-}{\frac{15}{2}}{+}{15}{\cdot}{\interval{0,\word{flag}}}$ is the loop invariant. 
\begin{figure}[th!]
\centering
\begin{minipage}[b]{0.29\textwidth}
\centering
\begin{lstlisting}[basicstyle=\footnotesize,escapeinside={@}{@}]
@$\devcomment{{q}{:=}{\frac{15}{2}}{\cdot}{\interval{0,n}}{-}{\frac{15}{2}}}$@
@$\devcomment{{.};{\frac{15}{2}}{\cdot}{\interval{0,n}}}$@
while (n>0) {
  @$\devcomment{{{n}{>}{0}};{\frac{15}{2}}{+}{q}}$@
  flag = 1;
    @$\devcomment{{{\word{flag}}{=}{1}};{15}{\cdot}{\interval{0,\word{flag}}}{+}{q}}$@
    while (flag > 0) {
      @$\devcomment{{{\word{flag}}{=}{1}};{3}{+}{q}}$@
        flag=0;
        @$\devcomment{{{\word{flag}}{=}{0}};{3}{+}{15}{\cdot}{\interval{0,\word{flag}}}{+}{q}}$@ 
        tick(3)
        @$\devcomment{{{\word{flag}}{=}{0}};{15}{\cdot}{\interval{0,\word{flag}}}{+}{q}}$@ 
      @$\oplus_{\frac{1}{3}}$@
      @$\devcomment{{{\word{flag}}{=}{1}};{\frac{-3}{2}}{+}{\frac{45}{2}}{+}{q}}$@
          @$\devcomment{{{\word{flag}}{=}{1}};{\frac{-5}{2}}{+}{\frac{45}{2}}{+}{q}}$@
          flag=1;
          @$\devcomment{{{\word{flag}}{=}{1}};{5}{+}{15}{\cdot}{\interval{0,\word{flag}}}{+}{q}}$@ 
          tick(5)
           @$\devcomment{{{\word{flag}}{=}{1}};{15}{\cdot}{\interval{0,\word{flag}}}{+}{q}}$@ 
        @$\oplus_{\frac{1}{2}}$@
          @$\devcomment{{{\word{flag}}{=}{1}};{\frac{-1}{2}}{+}{\frac{45}{2}}{+}{q}}$@
          flag=1;
          @$\devcomment{{{\word{flag}}{=}{1}};{7}{+}{15}{\cdot}{\interval{0,\word{flag}}}{+}{q}}$@  
          tick(7);
          @$\devcomment{{{\word{flag}}{=}{1}};{15}{\cdot}{\interval{0,\word{flag}}}{+}{q}}$@  
    }
    @$\devcomment{{{\word{flag}}{\leq}{0}};{15}{\cdot}{\interval{0,\word{flag}}}{+}{q}}$@
    @$\devcomment{{n}{>}{0};{q}}$@
  @$\oplus_{\frac{1}{2}}$@
    @$\devcomment{{n}{>}{0};{\frac{15}{2}}{\cdot}{\interval{0,n}}{-}{\frac{15}{2}}}$@
    skip; 
  @$\devcomment{{n}{>}{0};{q}}$@
  n=n-1;
  @$\devcomment{{n}{\geq}{0};{\frac{15}{2}}{\cdot}{\interval{0,n}}}$@
}
\end{lstlisting}
$\progname{miner}$\\
${\frac{15}{2}}{\cdot}{\interval{0,n}}$
\end{minipage}%
\begin{minipage}[b]{0.20\textwidth}
\centering
\begin{lstlisting}[basicstyle=\footnotesize,escapeinside={@}{@}]
@$\devcomment{{\word{inv}}{:=}{6}{\cdot}\binom{\interval{0,n}}{2} {+} {3}{\cdot}{\interval{{0},{m}}}}$@
@$\devcomment{{.};{3}{\cdot}{\interval{{0},{n}}}^{2}}$@
while (n>0) {
  @$\devcomment{{n}{>}{0};{3}{\cdot}{\interval{{0},{n}}}^{2}}$@
  @$\devcomment{{n}{>}{0};{6}{\cdot}\binom{\interval{0,n}}{2}}$@
  n=n-unif(0,1);
  @$\devcomment{{n}{\geq}{0};{6}{\cdot}\binom{\interval{0,n}}{2} {+} {3}{\cdot}{\interval{{0},{n}}}}$@
  m=n
  @$\devcomment{{n}{=}{m};{\word{inv}}}$@
  while (m>0) {
    @$\devcomment{{m}{=}{n}{\land}{m}{>}{0};{-2}{+}{\word{inv}}}$@
    m=m-1
    @$\devcomment{{m}{=}{n}{\land}{m}{\geq}{0};{1}{+}{\word{inv}}}$@
    @$\oplus_{\frac{1}{3}}$@
    @$\devcomment{{m}{=}{n}{\land}{m}{>}{0};{1}{+}{\word{inv}}}$@
    skip;
    @$\devcomment{{m}{=}{n}{\land}{m}{\geq}{0};{1}{+}{\word{inv}}}$@
    tick(1);
    @$\devcomment{{m}{=}{n}{\land}{m}{\geq}{0};{0}{+}{\word{inv}}}$@
  }
  @$\devcomment{{m}{\leq}{0};{\word{inv}}}$@
}
\end{lstlisting}
$\progname{rdbub}$\\
${3}{\cdot}{\interval{{0},{n}}}^{2}$
\end{minipage}%

\caption{More representative examples where probabilistic branching is used with compound statements and polynomial bound.}
\label{fig:moredevration02}
\end{figure}

Example $\progname{rdbub}$ models a probabilistic bubble sort algorithm. For each iteration of the inner loop, a pair of adjacent elements in the input array can be swapped. However, the algorithm performs the swapping only with probability $\frac{1}{3}$ and skips the swap with probability $\frac{2}{3}$.
The derivation proves the polynomial bound ${3}{\cdot}{\interval{{0},{n}}}^{2}$
on the expected number of iterations. This example shows the linear
potential ${3}{\cdot}{\interval{{0},{m}}}$ is spilled from the a quadratic potential
${6}{\cdot}\binom{\interval{0,n}}{2}$ after the sampling in the outer loop. This linear potential is used to
pay for the ticks
in the inner loop. The derivation of the bound for the inner loop is
similar to the bound derivations for the previous program.
However, the annotation carries an extra quadratic part which is
invariant in the inner loop and just passed along in the derivation. Note that the derivation of the potential 
${3}{\cdot}{\interval{{0},{n}}}^{2}$ from the potential ${6}{\cdot}\binom{\interval{0,n}}{2}$ reflects the use of the rule \rul{Q:Weaken} in the derivation system. 

\end{document}